%% file: ms.tex
\documentclass[final]{siamart190516}
\usepackage[margin=1in]{geometry} 

\usepackage{calc}
\usepackage{subfig}
\usepackage{float}
\usepackage{amsmath,amssymb,amsfonts}
\usepackage{comment}
\usepackage{xcolor}
\usepackage{graphicx}
\usepackage[stable]{footmisc}
\usepackage{svg}
\usepackage[utf8]{inputenc}
\usepackage{amssymb}
\usepackage{mwe}
\usepackage{hyperref, cleveref}

\crefformat{chapter}{\S#2#1#3}
\crefmultiformat{chapter}{\S\S#2#1#3}{and~#2#1#3}{, #2#1#3}{, and~#2#1#3}

\crefformat{section}{\S#2#1#3}
\crefmultiformat{section}{\S\S#2#1#3}{and~#2#1#3}{, #2#1#3}{, and~#2#1#3}

\crefformat{subsection}{\S#2#1#3}
\crefmultiformat{subsection}{\S\S#2#1#3}{and~#2#1#3}{, #2#1#3}{, and~#2#1#3}

\crefformat{subsubsection}{\S#2#1#3}
\crefmultiformat{subsubsection}{\S\S#2#1#3}{and~#2#1#3}{, #2#1#3}{, and~#2#1#3}

\usepackage{xcolor}
\hypersetup{
    colorlinks,
    linkcolor={red!50!black},
    citecolor={blue!50!black},
    urlcolor={blue!80!black}
}
\usepackage{graphicx}
\usepackage{fullpage}
\usepackage{enumerate}
\usepackage{graphicx}

\usepackage[sort]{cite}

\usepackage{amssymb, amsmath, amscd, mathrsfs,units}
\usepackage{mathtools}
\DeclarePairedDelimiter{\ceil}{\lceil}{\rceil}
\DeclarePairedDelimiter{\floor}{\lfloor}{\rfloor}
\usepackage{ stmaryrd }
\usepackage{verbatim}
\usepackage[english]{babel}

\usepackage[noend]{algcompatible}

\makeatletter
\def\myfnt{\ifx\protect\@typeset@protect\expandafter\footnote\else\expandafter\@gobble\fi}
\makeatother

\makeatletter
\newcommand{\myref}[1]{\cref{#1}\mynameref{#1}{\csname r@#1\endcsname}}
\newcommand{\Myref}[1]{\Cref{#1}\mynameref{#1}{\csname r@#1\endcsname}}

\def\mynameref#1#2{%
  \begingroup
    \edef\@mytxt{#2}%
    \edef\@mytst{\expandafter\@thirdoffive\@mytxt}%
    \ifx\@mytst\empty\else
    \space(\nameref{#1})\fi
  \endgroup
}
\makeatother

\makeatletter
\def\BState{\State\hskip-\ALG@thistlm}
\makeatother
\usepackage{enumitem}
\setlist{leftmargin=*,noitemsep}

\newcommand{\pluseq}{\mathrel{+}=}
\newcommand{\parallelsum}{\mathbin{\|}}
\newcommand{\cut}[0]{\mathrm{\mathrm{cut}}}
\newcommand{\Flats}[0]{\mathrm{\mathrm{Flats}}}
\newcommand{\comp}[0]{\mathrm{\mathrm{comp}}}
\newcommand{\supp}[0]{\mathrm{\mathrm{supp}}}
\newcommand{\Cuts}[0]{\mathrm{Cuts}}
\newcommand{\im}[0]{\mathrm{Im}}
\newcommand{\In}[0]{\mathrm{In}}
\newcommand{\inc}[0]{\mathrm{inc}}
\newcommand{\connE}[0]{\mathrm{connE}}
\newcommand{\ep}[0]{\mathrm{ep}}
\newcommand{\Bridge}[0]{\mathrm{Bridge}}
\newcommand{\Mod}[0]{\mathrm{Mod}}
\newcommand{\MST}[0]{\mathrm{MST}}
\newcommand{\UST}[0]{\mathrm{UST}}
\newcommand{\Tree}[0]{\mathrm{Tree}}
\newcommand{\Uniform}[0]{\mathrm{Uniform}}
\newcommand{\NP}[0]{\mathrm{NP}}
\newcommand{\RP}[0]{\mathrm{RP}}
\newcommand{\BPP}[0]{\mathrm{BPP}}
\newcommand{\Poly}[0]{\mathrm{P}}
\newcommand{\FPT}[0]{\mathrm{FPT}}
\newcommand{\EMS}[0]{\mathrm{EMS}}
\newcommand{\MSO}[0]{\mathrm{MSO}}
\newcommand{\XP}[0]{\mathrm{XP}}
\renewcommand{\arraystretch}{1.5}
\usepackage[at]{easylist}
\usepackage{environ}

\NewEnviron{computationalproblem}[1]{%
\begin{center}\fbox{\parbox{3in}{%
    {\centering\scshape #1\par}%
    \parskip=1ex
    \everypar{\hangindent=1em}%
    \BODY
}}\end{center}}

\theoremstyle{definition}
\newtheorem{thm}{Theorem}[section]
\newtheorem{cor}[thm]{Corollary}
\newtheorem{lem}[thm]{Lemma}

\newtheorem{prop}[thm]{Proposition}

\newtheorem{example}[thm]{Example}

\newtheorem{defn}[thm]{Definition}

\theoremstyle{remark}

\newtheorem{remark}[thm]{Remark}
\numberwithin{equation}{section}

\def\tcom#1\par{}

\newcommand\blfootnote[1]{%
  \begingroup
  \renewcommand\thefootnote{}\footnote{#1}%
  \addtocounter{footnote}{-1}%
  \endgroup
}

\graphicspath{{Images/}}

\title{Complexity and Geometry of Sampling Connected Graph Partitions}

\author{Lorenzo Najt$^*$, Daryl DeFord$^{\dagger}$, Justin Solomon$^{\dagger}$}
\begin{document}
\raggedbottom
\maketitle
\vspace{.25cm}
\begin{abstract}In this paper, we prove intractability results about sampling from the set of partitions of a planar graph into connected components. Our proofs are motivated by a technique introduced by Jerrum, Valiant, and Vazirani. Moreover, we use gadgets inspired by their technique to provide families of graphs where the ``flip walk'' Markov chain used in practice for this sampling task exhibits exponentially slow mixing. Supporting our theoretical results we present some empirical evidence demonstrating the slow mixing of the flip walk on grid graphs and on real data. Inspired by connections to the statistical physics of self-avoiding walks, we investigate the sensitivity of certain popular sampling algorithms to the graph topology. Finally, we discuss a few cases where the sampling problem is tractable. Applications to political redistricting have recently brought increased attention to this problem, and we articulate open questions about this application that are highlighted by our results.\blfootnote{$^*$University of Wisconsin-Madison. Corresponding Author: \href{mailto:LNAJT@math.wisc.edu}{Lnajt@math.wisc.edu}
}\blfootnote{$^{\dagger}$Massachusetts Institute of Technology.}%
\end{abstract}

\input{Sections/1IntroductionMotivation/1Introduction.tex}

\input{Sections/1IntroductionMotivation/2NotationDefns.tex}

\input{Sections/1IntroductionMotivation/3Motivation.tex}

\input{Sections/2Complexity/1Header.tex}

\input{Sections/2Complexity/2Preliminaries.tex}

\input{Sections/2Complexity/3SamplingSimpleCycles.tex}

\input{Sections/2Complexity/4BasicIntractibility.tex}

\input{Sections/2Complexity/5Balanced.tex}

\input{Sections/2Complexity/6Triangulations.tex}

\input{Sections/2Complexity/7kPartitions.tex}

\input{Sections/3Bottlenecks/1Header.tex}

\input{Sections/3Bottlenecks/2TheFlipWalk.tex}

\input{Sections/3Bottlenecks/3MixingTimeBackground.tex}

\input{Sections/3Bottlenecks/4FirstExample.tex}

\input{Sections/3Bottlenecks/5MaximalPlanarExample.tex}

\input{Sections/4Empirical/1Header.tex}

\input{Sections/4Empirical/2Grid.tex}

\input{Sections/5PositiveResults/ShortVersion.tex}

\input{Sections/6ConclusionsAcknowledgements/1Conclusions.tex}

\input{Sections/6ConclusionsAcknowledgements/2Acknowledgements.tex}

\bibliography{ref}
\appendix

\input{Sections/2Complexity/Appendix/TediousLemmaVerificationPicture.tex}

\input{Sections/2Complexity/Appendix/6bProvingThatRdPreserves3CCP.tex}

\input{Sections/2Complexity/Appendix/7bProofOfDualityTheorem.tex}

\input{Sections/4Empirical/Appendix/EmpiricalAppendix.tex}
\input{Sections/5PositiveResults/Appendix/SamplingUsingMarginals.tex}

\input{Sections/5PositiveResults/Appendix/SeriesParallelGraphsBackground.tex}

\input{Sections/5PositiveResults/Appendix/SimpleCycleDP.tex}

\input{Sections/5PositiveResults/Appendix/DynamicProgramForBalancedPartitions.tex}

\input{Sections/5PositiveResults/Appendix/9NonSelfReducible.tex}

\end{document}

%% file: Sections/1IntroductionMotivation/1Introduction.tex
\vspace{.25cm}
\section{Introduction}

The problem of \emph{graph partitioning}, or dividing the vertices of a graph into a small number of connected subgraphs that extremize an objective function, is a classical task in graph theory with application to network analytics, machine learning, computer vision, and other areas.  Whereas this task is well-studied in computation and mathematics, a related problem remains relatively understudied:  understanding how a given partition compares to other members of the set of possible partitions.  In this case, the goal is not to generate a partition with favorable properties, but rather to compare a given partition to some set of alternatives. Recent analysis of political redistricting have invoked such comparisons (see \Cref{section:CongressionalMotivation}), motivating the investigation of this general problem.%

Consider a connected graph $G = (V,E)$, and let $P_k(G)$ denote the collection of $k$-partitions of $V$ such that each block induces a connected subgraph. One approach to understanding how a given element of $P_2(G)$ compares to the other elements proceeds by uniformly sampling from $P_2(G)$, then gathering statistics about this sample and comparing them to the partition under consideration.%
While this is an attractive approach, this uniform sampling problem is computationally intractable, assuming $\NP \not = \RP$.

We open the paper by reviewing this fact. Our first new result is that this intractability persists even if we consider partitions of equal size. Then, motivated to produce a result that is more relevant to the classes of graphs that arise in redistricting, we show that uniformly sampling $P_2(G)$ remains intractable even if $G$ is a maximal plane graph with a constant bound on the vertex degree. Beyond sampling from the uniform distribution, we also prove results about the intractability of sampling from a broader class of distributions over connected $k$-partitions.  

Such worst case results should not be considered proof that uniformly sampling from $P_k(G)$ is \emph{always} impossible. However, it does indicate that algorithm designers should examine sampling heuristics with some skepticism. Driven by this philosophy, we follow up our investigation of the worst-case complexity with an investigation into applicability of a general and often extremely useful sampling tool, which is Markov chain Monte Carlo.

In the context of redistricting, Markov chains have rapidly become a popular tool for sampling from $P_k(G)$ to compare a districting plan (\cref{section:CongressionalMotivation}) to the space of possible plans \cite{chikina_assessing_2017,pegden1,mattingly,mattingly1,chen1}. The most commonly used Markov chain moves randomly on $P_2(G)$ by proposing to change the block assignment of a uniformly chosen node, and accepting such moves only if the connectivity of each block is preserved \cite{mattingly, chikina_assessing_2017}. %
We call this chain the flip walk (\cref{defn:flipwalk}). If $G$ is $2$-connected, then the flip walk on $P_2(G)$ is irreducible, and the stationary distribution is uniform \cite{akitaya2019reconfiguration}. In principle, running the flip walk on $P_2(G)$ for a long time will produce a uniformly random element of $P_2(G)$.  For this approach to be computationally feasible, however, one must guarantee that the mixing time of the Markov chain on $P_2(G)$ is not too large compared to $|G|$.%

Pursuing this angle, we explain how to engineer a family of graphs $G \in \mathcal{G}$ so that the mixing time of the flip walk on $P_2(G)$ grows exponentially quickly in $|G|$. Based on this, as well as some empirical work motivated by the bottlenecks we discover, we can conclude that there are strong reasons to doubt that Markov chain methods based on the flip walk mix in polynomial time. %

In addition to this, we make a connection with the literature on self-avoiding walks \cite{duminil2014supercritical} that demonstrates the existence of dramatic phase transitions in the qualitative behavior of distributions on $P_2(G)$. We provide experiments illustrating the relevance of these phase transitions to redistricting and, inspired by the ideas in those experiments, we examine the robustness of other popular approaches to sampling from $P_2(G)$, including some methods based on spanning trees \cite{recomb, deford_total_2018}. Overall, the observations we make highlight interesting and difficult %
challenges for the sampling algorithms and inference principles being used in statistical analysis of redistricting plans. %

Finally, we discuss a few classes of graphs on which it is possible to sample uniformly from $P_2(G)$ in polynomial time, but which are far from the kinds of graphs relevant to redistricting. The large gap between where we know that uniform sampling is intractable and where we know it is tractable, along with some connections to outstanding problems from statistical physics that seem to be on par with the intended redistricting application, indicates that there are many challenging questions remaining about sampling from $P_k(G)$.

\paragraph*{Overview and contributions}  As part of a broader effort to establish mathematical underpinnings for the analytical tools used in redistricting \cite{chikina_assessing_2017, mattingly1, akitaya2019reconfiguration, altman1997automation, kueng2019fair}%
, we identify challenges and opportunities for further improvement related to random sampling in the space of graph partitions. In addition to the technical material listed below, we articulate some implicit assumptions behind outlier methods used in the analysis of gerrymandering (\Cref{section:EOH}), and offer some suggestions for future work around computational redistricting. %

\begin{itemize}
\item Sampling intractability results, and bottlenecks:
\begin{itemize}

    \item We review why it is intractable to sample uniformly from $P_2(G)$ (\Cref{section:basichard}). As is typical, our strategy will be to engineer graphs so that uniform samples from their connected $2$-partitions are likely to solve an NP-hard problem. We will work with planar graphs to leverage bond-cycle duality, which gives a bijection between $P_2(G)$ and the set of simple cycles of the dual.%

\item One realistic condition to put on samplers from $P_2(G)$ is to restrict to the set of \emph{balanced} partitions, partitions for which both blocks have the same number of nodes. We next show that uniformly sampling balanced $2$-partitions remains NP-hard (\Cref{section:balancedhard}).%

\item We also prove the intractability of uniformly sampling from $P_2(G)$ for an even more constrained family of graphs: planar triangulations with bounded vertex degree (\Cref{section:maxhard}).%

\item We prove that uniformly sampling $k$-partitions is intractable, using %
a generalization of bond-cycle duality (\cref{appendix:duality}).

\item The gadgets used in the intractability proofs provide a means for constructing families of graphs such that the flip walk Markov chain on $P_2(G)$ has exponentially large mixing time (\Cref{Section:FlipChain}). We describe a family of plane triangulations with vertex degree bounded by $9$ such that the flip walk on $P_2(G)$ mixes torpidly. %

\end{itemize}
\item Empirical results:
\begin{itemize}
\item We include some empirical evidence indicating that Markov chains based on the flip walk mix slowly on grid graphs (\Cref{subsection:gridgraphempirical}) and on the graphs used in analysis of redistricting (\Cref{subsection:empiricalrealstates}).
\item We mention a link between our sampling problem on grid graphs and long standing challenges regarding the self-avoiding walk model from statistical physics (\Cref{Section:GridGraph}). We use this connection to motivate and demonstrate phase transitions in the qualitative properties of $P_2(G)$.

\item %
We provide experiments (\Cref{section:thechoiceofmodel} and \Cref{section:phasetransitions}) demonstrating that popular methods for sampling from geographic partitions via discretizing the geography as a graph $G$ and sampling from $P_2(G)$ are impacted in surprising ways by the discretization used.
\end{itemize}
\item Positive results:
\begin{itemize}
    \item We prove that there are efficient and implementable dynamic programming algorithms that can be used to sample uniformly from $P_2(G)$ and to sample uniformly from the balanced partitions in $P_2(G)$, provided that $G$ is a series-parallel graph. This algorithm succeeds in some cases where the flip walk is unreliable. We observe that these sampling problems are tractable on graphs of bounded treewidth.%
    
\end{itemize}

\end{itemize}

\subsection{Related Work}

We are not the first team of researchers to have considered redistricting problems from a complexity point of view. Indeed, there are many papers showing that optimization problems related to designing the most ``fair'' or ``unfair'' districts are $\NP$-hard, for various meanings of the word fair; works in this category include \cite{kueng2019fair,  altman1997automation}. Other researchers have explored the complexity of findings paths through the flip walk state space \cite{akitaya2019reconfiguration}. We will discuss other related work in the body of the paper, such as the connection to self-avoiding walks (\cref{Section:GridGraph}), and related sampling problems (\cref{section:markovchainrelatedwork}).%

%% file: Sections/1IntroductionMotivation/2NotationDefns.tex
\subsection{Basic Notation}

Let $G=(V,E)$ be a graph; unless otherwise specified, all of our graphs will be undirected, finite and simple. If unspecified, usually $n := |V|$. Given a graph $G$, $V(G)$ denotes the set of nodes, and $E(G)$ the set of edges. An (ordered) \emph{$k$-partition} $P=(V_1,V_2,\ldots, V_k)$ of $G$ is an list of disjoint subsets $V_i\subseteq V$ whose union is $V$, while an unordered $k$-partition is a set $\{V_1, \ldots, V_k\}$ satisfying the same conditions. Throughout this paper we will be concerned with {\em connected $k$-partitions}, i.e., those $k$-partitions where each $V_i$ induces a connected subgraph. The set of ordered connected $k$-partitions of $G$ is denoted $P_k(G)$, and the set of unordered connected $k$-partitions of $G$ is denoted $\mathscr{P}_k(G)$.  If $A \subset V$, then we will use $\partial_E A$ to denote the \emph{edge boundary} of $A$: $\partial_E A = \{ \{u,v\} \in E : u \in A, v \not \in A \}$. %

%% file: Sections/1IntroductionMotivation/3Motivation.tex
\subsection{Motivation from Redistricting}\label{section:CongressionalMotivation}

    \begin{figure}
        \centering
        \begin{tabular}{cc}
        \includegraphics[scale = .15]{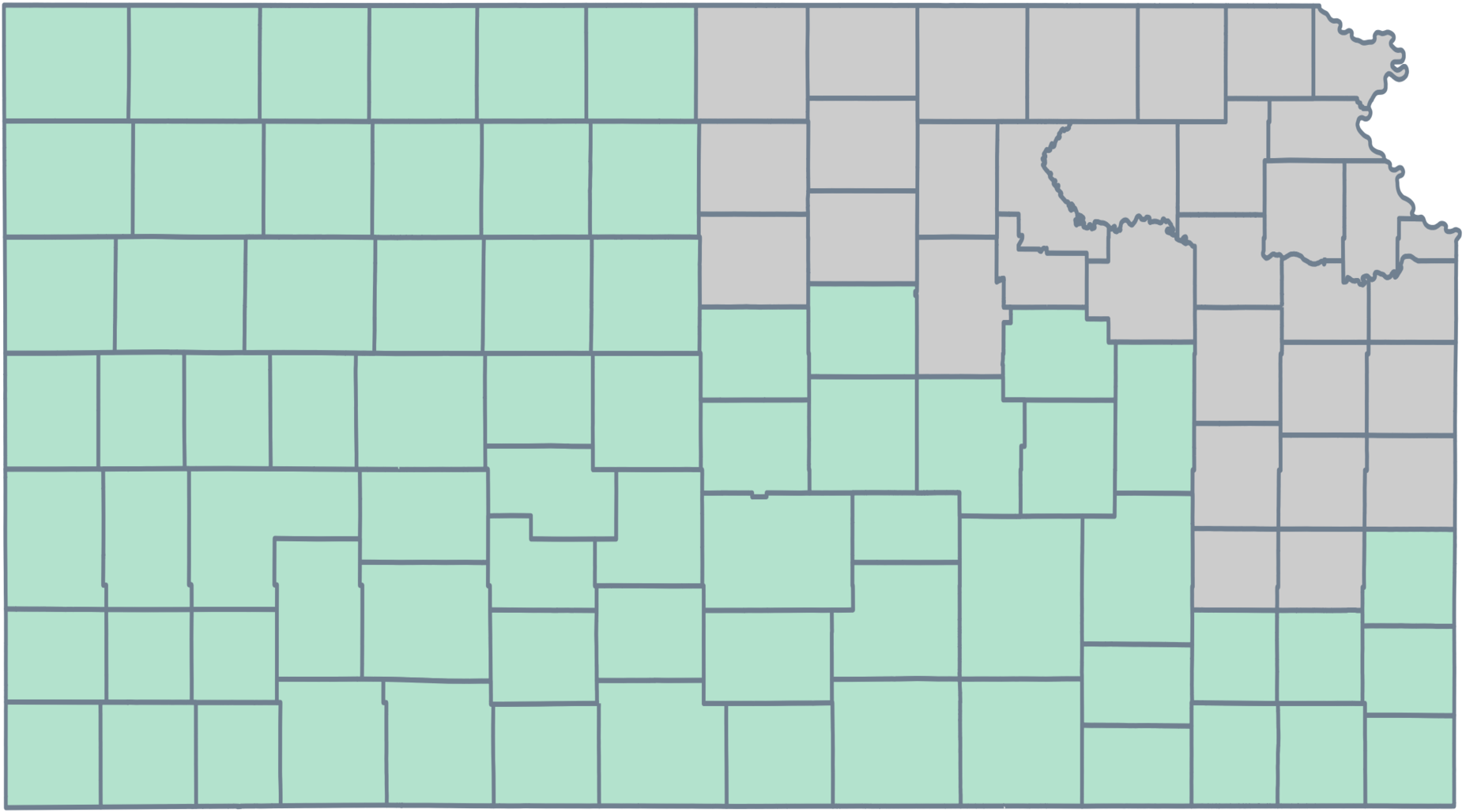} &
        \includegraphics[scale = .15]{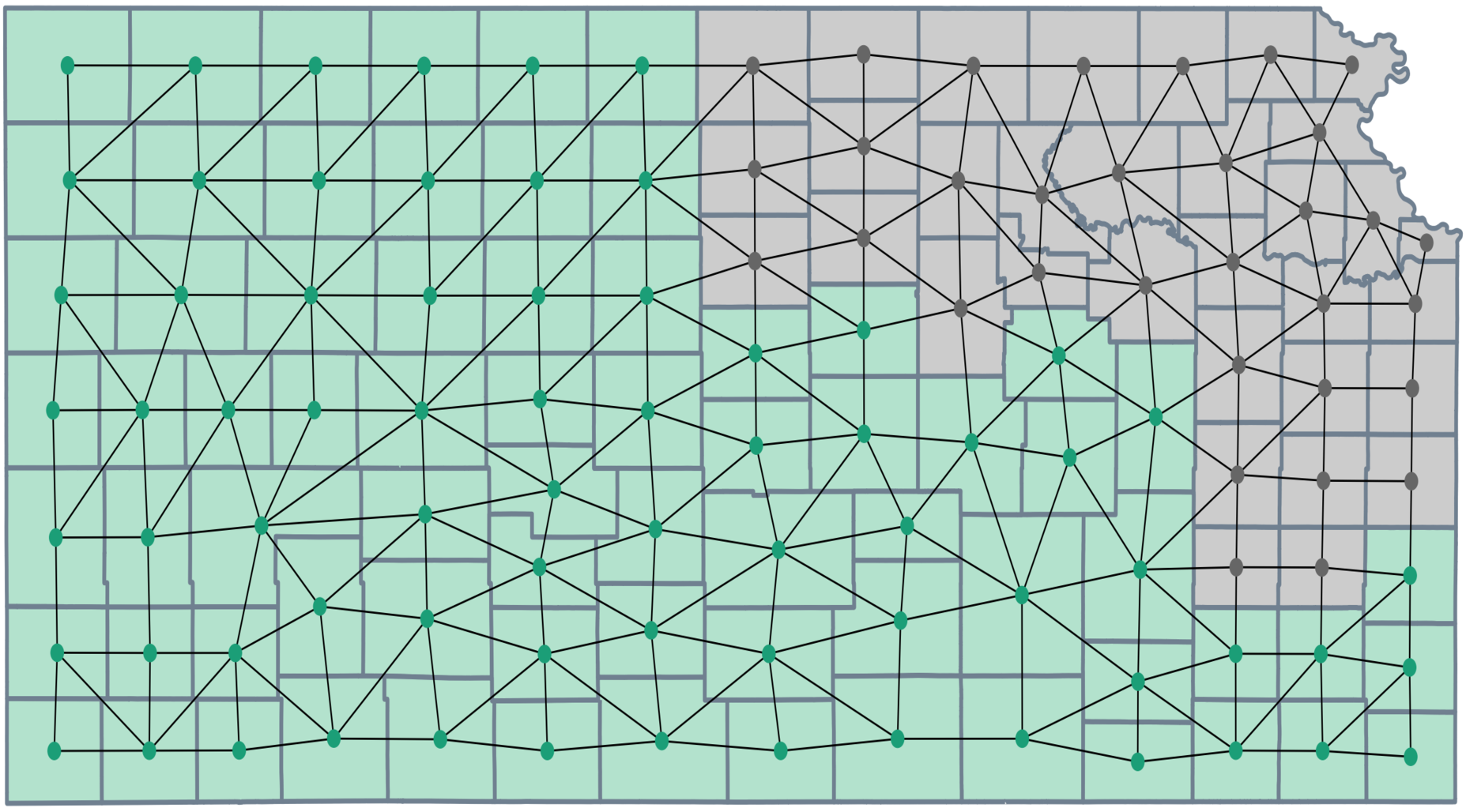}
        \\
        a) & b)
        \end{tabular}
        \caption{a) Kansas with county units \cite{shapefiles}, along with a connected $2$-partition. b) The corresponding state dual graph overlayed. Much more granular subdivisions of a state are often used. %
        }\label{fig:kansas}
    \end{figure}

In the United States, states are divided into small geographical units, such as in \Cref{fig:kansas}a);
these units are combined into voting districts, each of which elects a single representative. An assignment of these units to a district is called a districting plan. The units can be represented by the nodes of a graph, where units that share common boundaries are adjacent, as in \Cref{fig:kansas}. This graph is called the state dual graph%
. Assuming that voting districts must be contiguous, as is usually the case, a districting plan with $k$-districts is modeled by a connected $k$-partition of the state dual graph.

It was quickly observed \cite{martis2008original, hunter2011first} that by a clever choice of districting plan one could engineer aspects of electoral outcomes, a practice known as ``gerrymandering.'' In an effort to counteract this, there have been many proposals to design districting plans algorithmically, a process which often involves grappling with computationally intractable problems \cite{kueng2019fair,  altman1997automation}. The reality of redistricting, however, is that the power to draw the graph partition is in the hands of a legislature, dedicated committee, or hired expert---rather than a piece of software. For this reason, rather than using an algorithm to draw plans in the first place, some have suggested to analyze already drawn plans for compliance with civil rights law or desirability relative to alternatives.  

Arguments for or against districting plans are facilitated by understanding a plan in the context of what is possible.  For instance, an argument that a plan was drawn with the intent to discriminate might calculate that the proposed plan has more discriminatory properties than the vast majority of plans from a randomly generated collection of comparable plans; more specifically, the claim is that a particular map is an outlier compared to the other possibilities \cite{mathematiciansbrief} (see also \cref{section:EOH}).
This contextual approach requires sampling a \emph{diverse} \emph{ensemble} of plans that are compliant with the principles laid out by the governing body. %
A variety of algorithms have been proposed to sample ensembles of graph partitions for this purpose, from genetic algorithms \cite{cho_toward_2016, liu_pear:_2016} to random walks \cite{chikina_assessing_2017, herschlag_quantifying_2018}; recent expert reports in redistricting cases have used these tools to generate quantitative assessments of proposed plans \cite{chen1,herschlag_quantifying_2018,herschlag_evaluating_2017, pegden1}.

While random walk methods like \cite{mattingly1, herschlag_quantifying_2018} are guaranteed to sample from an explicitly designed distribution if run for long enough, practical computational constraints make it impossible to reach that point if there are no guarantees on the mixing time. %
On the other hand, algorithms like \cite{chen_unintentional_2013, magleby_new_2018, recomb}, which are
intended to generate a diverse set of partitions, sample from unknown distributions whose properties are hard to characterize.%

Thus, two critical open problems arise when relying on measurements derived from random ensembles of districting plans.  First, it is difficult to verify whether ensemble generation algorithms produce a statistically-representative sample from a targeted distribution. We study this problem by asking whether certain distributions over partitions are efficiently sampleable (\Cref{sec:initialintractability} and \Cref{Section:PositiveResults}), and whether certain sampling algorithms run efficiently (\Cref{Section:FlipChain} and \Cref{Section:Empirical}). Second, the qualitative properties of distributions over partitions are challenging to characterize, and it is difficult to determine the extent to which an outlier classification is affected by modelling decisions, including the choice of sampling algorithm or discretization. We study this problem in \Cref{section:phasetransitions} and \Cref{section:thechoiceofmodel}.

%% file: Sections/2Complexity/1Header.tex
\section{Sampling Intractability}\label{sec:initialintractability}

In this section, we present our results about the intractability of various general sampling problems associated with connected $k$-partitions.  The key idea \cite{JVV} behind proving that some uniform sampling problem is intractable is to show that one can modify an algorithm that solves it into an algorithm that samples from the solutions to some hard problem. We begin by setting up some language (\Cref{section:preliminariesonintract}), some of which is standard, and some of which we have created to organize our results. Then, we review the intractability of uniformly sampling from $P_2(G)$ (\Cref{sec:initialintractability}). Next, we show that uniformly sampling from balanced connected $2$-partitions is intractable (\Cref{section:balancedhard}). Then, we show that uniformly sampling from $P_2(G)$ remains intractable under certain constraints on the topology of $G$ \Cref{section:maxhard}. Finally, we will show that, for any fixed $k$, certain generalizations of the uniform distribution on $P_k(G)$ are intractable to sample from (\Cref{section:kpartitionshard}).%

%% file: Sections/2Complexity/2Preliminaries.tex
\subsection{Preliminaries on Intractability of Sampling}\label{section:preliminariesonintract}

In this section, we discuss some background on sampling problems, the class $\RP$, why $\RP \not = \NP$ is a reasonable assumption, and what it means for a sampling problem to be intractable. We also prove lemmas that will be used throughout.

The formalism for sampling problems, which goes back to at least \cite{JVV}, begins with a finite alphabet $\Sigma$ and a binary relation between words in this alphabet $R \subseteq \Sigma^* \times \Sigma^*$. We interpret $(x,y) \in R$ as asserting that $y$ is a solution to the instance $x$. For example, we can define a binary relation $R$ as those $(x,y)$ such that $x$ encodes a graph $G(x)$ and $y$ encodes the edges of a simple cycle of $G(x)$. We will consider only those relations that can be verified efficiently, which are called $p$-relations:

\begin{defn}[$p$-relations, \cite{JVV}] A relation $R \subseteq \Sigma^* \times \Sigma^*$ is a $p$-relation if there is a deterministic polynomial time Turing machine that recognizes $R \subseteq \Sigma^* \times \Sigma^*$ and if there is a polynomial $p$ such that $\forall x$, $(x,y) \in R$ implies that $|y| \leq p( |x|)$. We define $R(x) = \{y \in \Sigma^* : (x,y) \in R\}$.%
\end{defn}

Now we define the sampling problems we will be considering:

\begin{defn}[Family of $p$-distributions]
A family of $p$-distributions is defined by a $p$-relation $R$ and %
function $f : \Sigma^* \to \mathbb{Q}_{\geq 0}$. %
For each instance $x \in \Sigma^*$ with $R(x) \not = \emptyset$, we require that $f$ is not identically zero on $R(x)$. For such an instance $x$, we associate a probability distribution $p_x$ on $R(x)$, where $y \in R(x)$ has weight proportional to $f(y)$. The uniform distribution on $R$ is defined by taking $f$ to be identically $1$.
\end{defn}

\begin{defn}[Sampling problem]
To each family of $p$-distributions $(R,p_x)$, there is an associated sampling problem, which we also refer to as $(R,p_x)$:

\begin{computationalproblem}{ $P = (R,p_x)$ Sampling}

Input: $x \in \{ x \in \Sigma^* : R(x) \not = \emptyset \}$

Output: A sample drawn according to $p_x$.
\end{computationalproblem}
\end{defn}

Similar to approximation algorithms in the deterministic case, we can ask if Turing machine ``almost'' solves a sampling problem:

\begin{defn}[$\alpha$-almost solving a sampling problem]\label{def:alphaalmostsolve}
Suppose that $P = (R, p_X)$ is some sampling problem. Let $\alpha \in [0,1]$. We say that a probabilistic Turing machine $M$ \emph{$\alpha$-almost solves} \textsc{$P = (R,p_x)$ Sampling} if for all instances $X$ with $R(X) \not = \emptyset$, $M(X)$ accepts $X$ at least half the time and then outputs a sample from a distribution $q^M_X$, where $||q^M_X - p_X||_{TV} \leq \alpha$. In the case $\alpha = 0$, we say that $M$ solves the sampling problem.%
\end{defn}

We will use the complexity class $\RP$ to describe the intractability of a sampling problem.

\begin{defn}[The class $\RP$ \cite{arora2009computational}]
$\RP$ is the class of languages $L \subseteq \Sigma^*$ such that there is a polynomial time probabilistic Turing machine $M$ and a constant $\epsilon > 0$ so that, if $x \not \in L$, $M(x)$ always rejects, and if $x \in L$, $M(x)$ accepts with probability at least $\epsilon$.
\end{defn}

It is widely believed that $\RP \not = \NP$; this belief follows from the widely believed conjectures that $\NP \not = \Poly$ \cite{aaronson2016p} and $\BPP = \RP = \Poly$ \cite{impagliazzo1997p}. Based on this reasoning, and arguments in the style of \cite[Proposition 5.1]{JVV} or \cite[Theorem 1.17]{sinclair1988randomised}, which argue that there is likely no efficient algorithm for a sampling problem by showing that the existence of an efficient sampler would imply $\RP = \NP$, we make the following definition for when a sampling problem is intractable:

\begin{defn}[Intractability of a sampling problem]\label{defn:intractable}
We say that a sampling problem $P$ is \emph{intractable} on a language (or class) of instances $\mathscr{C}$ if for all $\alpha < 1$, the existence of a polynomial time probabilistic Turing machine that $\alpha$-almost uniformly samples from $P$ for all instances in $\mathscr{C}$ implies that $\RP = \NP$.
\end{defn}

\Cref{lem:luckyguess} below abstracts the repetitive part of most proofs showing that a sampling problem is intractable. %
To state it cleanly, we make the following definition, which takes a $p$-relation $Q$, contained inside a $p$-relation $S$, and describes the set of instances that have solutions:

\begin{defn}[Decision problem on a $p$-relation]
Let $S$ be a $p$-relation. A \emph{decision problem} in $S$ is a $p$-relation $Q$, such that $Q \subseteq S$, with an associated language $L_Q = \{ x \in \Sigma^* :  Q(x) \not = \emptyset\}$. If $\mathscr{C} \subseteq \Sigma^*$ is a language, and $L_Q(\mathscr{C}) := \{ x \in \mathscr{C} : Q(x) \not = \emptyset \}$ is $\NP$-complete, then we will say that $Q$ is a decision problem in the $p$-relation $S$ which is $\NP$-complete on the language $\mathscr{C}$. %
\end{defn}

\begin{lem}[Lucky guess lemma]\label{lem:luckyguess}
Consider some sampling problem $P = (R, p_x)$. %
Let $Q$ be some decision problem in a $p$-relation $S$, which is $\NP$-complete on the language $\mathscr{C}$. %
Suppose the following assumptions hold for some polynomials $p_m(n)$, $m \in \mathbb{N}_{\geq 1}$:
\begin{itemize}
    \item There is a $p_m$-time Turing machine $B_m$ such that for any instance $x \in \mathscr{C}$, $B_m$ constructs some $B_m(x) \in \Sigma^*$ with $R(B_m(x)) \not = \emptyset$.
    \item There is another $p_m$-time Turing machine $M_m$ that computes a map $\pi_m : R(B_m(x)) \to S(x)$.
    \item{(Probability Concentration)} If $|Q(x)| \geq 1$ and if $C$ is a random variable distributed according to $P$ on $R(B_m(x))$, then $\mathbb{P} ( \pi_m(C) \in Q(x) ) \geq 1 - \nicefrac{1}{m}.$
\end{itemize}
Then, $P$ is intractable on the language $B(\mathscr{C}) = \{B(x) : x \in \mathscr{C} \}$.
\end{lem}
\begin{proof}

Fix $\alpha < 1$, and take $m = \ceil{ \frac{2}{1 - \alpha}} $. We fix $B = B_m, M = M_m$. Assume that there exists a polynomial time probabilistic Turing machine $G$ that $\alpha$-almost solves $P$ on $B( \mathscr{C})$. We claim that \Cref{alg:LuckyGuess} gives an $\RP$-algorithm for $L_Q(\mathscr{C})$. %
\Cref{alg:LuckyGuess} runs in polynomial time, since for any $X \in \mathscr{C}$, constructing $B(X)$, sampling $C$ with $G$ and computing $\pi(C)$ with $M$ takes time polynomial in $|X|$. %
Thus, we only have to prove that the algorithm succeeds with the correct error bounds. Since \Cref{alg:LuckyGuess} clearly has no false positives, we only need to check that there is a constant lower bound on the true positive rate. We will show that if $|Q(X)| \geq 1$, then the probability of success is at least $\nicefrac{1}{m}$. %
Suppose that $q_Y$ is the distribution over $R(Y)$ of outputs of $G$ on input $Y$. Suppose that $A = \{ C \in R(B(X)) : \pi(C) \in Q(X) \}$. Since $\|p_{B(X)} - q_{B(X)}||_{TV} < \alpha$, and in particular $p_{B(X)}(A) - q_{B(X)}(A) < \alpha$, it follows that $q_{B(X)}(A) > p_{B(X)}(A) - \alpha
\geq 1 - \nicefrac{1}{m} - \alpha \geq \nicefrac{1}{m}$.  Hence, with probability at least $\nicefrac{1}{m}$, the sample drawn by $G$ from $R(B(X))$ will land in $A$. In other words, if $|Q(X)| \geq 1$, then \Cref{alg:LuckyGuess} will answer YES with probability at least $\nicefrac{1}{m}$. Since $L_Q(\mathscr{C})$ is $\NP$-complete, it would follow that $\NP = \RP$. Since this argument holds for all $\alpha < 1$, $P$ is intractable on $B(\mathscr{C})$.
\end{proof}

\begin{algorithm}[H]
\caption{Lucky Guess}\label{alg:LuckyGuess}
\textbf{Input:}  $G, M, B$ and $x \in \mathcal{C}$ as in the proof of \Cref{lem:luckyguess}.
\begin{algorithmic}[1]
\STATE{Construct $B(x)$}
\STATE{Let $C$ be the output of $G$ on $B(x)$}
\IF{$M(C) \in Q(x)$} \STATE{ return \texttt{YES}} \ELSE \STATE{ Return \texttt{NO}.}  \ENDIF
\end{algorithmic}
\end{algorithm}

Certain calculations appear repeatedly when checking the probability concentration hypothesis of \Cref{lem:luckyguess}. We isolate them here:

\begin{lem}\label{lem:HgeqDN}
If $H, N \geq 0$ and $H \geq D N$ for some $ D > 0$, then $\frac{H}{H + N} \geq \frac{D}{1 + D}$.
\end{lem}

\begin{lem}\label{lem:polylarge}
Fix $q  \geq 2$. Then for any $e \in \mathbb{N}$ and $S \geq 1$ if $d \geq 2 $ and
$d \geq 4 (\frac{ \log_2(S) + e}{ \log_2(q)})^2$ then $q^d \geq S d^e$.

\end{lem}
\begin{proof}

It suffices to pick $d$ so that $\frac{ d}{ \log(d)} \geq \frac{ \log(S) + e } { \log(q)}$. Since $\frac{d}{ \log_{2}(d)} \geq \frac{1}{2}\sqrt{d}$ for $d \geq 2$, the claim follows.

\end{proof}

%% file: Sections/2Complexity/3SamplingSimpleCycles.tex
\subsection{Intractability of Uniformly Sampling Simple Cycles}

The seminal paper \cite{JVV} proves that the following sampling problem is intractable:

\begin{computationalproblem}{GenDirectedCycle}
Input: A directed graph $G$.

Output: An element of the set of directed simple cycles of $G$, selected uniformly at random.
\end{computationalproblem}

\begin{figure}
    \centering
    \begin{tabular}{cc}
    \def\svgscale{.3}{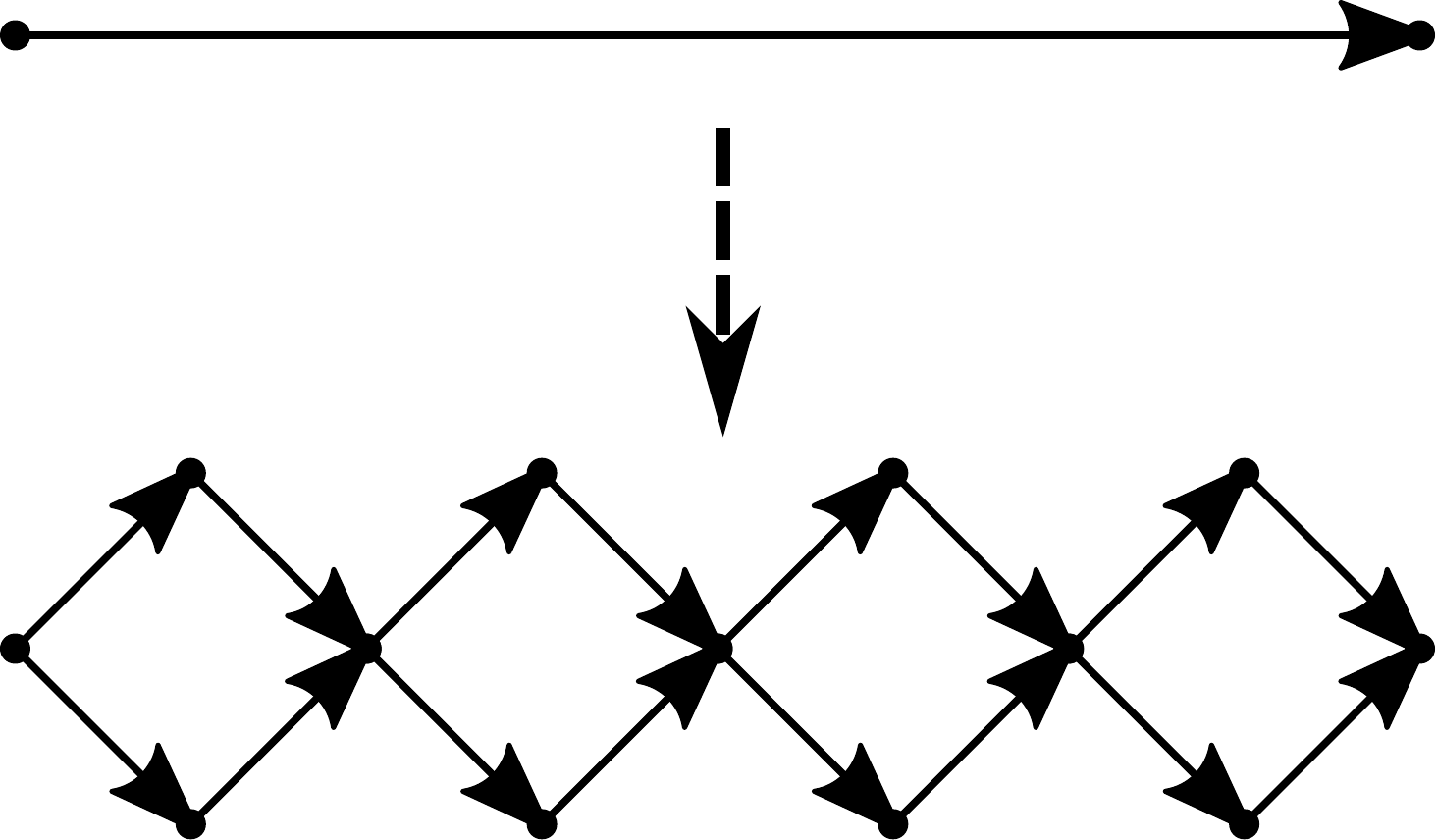 }
    &
    \raisebox{.23cm}{
    \def\svgscale{.3}{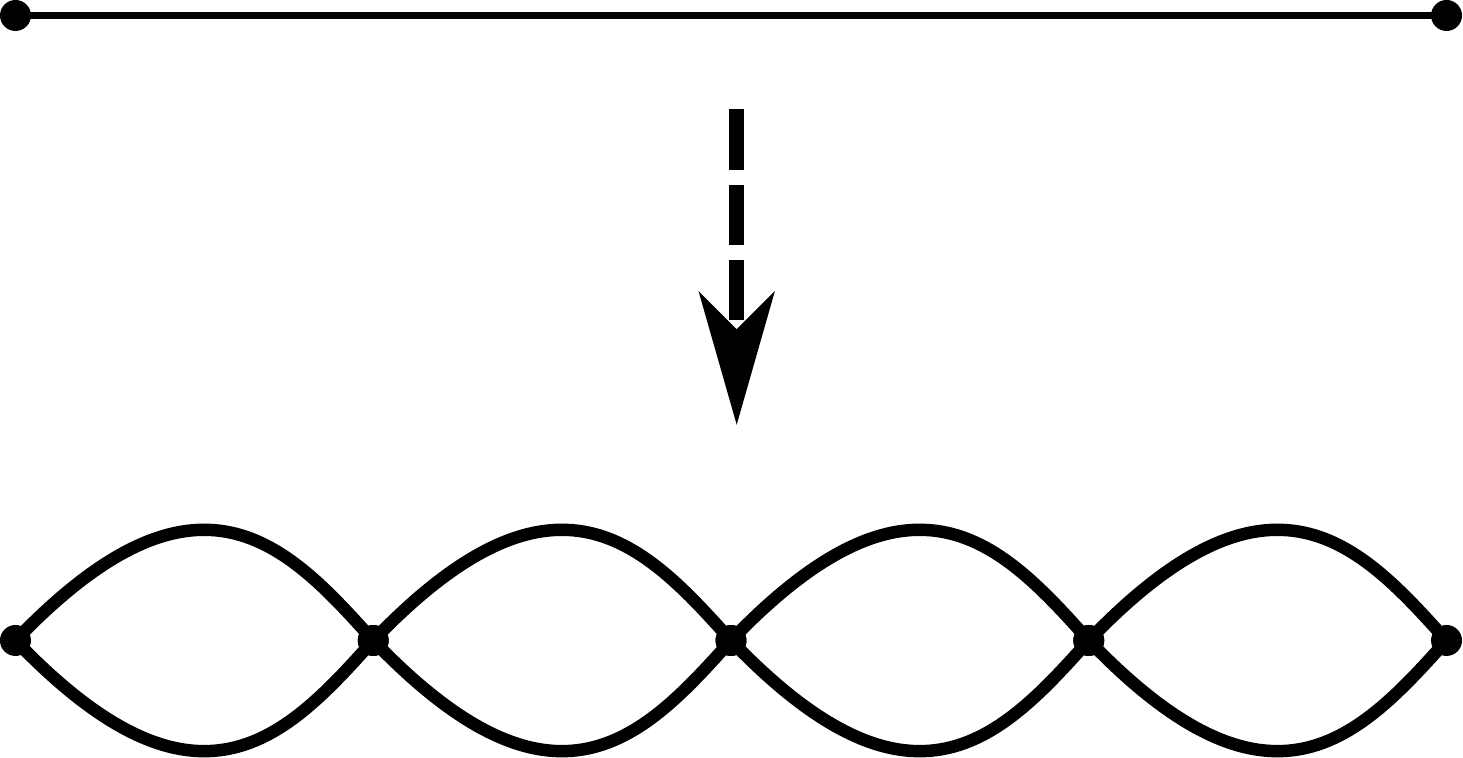}}
    \\
    (a) Directed & (b) Undirected
    \end{tabular}
    \caption{The directed version from \cite{JVV} and its undirected bigons counter part.}
    \label{fig:chainofdiamonds}
\end{figure}

\begin{thm}[\cite{JVV}, Proposition 5.1]\label{thm:gendirectedcycle}
The sampling problem \textsc{GenDirectedCycle} is intractable on the class of directed graphs.\footnote{We want to acknowledge a \href{https://cstheory.stackexchange.com/a/41367/44995}{helpful Stack Exchange conversation} \cite{hengguostack} with Heng Guo that alerted us to this theorem.}
\end{thm}
\noindent We refer the reader to the original paper for the proof. The idea is to concentrate probability on longer cycles by replacing edges with chains of diamonds, as in \Cref{fig:chainofdiamonds}(a), which has the effect of increasing the number of ways to traverse a cycle by a quantity which grows at a exponential rate proportional to the length of the cycle.

Similarly, one can consider an undirected, plane graph version of the same problem:%

\begin{computationalproblem}{GenSimpleCycle}
Input: An undirected, plane graph $G$.

Output: An element of $SC(G)$, selected uniformly at random.
\end{computationalproblem}

Although similar to \Cref{thm:gendirectedcycle}, will include the proof that \textsc{GenSimpleCycle} is intractable on the class of plane graphs as a preface to our other results. To do so, we formalize an analog of the chain of diamonds construction:

\begin{defn}[Chain of bigons, projection map]\label{def:chainofbigons}
Let $G$ be an (undirected) graph. Let $B_d(G)$ denote the graph obtained from $G$ by replacing each edge by a \emph{chain of $d$ bigons}.  Formally, we subdivide each edge into $d$ edges and then add a parallel edge for every edge. %
This construction is illustrated in \Cref{fig:chainofdiamonds}(b). There is a natural map $ \pi_d : SC(B_d(G)) \to 2^{E(G)}$, which collapses the chains of bigons. Formally, $\pi_d(C) = \{ e \in E(G) : B_d(e) \cap C \not = \emptyset \} $.\footnote{For there to be a Turing machine $M$ as in \cref{lem:luckyguess} that simulates $\pi_d$, the collection of $B_d(e)$ for $e \in E(G)$ has to be part of the encoding of the graph $B_d(G)$ as a string of $\Sigma^*$, and $B$ must construct an encoding of $B_d(G)$ with this information. We assume that all of this is true, and make sense of $SC(B_d(G))$ and similar objects in terms of the underlying graph. In general we will omit discussion of this level of detail.}
\end{defn}

\begin{prop}[Adaptation of \cite{JVV}, Proposition 5.1]\label{prop:SimpleCycleUniformHard}
\textsc{GenSimpleCycle} is intractable on the class of simple plane graphs.
\end{prop}

\begin{proof} %
Given a polynomial time probabilistic Turing machine $M$ which $\alpha$-solves  \textsc{GenSimpleCycle} on the class of simple plane graphs, we obtain one that $\alpha$-almost solves  \textsc{GenSimpleCycle} on the class of plane graphs by subdividing each edge of a given graph. We thus let $\mathscr{C}$ be the class of plane graphs, and we will now work on checking that the conditions for \Cref{lem:luckyguess} can be satisfied. For $G \in \mathscr{C}$, with $n = |V(G)$, we fix $m$ and take $d = n^2 + m$. Observe that  $\im ( \pi_d ) = SC(G) \cup E(G)$. If $X \in SC(G)$ is a simple cycle with $m$ edges, then $| \pi_d^{-1}(X) | = 2^{dm}$. For $e \in E(G)$, $|\pi_d^{-1}(e)| = d$. Thus, if $\mathscr{H} \subset SC(G)$ is the set of Hamiltonian cycles of $G$, and $|\mathscr{H}| \not = 0$, then $|\pi_d^{-1} (\mathscr{H})| \geq 2^{dn} \geq 2^{m} 2^{n^2} 2^{d(n - 1)} \geq 2^m |\pi_d^{-1} ( 2^{E(G)} \setminus \mathscr{H})|$; here we have used $2^{n^2}$ as a crude upper bound on $|SC(G) \cup E(G)|$. Thus, by \Cref{lem:HgeqDN}, $\frac{|\pi_d^{-1} (\mathscr{H})|}{ |SC( B_d(G))|} \geq \frac{2^m}{1 + 2^m} \geq 1 - \nicefrac{1}{m}$. Define polynomial-time Turing machines $B$ and $M$ so that $B(G) = B_d(G)$ and $M(C) = \pi_d(C)$, and set $S(G) = E(G) \cup SC(G)$, $Q(G) = \textrm{HamiltonianCycles}(G)$ and $R(G) = SC(G)$ for all $G \in \mathscr{C}$. Since the Hamiltonian cycle problem is $NP$-complete on $\mathscr{C}$ \cite{garey_plane_1976}, the conditions of \Cref{lem:luckyguess} are satisfied.

\end{proof}

%% file: Images/ChainOfDirectedDiamonds.pdf_tex
\begingroup%
  \makeatletter%
  \providecommand\color[2][]{%
    \errmessage{(Inkscape) Color is used for the text in Inkscape, but the package 'color.sty' is not loaded}%
    \renewcommand\color[2][]{}%
  }%
  \providecommand\transparent[1]{%
    \errmessage{(Inkscape) Transparency is used (non-zero) for the text in Inkscape, but the package 'transparent.sty' is not loaded}%
    \renewcommand\transparent[1]{}%
  }%
  \providecommand\rotatebox[2]{#2}%
  \newcommand*\fsize{\dimexpr\f@size pt\relax}%
  \newcommand*\lineheight[1]{\fontsize{\fsize}{#1\fsize}\selectfont}%
  \ifx\svgwidth\undefined%
    \setlength{\unitlength}{421.05956571bp}%
    \ifx\svgscale\undefined%
      \relax%
    \else%
      \setlength{\unitlength}{\unitlength * \real{\svgscale}}%
    \fi%
  \else%
    \setlength{\unitlength}{\svgwidth}%
  \fi%
  \global\let\svgwidth\undefined%
  \global\let\svgscale\undefined%
  \makeatother%
  \begin{picture}(1,0.58505651)%
    \lineheight{1}%
    \setlength\tabcolsep{0pt}%
    \put(0,0){\includegraphics[width=\unitlength,page=1]{ChainOfDirectedDiamonds.pdf}}%
  \end{picture}%
\endgroup%

%% file: Images/ChainOfBigons.pdf_tex
\begingroup%
  \makeatletter%
  \providecommand\color[2][]{%
    \errmessage{(Inkscape) Color is used for the text in Inkscape, but the package 'color.sty' is not loaded}%
    \renewcommand\color[2][]{}%
  }%
  \providecommand\transparent[1]{%
    \errmessage{(Inkscape) Transparency is used (non-zero) for the text in Inkscape, but the package 'transparent.sty' is not loaded}%
    \renewcommand\transparent[1]{}%
  }%
  \providecommand\rotatebox[2]{#2}%
  \newcommand*\fsize{\dimexpr\f@size pt\relax}%
  \newcommand*\lineheight[1]{\fontsize{\fsize}{#1\fsize}\selectfont}%
  \ifx\svgwidth\undefined%
    \setlength{\unitlength}{421.05957146bp}%
    \ifx\svgscale\undefined%
      \relax%
    \else%
      \setlength{\unitlength}{\unitlength * \real{\svgscale}}%
    \fi%
  \else%
    \setlength{\unitlength}{\svgwidth}%
  \fi%
  \global\let\svgwidth\undefined%
  \global\let\svgscale\undefined%
  \makeatother%
  \begin{picture}(1,0.51803677)%
    \lineheight{1}%
    \setlength\tabcolsep{0pt}%
    \put(0,0){\includegraphics[width=\unitlength,page=1]{ChainOfBigons.pdf}}%
  \end{picture}%
\endgroup%

%% file: Sections/2Complexity/4BasicIntractibility.tex
\subsection{Intractability of Sampling from $P_2(G)$}\label{section:basichard}

As we discussed in \Cref{section:CongressionalMotivation}, we are interested in the problem of uniformly sampling connected $2$-partitions.

\begin{computationalproblem}{GenConnected2Partition}
Input: A graph $G$.

Output: An element of $P_2(G)$, selected uniformly at random.
\end{computationalproblem}

We recall a fact about plane duality, which will connect \Cref{prop:SimpleCycleUniformHard} above to \textsc{GenConnected2Partition}.
\begin{thm}[\cite{Erickson}]\label{duality}
Let $G$ be a plane graph and $G^*$ its plane dual. Then, there is an polynomial time computable bijection between $SC(G^*)$ and $\mathscr{P}_2(G)$.\footnote{We wish to acknowledge a helpful \href{https://cstheory.stackexchange.com/a/41272/44995}{Stack Exchange discussion} with Mikhail Rudoy which first drew our attention to this theorem \cite{firstduality}.}
\end{thm}

\begin{thm}
\textsc{GenConnected2Partition} is intractable on the class of plane graphs.
\end{thm}\label{thm:basicintractbilitypartitions}
\begin{proof}
If $G$ is a plane graph, then a polynomial algorithm to sample uniformly $P_2(G^*)$ gives an algorithm to sample uniformly from $SC(G)$ by \Cref{duality}. Now the result follows from \Cref{prop:SimpleCycleUniformHard}.
\end{proof}

This theorem implies that the broadest version of uniform sampling from the space of graph partitions is intractable. While already this observation is notable given how little is known about the graphs that appear in redistricting, our discussion does not end here. Rather, \Cref{thm:basicintractbilitypartitions} will be strengthened in \Cref{thm:balancedpartition} and \Cref{thm:triangulationhard}. In \Cref{Section:FlipChain}, we will also highlight how the probability concentration gadgets identify concrete issues with Markov chains used for sampling partitions.

%% file: Sections/2Complexity/5Balanced.tex
\subsection{Intractability of Uniformly Sampling Balanced Partitions}\label{section:balancedhard}

For applications in redistricting, the blocks of a connected partition should be roughly equal in population. This motivates studying the problem of sampling \emph{balanced} partitions: %

\begin{defn}[$\epsilon$-balanced simple cycles and $2$-partitions]
Let $SC^{\epsilon}(G)$ be the set of \emph{$\epsilon$-balanced simple cycles} of a plane graph $G$, such that if $\{A,B\}$ is the dual connected 2-partition of $G^*$ then $1 - \epsilon \leq \frac{ |A| }{|B|} \leq 1 + \epsilon$ and $1 - \epsilon \leq \frac{ |B| }{|A|} \leq 1 + \epsilon$. Similarly, we say that a partition $(A,B) \in P_2(G)$ is $\epsilon$-balanced if these inequalities hold for $\{A,B\}$; we define $P_2^{\epsilon}(G)$ to be the set of such partitions.  %
\end{defn}

\begin{computationalproblem}{$\epsilon$-Balanced Uniform 2-Partition Sampling}
Input: A plane graph $G$.

Output: An element of $P_2^{\epsilon}(G)$ selected uniformly at random.
\end{computationalproblem}

The existence of a $0$-\emph{balanced} 2-partition in a graph $G$ is not obvious. In fact, determining if there exists a balanced connected $2$-partition of a given graph $G$ is $NP$-complete:

\begin{thm}[\cite{dyer_complexity_1985}, Theorem 2.2\label{thm:balancedpartition}\footnote{The phrase ``$k$-partition'' in \cite{dyer_complexity_1985} has a different meaning from the way we use it here. In their notation, a $k$-partition is a connected partition where each piece has size $k$. What we call a balanced connected $2$-partition, they call an $\nicefrac{n}{2}$ partition. The theorem is stated here in our notation.}] The decision problem of whether a given connected graph $G$ has a $0$-balanced, connected $2$-partition is $\NP$-complete.
\end{thm}

Given Theorem~\ref{thm:balancedpartition}, the problem of uniformly sampling from the set of $0$-balanced connected $2$-partitions is vacuously intractable. To circumvent this issue, we focus on the case where $G$ is $2$-connected, since in that case a $0$-balanced 2-partition always exists and can be constructed in polynomial time by constructive versions \cite{suzuki1990linear} of the Gy{\H o}ri--Lov\'asz Theorem  \cite{gyori,lovasz1977homology} for $2$-partitions. As in the previous section, our strategy is to work with simple cycles, rather than connected partitions. In particular, the following classical theorem will let us translate a statement about Hamiltonian cycles into a statement about balanced Hamiltonian cycles:

\begin{thm}[Grinberg's Theorem, \cite{Grinberg}]%
Let $G$ be a plane graph. Assign to each face $F$ a weight $deg(F) - 2$, where $deg(F)$ is the number of edges in $F$. If $H$ is a Hamiltonian cycle of $G$, then the total weight of the faces inside of $H$ is equal to the total weight of the faces outside of $H$.
\end{thm}

Grinberg's Theorem implies that every Hamiltonian cycle of a maximal plane graph is balanced, since every face is a triangle. Since the problem of determining whether a maximal plane graph has a Hamiltonian cycle is NP-complete \cite{wigderson_maximal}, it follows that the problem of determining whether a maximal plane graph has a \emph{balanced} Hamiltonian cycle is NP-complete.\footnote{The authors wish to thank Gamow from Stack Exchange for \href{https://cstheory.stackexchange.com/q/41999} { a helpful comment} \cite{gamow1} and \href{https://cstheory.stackexchange.com/a/41446/44995}{directing us towards} \cite{wigderson_maximal}. }

We begin by defining the map $\pi$ and proving the probability concentration result used to apply \Cref{lem:luckyguess}. %
In particular, recalling the definition of $\pi_d$ (\Cref{def:chainofbigons}), we define $\pi_{d}^{\epsilon} : SC^{\epsilon} ( B_d(G)) \to 2^{E(G)}$ as the restriction of $\pi_d : SC(B_d(G)) \to 2^{E(G)}$ to the $\epsilon$-balanced simple cycles of $G$. 
We derive the necessary inequalities in the following lemma:

\begin{lem}
Let $G  = (V,E)$ be a maximal plane graph, with $n = |V(G)|$. Let $C$ be a Hamiltonian cycle of $G$, and let $A$ be any non-Hamiltonian cycle of $G$. Let $\epsilon \geq 0$. Then we have
\begin{align} | (\pi_{2d}^{\epsilon})^{-1}( C) | &\geq | (\pi_{2d}^0)^{-1}( C) | = { 2dn \choose dn} \geq \frac{ 2^{2dn}}{2dn + 1},\textrm{ and}
\label{lowerbound:longestbalanced}
\\
\label{upperbound:shorterunbalanced} 
| (\pi_{2d}^{\epsilon})^{-1}(A) | & \leq | (\pi_{2d}^{\infty})^{-1} (A) | \leq  2^{2d(n-1)}.
\end{align}

If $\mathscr{H}$ is the set of Hamiltonian cycles of $G$, and $C \in \mathscr{H}$, then for $d \geq 4 \ceil{(m + n^2/2 + 3/2 + \log(n))^2}$, \begin{equation}\label{eqn:HDNbalanced}
| (\pi_{2d}^{\epsilon})( \mathscr{H}) | \geq 2^{2m} | \pi_{2d}^{-1}( 2^{E(G)} \setminus \mathscr{H})|.
\end{equation}

\end{lem}

\begin{proof}[Proof of \Cref{lowerbound:longestbalanced} and \Cref{upperbound:shorterunbalanced}]
Since every Hamiltonian cycle of $G$ is balanced, the only way to lift the cycle to a balanced simple cycle of $B_{2d}(G)$ is to take the inward edge along exactly half of the bigons. For the lifts of non-Hamiltonian cycles, we can bound the number of lifts to a balanced simple cycle by the number of lifts to a cycle. 
These inequalities then follow from the standard fact that $\frac{ 2^{2r} }{ 2r + 1} \leq { 2r \choose r }$%
\end{proof}
\begin{proof}[Proof of \Cref{eqn:HDNbalanced}]
By \Cref{lem:polylarge}, for $d \geq 4 \ceil{(m + n^2/2 + 3/2 + \log(n))^2}$ we have $$2^{2dn} \geq (2^{2m + n^2 + 2}n) d 2^{2d(n-1)} \geq 2^{2m} 2^{n^2}(2dn + 1) 2^{2d(n-1)},$$ and thus
$$| (\pi_{2d}^{\epsilon})( \mathscr{H}) | \geq | (\pi_{2d}^{\epsilon})(C)| \geq \frac{2^{2dn}}{2 dn + 1} \geq 2^{2m} 2^{n^2} 2^{2d(n-1)} \geq 2^{2m} | \pi_{2d}^{-1}( 2^{E(G)} \setminus \mathscr{H})|.$$
\end{proof}

\begin{prop}\label{prop:balancedhard} Fix $\epsilon \geq 0$. Then, \textsc{$\epsilon$-Balanced Uniform Simple Cycle Sampling} is intractable on the class of graphs of the form $B_d(G)$, where $d \geq 1$ and $G$ is any maximal plane graph.
\end{prop}
\begin{proof}
To prove this, we fit what we have calculated into the format of \Cref{lem:luckyguess}. Fix $m$. Then we take $d = 4 \ceil{(m + n^2/2 + 3/2 + \log(n))^2}$, which is polynomial in $G$, and set $B(G) = B_{2d}(G)$, and $M$ to compute $\pi_{2d}$. The probability concentration hypothesis follows from \Cref{lem:HgeqDN} and \Cref{eqn:HDNbalanced}, since they show that $\frac{| (\pi_{2d}^{\epsilon})^{-1}( \mathscr{H}) | } { |SC^{\epsilon}(G) | } \geq \frac{4^m}{1 + 4^m} \geq 1 - 1/m$. Finally, we set $Q$ to be the Hamiltonian cycles, and $\mathscr{C}$ to be the class of maximal plane graphs.

\end{proof}

\begin{thm}
Fix  $\epsilon \geq 0$. Then \textsc{$\epsilon$-Balanced Uniform 2-Partition Sampling} is intractable on the class of $2$-connected plane graphs.
\end{thm}
\begin{proof}
$(B_d(G))^*$ is $2$-connected when $G$ is a maximal plane graph. %
The claim now follows by \Cref{prop:balancedhard} using \Cref{duality}.
\end{proof}

\begin{remark} These arguments work for any reasonable definition of nearly balanced that considers partitions with $|A| = |B|$ to be balanced.
\end{remark}

%% file: Sections/2Complexity/6Triangulations.tex
\subsection{Sampling intractability on maximal plane Graphs with bounded degree}\label{section:maxhard}

In this section, we improve the results from \Cref{section:basichard} by showing that the connected $2$-partition sampling problem is intractable on the class of maximal plane graphs with vertex degree bounded by $531$. 
This will shrink the gap between our theoretical intractability statements and the graphs used to study redistricting. To obtain our result, we will start in \Cref{strengthening} by proving a corresponding $NP$-completeness theorem, building on the results in \cite{garey_plane_1976}. Second, in \Cref{section:Recursive}, we will describe a construction for concentrating probability on the longer simple cycles by providing a simple cycle with many paths through a vertex, rather than with many paths through an edge. Finally, in \Cref{Max2Part} we will show how to tie the intractability argument together. We will reuse the gadgets from this section in \Cref{betterbehavedexample} to construct explicit counterexamples to heuristic sampling algorithms.

\subsubsection{Hamiltonian cycle is $\NP$-complete on cubic, 3-connected plane graphs with face degree $\leq 177$}\label{strengthening}

\begin{defn}[$3CCP$ graphs with bounded face degree]\label{defn:Cm3CCP}
Recall that a $3CCP$ graph is one that is $3$-connected, cubic, simple, and plane. Let $\mathscr{C}_m$ denote the collection of $3CCP$ graphs with face degree $\leq m$. Let $\mathscr{C} = \bigcup_{k \geq 3} \mathscr{C}_k$.
\end{defn}

\begin{defn}[Hamiltonian cycle problem]
If $\mathscr{D}$ is a language of graphs, let $\mathscr{D}\textrm{-HAM} = \{ G \in \mathscr{D} : G \text{ is Hamiltonian} \}$
\end{defn}

Our goal in this section is to prove the following theorem:\footnote{The authors are grateful to \href{http://domotorp.web.elte.hu/}{Pálvölgyi Dömötör Honlapja} for \href{https://cstheory.stackexchange.com/questions/42565/np-completeness-of-hamiltonicity-of-cubic-polyhedral-plane-graphs-with-bounded}{suggesting the proof strategy} used in this section \cite{honlapja}. We are of course responsible for all errors in our execution of the strategy. }%

\begin{thm}\label{thm:facebounded3CCP}
$\mathscr{C}_{177}\textrm{-HAM}$ is $\NP$-complete.  %
\end{thm}
 We will prove a reduction from the following theorem: %

\begin{thm}[\cite{garey_plane_1976}]\label{thm:GJT3CCP}$\mathscr{C}\textrm{-HAM}$ is $\NP$-complete
\end{thm}

To obtain \Cref{thm:facebounded3CCP} from \Cref{thm:GJT3CCP}, we will show that we can take $G \in \mathscr{C}$ and construct a $G' \in \mathscr{C}_{177}$ such that $G'$  has Hamiltonian cycle if and only if $G$ does. We will obtain $G'$ from $G$ by using an algorithm that subdivides the large faces repeatedly (\Cref{alg:Subdivisions}) in polynomial time (\Cref{prop:subdivisioninpoly}). The gadget that we use to subdivide large faces comes from the proof of \Cref{thm:GJT3CCP} in \cite{garey_plane_1976}, in particular, we use their $3$-way OR gate gadget. We now review some relevant properties of that $3$-way $OR$ (3OR) that will be used in the reduction:

\begin{defn}[$3OR$ Gadget, \cite{garey_plane_1976}]
The $3OR$ gadget is pictured in  \Cref{fig:3ORDetailed}. This gadget has three distinguished sets of attaching nodes, each of which consists of a path graph with $6$ nodes. %
\end{defn}

\begin{figure}
    \centering
    \def\svgscale{.3}{
    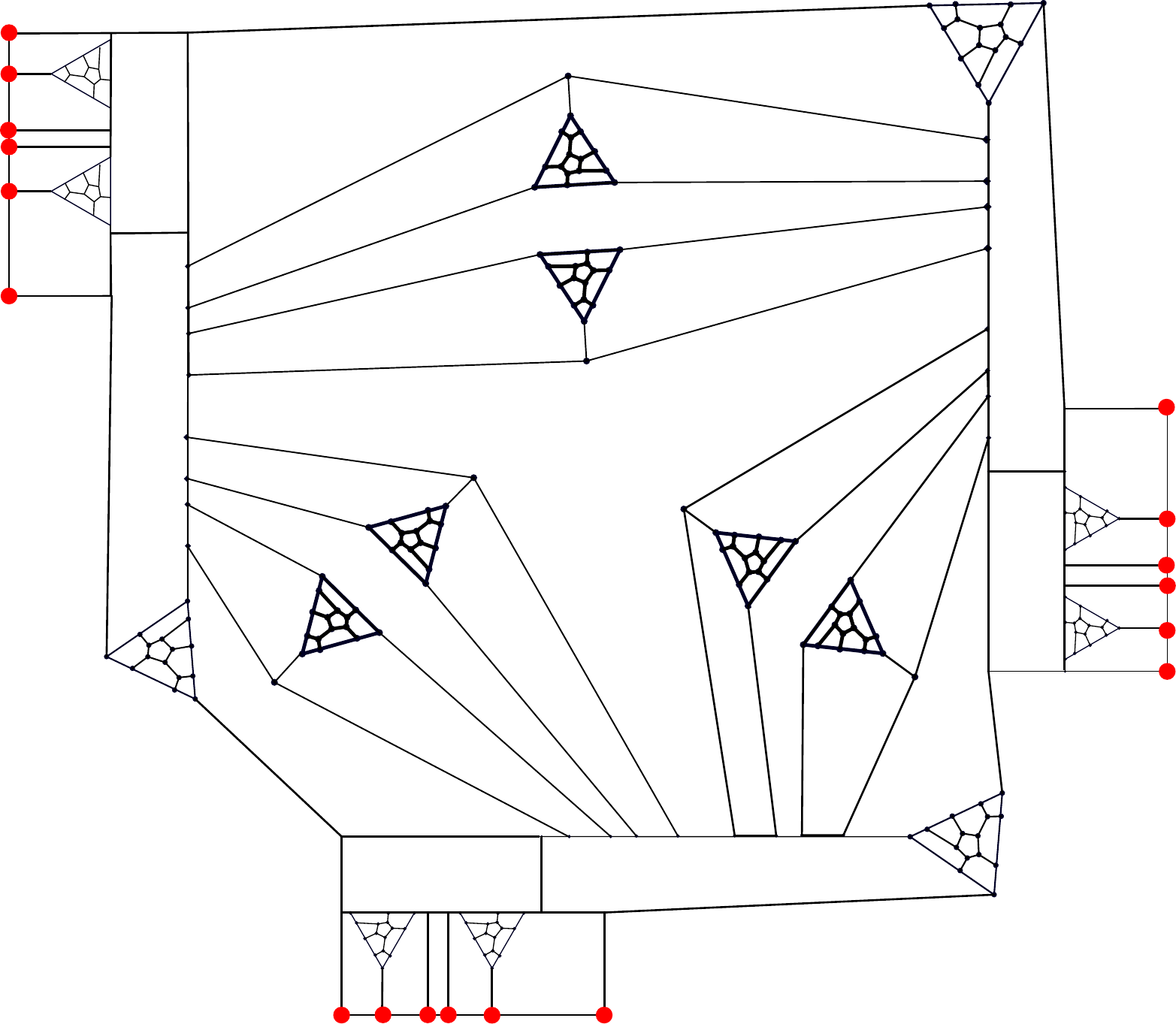}
    \caption{The 3OR gadget. See \Cref{fig:3ORDetailedForTedious} in the appendix for more detail. The distinguished attaching nodes are colored red.
    }
    \label{fig:3ORDetailed}
\end{figure}

\begin{defn}[$3OR$ insertion]
 A $3OR$ gadget can be inserted into a face $F$ of a plane graph by picking 3 edges $\{e_1, e_2,e_3\}$ of $F$, replacing each edge $e_i$ with a path containing $6$ nodes, and gluing each of the distinguished segments of the $3OR$ to one of those subdivided edges. For the plane embedding, we place the rest of the $3OR$ into the interior of $F$, as in \Cref{fig:Detailed3ORInsertion}. %
We will refer to this operation as \emph{inserting a $3OR$} into $F$ at the edges $\{e_1, e_2, e_3\}$.
\end{defn} %

\begin{figure}
\centering
\rotatebox{90}{\def\svgscale{.3}{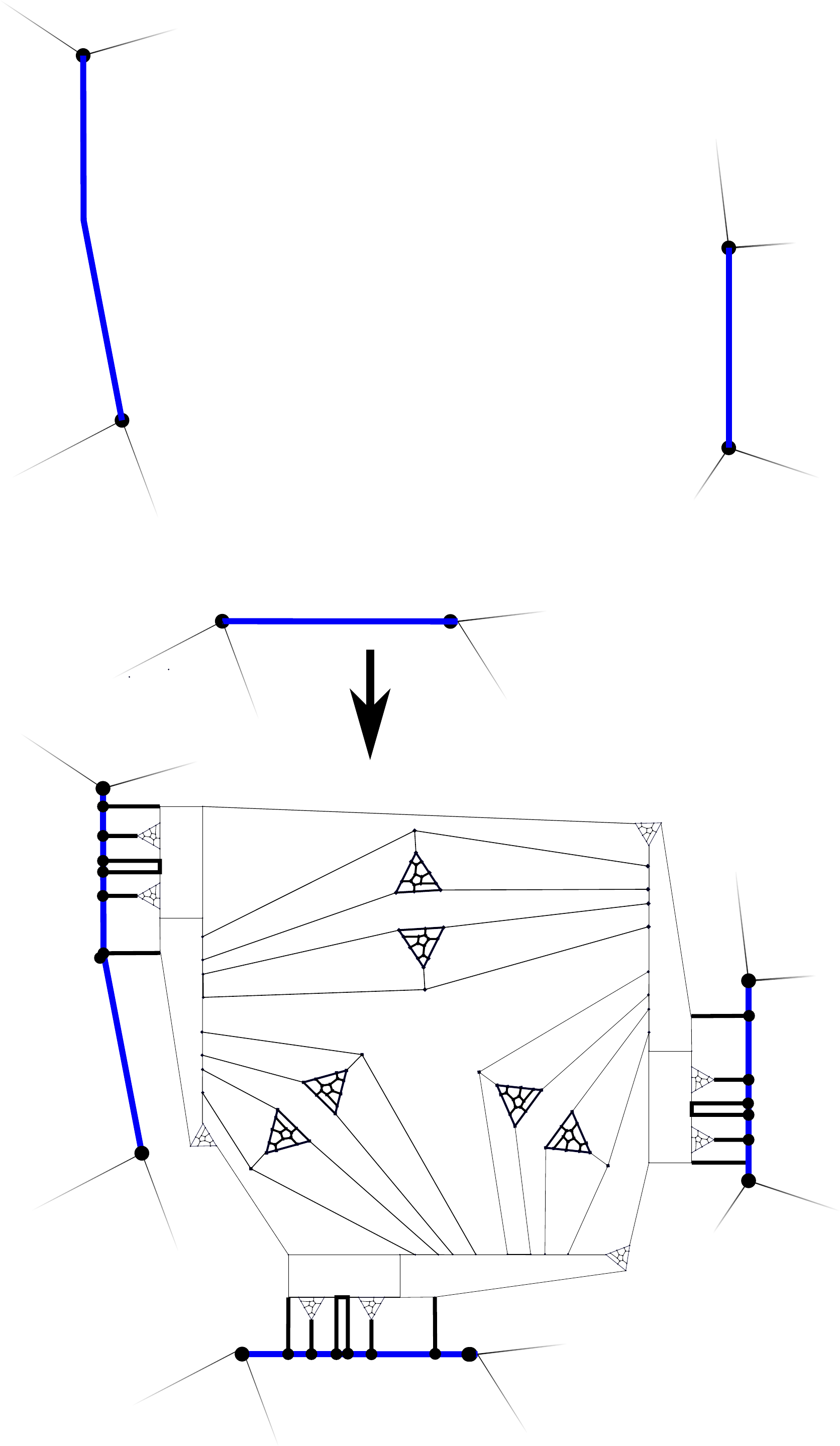} }
\vspace{.8cm}\caption{Inserting a $3OR$. See \Cref{fig:3ORDetailedForTedious} for more detail.} \label{fig:Detailed3ORInsertion}\label{fig:Detailed3ORInsertion}
\end{figure}
\begin{lem}\label{lem:GJT3ORProperty}
Let $H$ be a $3CCP$ graph, $F$ a face of $H$, and  $\{e_1, e_2, e_3\}$ edges of $H$. Construct a graph $H'$ by attaching a $3OR$ gadget to the edges $\{e_1, e_2, e_3\}$. Then $H'$ has a Hamiltonian cycle if and only if $H$ has a Hamiltonian cycle containing at least one of the $e_i$. Additionally, $H'$ is a $3CCP$ graph.
\end{lem}
\begin{proof}[Proof \cite{garey_plane_1976}]
This amounts to an analysis of the local states, which are described in \cite[Fig.\ 6B]{garey_plane_1976}. The proof that these are the only possible local states exploits detailed properties of the $3OR$ gadget; the reader can consult \cite{garey_plane_1976} for details. The proof that $H'$ is $3$-connected follows from \Cref{lem:3_connected_lemma}, and checking that $H'$ is cubic and planar is straightforward. 
\end{proof}

Our strategy will be to take large faces and subdivide them using $3OR$ gadgets: %
\begin{defn}[Subdivision]\label{def:subdiv}
Let $F$ be a face of $H$. A \emph{subdivision of $F$} is a graph $H'$ obtained by taking $3$ edges of $F$, say $e_1, e_2, e_3$, where $e_1$ and $e_2$ \emph{share a vertex}, and inserting a $3OR$ into $F$ at $e_1, e_2, e_3$. See \Cref{fig:Detailed3ORInsertion} for an illustration of this definition. In \cref{fig:3ORDetailedForTedious} we label a few regions of this subdivision for reference: $3$ \emph{adjacent faces}, a \emph{pocket}, and $2$ \emph{large faces}.

\end{defn}

The difference between `subdivision' and `inserting a $3OR$' is that we require that $2$ of the edges used are adjacent. The following proposition shows that we can subdivide faces without changing Hamiltonicity:
\begin{prop}
Let $H$ be any $3CCP$ graph, and let $F$ be a face of $H$. Let $v$ be a vertex of $F$, let $e_1$ and $e_2$ be the edges of $F$ adjacent to $v$, and let $e_3$ be any other edge of $F$. Let $H'$ be the graph obtained by subdividing $H$ at $e_1,e_2$ and $e_3$. Then, $H$ has a Hamiltonian cycle if and only if $H'$ has a Hamiltonian cycle. See \Cref{fig:3ORDetailedForTedious}.
\end{prop}
\begin{proof}
Since $H$ is cubic, every Hamiltonian cycle of $H$ uses all but one of the edges at each vertex. %
Thus, any Hamiltonian cycle uses at least one of $\{e_1,e_2\}$, allowing us to apply \Cref{lem:GJT3ORProperty}.%
\end{proof}

Since our goal is to reduce the face degree, we define the subdivisions that optimally decrease degree: %

\begin{defn}[Optimal subdivision]\label{defn:optimalsubdivision}
Let $F$ be a face of a $3CCP$ graph $G$. An \emph{optimal subdivision} of $F$ is any subdivision of $F$ that minimizes the maximum degree of the two large faces (cf.\ \Cref{def:subdiv}) of the subdivision.
\end{defn}

\Cref{alg:Subdivisions} takes a graph $H$ and a parameter $d$, and---if it terminates---returns a graph that has no faces of degree $> d$. We will use this algorithm to reduce to an instance of $\mathscr{C}_{177}\textrm{-HAM}$ from one of $\mathscr{C}\textrm{-HAM}$, so we must show that \Cref{alg:Subdivisions} terminates in polynomial time for an appropriate choice of $d$. To determine the necessary $d$, we need the following lemma:

\begin{algorithm}[t]
\caption{Subdivision of faces with more than $d$ edges}\label{alg:Subdivisions}
Input: A plane graph $H$ and $d \in \mathbb{N}$.
\begin{algorithmic}[1]
\STATE{  Let $F$ be any face of $H$ with maximum degree}
\IF{ $deg(F) \leq d$} \STATE{terminate and return $H$} \ELSE \STATE{Use a $3OR$ gadget to optimally subdivide $F$ (\Cref{defn:optimalsubdivision}). Set this subdivided graph as $H$.} \ENDIF
\STATE{ Return to $1$.}
\end{algorithmic}
\end{algorithm}

\begin{lemma}%
\label{lem:tedious} Let $G$ be a $3CCP$ graph. Suppose that $F$ is a face of degree $f$. Suppose we make an optimal subdivision of $F$ at edges $e_1, e_2, e_3$, where $e_1$ is adjacent to $e_2$. Let $F_i$ be the face in $G$ adjacent to $e_i$ for $i = 1,2,3$. Then, the following hold:
\begin{itemize}
    \item Each $F_i$ is distinct. Moreover, the degree of each $F_i$ increases by $6$.
    \item The two large faces (\cref{fig:3ORDetailedForTedious}) inside of what was originally $F$ each have degree $\floor{f / 2} + 10$ and $\ceil{f/2} + 10$
    \item  The gadget itself introduces $33$ faces of degree $4$, $75$ faces of degree $5$, $9$ faces of degree $7$, $15$ faces of degree $8$, $4$ faces of degree $9$, $3$ faces of degree $10$, $3$ faces of degree $12$, and $3$ of degree $14$. We call these the ``small faces.''%
    \item A face of degree 10 is introduced, which is labelled ``the pocket'' in \Cref{fig:3ORDetailedForTedious}.
\end{itemize}
\end{lemma}
\begin{proof}
\Cref{fig:3ORDetailedForTedious} in the appendix can be used to count the degrees of these faces, and the number of the small faces of different sizes. The fact that each $F_i$ is distinct follows from the fact that the graph is $3CCP$. %
\end{proof}

\begin{prop}\label{prop:subdivisioninpoly}
Let $H$ be a $3CCP$ graph with $n$ nodes. As long as $d \geq 178 $, Algorithm~\ref{alg:Subdivisions} terminates in time polynomial in $|H|$.%
\end{prop}
\begin{proof}
We will consider a single step in the subdivision algorithm and show that in each step a certain nonnegative energy function decreases by at least one. Since the energy function starts off with value $O(n^2)$, the proposition will follow.

Let $f_j$, $j = 1, \ldots, |F(H)|$, and $f'_k$, $k = 1, \ldots, |F(H')|$ refer to the face degrees in some enumerations of the faces of $H$, %
before and after one step of Algorithm~\ref{alg:Subdivisions}, respectively. Assume that $f_1$ corresponds to the face $F_1$ being subdivided during that step and that $f_2, f_3, f_4$ correspond to the faces adjacent to $F_1$ along the edges where the $3OR$ gadget is being added. %
Let $S = \sum f_i^2$ and $S' = \sum f'_k$. 
By \Cref{lem:tedious} and notating the degrees of all the small faces by $c_i$, we have that $S' = S - f_1^2 + \sum c_i^2 + 10^2 + (\floor{f_1/2} + 10)^2 + (\ceil{ f_1/2}+10)^2 + \sum_{j \in \{2,3,4\} } ((f_j + 6)^2 - f_j^2).$

Thus, if $d$ is a positive integer so that $- d^2 + \sum c_i^2 + 10^2 + (\floor{d/2} + 10)^2 + (\ceil{ d/2}+10)^2  +  3 ((d + 6)^2 - d^2) \leq -1$, %
 then whenever there is a face of degree $> d$ one step of the subdivision algorithm reduces the energy $S$ by at least one. %
The precise computation of the smallest such $d$ depends on the counts of the $c_i$ listed in \Cref{lem:tedious}. Using these counts, it can be checked that taking $d = 178$ suffices to ensure that the energy decreases by at least one in each step.

Finally, note that the initial energy $S$ is bounded by $(\sum_i f_i)^2 = (2|E(H)|)^2 = O( |V(H)|^4)$. Therefore, after $O( |V(H)|^4)$ subdivision steps, \Cref{alg:Subdivisions} with $ d = 178$ terminates. Since each step of \Cref{alg:Subdivisions} takes polynomial time, the result follows.
\end{proof}

In particular, we can use Algorithm~\ref{alg:Subdivisions} to eliminate all of the faces of degree greater than or equal to $178$. Combining this with \Cref{thm:GJT3CCP}, we can now prove \Cref{thm:facebounded3CCP}:

\begin{proof}[Proof of \Cref{thm:facebounded3CCP}]
Let $H \in \mathscr{C}$. Apply \Cref{alg:Subdivisions} to obtain an $H' \in \mathscr{C}_{177}$, such that $H'$ has a Hamiltonian cycle if and only if $H$ has one. Constructing $H'$ takes polynomial time, by \Cref{prop:subdivisioninpoly}. %
\end{proof}

\subsubsection{The Node Replacing Gadgets $R_d$}\label{section:Recursive}

In this section, we will construct the corresponding gadget to the chain of bigons, which will allow us to concentrate probability on the longer cycles while remaining $3CCP$. Instead of replacing edges with a chain of $d$ bigons, which allowed for $2^d$ choices of ways to route through that edge, the gadgets $R_d$ we construct here will replace cubic vertices and allow for $\Theta(5^d)$ choices through that vertex.\footnote{The inspiration for this construction came from \cite{wigderson_maximal}, wherein one step a reduction from Hamiltonian cycle on $3CCP$ graphs to Hamiltonian cycle on maximal plane graphs is to replace cubic vertices by a certain gadget.} The first few $R_d$'s are displayed in \Cref{fig:recursiveinitial}, and we give a definition below:

\begin{figure}
    \centering
    \begin{tabular}{cc}
    \def\svgscale{.3}{
    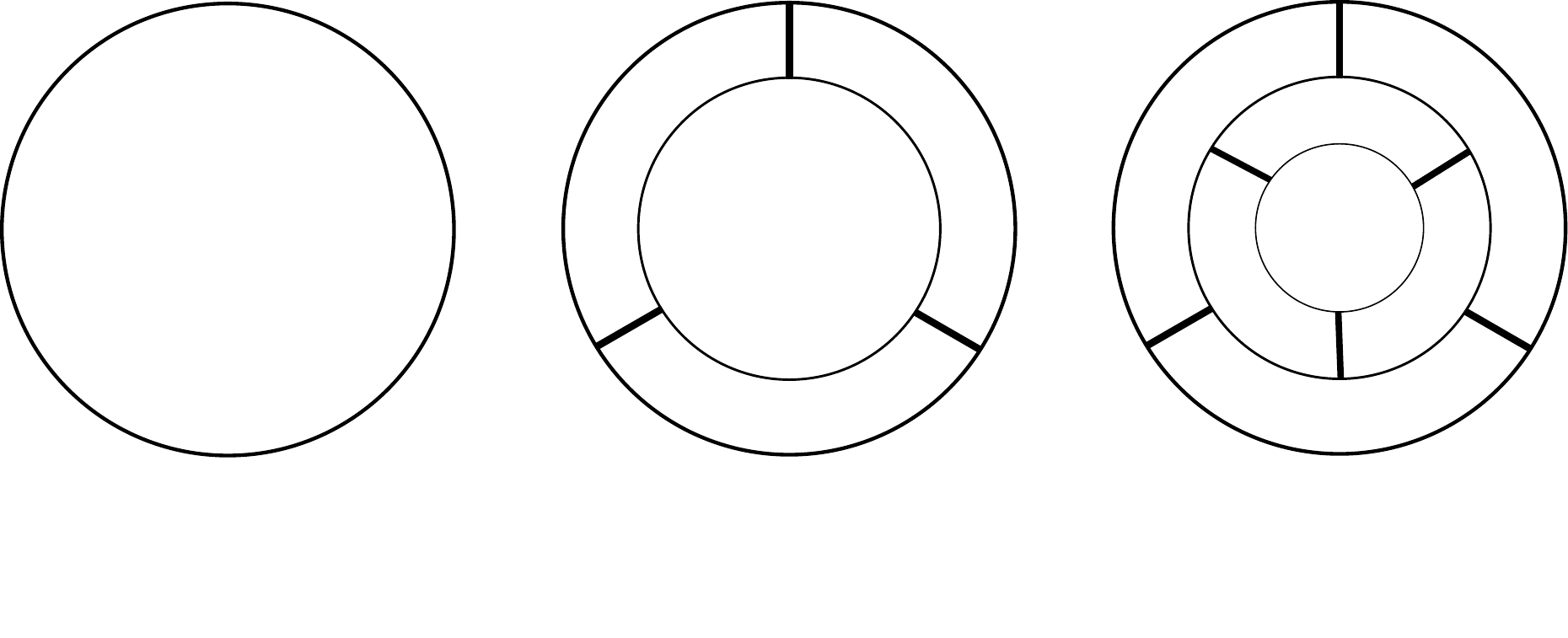}
    \end{tabular}
    \caption{The first $3$ terms in the sequence of gadgets.}\label{fig:recursiveinitial}
\end{figure}

\begin{defn}[The probability concentration gadgets $R_d$, $C_d$, $R'_d$%
]\label{defn:RDDef}
Define $R_0$ to be a $3$-cycle, with nodes labeled with $(a_0,b_0,c_0)$. %
For each $d \geq 0$, we will construct $R_{d+1}$ from $R_d$. First, we construct $R_d'$ from $R_d$ by subdividing the edges $\{x_d, y_d\}$ for all $x \not = y \in \{a,b,c\}$. The node that subdivided the edge $\{ x_d,y_d \}$ gets labelled $z'_d$, where $\{x,y,z\} = \{a,b,c\}$. For each $x \in \{a,b,c\}$, we attach a node $x_{d+1}$ and an edge $\{x'_d, x_{d+1}\}$. Then, we separately build a $3$ cycle $C_{d+1}$ with nodes labelled by $(a_{d+1}, b_{d+1},c_{d+1})$. We obtain $R_{d+1}$ by gluing $R'_d$ to $C_d$ by identifying the nodes with the same labels. See \Cref{fig:labelsrecursive} for an illustration of this construction.

\end{defn}

\begin{figure}
    \centering
    \def\svgscale{.7}{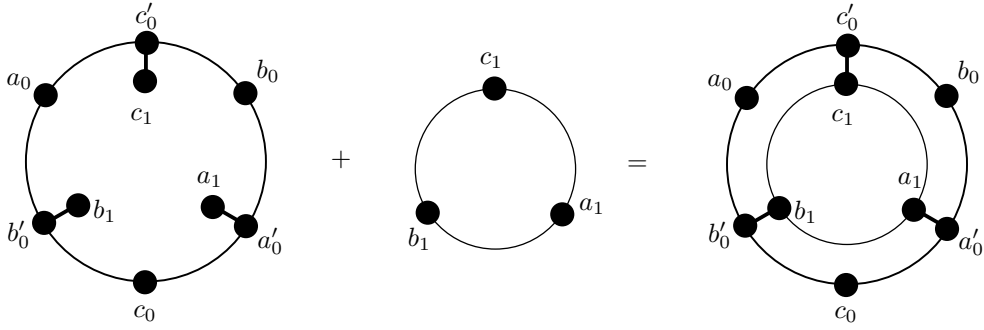}
    \caption{The node labels in the recursive construction}\label{fig:labelsrecursive}
    
\end{figure}

To show probability concentration, %
we will need to compute the number of choices a simple cycle will have when passing through an $R_d$ gadget, as well as the number of simple cycles internal to an $R_d$ gadget. The latter we know how to describe as $|SC(R_d)|$, and the former is captured by the following definition:

\begin{defn}[Simple boundary links]
We call the simple paths in $R_d$ that go from any two points $x \not = y \in \{a_0, b_0, c_0\}$ \emph{simple boundary links}, denoted $SBL(R_d;x,y)$, where $\{a_0, b_0, c_0\}$ are as in \Cref{defn:RDDef}. We denote $SBL(R_d) := SBL(R_d; a_0, b_0)$. %
\end{defn}

Because of the rotational symmetry of the gadget, $|SBL(R)|$ does not change if we choose a different two element subset of $\{a_0,b_0,c_0\}$ as the start and stop vertices. %
We next introduce notation that will be useful when computing $|SC(R_n)|$ and $|SBL(R_n)|$:

\begin{defn}[Simple paths from $X$ to $Y$]
For any graph $G$, with $X , Y \subseteq V(G)$, we let $SP_{X,Y}(G)$ denote the set of simple paths in $G$ that start in $X$, and stop at the first positive time they reach $Y$: $$SP_{X,Y}(G) = \{ \gamma = (x_0, x_1, \ldots, x_n) : x_i \in V(G), \{x_i, x_{i+1} \} \in E(G), x_0 \in X, x_n \in Y, x_i \not \in Y \text{ for } 0 < i < n \}.$$
\end{defn}

\begin{thm}\label{thm:RDUniformCounts}
Let $R_d$ be as in \Cref{defn:RDDef}. Then: %

\begin{align}
    &|SC(R_d)| =  \frac{1}{4} ( 3 \cdot 5^{d + 1} - 8 d-11)\label{eq:SCformula}\\
    5^d \leq &|SC(R_d)| \leq 5^{d+1}\label{eq:SCbounds}\\
    &|SBL(R_d)| = \frac{1}{2} (5^{d+1} - 1) \label{eq:SBLformula}\\
    5^d \leq &|SBL(R_d)| \leq 5^{d+1}\label{eq:SBLbounds}.
\end{align}

\end{thm}
\begin{proof}[Proof of \Cref{eq:SCformula}]
We partition the simple cycles in $R_d$ into those that touch $C_d$ and those that do not. Those simple cycles that do not touch $C_d$ can be identified with simple cycles in $R_{d-1}$. To describe the simple cycles that touch $C_d$, we start by defining $S_d = \{ X \in SP_{V(C_d), V(C_d)}( R_d) : X \cap E( C_d)  = \emptyset \}$. %
Among the cycles that touch $C_d$, there is $C_d$ itself, and there are the cycles that can be decomposed into an element of $SP_{ V(C_d) , V(C_d)} ( R_d[ C_d] )$ along with an element of $S_d$, as in \Cref{fig:simplecyclerecursivedecomposition}. %
Thus, $SC(R_d) = SC_{d-1} + 1 + 2S_d$.

\begin{figure}
    \centering
    \begin{tabular}{cc}
    \def\svgscale{.2}{
    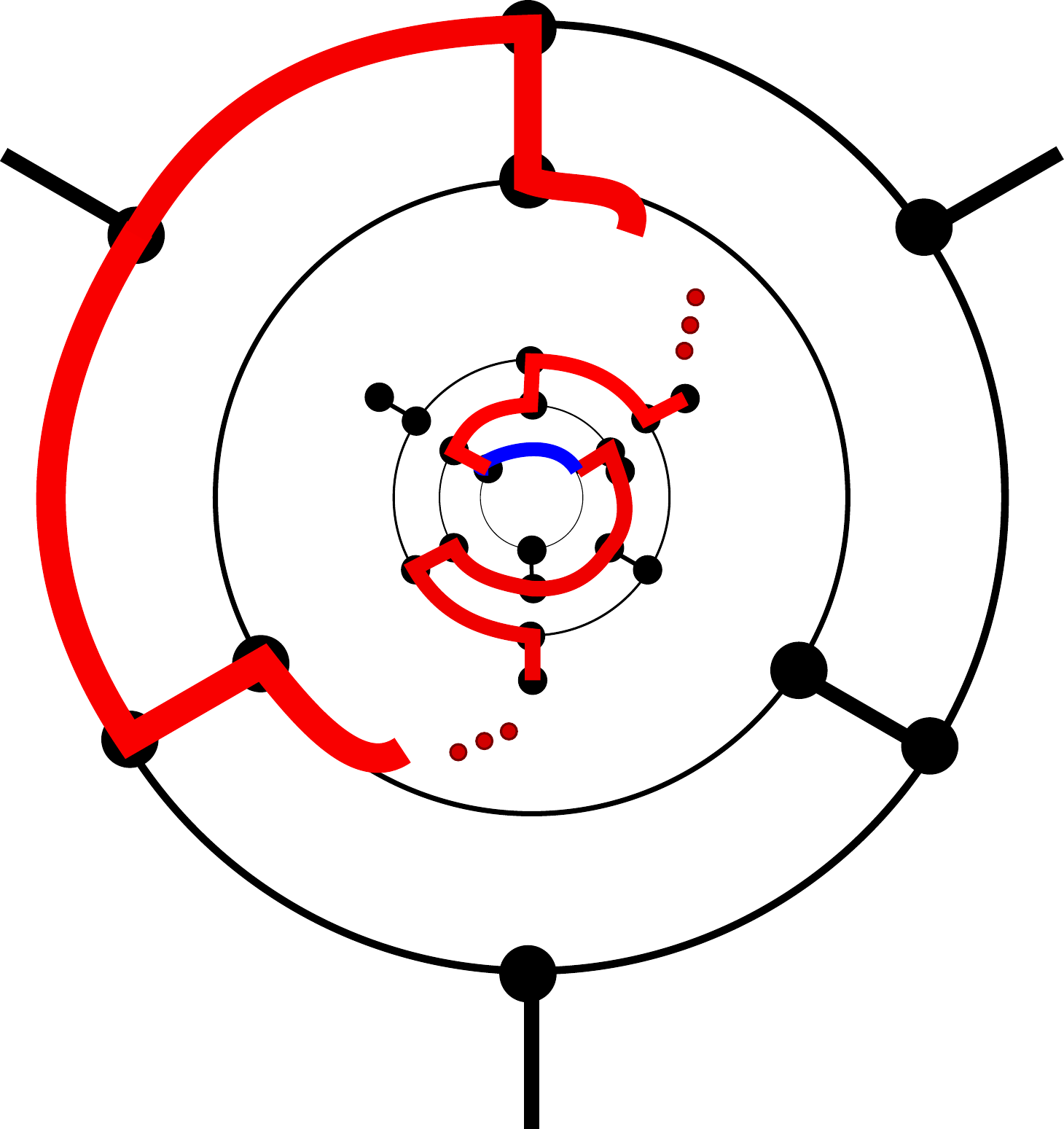} &    \def\svgscale{.17}{ 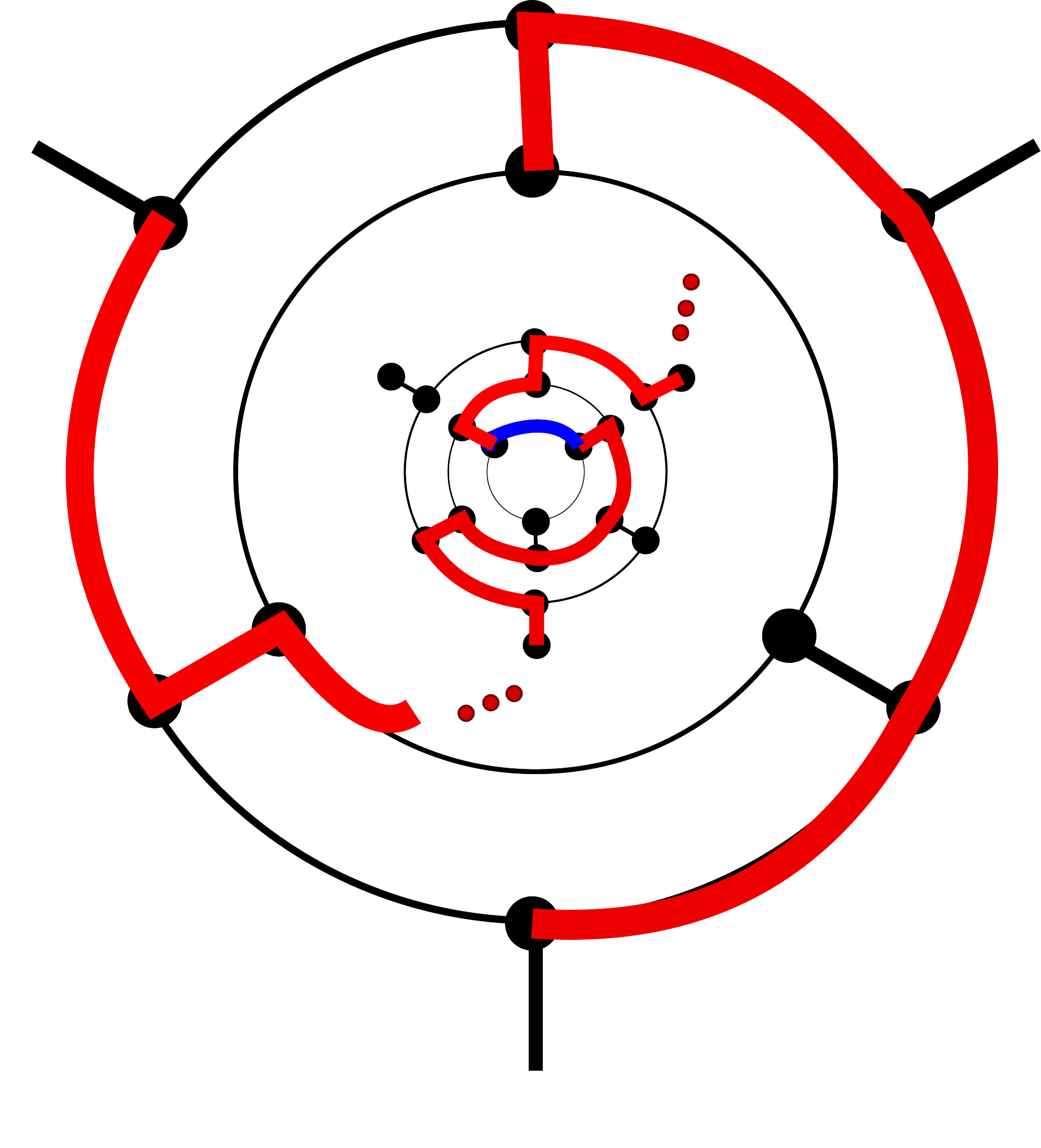 }
    \end{tabular}
    \caption{The decompositions used for \Cref{eq:SCformula} and \Cref{eq:SBLformula}. The $S_d$ (resp. $D_d$) part is in red. Observe that there  are two options for the simple path $R_d[C_d]$, one of which is colored blue.}\label{fig:simplecyclerecursivedecomposition}\label{fig:RecursiveGadgetSBLCounting}
    
\end{figure}

It can be checked that $S_d = 5  S_{d-1} + 6 $, by analyzing the number of ways to extend an element of $S_{d-1}$ to an element of $S_d$ and by accounting for the six elements of $S_d$ not obtained by extensions of $S_{d-1}$, as in \Cref{fig:6Extensions}. This second calculation also shows that $S_1 = 6$. We can solve this recurrence relation to conclude that $S_d = (3/2) (5^d - 1)$. Hence, from $SC_d = SC_{d-1} + 1 + 2S_d$ we have that $SC_d = SC_{d-1} + 1 + 3(5^d - 1)$. As $SC_0 = 1$, we conclude that $SC_d = 1/4 (3 \cdot 5^{1 + d} - 8 d-11)$.%
\end{proof}
\begin{proof}[Proof of \Cref{eq:SBLformula}]

We partition the simple boundary links in $R_d$ into those that pass through $C_d$ and those that do not. Those elements of $SBL(R_d)$ that do not touch $C_d$ can be identified with $SBL ( R_{d-1})$. To analyze those that pass through $C_d$, we define $D_d$ to be the set of pairs of disjoint simple paths, one from $a_0$ that stops at the first point it touches $C_d$, and the other from $b_0$ that stops at the first point it touches $C_d$.  The elements of $SBL(R_d)$ that touch $C_d$ can be decomposed into an element $\gamma$ of $D_d$ and one of the two simple paths in $R_d[C_d]$ that connect the points where $\gamma$ meets $C_d$, as in \Cref{fig:RecursiveGadgetSBLCounting}. Thus, $BL_{d+1} = BL_d + 2 D_d$.

We can compute that $D_{d+1} = 5D_d$, by analyzing how elements of $D_d$ can be extended to elements of $D_{d+1}$, as in \Cref{fig:TheFiveExtensions}. As $D_0 = 1$, we have that $D_d = 5^d$. From $BL_{d+1} = BL_d + 2 D_d$, we have that $BL_{d+1} = BL_d + 2 ( 5^{d + 1})$. As $BL_0  =2$, we can solve the recurrence to find that $BL_d =  (1/2)( 5^{d+1} - 1)$.
\end{proof}

\begin{proof}[Proof of \Cref{eq:SCbounds} and \Cref{eq:SBLbounds}]These follow directly from \Cref{eq:SCformula,eq:SBLformula}.%
\end{proof}

\begin{figure}
    \centering
    \begin{tabular}{cc}
    \def\svgscale{.3}{
    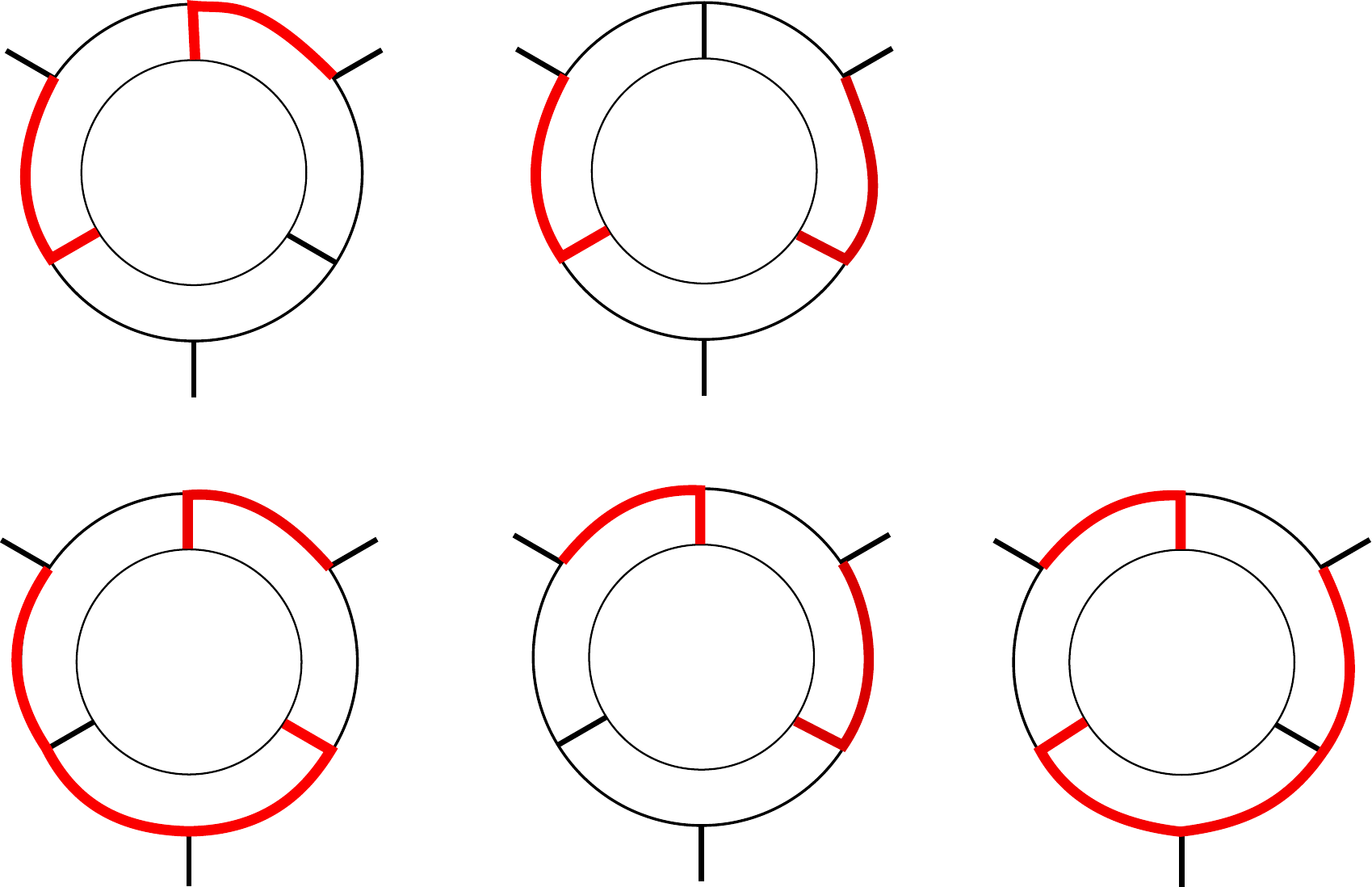} &\def\svgscale{.3}{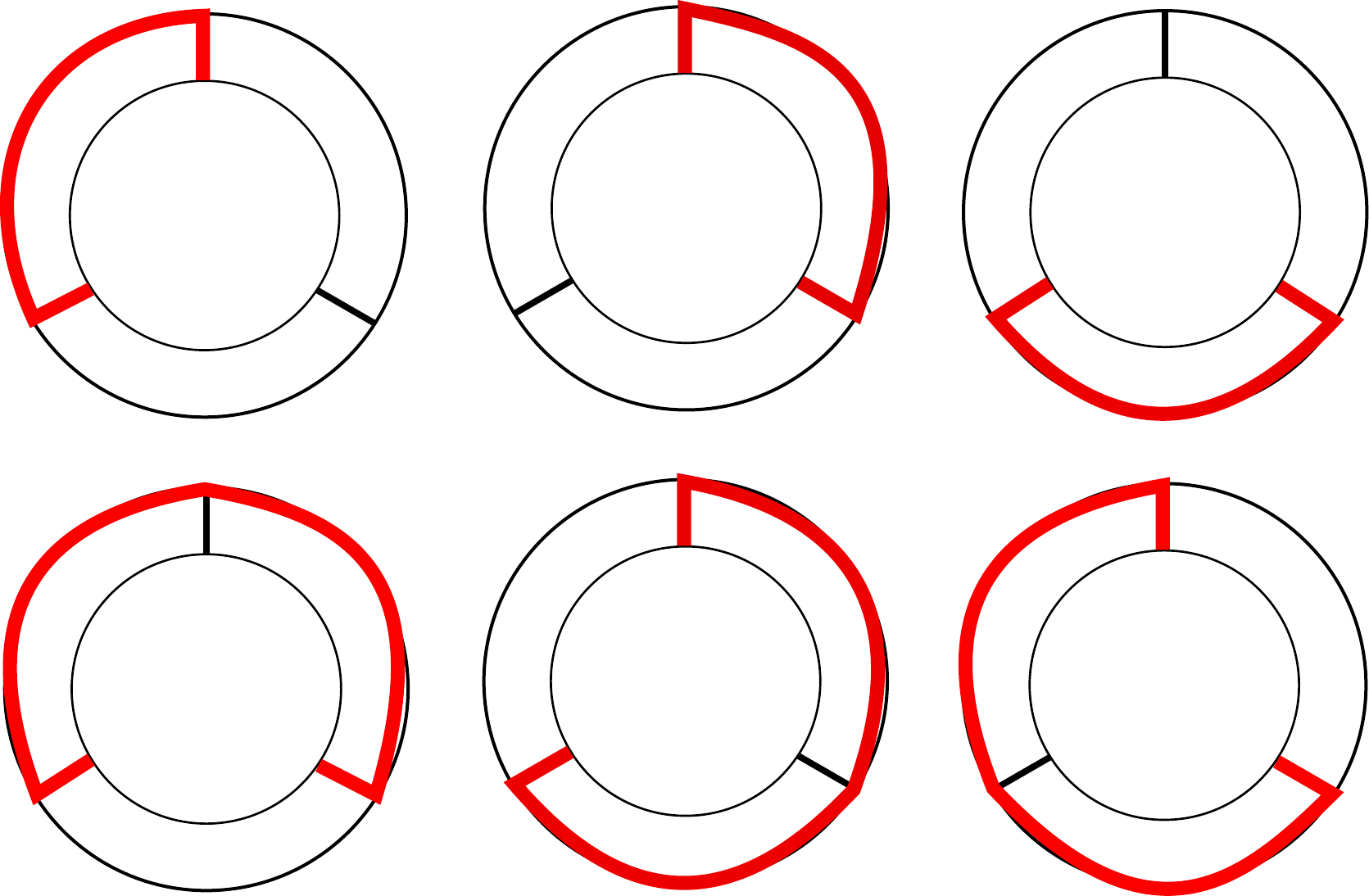}\\
    a) & b) 
    \end{tabular}
    \caption{a) The five ways to extend an element of $D_n$ to one of $D_{n+1}$, or one of $S_n$ to one of $S_{n+1}$. b) The $6$ elements of $S_n$ that are not extensions of elements of $S_{n-1}$. The inner $3$-cycle is $C_n$; most of $R_n$ is not pictured.
    }\label{fig:TheFiveExtensions}\label{fig:6Extensions}
    
\end{figure}

\subsubsection{Intractability on Maximal plane Graphs}\label{Max2Part}

We establish the analogous construction to replacing edges by chains of bigons and state a few results that we will need to finish the intractability proof.

\begin{defn}[$G \to R_d(G)$ vertex replacement construction]\label{defn:GtoRdGconstruction}
Given a cubic graph $G$, we let $R_d(G)$ denote the graph obtained by keeping the edges and replacing each vertex of $G$ with a copy of $R_d$ (\Cref{defn:RDDef}). %
\end{defn}

\begin{prop}\label{prop:Rdkeeps3CCP}
The construction $G \mapsto R_d(G)$ sends $\mathscr{C}_{m}$ to $\mathscr{C}_{3m}$.
\end{prop}
\noindent See \Cref{section:preserves3CCP} for proof of this proposition. %

With the $R_d$ construction in place, we analyze the relationship between $R_d(G)$ and $G$ to show that $R_d$ can be used to concentrate probability onto the longer simple cycles of $G$.

\begin{defn}[Original edges, projection map]
There is a natural inclusion map $i : E(G) \to E( R_d(G))$. We will call the edges in the image of $i$ the \emph{original edges}. This lets us define a map $\pi_d: 2^ { E(R_d(G))} \to 2^{E(G)}$ by $\pi_d(X) = i^{-1}(X)$.
\end{defn}

\begin{lem}\label{surjective}
For any cubic graph $G$, $\im(\pi_d) =  SC(G) \cup \{ \emptyset \}$ for $d \geq 0$.
\end{lem}
\begin{proof}%
That $SC(H) \cup \{ \emptyset \}$ is contained in the image is straightforward. 
Now, let $\beta \in SC(R_d(G))$. %
If $\pi(\beta)$ has degree $1$ or $3$ at a node $v \in V(G)$, then $\beta$ had an odd degree node in $R_d(v)$, which impossible as $\beta$ is a simple cycle.
Moreover, $\pi(\beta)$ is connected. %
Since the simple cycles can be characterized as the non-empty connected edge subgraphs such that all nodes are degree $2$, this concludes the proof.
\end{proof}

We now compute the probability concentration lemma necessary for applying \Cref{lem:luckyguess}:

\begin{lem}\label{lem:3CCPconcentration}
Suppose that $C$ is a Hamiltonian cycle of $G$, where $G$ has $n \geq 2$ nodes. If $d \geq n^2 + n + m $, and $X$ is a uniform sample from $SC( R_d(G))$, then $\mathbb{P} ( \pi_d(X) \text{ is a Hamiltonian cycle of $G$} )  \geq \frac{ 5^m}{1 + 5^m}$.
\end{lem}
\begin{proof}

Let $H$ be the set of Hamiltonian cycles of $G$. %
From \Cref{eq:SCbounds} and \Cref{eq:SBLbounds} it follows that $|\pi_d^{-1} ( \emptyset \cup SC(G) \setminus H)| \leq 2^{n^2} 5^{(d + 1)(n-1)} + n 5^{d + 1} \leq  2^{n^2  + 1} 5^{(d + 1)(n-1)}$ and $|\pi_d^{-1}(H)| \geq 5^{dn}$. 
For $d \geq n^2 + n + m$, $5^{dn} \geq  5^m 2^{n^2 + 1} 5^{n-1} 5^{d(n-1)}$. The claim now follows by \Cref{lem:HgeqDN}.
\end{proof}

\begin{thm} If the Hamiltonian cycle problem is NP-complete on $\mathscr{C}_d$, then the problem of uniformly sampling simple cycles is intractable on the class of graphs $ \mathscr{C}_{3d}$.\label{thm:generalintractability3CCP}
\end{thm}
\begin{proof}
We follow the notation of \Cref{lem:luckyguess}. For fixed $m$, we take $d = n^2 + n + m$. Define $B$ by $B(G) = R_d(G)$, and $M$ by the map $\pi_d$, and $Q$ is the set of Hamiltonian Cycles of $G$. \Cref{lem:3CCPconcentration} assures that the conditions for \Cref{lem:luckyguess} are satisfied.
\end{proof}

\begin{cor}\label{thm:IntractibilityofSCSamplingonC528}
The $SC(G)$ uniform sampling problem is intractable on $\mathscr{C}_{ 531}$.
\end{cor}
\begin{proof}
Immediate from \Cref{thm:facebounded3CCP} and \Cref{thm:generalintractability3CCP}.
\end{proof}

We can now prove the main result of this section:

\begin{thm}\label{thm:triangulationhard}
The $P_2(G)$ uniform sampling problem is intractable on the class of maximal plane graphs with vertex degree $\leq 531$. %
\end{thm}
\begin{proof}
Since the dual graphs of those graphs in $\mathscr{C}_{531}$ are exactly the maximal plane graphs with vertex degree $\leq 531$ and since simple cycles correspond bijectively to (unordered) connected 2-partitions under that duality, this result follows from \Cref{thm:IntractibilityofSCSamplingonC528}.
\end{proof}

%% file: Images/3ORDetailedMarkedAttachingNodes.pdf_tex
\begingroup%
  \makeatletter%
  \providecommand\color[2][]{%
    \errmessage{(Inkscape) Color is used for the text in Inkscape, but the package 'color.sty' is not loaded}%
    \renewcommand\color[2][]{}%
  }%
  \providecommand\transparent[1]{%
    \errmessage{(Inkscape) Transparency is used (non-zero) for the text in Inkscape, but the package 'transparent.sty' is not loaded}%
    \renewcommand\transparent[1]{}%
  }%
  \providecommand\rotatebox[2]{#2}%
  \newcommand*\fsize{\dimexpr\f@size pt\relax}%
  \newcommand*\lineheight[1]{\fontsize{\fsize}{#1\fsize}\selectfont}%
  \ifx\svgwidth\undefined%
    \setlength{\unitlength}{450.66965982bp}%
    \ifx\svgscale\undefined%
      \relax%
    \else%
      \setlength{\unitlength}{\unitlength * \real{\svgscale}}%
    \fi%
  \else%
    \setlength{\unitlength}{\svgwidth}%
  \fi%
  \global\let\svgwidth\undefined%
  \global\let\svgscale\undefined%
  \makeatother%
  \begin{picture}(1,0.87070432)%
    \lineheight{1}%
    \setlength\tabcolsep{0pt}%
    \put(0,0){\includegraphics[width=\unitlength,page=1]{3ORDetailedMarkedAttachingNodes.pdf}}%
  \end{picture}%
\endgroup%

%% file: Images/Detailed3ORInsertion.pdf_tex
\begingroup%
  \makeatletter%
  \providecommand\color[2][]{%
    \errmessage{(Inkscape) Color is used for the text in Inkscape, but the package 'color.sty' is not loaded}%
    \renewcommand\color[2][]{}%
  }%
  \providecommand\transparent[1]{%
    \errmessage{(Inkscape) Transparency is used (non-zero) for the text in Inkscape, but the package 'transparent.sty' is not loaded}%
    \renewcommand\transparent[1]{}%
  }%
  \providecommand\rotatebox[2]{#2}%
  \newcommand*\fsize{\dimexpr\f@size pt\relax}%
  \newcommand*\lineheight[1]{\fontsize{\fsize}{#1\fsize}\selectfont}%
  \ifx\svgwidth\undefined%
    \setlength{\unitlength}{481.2913722bp}%
    \ifx\svgscale\undefined%
      \relax%
    \else%
      \setlength{\unitlength}{\unitlength * \real{\svgscale}}%
    \fi%
  \else%
    \setlength{\unitlength}{\svgwidth}%
  \fi%
  \global\let\svgwidth\undefined%
  \global\let\svgscale\undefined%
  \makeatother%
  \begin{picture}(1,1.72180409)%
    \lineheight{1}%
    \setlength\tabcolsep{0pt}%
    \put(0,0){\includegraphics[width=\unitlength,page=1]{Detailed3ORInsertion.pdf}}%
  \end{picture}%
\endgroup%

%% file: Images/RecursiveGadget.pdf_tex
\begingroup%
  \makeatletter%
  \providecommand\color[2][]{%
    \errmessage{(Inkscape) Color is used for the text in Inkscape, but the package 'color.sty' is not loaded}%
    \renewcommand\color[2][]{}%
  }%
  \providecommand\transparent[1]{%
    \errmessage{(Inkscape) Transparency is used (non-zero) for the text in Inkscape, but the package 'transparent.sty' is not loaded}%
    \renewcommand\transparent[1]{}%
  }%
  \providecommand\rotatebox[2]{#2}%
  \newcommand*\fsize{\dimexpr\f@size pt\relax}%
  \newcommand*\lineheight[1]{\fontsize{\fsize}{#1\fsize}\selectfont}%
  \ifx\svgwidth\undefined%
    \setlength{\unitlength}{535.493205bp}%
    \ifx\svgscale\undefined%
      \relax%
    \else%
      \setlength{\unitlength}{\unitlength * \real{\svgscale}}%
    \fi%
  \else%
    \setlength{\unitlength}{\svgwidth}%
  \fi%
  \global\let\svgwidth\undefined%
  \global\let\svgscale\undefined%
  \makeatother%
  \begin{picture}(1,0.40669431)%
    \lineheight{1}%
    \setlength\tabcolsep{0pt}%
    \put(0,0){\includegraphics[width=\unitlength,page=1]{RecursiveGadget.pdf}}%
    \put(0.14393614,0.00836055){\color[rgb]{0,0,0}\makebox(0,0)[t]{\lineheight{1.25}\smash{\begin{tabular}[t]{c}$R_0$ \end{tabular}}}}%
    \put(0.5020736,0.00872265){\color[rgb]{0,0,0}\makebox(0,0)[t]{\lineheight{1.25}\smash{\begin{tabular}[t]{c}$R_1$ \end{tabular}}}}%
    \put(0.85312119,0.00924849){\color[rgb]{0,0,0}\makebox(0,0)[t]{\lineheight{1.25}\smash{\begin{tabular}[t]{c}$R_2$ \end{tabular}}}}%
    \put(0,0){\includegraphics[width=\unitlength,page=2]{RecursiveGadget.pdf}}%
  \end{picture}%
\endgroup%

%% file: Images/RecursiveGluing.pdf_tex
\begingroup%
  \makeatletter%
  \providecommand\color[2][]{%
    \errmessage{(Inkscape) Color is used for the text in Inkscape, but the package 'color.sty' is not loaded}%
    \renewcommand\color[2][]{}%
  }%
  \providecommand\transparent[1]{%
    \errmessage{(Inkscape) Transparency is used (non-zero) for the text in Inkscape, but the package 'transparent.sty' is not loaded}%
    \renewcommand\transparent[1]{}%
  }%
  \providecommand\rotatebox[2]{#2}%
  \newcommand*\fsize{\dimexpr\f@size pt\relax}%
  \newcommand*\lineheight[1]{\fontsize{\fsize}{#1\fsize}\selectfont}%
  \ifx\svgwidth\undefined%
    \setlength{\unitlength}{602.10880039bp}%
    \ifx\svgscale\undefined%
      \relax%
    \else%
      \setlength{\unitlength}{\unitlength * \real{\svgscale}}%
    \fi%
  \else%
    \setlength{\unitlength}{\svgwidth}%
  \fi%
  \global\let\svgwidth\undefined%
  \global\let\svgscale\undefined%
  \makeatother%
  \begin{picture}(1,0.29041532)%
    \lineheight{1}%
    \setlength\tabcolsep{0pt}%
    \put(0,0){\includegraphics[width=\unitlength,page=1]{RecursiveGluing.pdf}}%
    \put(0.08710246,0.21464913){\color[rgb]{0,0,0}\makebox(0,0)[rt]{\lineheight{1.25}\smash{\begin{tabular}[t]{r}$a_0$\end{tabular}}}}%
    \put(0.28755772,0.22043515){\color[rgb]{0,0,0}\makebox(0,0)[lt]{\lineheight{1.25}\smash{\begin{tabular}[t]{l}$b_0$\end{tabular}}}}%
    \put(0.18795581,0.00694298){\color[rgb]{0,0,0}\makebox(0,0)[t]{\lineheight{1.25}\smash{\begin{tabular}[t]{c}$c_0$\end{tabular}}}}%
    \put(0.08537355,0.07504152){\color[rgb]{0,0,0}\makebox(0,0)[rt]{\lineheight{1.25}\smash{\begin{tabular}[t]{r}$b'_0$\end{tabular}}}}%
    \put(0.18961343,0.27037893){\color[rgb]{0,0,0}\makebox(0,0)[t]{\lineheight{1.25}\smash{\begin{tabular}[t]{c}$c'_0$\end{tabular}}}}%
    \put(0.2884213,0.06903705){\color[rgb]{0,0,0}\makebox(0,0)[lt]{\lineheight{1.25}\smash{\begin{tabular}[t]{l}$a'_0$\end{tabular}}}}%
    \put(0.18551365,0.18355569){\color[rgb]{0,0,0}\makebox(0,0)[t]{\lineheight{1.25}\smash{\begin{tabular}[t]{c}$c_1$\end{tabular}}}}%
    \put(0.1415427,0.09548437){\color[rgb]{0,0,0}\makebox(0,0)[lt]{\lineheight{1.25}\smash{\begin{tabular}[t]{l}$b_1$\end{tabular}}}}%
    \put(0.24577811,0.12653065){\color[rgb]{0,0,0}\makebox(0,0)[t]{\lineheight{1.25}\smash{\begin{tabular}[t]{c}$a_1$\end{tabular}}}}%
    \put(0,0){\includegraphics[width=\unitlength,page=2]{RecursiveGluing.pdf}}%
    \put(0.4976572,0.23345616){\color[rgb]{0,0,0}\makebox(0,0)[t]{\lineheight{1.25}\smash{\begin{tabular}[t]{c}$c_1$\end{tabular}}}}%
    \put(0.5760863,0.10299809){\color[rgb]{0,0,0}\makebox(0,0)[lt]{\lineheight{1.25}\smash{\begin{tabular}[t]{l}$a_1$\end{tabular}}}}%
    \put(0.354094,0.14203723){\color[rgb]{0,0,0}\makebox(0,0)[lt]{\lineheight{1.25}\smash{\begin{tabular}[t]{l}+\end{tabular}}}}%
    \put(0.61915739,0.13870733){\color[rgb]{0,0,0}\makebox(0,0)[lt]{\lineheight{1.25}\smash{\begin{tabular}[t]{l}=\end{tabular}}}}%
    \put(0.44363394,0.07052823){\color[rgb]{0,0,0}\makebox(0,0)[rt]{\lineheight{1.25}\smash{\begin{tabular}[t]{r}$b_1$\end{tabular}}}}%
    \put(0.71388708,0.21379902){\color[rgb]{0,0,0}\makebox(0,0)[rt]{\lineheight{1.25}\smash{\begin{tabular}[t]{r}$a_0$\end{tabular}}}}%
    \put(0.91434229,0.21958507){\color[rgb]{0,0,0}\makebox(0,0)[lt]{\lineheight{1.25}\smash{\begin{tabular}[t]{l}$b_0$\end{tabular}}}}%
    \put(0.81474041,0.00621952){\color[rgb]{0,0,0}\makebox(0,0)[t]{\lineheight{1.25}\smash{\begin{tabular}[t]{c}$c_0$\end{tabular}}}}%
    \put(0.71215812,0.07419138){\color[rgb]{0,0,0}\makebox(0,0)[rt]{\lineheight{1.25}\smash{\begin{tabular}[t]{r}$b'_0$\end{tabular}}}}%
    \put(0.81639811,0.26952882){\color[rgb]{0,0,0}\makebox(0,0)[t]{\lineheight{1.25}\smash{\begin{tabular}[t]{c}$c'_0$\end{tabular}}}}%
    \put(0.9152059,0.06818687){\color[rgb]{0,0,0}\makebox(0,0)[lt]{\lineheight{1.25}\smash{\begin{tabular}[t]{l}$a'_0$\end{tabular}}}}%
    \put(0.81229827,0.18270562){\color[rgb]{0,0,0}\makebox(0,0)[t]{\lineheight{1.25}\smash{\begin{tabular}[t]{c}$c_1$\end{tabular}}}}%
    \put(0.76832735,0.09463429){\color[rgb]{0,0,0}\makebox(0,0)[lt]{\lineheight{1.25}\smash{\begin{tabular}[t]{l}$b_1$\end{tabular}}}}%
    \put(0.87256266,0.12568046){\color[rgb]{0,0,0}\makebox(0,0)[t]{\lineheight{1.25}\smash{\begin{tabular}[t]{c}$a_1$\end{tabular}}}}%
  \end{picture}%
\endgroup%

%% file: Images/SimpleCycleRecursiveDecomposition.pdf_tex
\begingroup%
  \makeatletter%
  \providecommand\color[2][]{%
    \errmessage{(Inkscape) Color is used for the text in Inkscape, but the package 'color.sty' is not loaded}%
    \renewcommand\color[2][]{}%
  }%
  \providecommand\transparent[1]{%
    \errmessage{(Inkscape) Transparency is used (non-zero) for the text in Inkscape, but the package 'transparent.sty' is not loaded}%
    \renewcommand\transparent[1]{}%
  }%
  \providecommand\rotatebox[2]{#2}%
  \newcommand*\fsize{\dimexpr\f@size pt\relax}%
  \newcommand*\lineheight[1]{\fontsize{\fsize}{#1\fsize}\selectfont}%
  \ifx\svgwidth\undefined%
    \setlength{\unitlength}{420.41730802bp}%
    \ifx\svgscale\undefined%
      \relax%
    \else%
      \setlength{\unitlength}{\unitlength * \real{\svgscale}}%
    \fi%
  \else%
    \setlength{\unitlength}{\svgwidth}%
  \fi%
  \global\let\svgwidth\undefined%
  \global\let\svgscale\undefined%
  \makeatother%
  \begin{picture}(1,1.06092587)%
    \lineheight{1}%
    \setlength\tabcolsep{0pt}%
    \put(0,0){\includegraphics[width=\unitlength,page=1]{SimpleCycleRecursiveDecomposition.pdf}}%
  \end{picture}%
\endgroup%

%% file: Images/RecursiveGadgetSBLCounting.pdf_tex
\begingroup%
  \makeatletter%
  \providecommand\color[2][]{%
    \errmessage{(Inkscape) Color is used for the text in Inkscape, but the package 'color.sty' is not loaded}%
    \renewcommand\color[2][]{}%
  }%
  \providecommand\transparent[1]{%
    \errmessage{(Inkscape) Transparency is used (non-zero) for the text in Inkscape, but the package 'transparent.sty' is not loaded}%
    \renewcommand\transparent[1]{}%
  }%
  \providecommand\rotatebox[2]{#2}%
  \newcommand*\fsize{\dimexpr\f@size pt\relax}%
  \newcommand*\lineheight[1]{\fontsize{\fsize}{#1\fsize}\selectfont}%
  \ifx\svgwidth\undefined%
    \setlength{\unitlength}{531.12463217bp}%
    \ifx\svgscale\undefined%
      \relax%
    \else%
      \setlength{\unitlength}{\unitlength * \real{\svgscale}}%
    \fi%
  \else%
    \setlength{\unitlength}{\svgwidth}%
  \fi%
  \global\let\svgwidth\undefined%
  \global\let\svgscale\undefined%
  \makeatother%
  \begin{picture}(1,1.1029527)%
    \lineheight{1}%
    \setlength\tabcolsep{0pt}%
    \put(0,0){\includegraphics[width=\unitlength,page=1]{RecursiveGadgetSBLCounting.pdf}}%
    \put(0.16608698,1.00582433){\color[rgb]{0,0,0}\makebox(0,0)[rt]{\lineheight{1.25}\smash{\begin{tabular}[t]{r}$a_0$\end{tabular}}}}%
    \put(0.4928986,0.01332116){\color[rgb]{0,0,0}\makebox(0,0)[lt]{\lineheight{1.25}\smash{\begin{tabular}[t]{l}$b_0$\end{tabular}}}}%
  \end{picture}%
\endgroup%

%% file: Images/TheFiveExtensions.pdf_tex
\begingroup%
  \makeatletter%
  \providecommand\color[2][]{%
    \errmessage{(Inkscape) Color is used for the text in Inkscape, but the package 'color.sty' is not loaded}%
    \renewcommand\color[2][]{}%
  }%
  \providecommand\transparent[1]{%
    \errmessage{(Inkscape) Transparency is used (non-zero) for the text in Inkscape, but the package 'transparent.sty' is not loaded}%
    \renewcommand\transparent[1]{}%
  }%
  \providecommand\rotatebox[2]{#2}%
  \newcommand*\fsize{\dimexpr\f@size pt\relax}%
  \newcommand*\lineheight[1]{\fontsize{\fsize}{#1\fsize}\selectfont}%
  \ifx\svgwidth\undefined%
    \setlength{\unitlength}{489.02470565bp}%
    \ifx\svgscale\undefined%
      \relax%
    \else%
      \setlength{\unitlength}{\unitlength * \real{\svgscale}}%
    \fi%
  \else%
    \setlength{\unitlength}{\svgwidth}%
  \fi%
  \global\let\svgwidth\undefined%
  \global\let\svgscale\undefined%
  \makeatother%
  \begin{picture}(1,0.6465524)%
    \lineheight{1}%
    \setlength\tabcolsep{0pt}%
    \put(0,0){\includegraphics[width=\unitlength,page=1]{TheFiveExtensions.pdf}}%
  \end{picture}%
\endgroup%

%% file: Images/6Extensions.pdf_tex
\begingroup%
  \makeatletter%
  \providecommand\color[2][]{%
    \errmessage{(Inkscape) Color is used for the text in Inkscape, but the package 'color.sty' is not loaded}%
    \renewcommand\color[2][]{}%
  }%
  \providecommand\transparent[1]{%
    \errmessage{(Inkscape) Transparency is used (non-zero) for the text in Inkscape, but the package 'transparent.sty' is not loaded}%
    \renewcommand\transparent[1]{}%
  }%
  \providecommand\rotatebox[2]{#2}%
  \newcommand*\fsize{\dimexpr\f@size pt\relax}%
  \newcommand*\lineheight[1]{\fontsize{\fsize}{#1\fsize}\selectfont}%
  \ifx\svgwidth\undefined%
    \setlength{\unitlength}{481.07587796bp}%
    \ifx\svgscale\undefined%
      \relax%
    \else%
      \setlength{\unitlength}{\unitlength * \real{\svgscale}}%
    \fi%
  \else%
    \setlength{\unitlength}{\svgwidth}%
  \fi%
  \global\let\svgwidth\undefined%
  \global\let\svgscale\undefined%
  \makeatother%
  \begin{picture}(1,0.65481212)%
    \lineheight{1}%
    \setlength\tabcolsep{0pt}%
    \put(0,0){\includegraphics[width=\unitlength,page=1]{6Extensions.pdf}}%
  \end{picture}%
\endgroup%

%% file: Sections/2Complexity/7kPartitions.tex
\subsection{Intractability of uniformly sampling connected $k$-partitions}\label{section:kpartitionshard}

To obtain an intractability theorem about uniformly sampling connected $k$-partitions, we follow a similar approach as in previous sections. First, we recall a plane duality theorem, and then we prove that a relevant optimization problem is $\NP$-complete. Finally, we introduce a gadget that concentrates samples on the certificates to that problem.

We remind the reader that $\mathscr{P}_k(G)$ denotes the set of unordered $k$-partitions of $G$, such that each block induces a connected subgraph. In this section, we will show that uniformly sampling from $\mathscr{P}_k(G)$ is intractable on the class of planar graphs. We will also show intractability for a family of probability distributions that weights partitions according to the size of their boundary. \input{Sections/2Complexity/7bshort.tex}

\subsubsection{The corresponding $\NP$-complete problem}

We will show that it is $\NP$-complete to decide if $\mathscr{P}^*_k(G)$ has length maximizing elements.%

\begin{defn}[Spanning edge set]
Let $J$ be a subset of edges of a graph $G$. We say that $J$ \emph{spans} $G$ if every node of $G$ is incident to some edge of $J$.
\end{defn}

\begin{prop}\label{lem:maxnumedgesandhomology}
Let $G = (V,E)$ be a graph. The maximum number of edges any set $J \subseteq E$ with $h_1( G[J]) = k - 1$ can have is $|V| + k - 2$. Moreover, a $J \subseteq E$ with $h_1( G[J]) = k - 1$ has $|V| + k - 2$ edges if and only if $G[J]$ has one component and spans $G$.
\end{prop}
\begin{proof}
Let $E_J,V_J$ be the number of edges and vertices of $G[J]$, respectively.  From $E_J - V_J = h_1(J) - h_0(J)$ (\cref{prop:ranknullity}), we have $E_J - V_J = k - 1 - h_0(J)$. Thus, $E_J = V_J + k - 1 - h_0(J)%
\leq |V| + k - 2$, as $h_0(J) \geq 1$ and $V_J \leq |V|$. This establishes the upper bound. Moreover, these inequalities become equalities if and only if $V_J = |V|$ and $h_0(G[J]) = 1$, which is to say, if and only if $J$ spans and has one component.
\end{proof}

To define the $\NP$-complete decision problem we will use for the sampling intractability proof, we single out the elements of $\mathscr{P}_k(G)$ that achieve the upper bound of \Cref{lem:maxnumedgesandhomology} and define a corresponding decision problem.

\begin{defn}[Maximal dual $k$-partition]
Let $\mathscr{P}_k^*(G)_{m}$ be the subset of $\mathscr{P}_k^*(G)$ consisting of the subgraphs that have $|V(G)| + k - 2$ edges, which is the maximal number of edges possible.%
\end{defn}

\begin{computationalproblem}{PlanarMaxEdgesDualkpartition}\label{kpartitionNPC}

Input: A planar graph $G$

Output: YES if  $\mathscr{P}_k^*(G)_m \not = \emptyset$, NO, otherwise.
\end{computationalproblem}

We will prove that \textsc{PlanarMaxEdgesDualkpartition} is $\NP$-complete, by reducing from the Hamiltonian cycle problem on grid graphs:

\begin{thm}[\cite{itai1982hamilton}]
Let $G$ be a finite subgraph of the square grid graph $\mathbb{Z}^2$, wherein integer points are adjacent if and only if their Euclidean distance is $1$. Deciding whether $G$ has a Hamiltonian cycle is $\NP$-complete.
\end{thm}

\begin{prop}\label{prop:GridGraphReduction} The problem
\textsc{PlanarMaxEdgesDualkPartition} is $\NP$-complete.
\end{prop}
\begin{proof} The language of graphs that have $\mathscr{P}_k^*(G)_m \not = \emptyset$ is in $\NP$, since checking if a given set of edges is in $\mathscr{P}_k^*(G)_m$ can be done in polynomial time. We now show the reduction from Hamiltonicity of grid graphs. Let $G$ be a some subgraph of the grid graph, which we assume without loss of generality is $2$-connected, since otherwise it has no Hamiltonian cycle. %
Because $G$ is $2$-connected, the lexicographic upper-left %
most node $v$ of $G$ must have degree $2$. %
We build a new graph, $G'$, by removing $v$ and connecting its neighbors with a chain of $k-2$ diamonds, as in \Cref{fig:GridGraphReduction}.

We will show that $G'$ has an element of $P_k^*(G)_m$ if and only if $G$ has a Hamiltonian cycle. If $G$ has a Hamiltonian cycle, that cycle had to pass through both edges of $v$, hence we can replace those edges with the diamonds in $G'$. The result has $h_1 = k-1$, spans $G'$, and is connected. Thus it is an element of $P_k^*(G')_m$.

Going in the other direction, suppose that there is an $X \in \mathscr{P}_k^*(G')_m$. Since $X$ must span $G'$, $X$ must contain each node in the chain of diamonds. Moreover, since elements of $\mathscr{P}_k^*(G)$ have no bridge edges (\cref{lem:equivbridgeunion}), this implies that $X$ contains all the edges in that chain of diamonds. Thus, $k-2$ of $h_1(X)$ is accounted for by the diamonds, and for $h_1(X) = k-1$, the rest of $X$ must be a simple cycle, which spans $G \setminus v$ because $X$ spans $G'$. Replacing those diamonds by the path $\{a,v,b\}$, we obtain a Hamiltonian cycle of $G$.
\end{proof}

\begin{figure}[!htbp]
    \centering
    \def\svgscale{.5}{
    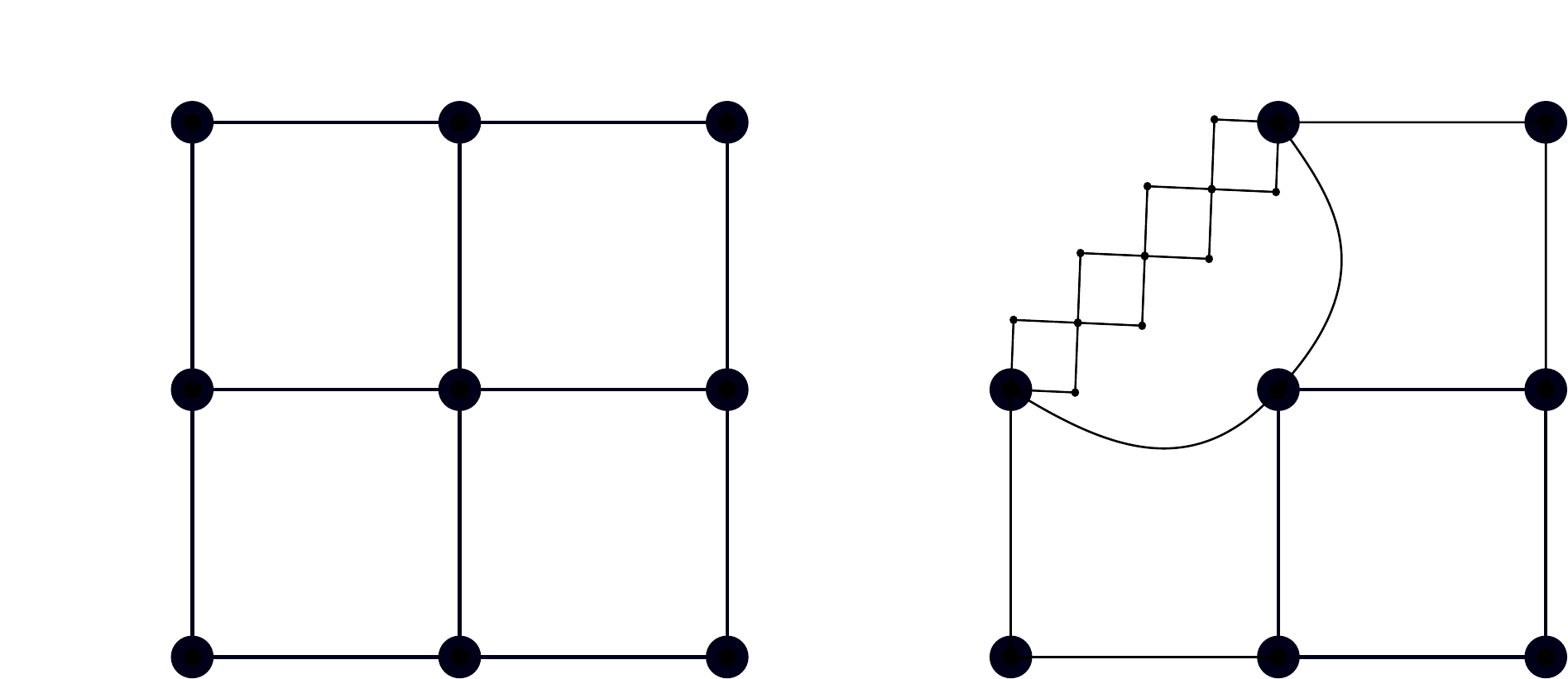}
    \caption{A step in the reduction in \Cref{prop:GridGraphReduction} %
    }
    \label{fig:GridGraphReduction}
\end{figure}

\subsubsection{The Chain of Dipoles Construction, Degeneration of $k$-Partitions} 

We will concern ourselves with the following sampling problem, where we fix $\lambda> 0$:

\begin{computationalproblem}{$\lambda$-sampling connected dual $k$-partitions}

Input: A graph $G$.

Output: An element of $\mathscr{P}^*_k(G)$, drawn according to the probability distribution that assigns a set $J \in \mathscr{P}^*_k(G)$ weight proportional to $\lambda^{|J|}$.

\end{computationalproblem}

To make some estimates, we give the probability distribution in this problem a name:

\begin{defn}[Measures $\nu_{\lambda}$ and $N_\lambda$]\label{defn:lambdameasures}
Let $G$ be a graph, and let $\lambda > 0$. For any collection of sets of edges $Y \subseteq 2^{E(G)}$, we define the measure $N_{\lambda}$ on $Y$ by $N_{\lambda}(J) =  \lambda^{ |J| }$ for all $J \in Y$. %
We define a probability measure $\nu_{\lambda}$ by normalizing $N_{\lambda}$: $\nu_{\lambda}(J) = \frac{ N_{\lambda}(J) }{N_{\lambda}(Y)}$. %
\end{defn}

To prove that sampling from $\nu_{\lambda}$ is intractable, instead of using chains of bigons like we did for the uniform distribution, we will use chains of order-$r$ dipoles, where $r$ will be chosen so that $r \lambda \geq 2$: %

\begin{defn}[Chain of order-$r$ dipoles]
Define $B_{r,d}(G)$ as the graph obtained from $G$ by subdividing each edge of $G$ into $d$ segments, and then replacing each edge of the resulting graph by $r$ parallel edges (i.e. order-$r$ dipoles). Let $B_{r,d}(e)$ denote the chain of $d$ order-$r$ dipoles that replaces the edge $e$.
\end{defn}

\begin{defn}[Dipole projection map]\label{defn:projectionmap}
We define a map $\pi_d : P_k^*( B_{r,d}(G) ) \to 2^{E(G)}$ by $\pi_d(X) = \{ e \in E(G) : X \cap B_{r,d}(e) \not = \emptyset \}$.
\end{defn}

Now we discuss the main technical hurdle to overcome in this section. Recall from \Cref{sec:initialintractability} that when dealing with a simple cycle $C$, there was only one way that $\pi_d(C)$ could fail to be a simple cycle, namely, if $C$ was one of the bigons. Crucially, there were only $d |E|$ ways this could happen, which was negligible compared to the size of the domain of $\pi_d$. On the other hand, elements of $P_k^*(G)$ can degenerate in more complicated ways, as is shown in \Cref{fig:DegeneratingAndNonDegeneratingDual3s}. The next few propositions establish that the degenerating elements---and all of the other possibilities---remain negligible compared to the preimages of $P_k^*(G)_m$:

\begin{figure}
    \centering
    \begin{tabular}{cc}
         \def\svgscale{.3}{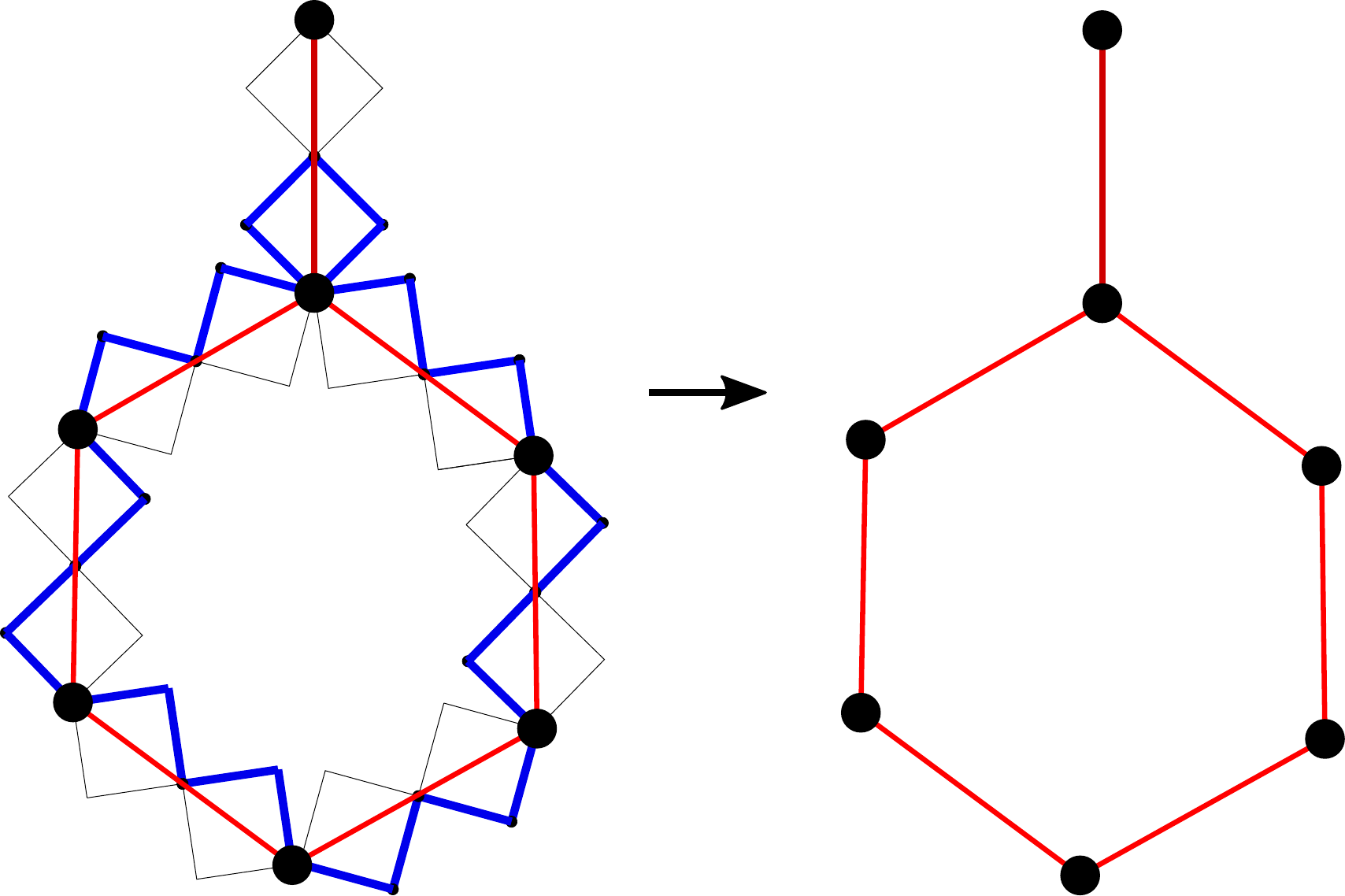} &
         \def\svgscale{.25}{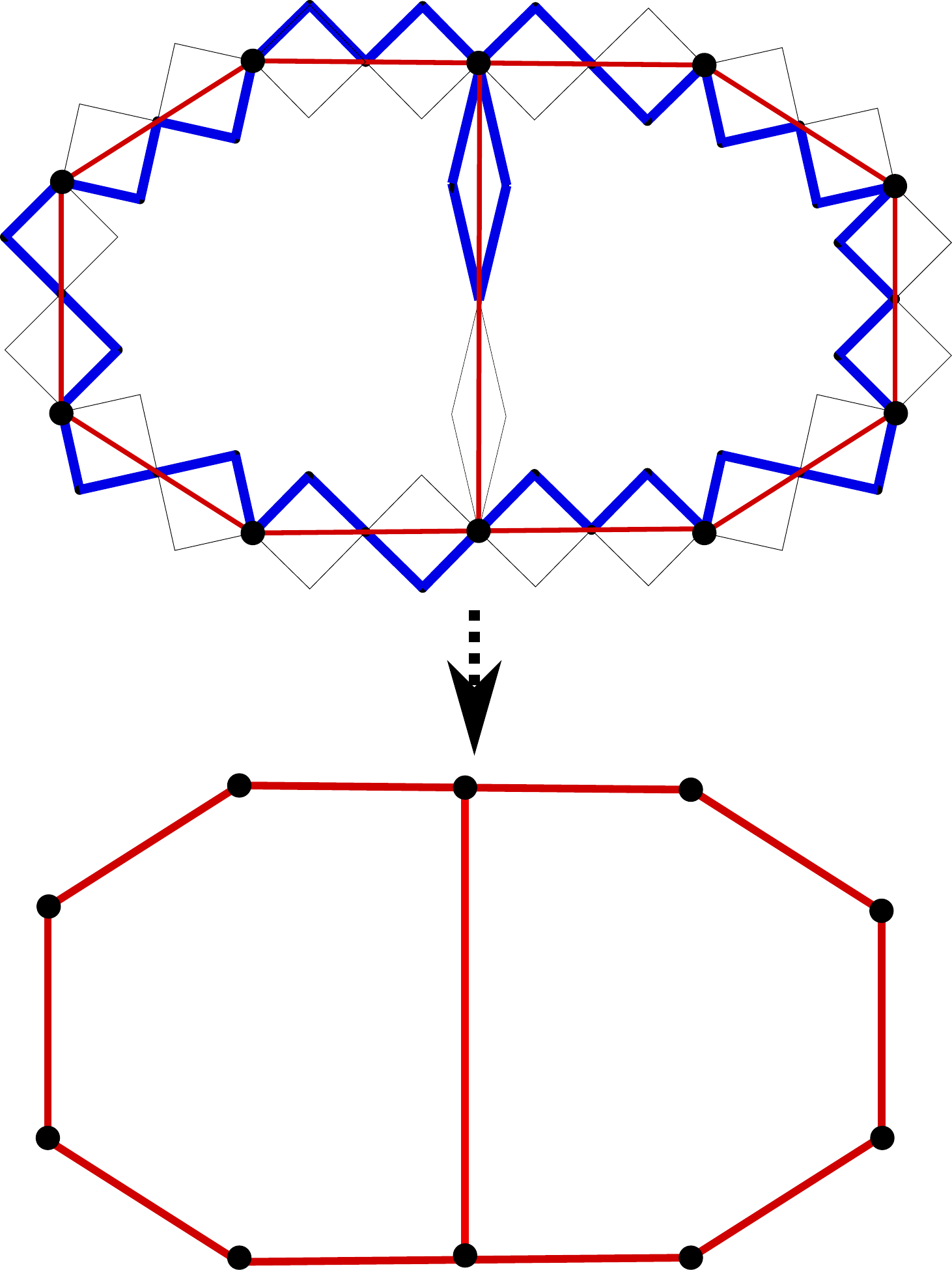} \\ a) A degenerating dual 3-partition & b) A Dual $3$-partition with a bigon from\\&  an $B_{r,d}$, but {which does not degenerate}
    \end{tabular}
    
    \caption{Example of degeneration and subtle non-degeneration.}
    \label{fig:DegeneratingAndNonDegeneratingDual3s}
\end{figure}

\begin{lem}[Surjectivity]\label{lem:lifts}
For each $Y \in P_k^*(G)$, there is a $\widetilde{Y} \in P_k^*( B_{r,d}(G))$ such that $\pi_d(\widetilde{Y}) = Y$.  
\end{lem}
\begin{proof}
Simply replace each edge $e \in Y$ with a simple path of length $d$ along $B_{r,d}(e)$. This does not change the topology, and hence the result has the same cycle rank and no bridges, which characterizes the elements of $P_k^*$ by \cref{lem:equivbridgeunion}.
\end{proof}

It will be convenient to single out the particular kind of $\pi_d$-preimage constructed in the proof of the last lemma:

\begin{defn}[Lift]\label{defn:lifts}
If $Y \in P_k^*(G)$, we refer to any of the $\widetilde{Y} \in P_k^*(B_{r,d}(G))$ obtained by replacing each edge $e \in Y$ with a simple path of length $d$ through $B_{r,d}(e)$ as a \emph{lift} of $Y$.
\end{defn}

To prove the probability concentration estimates (\Cref{prop:probabilityconcentrationkparts}) we need to characterize the elements of $\im(\pi_d)$ that have the most $N_{\lambda}$ mass above them as the elements of $P_k^*(G)_m$.
For that we will need to characterize the elements of the image with the most edges, which will be accomplished by \Cref{prop:boundsonprojectionelements} and \Cref{prop:maximizersinimagearecertificates}.

\begin{prop}\label{prop:boundsonprojectionelements}
The elements in the image of $\pi_d : P_k^*( B_{r,d} (G)) \to 2^{E(G)}$ all have $h_1 \leq k - 1$ and $\leq |V(G)| + k - 2$ edges. 
\end{prop}
\begin{proof}
The bound on the number of edges will follow from \Cref{lem:maxnumedgesandhomology} once we argue that all the elements in the image have $h_1 \leq k - 1$. So, let $X \in P_k^*( B_{r,d} (G)) $, and let $1_{C_1}, \ldots, 1_{C_m}$ be a basis for the cycle space of $\pi_d(X)$, where $1_Z$ is the indicator function of any set $Z \subseteq E(G)$. For each $i = 1, \ldots, m$, define $\widetilde{C}_i$ as a lift of $C_i$. %
We will show that the $1_{\widetilde{C}_i}$ are independent, by showing that any linear dependence gives a corresponding dependence between the $1_{C_i}$. Set $E_d = E( B_{r,d}(G))$. Define a linear map $T : \mathbb{R}^{E_d} \to \mathbb{R}^{E(G)}$, which, for each $e \in E(G)$, adds up the values along all edges of $B_{r,d}(e)$ and sets that as the value of $e$. In particular, $T( 1_{\widetilde{C}_i}) = d 1_{C_i}$. Suppose that we had that $0 = \sum_{i = 1}^m a_i 1_{\widetilde{C}_i}$, as functions on $E_d$. Then by applying $T$ to this equation, we obtain $0 = \sum_{i = 1}^m a_i d 1_{C_i}$, which implies that $a_i = 0$ for $i = 1, \ldots, m$. Hence $m \leq h_1(X)$, which implies that $m \leq k - 1$, so $h_1( \pi_d(X)) = m \leq k -1$.
\end{proof}

\begin{prop}\label{prop:maximizersinimagearecertificates}
The elements of $\im(\pi_d)$ that have $|V(G)| + k - 2$ edges are the elements of $P_k^*(G)_m$. 
\end{prop}
\begin{proof}
Suppose $K = \pi_d(X)$ has $|K| = |V(G)| + k - 2$. By \Cref{prop:boundsonprojectionelements}, $h_1(K) \leq k -1$. By \Cref{lem:maxnumedgesandhomology}, $h_1(K) \geq k - 1$, so $h_1(K) = k-1$. To finish the claim, we prove that $K$ has no bridge edges. 
But suppose that $e$ is a bridge edge, and let $\widetilde{e}$ be any edge in $B_{r,d}(e) \cap X$. $X$, having no bridges, must have a simple cycle $C$ that contains $\widetilde{e}$. Now, let $C_1, \ldots, C_{k-1}$ be a cycle basis for $K$. Since $e$ is a bridge edge, none of the $C_i$ contain $e$, hence their lifts do not contain $\widetilde{e}$. This implies that $C$ is not in the span of $C_1, \ldots, C_{k-1}$, hence $h_1 ( X) \geq k$, which contradicts $X \in P_k^*(B_{r,d}(G))$. 
This shows that every element in $\im(\pi_d)$ with $|V(G)| + k - 2$ edges is in $P_k^*(G)_m$. Since we already showed in \Cref{lem:lifts} that every element of $P_k^*(G)$ has a lift to an element of $P_k^*( B_{r,d} (G))$, the claim follows.
\end{proof}

For the probability concentration lemma we will need the upper and lower bounds on the $N_{\lambda}$ mass above an $m$ edge element of $\im( \pi_d)$ provided by the next two propositions:

\begin{prop}\label{prop:kpartitionfibermasslower}
Let $K \in P_k^*(G)_m$. Then $N_\lambda ( \pi_d^{-1} (K) )  \geq (r \lambda) ^{ d ( |V(G)| + k - 2)}$.
\end{prop}
\begin{proof}
Since $K$ has $|V(G)| + k - 2$ edges, and along each $B_{r,d}(e)$ there are $r^d$ simple paths, there are $r^{d ( |V(G) + k - 2)}$ lifts obtained by choosing one of those paths for each edge. Since the length of each such lift is $d ( |V(G) + k - 2|)$, the total $N_{\lambda}$ mass is $r^{d ( |V(G)| + k - 2)} \lambda^{d ( |V(G) + k - 2)} = (r \lambda) ^{ d ( |V(G)| + k - 2)}$. Since there may be other preimages, as shown in \Cref{fig:DegeneratingAndNonDegeneratingDual3s}, we only get a lower bound.
\end{proof}

\begin{prop}\label{prop:kpartitionfibermassupper}
Let $K \in \im(\pi_d)$, with $|E(K)| = m$. Assume $\lambda \in (0,1]$ and assume $r$ is such that that $r \lambda \geq 2$. Then $N_{ \lambda} ( \pi_d^{-1} (K) ) \leq 2^m r^{2km} d^{2km} ( r \lambda)^{dm}$.
\end{prop}
\begin{proof}
We begin by bounding the number of possible configurations of $X \cap B_{r,d}(e)$ for any $e \in E(G)$, and $X \in P_k^*( B_{r,d}(G))$. First, observe that there are $r^d$ simple paths across $B_{r,d}(e)$. We treat two cases, depending on whether or not $X$ contains one of these paths.

If $X$ contains one of those paths then $X$ may contain at most $k - 1$ additional edges from $B_{r,d}(e)$, because each one increases the rank of $h_1$. Thus, $r^d { rd \choose k - 1}$ upper bounds the number of configurations which includes a path through $B_{r,d}(e)$. Moreover, each such configuration contains at least $d$ edges, hence each one has $N_{\lambda}$ mass at most $\lambda^d$, as $\lambda \leq 1$.

Alternatively, $X \cap B_{r,d}(e)$ may not contain any of the simple paths crossing $B_{r,d}(e)$. However, in this case, we have that $|X \cap B_{r,d}(e)| \leq 2(k-1)$, since every pair of edges will increase $h_1$ by one, and thus an upper bound on the number of such configurations is ${ rd \choose 2(k - 1) }$. Moreover, each one of these configurations contains at least $2$ edges, so has mass at most $\lambda^2$. Thus, the total $N_{\lambda}$ mass obtained from this case is ${ rd \choose 2(k - 1) }\lambda^2$.

Combining these two cases, we have a bound for the total $N_{\lambda}$ mass of $\pi_d^{-1}(K)$ restricted to $B_{r,d}(e)$, namely:  $$\sum_{X \in \pi_d^{-1}(K)} N_{\lambda} ( X \cap B_{r,d}(e) ) \leq  { rd \choose k - 1} r^d \lambda^d + { rd \choose 2(k - 1) }\lambda^2 \leq  2 (rd)^{2k} ( r \lambda)^{d}.$$
Here, the last inequality follows from $\max ( {rd \choose k - 1}, {rd \choose 2(k - 1) } ) \leq (rd)^{2k}$ and from $r \lambda \geq 2 \geq \lambda^2$. To determine $N_{\lambda}(\pi_d^{-1}(K))$, we first observe that %
$N_{\lambda}(X) = \lambda^{|X|} = \prod_{e \in E(K) } \lambda^{|X \cap E(B_{r,d}(e))|} =  \prod_{e \in E(K) } N_{\lambda} ( X \cap B_{r,d}(e))$ for any $X \in \pi_d^{-1}(K)$. From this, the result follows:
\begin{align*}
    N_{\lambda} ( \pi_d^{-1}(K) ) &= \sum_{X \in \pi_d^{-1}(K) } N_{\lambda}(X) \\
    &= \sum_{X \in \pi_d^{-1}(K) } \prod_{e \in E(K) } N_{\lambda} ( X \cap B_{r,d}(e)) \\
    &\leq  \prod_{e \in E(K) } \sum_{X \in \pi_d^{-1}(K) }N_{\lambda} ( X \cap B_{r,d}(e)) \\
    &\leq \prod_{e \in E(K) }  2 (rd)^{2k} ( r \lambda)^{d} \\
    &=  2^m (rd)^{2km} ( r \lambda)^{dm}
\end{align*} 
\end{proof}

We now prove the probability concentration lemma necessary for applying the lucky guess algorithm, \Cref{lem:luckyguess}:

\begin{lem}[Probability concentration lemma]\label{prop:probabilityconcentrationkparts}
Let $G = (V,E)$ have $n = |V|$. Let $\lambda \in (0,1]$ and fix an integer $r$ such that $r \lambda \geq 2$. Assuming that $P_k^*(G)_m$ is non-empty, then for $d = 4 \ceil{(\frac{ a + 2n^2 + 2k n^2( \log(r) + 1) }{\log(r \lambda)})^2}$ the probability under $\nu_{\lambda}$ that an element of $P_k^*( B_{r,k} (G))$ maps via $\pi_d$ to an element of $P_k^*(G)_m$ is at least $\frac{2^a}{1 + 2^a}$. %

\end{lem}
\begin{proof}
Let $M = n + K - 2$. %
Since elements of $\im( \pi_d) \setminus P_k^*(G)_m$ have $\leq M -1 $ edges (\Cref{prop:maximizersinimagearecertificates}), and $|Im(\pi_d)| \leq 2^{n^2}$, we have $N_d : = N_{\lambda} ( \pi_d^{-1} ( 2^{E(G)} \setminus P_k^*(G)_m ) ) \leq 2^{n^2}  2^{M-1} (rd)^{2k(M-1)} ( r \lambda)^{d(M-1)}$ (\cref{prop:kpartitionfibermassupper}). We also have that $H_d: = N_{\lambda} (\pi_d^{-1}(P_k^*(G)_m)) \geq ( \lambda r) ^{dM}$ (\cref{prop:kpartitionfibermasslower}). %
Using \Cref{lem:polylarge} and $M \leq n^2$,  %
for $d \geq 4 (\frac{ \log(S) + 2 k n^2 }{\log(q)})^2$, where $S = 2^a 2^{2n^2} r^{2kn^2}$ and $q = r \lambda$,
we have that $H_d \geq 2^a N_d$. Hence, by \Cref{lem:HgeqDN}, this $d$ suffices for $\frac{ H_d}{N_d + H_d} \geq \frac{2^a}{1 + 2^a }$.
\end{proof}

\subsubsection{Uniformly sampling connected $k$-partitions is intractable}

In this section we prove intractability of sampling dual $k$-partitions, and connect this to the intractability of sampling connected $k$-partitions.

\begin{thm}\label{thm:lambdadualkintractable}
For any fixed $\lambda \in (0,1]$, $\lambda$-sampling connected dual $k$-partitions is intractable on the class of 2-connected planar graphs.
\end{thm}
\begin{proof}
We will assemble the ingredients for \Cref{lem:luckyguess}, the Lucky Guess lemma. %
We fix a choice of $a$ so that $\frac{2^a}{2^a + 1} \geq 1 - 1/m$. Define $B$ is by the construction $G \to B_{r,d}(G)$ where $r$ is chosen so that $r \lambda \geq 2$ and $d = \ceil{(\frac{ a + 2n^2 + 2k n^2( \log(r) + 1) }{\log(r \lambda)})^2}$. The map $M$ is given by $\pi_d$, where we have $S(G) = 2^{E(G)}$. %
The problem $Q$ is defined by $Q(G)  = P_k^*(G)_m$, which is $\NP$-complete by \Cref{prop:GridGraphReduction}.
Moreover, by \Cref{prop:probabilityconcentrationkparts}, we have that $\mathbb{P} ( \pi_d(C) \in Q(G) : C \textrm{ is distributed according to  } \nu_{\lambda} \textrm { on } P_k^*( B_{r,d}(G) ) \geq 1 - 1/m$.
Thus, we obtain the result from \Cref{lem:luckyguess}.
\end{proof}

For any fixed $\lambda > 0$, we define a distribution on connected $k$-partitions. 

\begin{computationalproblem}{$\lambda$-sampling connected $k$-partitions}

Input: A graph $G$.

Output: An element of $\mathscr{P}_k(G)$, drawn according to the probability distribution that assigns a partition $P \in \mathscr{P}_k(G)$ weight proportional to $\lambda^{| \cut(P) |}$.

\end{computationalproblem}

Now we state the main theorem of this section:

\begin{thm}
Fix $\lambda \in (0,1]$. Then $\lambda$-sampling connected $k$-partitions is intractable on the class of $2$-connected planar graphs.
\end{thm}
\begin{proof}
This follows as a corollary to \Cref{thm:lambdadualkintractable}, using \Cref{thm:kpartitionduality:ref}.
\end{proof}

%% file: Sections/2Complexity/7bshort.tex
\subsubsection{Duality for connected $k$-partitions}

Key to our proof will be a duality theorem \cref{thm:kpartitionduality:ref}, which is proven in detail in the appendix. We now list the definitions necessary for its statement.

\begin{defn}[Unordered connected partitions]
Let $\mathscr{P}^c(G)$ %
be the set of \emph{unordered} partitions of $V(G)$ such that each block induces a connected subgraph. %
That is, $\mathscr{P}^c(G) = \bigcup_{k = 1}^{|V(G)|} \mathscr{P}_k(G)$.%
\end{defn}

\begin{defn}[Edge cut]
Let $P$ be a partition of $V(G)$. If $P = \{A_1, \ldots, A_k\}$, then we refer to the $A_i$ as the \emph{blocks} of $P$. Let $\cut(P)$ denote the set of edges of $G$ with endpoints in different blocks of $P$.  %
\end{defn}

\begin{defn}[Component map]
Given $J \subseteq E(G)$, define a partition $\comp(J) \in \mathscr{P}^c(G)$ as the partition into the connected components of $G \setminus J$.
\end{defn}

\begin{defn}[Dual connected partitions] Let $\mathscr{E}_2(G)$ denote the set of subsets of edges of $G$ such that each connected component of the induced subgraph is $2$-edge connected: $$\mathscr{E}_2(G) = \{ J \subseteq 2^{E(G)} : \textrm{ Each component of } G[J] \textrm{  is $2$-edge connected} \}.$$ %
\end{defn}

\begin{defn}[Circuit rank, number of components]
Let $G$ be a graph. Then $h_1(G)$ denotes the circuit rank of $G$, that is, the minimum number of edges that must be removed to make $G$ a forest%
. Also, $h_0(G)$ denotes the number of connected components of $G$. 
\end{defn}

\begin{defn}[Dual $k$-partition]
We define $\mathscr{P}_k^*(G) = \{ J \in \mathscr{E}_2(G) : h_1( G[J]) = k - 1 \}$. We call the elements of $\mathscr{P}_k^*(G)$ dual $k$-partitions.
\end{defn}

\begin{thm}[Duality between $\mathscr{P}_k(G)$ and $\mathscr{P}_k^*(G^*)$]\label{thm:kpartitionduality:ref}
Let $G$ be a plane graph, and $G^*$ its planar dual. Let $D : 2^{E(G)} \to 2^{E(G^*)}$ be the natural bijection. The map $D \circ \cut : \mathscr{P}_k(G) \to \mathscr{P}_k^*(G^*)$ is a bijection, with $\comp \circ D^{-1} : \mathscr{P}_k^*(G^*) \to \mathscr{P}_k(G)$ as its inverse. Both are computable in polynomial time. %
\end{thm}

%% file: Images/GridGraphReduction.pdf_tex
\begingroup%
  \makeatletter%
  \providecommand\color[2][]{%
    \errmessage{(Inkscape) Color is used for the text in Inkscape, but the package 'color.sty' is not loaded}%
    \renewcommand\color[2][]{}%
  }%
  \providecommand\transparent[1]{%
    \errmessage{(Inkscape) Transparency is used (non-zero) for the text in Inkscape, but the package 'transparent.sty' is not loaded}%
    \renewcommand\transparent[1]{}%
  }%
  \providecommand\rotatebox[2]{#2}%
  \newcommand*\fsize{\dimexpr\f@size pt\relax}%
  \newcommand*\lineheight[1]{\fontsize{\fsize}{#1\fsize}\selectfont}%
  \ifx\svgwidth\undefined%
    \setlength{\unitlength}{545.47630594bp}%
    \ifx\svgscale\undefined%
      \relax%
    \else%
      \setlength{\unitlength}{\unitlength * \real{\svgscale}}%
    \fi%
  \else%
    \setlength{\unitlength}{\svgwidth}%
  \fi%
  \global\let\svgwidth\undefined%
  \global\let\svgscale\undefined%
  \makeatother%
  \begin{picture}(1,0.43302575)%
    \lineheight{1}%
    \setlength\tabcolsep{0pt}%
    \put(0,0){\includegraphics[width=\unitlength,page=1]{GridGraphReduction.pdf}}%
    \put(0.10489196,0.36909566){\color[rgb]{0,0,0}\makebox(0,0)[rt]{\lineheight{1.25}\smash{\begin{tabular}[t]{r}$v$\end{tabular}}}}%
    \put(0,0){\includegraphics[width=\unitlength,page=2]{GridGraphReduction.pdf}}%
    \put(0.09920014,0.18151573){\color[rgb]{0,0,0}\makebox(0,0)[rt]{\lineheight{1.25}\smash{\begin{tabular}[t]{r}$a$\end{tabular}}}}%
    \put(0.2948731,0.39124031){\color[rgb]{0,0,0}\makebox(0,0)[lt]{\lineheight{1.25}\smash{\begin{tabular}[t]{l}$b$\end{tabular}}}}%
  \end{picture}%
\endgroup%

%% file: Images/DegenerateDual3.pdf_tex
\begingroup%
  \makeatletter%
  \providecommand\color[2][]{%
    \errmessage{(Inkscape) Color is used for the text in Inkscape, but the package 'color.sty' is not loaded}%
    \renewcommand\color[2][]{}%
  }%
  \providecommand\transparent[1]{%
    \errmessage{(Inkscape) Transparency is used (non-zero) for the text in Inkscape, but the package 'transparent.sty' is not loaded}%
    \renewcommand\transparent[1]{}%
  }%
  \providecommand\rotatebox[2]{#2}%
  \newcommand*\fsize{\dimexpr\f@size pt\relax}%
  \newcommand*\lineheight[1]{\fontsize{\fsize}{#1\fsize}\selectfont}%
  \ifx\svgwidth\undefined%
    \setlength{\unitlength}{491.85412302bp}%
    \ifx\svgscale\undefined%
      \relax%
    \else%
      \setlength{\unitlength}{\unitlength * \real{\svgscale}}%
    \fi%
  \else%
    \setlength{\unitlength}{\svgwidth}%
  \fi%
  \global\let\svgwidth\undefined%
  \global\let\svgscale\undefined%
  \makeatother%
  \begin{picture}(1,0.66581288)%
    \lineheight{1}%
    \setlength\tabcolsep{0pt}%
    \put(0,0){\includegraphics[width=\unitlength,page=1]{DegenerateDual3.pdf}}%
  \end{picture}%
\endgroup%

%% file: Images/NonDegeneratingDual3.pdf_tex
\begingroup%
  \makeatletter%
  \providecommand\color[2][]{%
    \errmessage{(Inkscape) Color is used for the text in Inkscape, but the package 'color.sty' is not loaded}%
    \renewcommand\color[2][]{}%
  }%
  \providecommand\transparent[1]{%
    \errmessage{(Inkscape) Transparency is used (non-zero) for the text in Inkscape, but the package 'transparent.sty' is not loaded}%
    \renewcommand\transparent[1]{}%
  }%
  \providecommand\rotatebox[2]{#2}%
  \newcommand*\fsize{\dimexpr\f@size pt\relax}%
  \newcommand*\lineheight[1]{\fontsize{\fsize}{#1\fsize}\selectfont}%
  \ifx\svgwidth\undefined%
    \setlength{\unitlength}{421.00775995bp}%
    \ifx\svgscale\undefined%
      \relax%
    \else%
      \setlength{\unitlength}{\unitlength * \real{\svgscale}}%
    \fi%
  \else%
    \setlength{\unitlength}{\svgwidth}%
  \fi%
  \global\let\svgwidth\undefined%
  \global\let\svgscale\undefined%
  \makeatother%
  \begin{picture}(1,1.33374913)%
    \lineheight{1}%
    \setlength\tabcolsep{0pt}%
    \put(0,0){\includegraphics[width=\unitlength,page=1]{NonDegeneratingDual3.pdf}}%
  \end{picture}%
\endgroup%

%% file: Sections/3Bottlenecks/1Header.tex
\section{The Flip Markov Chain}\label{Section:FlipChain}

In the previous section, we examined the worst case complexity of the partition sampling problem. However, worst case intractability results do not necessarily mean that the problem is intractable on examples of interest, since there can be algorithms which are effective only on certain cases. In this section and the next, we examine the performance of one such algorithm, which is based on Markov chains. 

Markov chains provide a generic means of sampling from prescribed distributions over a state space $\Omega$.  This technique starts with a seed in $\Omega$ and randomly applies perturbations to walk around the space; the more steps in the random walk, the closer the sample is to being distributed according to the \emph{stationary} distribution of the chain rather than the seed point. %
While this approach provides an elegant means of sampling, a mathematical analysis of the \emph{mixing time} (see \Cref{section:MCbackground}) is needed to understand how many steps one must take before the output can be trusted as representative of the stationary distribution.  Without control over the mixing time, it is possible that the sample did not travel far from the initial seed, %
potentially yielding a biased sample and distorted measurements.

In this section, we discuss a commonly-used Markov chain for sampling from $P_2(G)$, which we call the \emph{flip walk}.  This chain has seen wide use in the analysis of gerrymandering \cite{mattingly,mattingly1,chikina_assessing_2017}. We know from our analysis in the previous sections that one cannot hope for this chain to mix rapidly on general graphs, unless one also believes that $\RP = \NP$. %
To make this more concrete in this section we show that the gadgets used in our complexity proofs directly yield bottlenecks impeding the mixing of the flip chain. Later, in \cref{Section:Empirical}, we will use ideas from this section to analyze the mixing of the flip chain on examples relevant to redistricting.%

\subsection{Related work}\label{section:markovchainrelatedwork}

The flip walk is analogous to Glauber dynamics and Potts models. Contiguity of the blocks is not usually considered in these physical settings, and it is part of what makes sampling districting plans challenging. A difficulty in analyzing the combinatorics of contiguity constraints is that it is defined through global, rather than local interactions; a physical model with similar challenges is the self-avoiding walk, which we consider further in (\cref{Section:GridGraph}). We now list a few works that have studied Markov chains similar to the flip walk:

\subsubsection{Sampling min-cuts, and cuts according to boundary length}

A Markov chain with similar proposal moves as in \Cref{defn:flipwalk} was studied in \cite{bezakova2016integrating}, but restricted to a state space of $st$-cuts instead of connected $k$-partitions. They show that this Markov chain mixes slowly, even if the underlying graph is of bounded treewidth. They prove a tree-width fixed parameter tractability results for the counting and sampling problems they consider which, similarly to our work in \Cref{Section:PositiveResults}, build on Courcelle's theorem and are based on dynamic programming. Their results therefore share some similarities with ours, except that we studied \emph{connected} $2$-partitions weighted by a function of the edge-cut, whereas they studied two different cases: min cuts, and \emph{all} cuts weighted by a function of the edge-cut. The example in their section $7.4$ has some similar features to our example in \Cref{fig:torpid_mixing}.

\subsubsection{Literature on sampling $st$-paths}

Another place in the literature where the flip walk appears is in \cite{montanari2015sampling}, where the problem of sampling simple $st$-paths using the flip walk is studied. Their paper provides another example where there is a bottleneck \cite[Theorem 7]{montanari2015sampling} and gives a proof of ergodicity for their version of the flip chain (which is restricted to $st$-paths for fixed $s$ and $t$). They also make the observation that if the $st$-path flip chain on the grid graph is restricted to paths that are \emph{monotone} in one direction, then the flip walk is rapidly mixing on that restricted state space. We remark that the techniques based on \cite[Prosition 5.1]{JVV} that we discussed in \Cref{sec:initialintractability} suffice to show that the sampling problem they consider is intractable on any class of graphs closed under the operation of replacing edges with chains of bigons, and where the Hamiltonian $st$-paths problem is NP-complete. Additionally, the techniques we discuss below in \Cref{Section:PositiveResults} should suffice to reduce the simple $st$-path sampling problem to a corresponding counting problem, which will be tractable on certain classes of graphs, such as series parallel graphs or graphs of bounded treewidth. %

The question of sampling simple paths has also received some attention in the literature: \cite{hoens2008counting} proves that a certain Markov chain on simple paths in a complete graph mixes rapidly (Theorem 4.1.2) but that a Metropolis-Hasting's version with weights has bottlenecks (Theorem 4.2.2), and repeats a similar analysis for sampling simple paths in trees (Theorem 4.3.2 and Theorem 4.4.1). The existence of a FPRAS for weighted simple paths on the complete graph, where weights can be set to zero to forbid edges, would imply the existence of a FPRUS for simple paths in any graph, which would imply that $\RP = \NP$ by using the chain of bigons trick from \cite[Prosition 5.1]{JVV} and the $\NP$-completeness of the Hamiltonian path problem; this negatively answers one of the open problems given in the conclusion of \cite{hoens2008counting}. \cite{hoens2008counting} also provides a dynamic program algorithm to count and uniformly sample weighted simple paths in trees and DAGs (Section 5). The undirected case can be extended using Courcelle's theorem, so it is likely that a reasonably implementable fixed-parameter in treewidth algorithm for sampling simple paths exists, perhaps along similar lines to \Cref{Section:PositiveResults}.

%% file: Sections/3Bottlenecks/2TheFlipWalk.tex
\subsection{The Flip Walk}

We put a graph structure on the set of connected $2$-partitions $P_2(G)$ as follows.  
Let $(A,B) \in P_2(G)$. Given any $x \in V(G)$, consider the partition $(A \cup \{x\}, B \setminus \{x\}) = (A',B')$. Provided that $(A',B') \in P_2(G)$, including the case when $(A',B') = (A,B)$, this defines an \emph{edge} between two elements of $P_2(G)$. If $(A',B') \not \in P_2(G)$, that is, if either $A'$ or $B'$ does not induce a connected subgraph of $G$, then we add a self loop to $(A,B)$. Do the same also for $(A \setminus \{x\}, B \cup \{x\})$. This defines a $|V(G)|$-regular graph structure on $P_2(G)$. Given this graph structure, we define the flip walk: %

\begin{defn}[Flip Walk]\label{defn:flipwalk}
The flip walk on $P_2(G)$ is the Markov chain obtained by performing a lazy simple random walk on $P_2(G)$, using the graph structure defined in the previous paragraph. We abuse notation and refer to the Markov chain, the graph, and the set by $P_2(G)$. %
\end{defn}

If $G$ is $2$-connected, then $P_2(G)$ is irreducible \cite{akitaya2019reconfiguration} and hence ergodic. Since every node of $P_2(G)$ has degree $|V(G)|$, the uniform distribution is the stationary distribution for the flip walk on $P_2(G)$. Thus, by standard Markov chain theory \cite{levin_markov_2009}, this flip walk \emph{eventually} produces nearly uniformly distributed elements in $P_2(G)$. However, we will see examples in this section where the flip walk on $P_2(G)$ can take exponential time in $|G|$ %
to generate a nearly uniform sample. %

%% file: Sections/3Bottlenecks/3MixingTimeBackground.tex
\subsection{Background: Mixing time of Markov Chains}\label{section:MCbackground}

We make a short digression to review a few notions from Markov chain theory. For details we have left out, see \cite{levin_markov_2009}. Since the goal of our discussion is to give examples where the random walk on $P_2(G)$ mixes slowly, we will recall the notion of \emph{mixing time} in the context of (discrete) Markov chains:

\begin{defn}[Total variation]
Given two probability distributions $\mu$ and $\nu$ on a finite set $\Omega$, the \emph{total variation} distance between $\mu$ and $\nu$ is given by $\|\mu - \nu\|_{TV} = \frac{1}{2} \sum_{x \in \Omega} |\mu(x) - \nu(x)|$. 
\end{defn}

\begin{defn}[Mixing time]
Let $\mu$ be the stationary distribution of the (discrete time) Markov chain $M = (\Omega, P)$. Let $P^t \delta_x$ denote the distribution at time $t$ of the Markov chain $M$ started at $x$. Define $$d^M(t) := \max_{x \in \Omega} ||P^t \delta_x - \mu||_{TV}.$$  Then, the \emph{mixing time} of $M$ is
\begin{equation}
t^M_{mix}(\epsilon) = \inf \{ t : d^M(t) \leq \epsilon \}.
\end{equation}
If the chain is clear from the discussion, we omit the superscript $M$.
\end{defn}
 
The definitions above help formalize what it means for a Markov chain to mix rapidly or torpidly:

\begin{defn}[Rapidly mixing]
A family of Markov chains $M \in \mathcal{M}$ is said to be \emph{rapidly mixing} if the there is a polynomial $p(x,y)$ so that $t^M_{mix}(\epsilon) \leq p( \log |M|, \log\epsilon),\ \forall M \in \mathcal{M}$, where $|M|$ denotes the size of the state space of $M$. %
\end{defn}

To prove rapid mixing, it is equivalent to find a polynomial $q(x)$ so that $t^M_{mix}(1/4) \leq q(\log (|M|))$, $\forall M \in \mathcal{M}$, as $t^M_{mix}(\epsilon) \leq \lceil \log_2( \epsilon^{-1}) \rceil t^M_{mix}(1/4)$ \cite[Equation (4.36)]{levin_markov_2009}. %

\begin{defn}[Torpidly mixing]
If there is an exponentially growing function, $f(n)$, such that $t_{mix}^M(1/4) \geq f( \log(|M|))$, for all $M \in \mathcal{M}$, then we say that $M$ is torpidly mixing.
\end{defn}

A standard means of arguing about mixing times for random walks on regular graphs comes from measuring bottlenecks, as in the next definition:

\begin{defn}[Conductance]
Let $G$ be a $d$-\emph{regular} graph, and $M$ the Markov chain obtained by a lazy random walk on $G$. We define the conductance of $M$\footnote{This is not the usual definition of the conductance, but this is the correct formula for the conductance of a lazy random walk on a $d$-regular graph \cite[p. 144]{levin_markov_2009}. The formula there has a typo, which was corrected in the errata.} as %
$$\Phi(M) = \min_{\substack{U \subset V(G)\\|U| \leq \frac{1}{2}|V(G)|}} \frac{ | \partial_E U |}{2d|U|}.$$
\end{defn}

Loosely speaking, such a set $U$ which proves that $\Phi(M)$ is small is called a bottleneck.

The following theorem connects mixing time and conductance and will be used to show that the chain $P_2(G)$ mixes torpidly for certain families of graphs, by building explicit bottleneck sets that upperbound the conductance:

\begin{thm}[\cite{levin_markov_2009}]\label{thm:conductancetomixing}For every Markov chain $M$, $t^M_{mix}(\nicefrac{1}{4}) \geq \frac{1}{4 \Phi(M)}$.
\end{thm}

%% file: Sections/3Bottlenecks/4FirstExample.tex
\subsection{Bottlenecks from the chain of bigons construction}\label{section:basicflips}

Due to the sampling intractability results (\cref{sec:initialintractability}), we know that by replacing edges with chains of bigons, we created graphs whose simple cycles should be expensive to sample. Likewise, it should be expensive to sample the connected $2$-partitions of the plane duals of these graphs. It is therefore natural to look for bottlenecks in the flip walk that arise through the plane dual of the chain of bigons construction (\Cref{fig:planardiamonds}(a)). In this section, we describe the dual of the chain of bigons construction (\Cref{defn:doubledstar}) and explain how it creates bottlenecks.

\begin{figure}
    \centering
    \begin{tabular}{cc}
    \def\svgscale{.3}{ 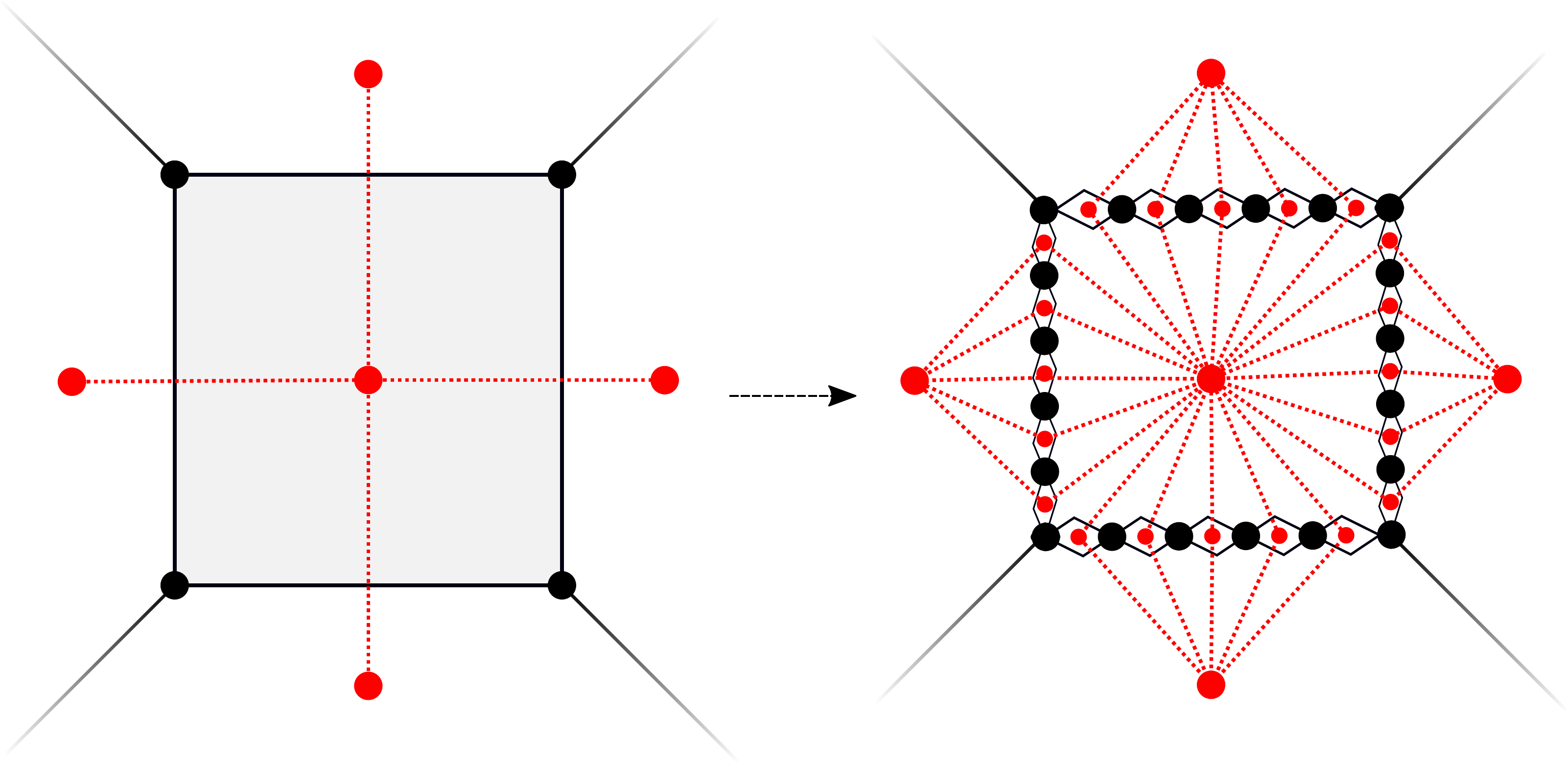 }&
    \scalebox{.4}{ 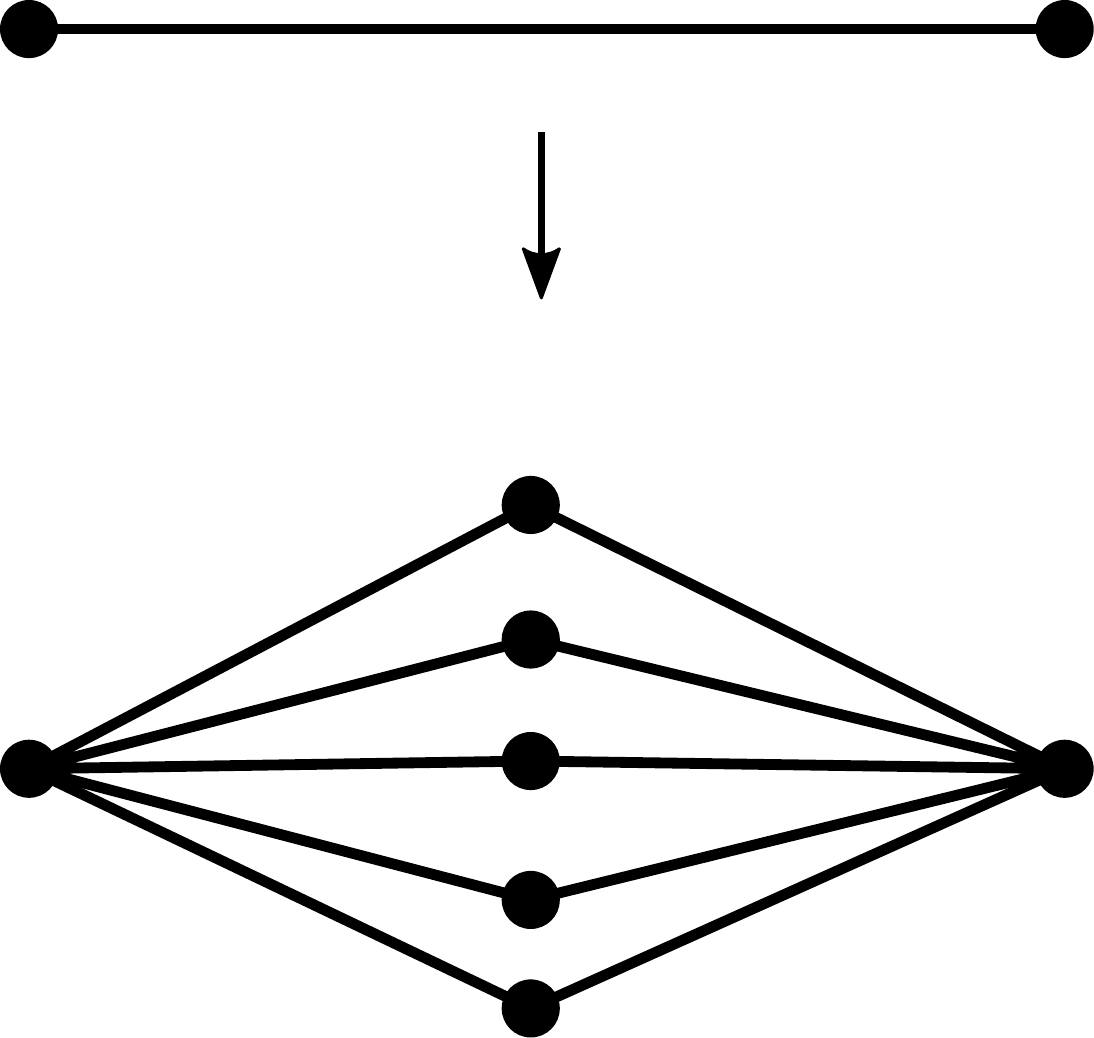 } 
    \\
    (a) Illustrating that $D_d(H) = (B_d(H^*))^*$ for plane $H$ & (b) Doubled $d$-star
    \end{tabular}
    \caption{The  chain of bigons construction and its dual }
    \label{fig:planardiamonds}
    
\end{figure}

\begin{defn}[Doubled $d$-star, original nodes]\label{defn:doubledstar}
Let $H$ be a graph. The \emph{doubled $d$-star} construction applied to $H$, notated $D_d(H)$, is obtained by replacing each edge of $H$ with $d$ parallel edges and then subdividing each new edge once. For $e \in E(G)$, we will let $D_d(e)$ denote the doubled $d$-star subgraph that replaced it. See \Cref{fig:planardiamonds}(b) for an illustration. 

There is an obvious inclusion $V(G) \hookrightarrow V(D_d(G))$, and we call the nodes in the image of that inclusion the original vertices. The other vertices in $D_d(G)$ are called new vertices.
\end{defn}

We will find bottlenecks in $P_2(D_d(G))$ by relating it to $P_2(G)$ using \Cref{lem:restrictionconnectededge}: %

\begin{lemma}\label{lem:restrictionconnectededge}
For any $(A,B) \in P_2( D_d(G))$, $(A \cap V(G), B \cap V(G) ) \in P_2(G)$. %
\end{lemma}
\begin{proof}
Let $x, y \in V(G)$ be members of the same block of $(A,B)$, say $A$. There exists a path $\gamma$ in $A$ from $x$ to $y$. This path alternates between new vertices and original vertices. Forgetting the new vertices in this path gives a path in $A \cap V(G)$ between $x$ and $y$.
\end{proof}

We use \cref{lem:restrictionconnectededge} to make the following definition:

\begin{defn}[Restriction map]\label{defn:Rdmap}
Define a map $R_d: P_2(D_d(G)) \to P_2(G)$ by setting $R_d( (A,B) ) = (A \cap V(G) , B \cap V(G))$. 
\end{defn}

We now explain the key intuition behind the bottlenecks. In order for the flip walk to move between the fibers of $R_d$---that is, to change the assignment of an old node---a certain rare event must occur. In particular, if $u$ and $v$ are adjacent old nodes, and $u \in A$ and $v \in B$, then to reassign $u$ to $B$, every new node in $D_d(\{u,v\})$ must already be in $B$. However, under the flip walk with $u \in A$, $v \in B$, the new nodes of $D_d(\{u,v\})$ behave like a random walk on a hypercube, and in particular, it is unlikely for them to become part of the same block. Pursuing this intuition, the next lemma proves that the fibers of $R_d$ have \emph{much} smaller edge boundary than size, which will mean that they are bottleneck sets:

\begin{lem}\label{lem:bottlenecks}
Suppose that $(A,B) \in P_2(G)$, with $A \not = \emptyset$ and $B \not = \emptyset$, and let $n = |V(G)|$. Then, 
\begin{equation}\label{eq:pdcut}
| R_d^{-1}( (A,B) )  | = 2^{ \cut(A,B) d}
\end{equation}
and
\begin{equation}\label{eq:pdboundarycut}
| \partial_E R_d^{-1}( (A,B) ) | \leq (d + 1) n 2^{ (\cut(A,B) - 1)d }.
\end{equation}
\end{lem}
\begin{proof}[Proof of \eqref{eq:pdcut}]
We will count the number of extensions of $(A,B)$ across the new nodes of $D_d(G)$, by considering each edge $e \in E(G)$ separately. If $e \in \cut(A,B)$, then one can assign new nodes of $D_d(e)$ arbitrarily without affecting contiguity, and therefore one has $2^d$ choices.
If $e \not \in \cut(A,B)$, suppose both endpoints of $e$ are in $A$. Since $B \not = \emptyset$, to preserve connectivity all the new nodes of $D_d(e)$ must be in $A$. Therefore, there is only one choice for how to extend $(A,B)$ along this edge.
Combining these two cases yields \eqref{eq:pdcut}.
\end{proof}
\begin{proof}[Proof of \eqref{eq:pdboundarycut}]

    Let $e \in \partial_E R_d^{-1}( (A,B)) $ be an edge between $(L,M), (L', M') \in P_2(D_d(G))$, with $ R_d ( (L', M') ) =: (A',B')$.  %
    There is an $x \in V(G)$ be such that $L' = L + x$, $M' = M - x$, $A' = A + x$ and $B' = B - x$. %
    Since $L \not = \emptyset$, for $L'$ to be connected there has to be at least one node $l \in L$ so that $l \sim x$. Moreover, since $M - x$ is connected  and $l,x \not \in M - x$, $M - x$ can contain at most one new node of $D_d( \{l,x\})$. Hence, there are at most $d +1$ extensions of $(A',B')$ onto the new nodes of $D_d( \{l,x\})$. %
    As there are most $\cut(A,B) - 1$ edges remaining where we might have the full $2^d$ range of extensions, it follows that there are \emph{at most} $(d + 1) 2^{(\cut(A,B) - 1)d}$ elements of $\partial_E  R_d^{-1}( (A,B))$ that map to $\{ (A,B), (A',B') \}$ under $R_d$. Finally, the claim follows because there are at most $n$ candidates for the original node $x$ that gets flipped when making a step across $\partial_E R_d^{-1} (A,B)$.
\end{proof}

We now use these computations to show the slow mixing of the flip chain. %

\begin{thm}\label{thm:bottlenecksimple}

Let $G$ be any $2$-connected graph with at least two distinct connected $2$-partitions $P, Q \in P_2(G)$, neither of which have the empty set as a block. Let $n = |V(G)|$. Then, the family $P_2(D_d(G))$, $d \geq 1$, is torpidly mixing. In particular, we have the following bounds:
\begin{align}
    \Phi( P_2(D_d(G)) &\leq (d+1) 2^{-d - 1}\label{eq:phibound},\\
    t_{mix}(1/4) (D_d(G)) &\geq \frac{ 2^{d - 1}}{(d+1) }\label{eq:mixingtimebound}, \text{ and }\\
    | P_2(D_d(G)) | \leq 2^ { |D_d(G)|} &\leq 2^{n + dn^2}.\label{eq:p2bound}
\end{align}

\end{thm}
\begin{proof}%
Without loss of generality, assume that $\cut(P) \leq \cut(Q)$. We have that $|R_d^{-1}(P)| \leq \nicefrac{1}{2} | P_2( D_d(G))|$ since $R_d^{-1}(P) \cap R_d^{-1}(Q) = \emptyset$ and $|R_d^{-1}(Q)| \geq |R_d^{-1}(P)|$ by \eqref{eq:pdcut}. Combining \eqref{eq:pdcut} with \eqref{eq:pdboundarycut} yields \eqref{eq:phibound}. Equation \eqref{eq:mixingtimebound} follows from this by \Cref{thm:conductancetomixing}. Finally, \eqref{eq:p2bound} follows from the construction of $D_d(G)$. From \eqref{eq:p2bound}, we have that $\log  |P_2(D_d(G))|  \leq n + dn^2$. %

\end{proof}

\noindent For a concrete example, take $G$ to be a $4$-cycle, and take the two $0$-balanced cuts. %
A snapshot of the evolution of the flip walk on $D_5(G)$ can be seen in \Cref{fig:torpid_mixing}.

\begin{figure}
    \centering
    
    \includegraphics[scale = .1]{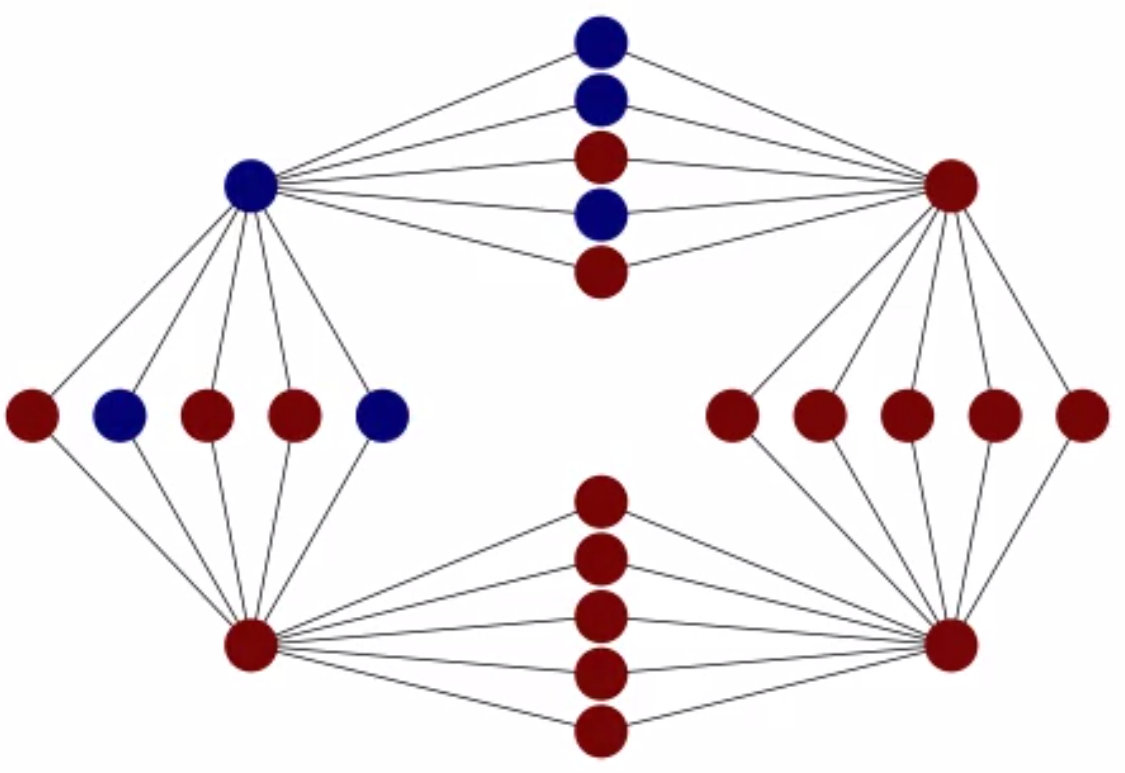}
    \caption{A snapshot of the flip walk evolving, illustrating the bottleneck of \Cref{thm:bottlenecksimple}.}
    \label{fig:torpid_mixing}
\end{figure}

We pause to note a key division between the intractability of uniformly sampling $P_2(G)$ and the mixing of the flip walk. 
First, observe that if $G$ is a series-parallel graph %
, then $D_d(G)$ is series parallel as well. We show in \Cref{Section:PositiveResults} that there is a polynomial time algorithm to uniformly sample from $P_2(G)$ on the class of series-parallel graphs, and yet our proof above shows that the flip walk still mixes slowly on this class of graphs.  Thus, even in cases where uniform sampling is tractable, the flip walk on $P_2(G)$ still may not be an efficient means of sampling. %

%% file: Images/DualityPicWithTails.pdf_tex
\begingroup%
  \makeatletter%
  \providecommand\color[2][]{%
    \errmessage{(Inkscape) Color is used for the text in Inkscape, but the package 'color.sty' is not loaded}%
    \renewcommand\color[2][]{}%
  }%
  \providecommand\transparent[1]{%
    \errmessage{(Inkscape) Transparency is used (non-zero) for the text in Inkscape, but the package 'transparent.sty' is not loaded}%
    \renewcommand\transparent[1]{}%
  }%
  \providecommand\rotatebox[2]{#2}%
  \newcommand*\fsize{\dimexpr\f@size pt\relax}%
  \newcommand*\lineheight[1]{\fontsize{\fsize}{#1\fsize}\selectfont}%
  \ifx\svgwidth\undefined%
    \setlength{\unitlength}{922.42202064bp}%
    \ifx\svgscale\undefined%
      \relax%
    \else%
      \setlength{\unitlength}{\unitlength * \real{\svgscale}}%
    \fi%
  \else%
    \setlength{\unitlength}{\svgwidth}%
  \fi%
  \global\let\svgwidth\undefined%
  \global\let\svgscale\undefined%
  \makeatother%
  \begin{picture}(1,0.48667401)%
    \lineheight{1}%
    \setlength\tabcolsep{0pt}%
    \put(0,0){\includegraphics[width=\unitlength,page=1]{DualityPicWithTails.pdf}}%
  \end{picture}%
\endgroup%

%% file: Images/DoubledDStar.pdf_tex
\begingroup%
  \makeatletter%
  \providecommand\color[2][]{%
    \errmessage{(Inkscape) Color is used for the text in Inkscape, but the package 'color.sty' is not loaded}%
    \renewcommand\color[2][]{}%
  }%
  \providecommand\transparent[1]{%
    \errmessage{(Inkscape) Transparency is used (non-zero) for the text in Inkscape, but the package 'transparent.sty' is not loaded}%
    \renewcommand\transparent[1]{}%
  }%
  \providecommand\rotatebox[2]{#2}%
  \newcommand*\fsize{\dimexpr\f@size pt\relax}%
  \newcommand*\lineheight[1]{\fontsize{\fsize}{#1\fsize}\selectfont}%
  \ifx\svgwidth\undefined%
    \setlength{\unitlength}{314.98935294bp}%
    \ifx\svgscale\undefined%
      \relax%
    \else%
      \setlength{\unitlength}{\unitlength * \real{\svgscale}}%
    \fi%
  \else%
    \setlength{\unitlength}{\svgwidth}%
  \fi%
  \global\let\svgwidth\undefined%
  \global\let\svgscale\undefined%
  \makeatother%
  \begin{picture}(1,0.94859413)%
    \lineheight{1}%
    \setlength\tabcolsep{0pt}%
    \put(0,0){\includegraphics[width=\unitlength,page=1]{DoubledDStar.pdf}}%
  \end{picture}%
\endgroup%

%% file: Sections/3Bottlenecks/5MaximalPlanarExample.tex
\subsection{Bottlenecks from the $R_d$ construction}\label{TriangleTorpid}

The example of \Cref{thm:bottlenecksimple} is not entirely satisfying, for example because the degrees of its nodes increase without bound. We will address some of its weaknesses in this section by producing a family of maximal plane graphs with vertex degree $\leq 9$, such that the corresponding family of flip walk chains is torpidly mixing. As in \Cref{section:basicflips}, our strategy will be to find a construction $G \to T_d(G)$ which uses gadgets to refine certain features of $G$, and a map $P_2(T_d(G) ) \to P_2(G)$, where we can count the size of the fibers and the size of the edge boundaries of the fibers. All graphs in this section are assumed to be embedded in the plane. Additionally, we will freely describe a partition $(A,B) \in P_2(G)$ by a map $p : V(G) \to \{a,b\}$.

\begin{defn}[$T_d$, original vertices, original triangles]
Let $G$ be a maximal plane graph. Let $T_d(G)$ be the graph defined by $T_d(G) = ( R_d( G^*))^*$, where $R_d$ is as in \Cref{defn:RDDef}.  There is a natural injection $i : V(G) \to V(T_d(G))$, since there is a natural injection $Faces(G^*)$ to $Faces( R_d(G^*))$, and we call the nodes in $\im(i)$ the original vertices. Moreover, if $F$ is any triangular face in $G$, then we call the original vertices of $F$ in $V(T_d(G))$ an original triangle. 
\end{defn}
\begin{figure}\centering
\def\svgscale{.4}{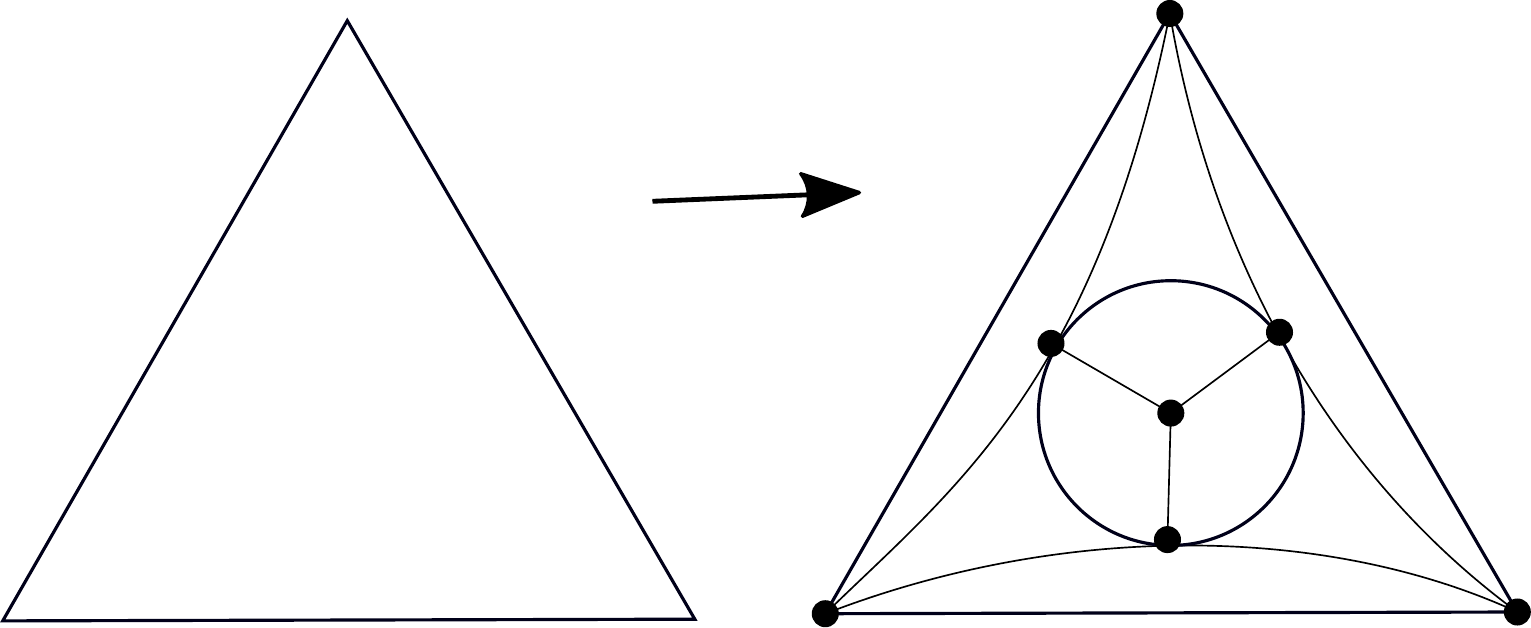}
\caption{The affect of $T_2$}\label{fig:dualrecursive}
\end{figure}

The affect of $T_d$ is to take every triangular face and refine it by gluing a graph hanging from the three nodes of the triangle. \Cref{fig:dualrecursive} shows the affect of applying $T_2$ on a single triangular face.

In \Cref{section:basicflips}, our strategy for finding bottleneck sets was to find a set of vertices such that if they all belonged to the same block, then a large number of vertices would also belong to that block. In this section, such a set of vertices will be the vertices of an original triangle. If all the nodes of an original triangle are in the same block, we will call that triangle pure:

\begin{defn}[Pure and mixed faces]
Let $G$ be a plane graph. Consider a partition $(A,B) \in P_2(G)$ defined by a function $p : V(G) \to \{a, b\}$. We will call a face $F$  \emph{pure} of assignment $a$ (resp.\ $b$), if $p$ takes the value $a$ (resp.\ $b$) on all of its nodes. We will call the face \emph{mixed} otherwise, that is, if $p$ takes on both values on the vertices of $F$. For a partition $(A,B) \in P_2(G)$, we let $\mathfrak{P}_{(A,B)}$ be the function on the set of faces of $G$ that assigns $a$ to all pure $a$-faces, $b$ to all pure $b$-faces, and $m$ to all mixed faces. %
Additionally, we define $M : P_2(G) \to \mathbb{N}$ as the number of mixed faces in a partition.%
\end{defn}

We are going to find bottlenecks in this section by defining sets of partitions of $T_d(G)$ by whether all original triangles are mixed. To leave such a set, some of the triangles will have to become pure, which will force the new nodes of that triangle to be in a specific arrangement. A convenient tool for expressing this will be to describe directed edges of $P_2(G)$ as being purifying or not.

\begin{defn}[Directed Configuration Space]
For a graph $G$, let $DP_2(G)$ be the directed graph version of the flip walk adjacency structure on $P_2(G)$. That is, it has a node for each node of $P_2(G)$ and for each edge $e = \{P, Q\}$ in $P_2(G)$, $DP_2(G)$ has two edges: $(P,Q)$ and $(Q,P)$.
\end{defn}

\begin{defn}[Purifying edges of the directed configuration space]
We call an edge $e = (P,Q) \in DP_2(G)$ purifying if there is a face $F$ of $G$ so that $\mathfrak{P}_P(F) = m$ but $\mathfrak{P}_Q(F) \in \{a,b\}$. Let $DP_2^C(G)$ be the graph obtained from $DP_2(G)$ by removing all purifying edges. %
For $Q \in P_2(G)$, we let $C_Q \subseteq P_2(G)$ be the set of all vertices strongly reachable from $Q$ in $DP_2^C(G)$.
\end{defn}

Finally, we will need a way to relate connected $2$-partitions of $T_d(G)$ to those of $G$, so that we can partition $P_2(T_d(G))$ based on which faces of $G$ are pure or mixed.

\begin{lemma}\label{restrictionconnectedtriangle}
For any partition $(A,B) \in P_2(T_d(G))$ defined by $p: V(T_d(G)) \to \{a,b\}$, let $p_o: V(G)\to\{a,b\}$ denote the restriction to the original vertices, which we identify with $V(G)$. Then $p_o: V(G)\to\{a,b\}$ defines a connected partition of $G$.
\end{lemma}
\begin{proof}
Let $x, y \in i(V(G)) \cap A$. Since $T_d(G)[A]$ is connected, there exist a path $\gamma$ in $A$ from $x$ to $y$. Forgetting the new vertices in this path gives a path in $G[i^{-1}(A)]$, since all the vertices on the boundary of any original triangle are adjacent.
\end{proof}

\begin{defn}[Restriction map]
Define $F_d: P_2(T_d(G)) \to P_2(G)$ to be the restriction map $F_d(p) = p_o$, with notation as in Lemma \ref{restrictionconnectedtriangle}
\end{defn}

\begin{lem}\label{lem:trianglebottlenecks}
Let $G$ and $F_d$ be as above. Let $P \in P_2(G)$. Then, 
\begin{equation}\label{eq:trianglearea}
| F_d^{-1} ( P) | \geq 5^{dM(P)} 
\end{equation}
and, supposing additionally that $P$ is such that $C_P = \{P \}$, 
\begin{equation}\label{eq:triangleboundary}| \partial_E F_d^{-1} ( P) | \leq  n  5^{(d + 1) (M(P) - 1)}.
\end{equation}

\end{lem}
\begin{proof}[Proof of \eqref{eq:trianglearea}]
A mixed face of a $2$-partition of $T_d(G)$ corresponds to a $SBL(R_d)$ segment of the simple cycle dual to that $2$-partition. Using the estimates in \Cref{eq:SBLbounds}, each of the mixed faces can be in at least in $5^d$ configurations, and so the claim follows.
\end{proof}
\begin{proof}[Proof of \eqref{eq:triangleboundary}] 
Since $C_P = \{P\}$ any edge out of $F_d^{-1}(P)$ must cause one mixed face of $P$ to become pure. This mixed face must be in a configuration where all but one node has the same block assignment, and that one exceptional node must be an original node. Since there are at most $n$ original nodes of $G$ which can switch during this step, and the other mixed faces have at most $5^{d+1}$ configurations each, the result follows by the bound on $SBL$ from \Cref{eq:SBLbounds}. %
\end{proof}

Taking $G=K_4$ yields the following corollary:

\begin{cor}\label{betterbehavedexample}
There is a family of graphs $H_d$, $d \in \mathbb{N}$, that are triangulations of the plane (maximal planar graphs) such the vertex degree is bounded by $9$ and $|V(H_d)| = O(d)$, and for which $P_2(H_d)$ and the unordered partition chain have mixing times at least $\frac{ 5^d } { 250}$. $H_2$ is shown in \Cref{fig:outdegreezero}b).
\end{cor}
\begin{proof}

\begin{figure}
\begin{tabular}{cc}
\def\svgscale{.4}{
    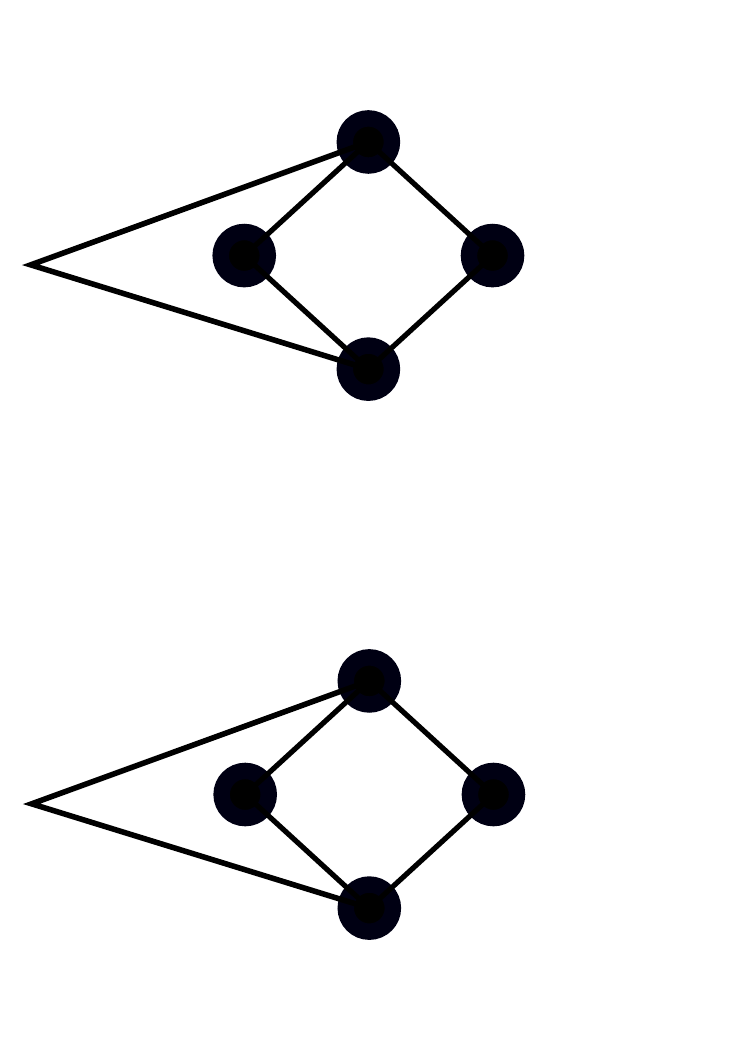} & \def\svgscale{.4}{ 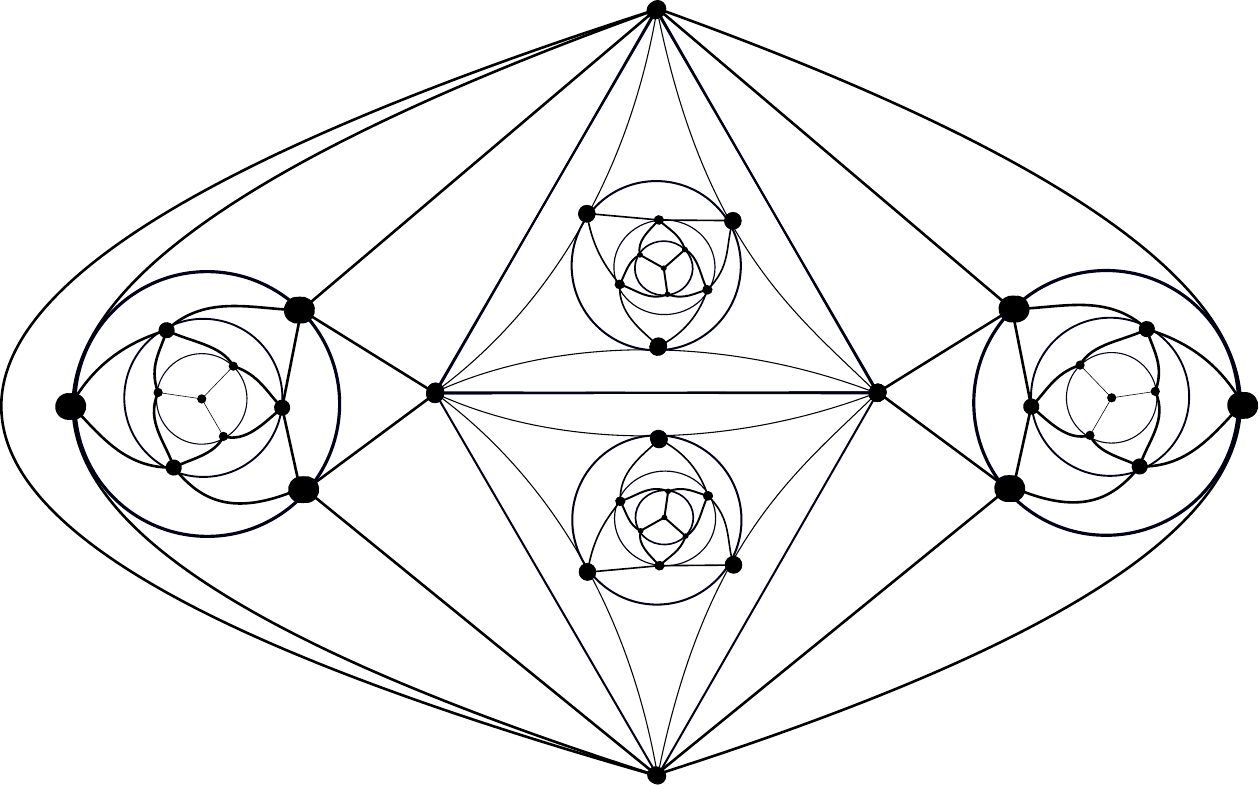} \\ 
    a) The elements of $DP_2^C(K_4)$ used in \Cref{betterbehavedexample} & b) $H_3$ of the family of \Cref{betterbehavedexample}.
\end{tabular}
    \caption{}
    \label{fig:outdegreezero}
\end{figure}
One can compute $DP_2^C(K_4)$ to find that there are three $P_i \in P_2(K_4)$ with $C_{P_i} = \{P_i \}$. \Cref{fig:outdegreezero}a) shows two such examples. %
Let $P$ be the top partition and $Q$ the bottom one in \Cref{fig:outdegreezero}a). By symmetry, we have that $|F_d^{-1}(P)| = |F_d^{-1}(Q)|$, and so since $F_d^{-1} (P) \cap F_d^{-1}(Q) = \emptyset$, it follows that $|F_d^{-1}(P)| \leq \nicefrac{ |P_2(T_d(G))|}{2}$. Hence, $F_d^{-1}(P)$, is a candidate bottleneck set. We compute, $$\frac{ | \partial_E F_d^{-1}(P) | }{ 2 n |F_d^{-1}(P)| } \leq \frac{1}{2} 5^{( d + 1) (M(P) - 1) - dM(P)}.$$ As $M(P)  = 4$, it follows that $\Phi( P_2(T_d(G)) \leq \frac{1}{2} (5^{3 - d})$. We obtain corresponding bottleneck sets in the quotient chain of unordered partitions. The result now follows by \Cref{thm:conductancetomixing}. 
\end{proof}

This last example illustrates that controlling neither the degree, nor the face degree, nor insisting on $3$-connectedness of $G$ can improve the mixing time of the flip walk on $P_2(G)$. On the other hand, these graphs still have a lot of area enclosed by length $3$ loops, which is arguably unrealistic for redistricting, except we could in principle see similar behavior around very dense cities. %
In the next section, we will use statistics inspired by the idea that certain nodes may change their assignment infrequently, as well as literature on self avoiding walks, to investigate the flip walk on connected partitions of a grid graph and on state dual graphs. %

%% file: Images/DualRecursive.pdf_tex
\begingroup%
  \makeatletter%
  \providecommand\color[2][]{%
    \errmessage{(Inkscape) Color is used for the text in Inkscape, but the package 'color.sty' is not loaded}%
    \renewcommand\color[2][]{}%
  }%
  \providecommand\transparent[1]{%
    \errmessage{(Inkscape) Transparency is used (non-zero) for the text in Inkscape, but the package 'transparent.sty' is not loaded}%
    \renewcommand\transparent[1]{}%
  }%
  \providecommand\rotatebox[2]{#2}%
  \newcommand*\fsize{\dimexpr\f@size pt\relax}%
  \newcommand*\lineheight[1]{\fontsize{\fsize}{#1\fsize}\selectfont}%
  \ifx\svgwidth\undefined%
    \setlength{\unitlength}{440.88953513bp}%
    \ifx\svgscale\undefined%
      \relax%
    \else%
      \setlength{\unitlength}{\unitlength * \real{\svgscale}}%
    \fi%
  \else%
    \setlength{\unitlength}{\svgwidth}%
  \fi%
  \global\let\svgwidth\undefined%
  \global\let\svgscale\undefined%
  \makeatother%
  \begin{picture}(1,0.40972038)%
    \lineheight{1}%
    \setlength\tabcolsep{0pt}%
    \put(0,0){\includegraphics[width=\unitlength,page=1]{DualRecursive.pdf}}%
  \end{picture}%
\endgroup%

%% file: Images/K4Example.pdf_tex
\begingroup%
  \makeatletter%
  \providecommand\color[2][]{%
    \errmessage{(Inkscape) Color is used for the text in Inkscape, but the package 'color.sty' is not loaded}%
    \renewcommand\color[2][]{}%
  }%
  \providecommand\transparent[1]{%
    \errmessage{(Inkscape) Transparency is used (non-zero) for the text in Inkscape, but the package 'transparent.sty' is not loaded}%
    \renewcommand\transparent[1]{}%
  }%
  \providecommand\rotatebox[2]{#2}%
  \newcommand*\fsize{\dimexpr\f@size pt\relax}%
  \newcommand*\lineheight[1]{\fontsize{\fsize}{#1\fsize}\selectfont}%
  \ifx\svgwidth\undefined%
    \setlength{\unitlength}{212.95795857bp}%
    \ifx\svgscale\undefined%
      \relax%
    \else%
      \setlength{\unitlength}{\unitlength * \real{\svgscale}}%
    \fi%
  \else%
    \setlength{\unitlength}{\svgwidth}%
  \fi%
  \global\let\svgwidth\undefined%
  \global\let\svgscale\undefined%
  \makeatother%
  \begin{picture}(1,1.42376541)%
    \lineheight{1}%
    \setlength\tabcolsep{0pt}%
    \put(0,0){\includegraphics[width=\unitlength,page=1]{K4Example.pdf}}%
    \put(0.49668607,1.31673504){\color[rgb]{0,0,0}\makebox(0,0)[t]{\lineheight{1.25}\smash{\begin{tabular}[t]{c}$a$\end{tabular}}}}%
    \put(0.5067484,0.75827386){\color[rgb]{0,0,0}\makebox(0,0)[t]{\lineheight{1.25}\smash{\begin{tabular}[t]{c}$a$\end{tabular}}}}%
    \put(0.71612677,1.06014434){\color[rgb]{0,0,0}\makebox(0,0)[lt]{\lineheight{1.25}\smash{\begin{tabular}[t]{l}$b$\end{tabular}}}}%
    \put(0.2571205,1.04505084){\color[rgb]{0,0,0}\makebox(0,0)[rt]{\lineheight{1.25}\smash{\begin{tabular}[t]{r}$b$\end{tabular}}}}%
    \put(0.49332048,0.02070446){\color[rgb]{0,0,0}\makebox(0,0)[t]{\lineheight{1.25}\smash{\begin{tabular}[t]{c}$a$\end{tabular}}}}%
    \put(0.50942027,0.57111519){\color[rgb]{0,0,0}\makebox(0,0)[t]{\lineheight{1.25}\smash{\begin{tabular}[t]{c}$b$\end{tabular}}}}%
    \put(0.74294828,0.3145256){\color[rgb]{0,0,0}\makebox(0,0)[lt]{\lineheight{1.25}\smash{\begin{tabular}[t]{l}$b$\end{tabular}}}}%
    \put(0.28394201,0.29943211){\color[rgb]{0,0,0}\makebox(0,0)[rt]{\lineheight{1.25}\smash{\begin{tabular}[t]{r}$a$\end{tabular}}}}%
    \put(9.20292179,-4.23239105){\color[rgb]{0,0,0}\makebox(0,0)[lt]{\begin{minipage}{0.34227842\unitlength}\raggedright \end{minipage}}}%
    \put(0,0){\includegraphics[width=\unitlength,page=2]{K4Example.pdf}}%
  \end{picture}%
\endgroup%

%% file: Images/MaximalPlaneCounterExample.pdf_tex
\begingroup%
  \makeatletter%
  \providecommand\color[2][]{%
    \errmessage{(Inkscape) Color is used for the text in Inkscape, but the package 'color.sty' is not loaded}%
    \renewcommand\color[2][]{}%
  }%
  \providecommand\transparent[1]{%
    \errmessage{(Inkscape) Transparency is used (non-zero) for the text in Inkscape, but the package 'transparent.sty' is not loaded}%
    \renewcommand\transparent[1]{}%
  }%
  \providecommand\rotatebox[2]{#2}%
  \newcommand*\fsize{\dimexpr\f@size pt\relax}%
  \newcommand*\lineheight[1]{\fontsize{\fsize}{#1\fsize}\selectfont}%
  \ifx\svgwidth\undefined%
    \setlength{\unitlength}{362.37946668bp}%
    \ifx\svgscale\undefined%
      \relax%
    \else%
      \setlength{\unitlength}{\unitlength * \real{\svgscale}}%
    \fi%
  \else%
    \setlength{\unitlength}{\svgwidth}%
  \fi%
  \global\let\svgwidth\undefined%
  \global\let\svgscale\undefined%
  \makeatother%
  \begin{picture}(1,0.62354017)%
    \lineheight{1}%
    \setlength\tabcolsep{0pt}%
    \put(0,0){\includegraphics[width=\unitlength,page=1]{MaximalPlaneCounterExample.pdf}}%
  \end{picture}%
\endgroup%

%% file: Sections/4Empirical/1Header.tex
\section{Empirical Examples}\label{Section:Empirical}

The torpid mixing of the flip walk highlighted in the previous section is not only a theoretical observation. In this section, we present several experiments showing slow mixing in practical applications of the flip walk to redistricting. The key statistic we study is closely linked to our bottleneck proofs, wherein we were able to identify sets of nodes that flip infrequently. For an empirical analysis, we will observe the frequency at which nodes flip during simulations of the flip walk on lattice like graphs (\cref{subsection:gridgraphempirical}) and on state dual graphs (\cref{subsection:empiricalrealstates}). %

To put our investigation in a larger context, we will begin by reviewing some related to self avoiding walks on the grid graph (\cref{Section:GridGraph}). This review will lead us to investigate the impact of the graph topology on the output the flip walk (\cref{section:phasetransitions}), and on the output of another popular partition sampling algorithm (\cref{section:thechoiceofmodel}).%

These experiments used the open source graph partition Markov chain software \href{https://github.com/mggg/GerryChain}{Gerrychain}. The code that produced these experiments \href{https://github.com/LorenzoNajt/Code-For-Complexity-and-Geometry-of-Sampling-Connected-Graph-Partitions}{is available online} \cite{mygithub}. %

%% file: Sections/4Empirical/2Grid.tex
\input{Sections/4GridGraph/2Grid.tex}

%% file: Sections/4GridGraph/2Grid.tex
\subsection{Grid Graph}\label{Section:GridGraph}

In this section, we will review some special features of connected partitions of the grid graph, in particular through the connection to the self-avoiding walk model from statistical physics. %

\subsubsection{Self-avoiding walks}

In one special case, the objects we are studying are closely linked to a famous topic from statistical physics, namely self-avoiding walks on lattices. A self-avoiding walk (resp. polygon) is a simple path (resp. cycle) in the integer lattice graph. These were introduced as models for polymers, and have shown themselves to be a difficult and rich object of mathematical investigation. An excellent reference for this topic is \cite{madrasself}; \cite{sokal1994monte} gives an overview of Monte Carlo methods used to investigate this topic. Self-avoiding walks on other lattices are also of interest, and we discuss one of those below. We will primarily be interested in walks that are constrained to lie in certain subsets of the lattice.

\begin{figure}
    \centering
    \def\svgscale{.2}{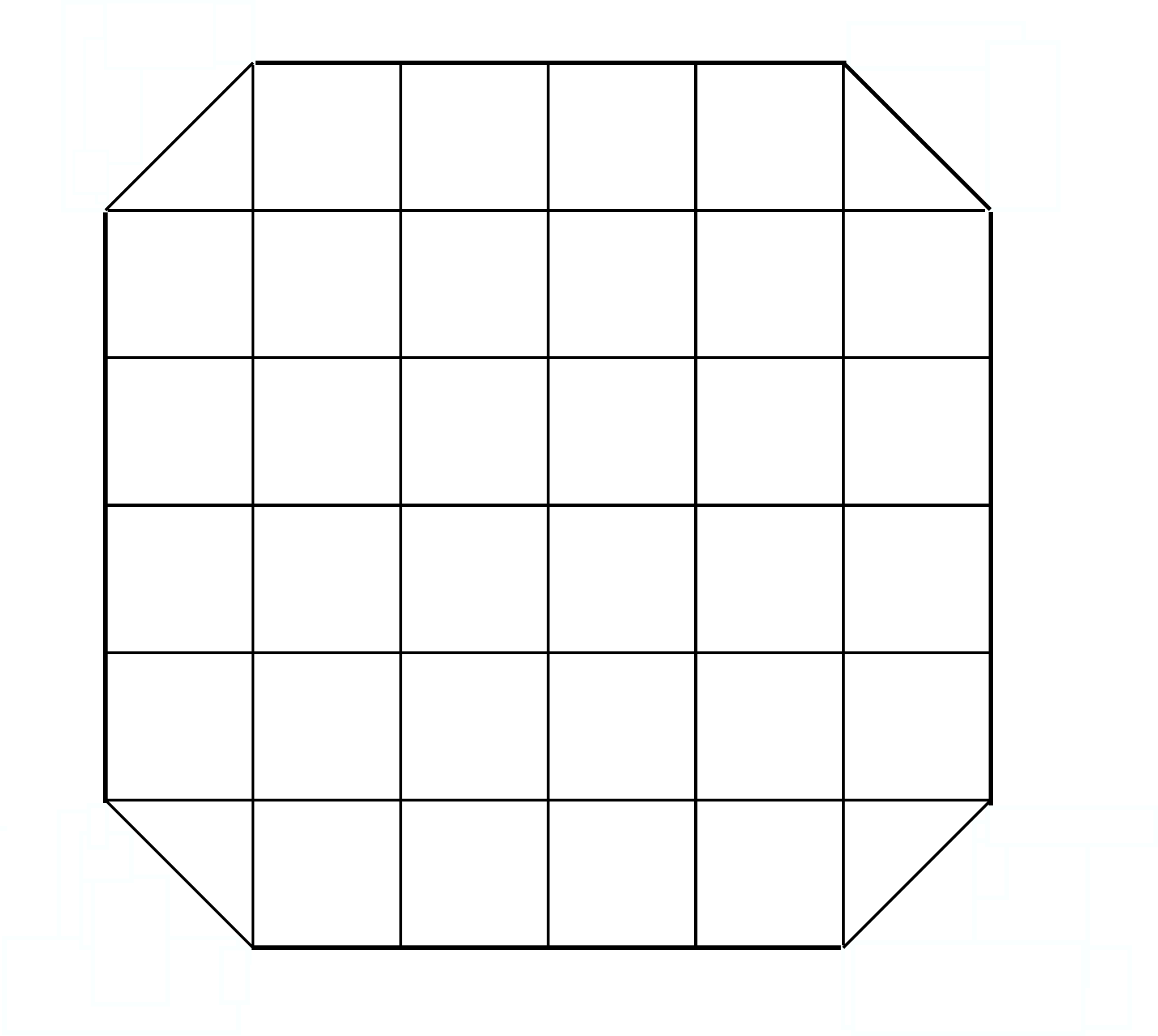}
    \caption{ $L_n$ }
    \label{fig:shavedgrid}
\end{figure}

Take $L_n$ to be the grid graph with shaved corners, as in \Cref{fig:shavedgrid}. The dual graph, $L_n^*$, is an $n-1 \times n-1$ grid graph $G_{n-1}$ with an additional ``supernode'' $V$ corresponding to the unbounded face. The simple cycles of $L_n^*$ break into two classes. First, there are those that do not contain the supernode. These can be thought of as self-avoiding polygons in $G_{n-1}$. Second, there are those that do contain the supernode. We can think of these as self-avoiding walks in $G_{n-1}$ between two points on the boundary. We will call these \emph{chordal} self-avoiding walks.

This connection to self-avoiding walks is important to us because such self-avoiding walks display phase transitions as one varies the preference for longer or shorter walks. As we will see, these phase transitions persist into the distribution $\nu_{\lambda}$ on $2$-partitions that we studied in \Cref{defn:lambdameasures}. After recalling the relevant facts and history about self-avoiding walks, we will present our experiments.

\begin{defn}[The family $P_{\lambda}$ of probability distributions on self-avoiding walks]
Fix $\lambda > 0$. Given a finite set of self-avoiding walks, $P_{\lambda}$ is a probability distribution that assigns mass to each walk $\omega$ proportionally to $x^{ |\omega|}$. Here $|\omega|$ counts the number of edges in the walk.
\end{defn}

It is known that the geometry of $P_{\lambda}$-typical chordal self-avoiding walks in $G_n$ display a phase transition as the parameter $\lambda$ is varied: depending on $\lambda$, for sufficiently large $n$, a path drawn from the distribution $P_{\lambda}$ will have certain properties with high probability.\footnote{We wish to thank a \href{https://mathoverflow.net/a/313003/41873}{helpful conversation} on MathOverflow for drawing our attention to this fact \cite{phasetransitionmathoverflow}.} To state this phase transition, we will recall an important constant called the connective constant of the square lattice.

\begin{thm}[\cite{madrasself}, Connective constant of the square lattice]
Let $c_n$ be the number of self-avoiding walks on the lattice $\mathbb{Z}^2$ that start at the origin.  The limit $\mu = \lim_{n \to \infty} \sqrt[n]{c_n}$ exists. $\mu$ is called the connective constant of the square lattice, and $\mu \approx 2.683$.
\end{thm}

The phase transition for the qualitative properties of $P_{\lambda}$ on chordal self-avoiding walks (SAW)
occurs at $\lambda = 1 / \mu$, which is called the critical ``fugacity''. In particular:

\begin{itemize}
    \item Subcritical fugacity $\lambda <  1 / \mu$: A $P_{\lambda}$-typical SAW resembles a geodesic on the grid graph \cite{duminil2014supercritical}. %
    
    \item Critical fugacity $\lambda = 1 / \mu$: A $P_{\lambda}$-typical SAW resembles a sample path from chordal $SLE_{8/3}$ \cite{lawler2002scaling, duminil2014supercritical}.
    
    \item Supercritical fugacity $\lambda > 1 / \mu$: A $P_{\lambda}$-typical SAW is ``space filling'', in a sense made precise in \cite{duminil2014supercritical}. %
    
\end{itemize}

\input{Sections/4GridGraph/6LitReview.tex}

\input{Sections/4GridGraph/5EmpiricalResults.tex}

%% file: Images/ShavedGrid.pdf_tex
\begingroup%
  \makeatletter%
  \providecommand\color[2][]{%
    \errmessage{(Inkscape) Color is used for the text in Inkscape, but the package 'color.sty' is not loaded}%
    \renewcommand\color[2][]{}%
  }%
  \providecommand\transparent[1]{%
    \errmessage{(Inkscape) Transparency is used (non-zero) for the text in Inkscape, but the package 'transparent.sty' is not loaded}%
    \renewcommand\transparent[1]{}%
  }%
  \providecommand\rotatebox[2]{#2}%
  \newcommand*\fsize{\dimexpr\f@size pt\relax}%
  \newcommand*\lineheight[1]{\fontsize{\fsize}{#1\fsize}\selectfont}%
  \ifx\svgwidth\undefined%
    \setlength{\unitlength}{588.65008671bp}%
    \ifx\svgscale\undefined%
      \relax%
    \else%
      \setlength{\unitlength}{\unitlength * \real{\svgscale}}%
    \fi%
  \else%
    \setlength{\unitlength}{\svgwidth}%
  \fi%
  \global\let\svgwidth\undefined%
  \global\let\svgscale\undefined%
  \makeatother%
  \begin{picture}(1,0.89467563)%
    \lineheight{1}%
    \setlength\tabcolsep{0pt}%
    \put(0,0){\includegraphics[width=\unitlength,page=1]{ShavedGrid.pdf}}%
  \end{picture}%
\endgroup%

%% file: Sections/4GridGraph/6LitReview.tex
\subsubsection{Relevant statistical physics literature}

It should come as no surprise that a Markov chain as natural as the flip walk has been investigated before, especially given the interest in the self avoiding walk model. Indeed, the plane dual of the flip walk moves were applied to the study of self avoiding walks in $\mathbb{Z}^2$ with fixed endpoints (but \emph{not} constrained to lie in a bounded region) in the BFACF algorithm \cite[Section 6.7.1]{sokal1994monte}. However, the state space of this walk is infinite, unlike our setting. It was proven that this walk has infinite exponential autocorrelation time \cite{sokal1988absence}. %
Various efforts were made to improve the mixing time of the BFACF algorithm \cite[Section 6.7.2]{sokal1994monte}, since physicists were interested in sampling from the stationary distribution in addition to observing the paths of the chain itself \cite{sokal1991beat}. Additionally, Markov chains on self avoiding walks have been considered in constrained domains \cite[p.69]{sokal1994monte} just as in our setting, but it appears that little is known. In a somewhat different direction, and conditional on conjectures about the asymptotics of $c_n$, there are rapidly mixing Markov chains for uniformly sampling from unconstrained, fixed length self avoiding walks starting at the origin with free endpoint, see \cite{sokal1989exponential, randall2000self}. %

Additionally, known bounds on self-avoiding walks provide estimates for the size of $P_2(L_n)$. For lower bounds, estimates on the number of self avoiding walks \cite{madras1995critical} can be used. For the upper bound, methods in section 5.1 of \cite{bousquet2005self} can be used. Interestingly, \cite{bousquet2005self} cites \cite{abbott1978lattice} as the inspiration for their method, and \cite{abbott1978lattice} was written to address the question of how many ways one could design districting plans for a grid-like state.

%% file: Sections/4GridGraph/5EmpiricalResults.tex
\subsection{Experiments around mixing}\label{subsection:gridgraphempirical}

Our goal in this section is to document experiments about MCMC methods built on top of the flip walk. These experiments were designed to investigate whether the chain was drawing samples from its stationary distribution. %
As the proposal distribution, we used the flip walk on the \emph{simply-connected} elements of $P_2(L_n)$. We used a Metropolis score function $S( (A,B) ) = \lambda^{- |\cut(A,B)|}$, so that the stationary distribution would be $\nu_{\lambda}$. %
To tune $\lambda$ around the critical fugacity $1/\mu$, we used the estimates $\mu \in [2.625622,  2.679193 ]$ and $\mu \approx 2.63815853031$ (\cite[p. 10]{jensen2004improved}, which references \cite{jensen2003parallel,ponitz2000improved}). We also imposed balancedness constraints that prevented any partition from having more than $X \%$ of the total number of nodes beyond the number of nodes in a perfectly balanced partition, which we call an allowed population deviation (APD) of $X \%$.

These experiments are recorded in \cref{fig:subandcriticalflips} and \cref{fig:supercriticalnotmixing}. %

\begin{figure}[t]
  \centering
  \captionsetup[subfigure]{width=.45\linewidth}
  \subfloat[\emph{Very low fugacity, without tight population constraints}
    Population deviation $90 \%$, $\lambda = 1/10$, Steps $=  5,072,065,569$.  This  quickly shrunk to a small bubble in the corner, around which the walk oscillated.]{\label{ref_label1}\includegraphics[width=.3\textwidth]{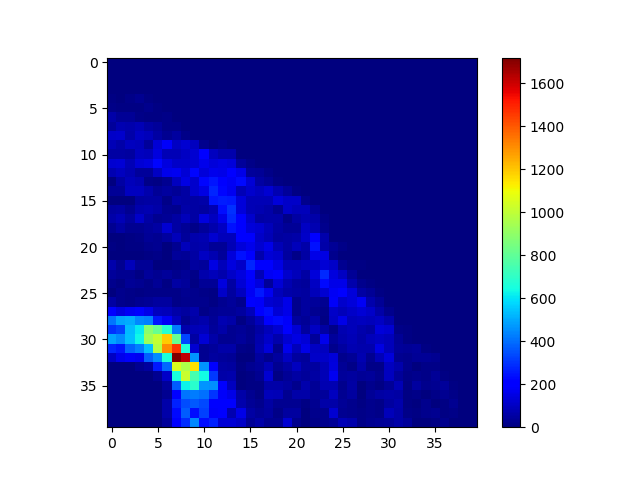}}
    \hspace{1.2in}
  \subfloat[\emph{Very low fugacity, with tight population constraints}
    Population deviation $10 \%$, $\lambda = 1/10 < 1 / \mu$, Steps $= 2,905,381,156$ This quickly rotated from the diagonal to be horizontal, and oscillated around that.]{\label{ref_label2}\includegraphics[width=.3\textwidth]{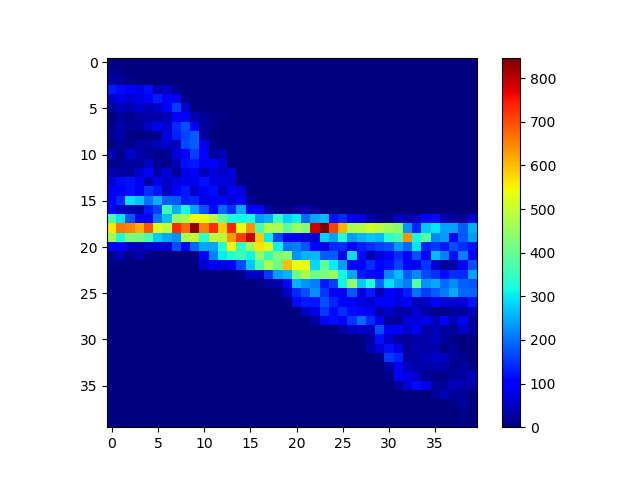}}

  \subfloat[\emph{Very low fugacity, with extremely tight population constraints}
    $APD = 1 \%$, $\lambda = 1/10 < 1 / \mu$, Steps $=2,413,374,064$.
    The population restriction made it more difficult to escape the diagonal configuration.]{\label{ref_label3}\includegraphics[width=.3\textwidth]{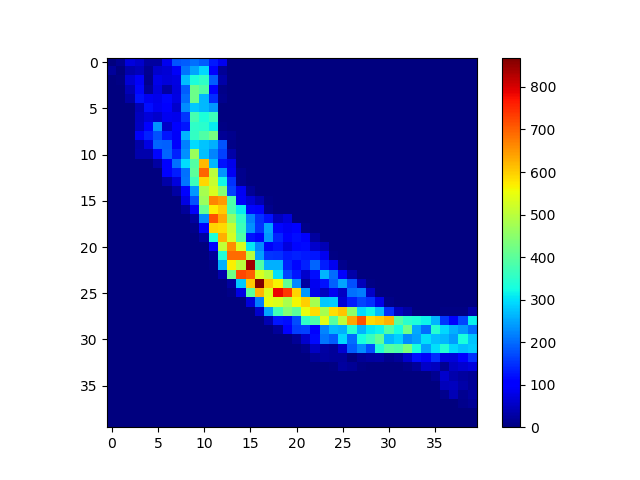}}
    \hspace{1.2in}
  \subfloat[    \emph{Critical fugacity, with loose population constraints}
    $APD = 50\%$, $\lambda = 1 / \mu$, Steps $= 1,405,646,093$.
    The walk oscillated around the diagonal.]{\label{fig:criticalfugacityloose}\includegraphics[width=.3\textwidth]{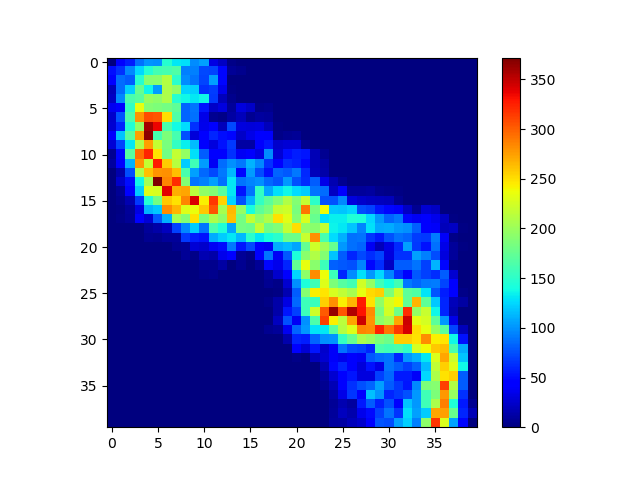}}
      \caption{\label{fig:subandcriticalflips}These examples start from an upper left to bottom right diagonal partition of a $40 \times 40$ grid graph. Each node kept track of the number of times it was flipped, and this number is reported and colored according to the key. Some interpretation of the history of the path revealed by the figures is provided. In all of these examples a symmetry argument demonstrates that the chain has evolved into some metastable region $P_2$.%
      }
\end{figure}

\begin{figure}
    \centering
    \begin{tabular}{cc}
        \includegraphics[width=.3\textwidth, height = .3\textwidth, angle = 90]{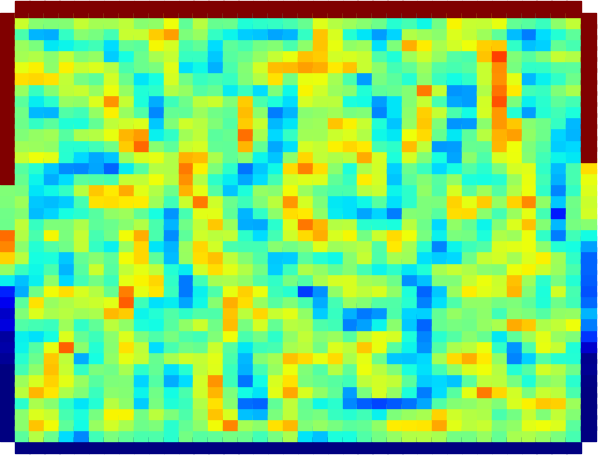} & \includegraphics[height=.3\textwidth]{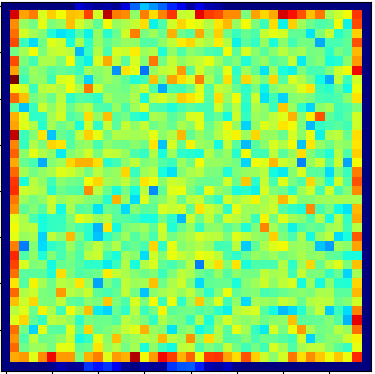}\\
        a) Average block & b) Flips
    \end{tabular}
  \caption{Two measurements of a single run. $APD = 90\%$, $x = 1$, steps $=194,390,536$, started from a vertical partition into red and blue blocks. a) Each node displays the color of its average block: more red if frequently in the red block, and more blue if frequently in the blue block. The left side, which started red, stays red through the entire run, while the right side stays blue. Considering the rotational symmetry and length of the run, this is a strong indication of a large mixing time. b) Each node records the number of times it flipped, with red indicating $100+$, green/yellow indicating around $60-80$, and blue indicating $<20$. One can see that most of the activity happens near the boundary of the square, and that the endpoints of the boundary of the partition barely move.}\label{fig:supercriticalnotmixing}
\end{figure}

\input{Sections/4Empirical/4RealStates.tex}

\subsection{Graph topology and phase transitions}\label{section:phasetransitions}

So far we have discussed phase transitions of self-avoiding walks and connected $2$-partitions in the grid graph, and remarked that the critical fugacity occurs at $1/ \mu \approx .379$. However, there are other lattices, and for them the phase transitions occur at different values. %
Thus, one may wonder about the behavior of a self avoiding walk in a ``Frankengraph'' such as in \Cref{fig:Frankengraph}, which consists of triangles on the top, and squares on the bottom. In \Cref{fig:FrankenGraphExp} we show that we can find a value $\lambda$ where the part of the partition boundary in the square grid acts super critically, and the part in the triangular acts subcritically.\footnote{Connected $2$-partitions of the triangular lattice correspond to self avoiding walks in the dual lattice, which is a hexagonal lattice. Since \cite{duminil2012connective} puts the connective constant of the hexagonal lattice at $\sqrt{2 + \sqrt{2}}$, the phase transition for partitions of the triangular part occurs around $.541$. The authors wish to thank \href{https://people.eecs.berkeley.edu/~sarah.cannon/index.html}{Sarah Cannon} for discussions that clarified this behavior.}

These experiments about phase transitions in geographic compactness scores fit in with other observations about compactness, namely that many features of the scores are not robust under changes of scale, geometry or other implementation parameters \cite{duchin2018discrete, bar2019gerrymandering, barnes2018gerrymandering}. Here we have raised an additional issue, which is that calibration of compactness score parameters in relation to the topology of the underlying graph can have dramatic affects on ensembles. These observations should be considered in light of \cite{barnes2018gerrymandering}, which analyzed the impact of decisions regarding the calculation of these compactness scores, and found that apparently small choices in implementation could have large effects.

\begin{figure}[H]
    \centering
    \includegraphics[scale = .3]{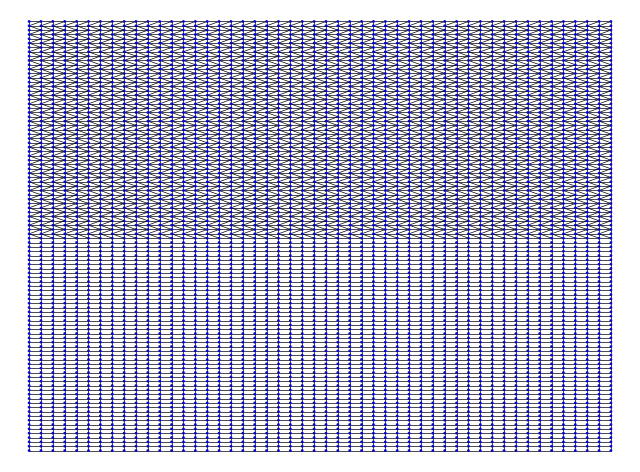}
    \caption{Frankengraph. The top is is a triangular lattice and bottom is a square lattice. Both pieces are $50 \times 50$. See \cref{fig:FrankenGraphExp} for results of running the flip walk.}
    \label{fig:Frankengraph}
\end{figure}

\begin{figure}
    \centering
    \begin{tabular}{cc}
       \includegraphics[width = .3\textwidth, height=.3\textwidth]{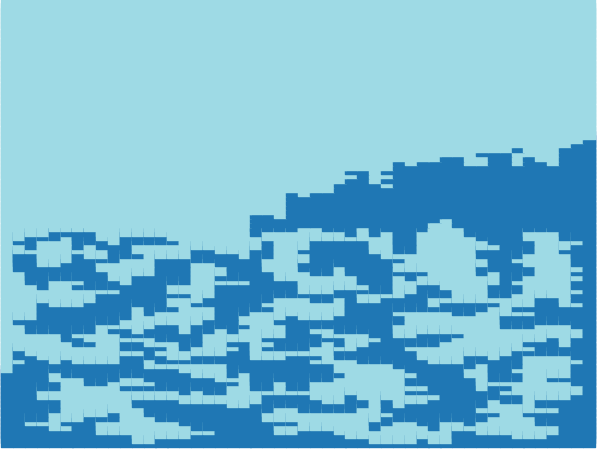}  &  \includegraphics[width =.3\textwidth, height=.3\textwidth]{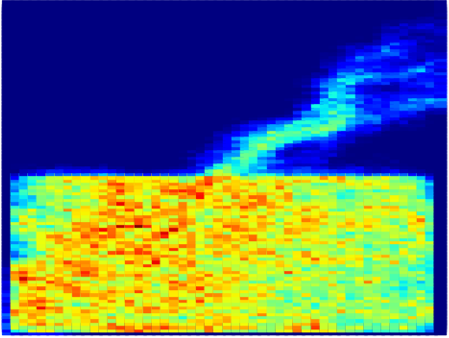}
    \end{tabular}
    \caption{A flip walk based MCMC run with $\lambda = 1/2$, starting from the diagonal partition of the Frankengraph. $90 \%$ population constraint. $17,768,956,990$ steps.}
    \label{fig:FrankenGraphExp}
\end{figure}

\input{Sections/4Empirical/5TheChoiceOfModel.tex}

%% file: Sections/4Empirical/4RealStates.tex
\subsubsection{Kansas State-Dual Graph}\label{subsection:empiricalrealstates}

In this section, we repeat the experiments we performed on the grid graph the state dual graphs (\Cref{section:CongressionalMotivation}) of Kansas, which was chosen because its state dual graph resembles the grid graph.
For ease of visualization, we display the partitions and flip statistics on the underlying map rather than the state dual graph. %
The main features to observe are the slow mixing of the chain around the state space, and the phase transitions. See \cref{fig:kansas_experiemnts}. %

\begin{figure}
    \centering
    \begin{tabular}{cc}
\includegraphics[width=.4\textwidth]{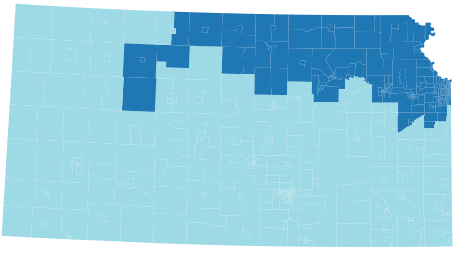}& \includegraphics[width=.4\textwidth]{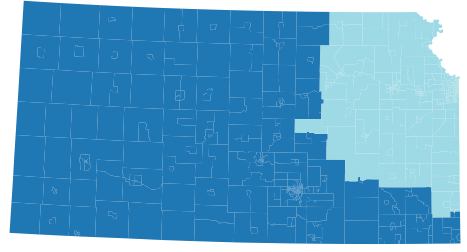} \\
\includegraphics[width=.4\textwidth]{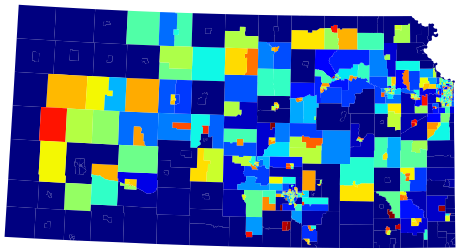}  &   \includegraphics[width=.4\textwidth]{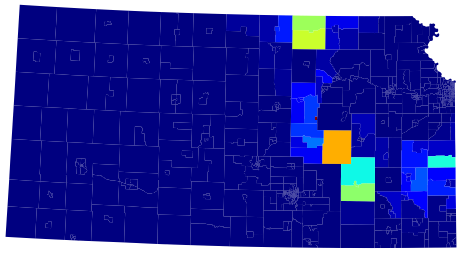} \\
\includegraphics[width=.4\textwidth]{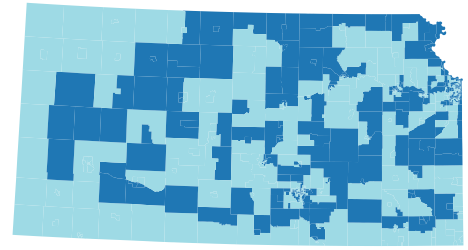} & \includegraphics[width=.4\textwidth]{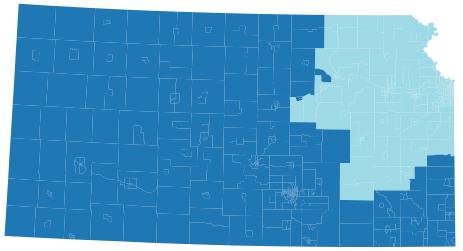}\\
    \text{a) $\lambda = 1$, $APD = 90 \%$. Steps$=93,415,894$} & \text{b) $\lambda = .379$, $APD = 90 \%$. Steps$=1,744,003,380$ } \\
    \end{tabular}
    \caption{Two runs of flip walk based $MCMC$ on the state dual graph of Kansas. Top: Starting plan. Middle: Counting Flips. Bottom: Ending plan.}
    \label{fig:kansas_experiemnts}
\end{figure}

%% file: Sections/4Empirical/5TheChoiceOfModel.tex
\subsubsection{The choice of model graph}\label{section:thechoiceofmodel}

The example of the $\lambda^{\cut}$ distribution and the Frankengraph teaches us that the choice of graph used to discretize the same underlying geography can dramatically affect the distribution over partitions produced by a fixed algorithm. However, in those cases the geographic reasonableness of the distribution also changed dramatically, making it easy to classify a \emph{single} partition as arising from one discretization or the other. In this section, we look at a different distribution over partitions of a rectangle, with the property that changing the discretization noticeably changes the distributions, but such that differences between the two distributions cannot be easily detected by observing natural geometric properties of individual plans.

We will now explain this other sampling algorithm, which underlies the method used in \cite{deford2019redistricting, recomb}. Let $\UST$ be an algorithm that takes a graph $G$ and returns a uniform spanning tree, for example using Wilson's algorithm \cite{wilson1996generating}. Let $\MST$ refer to the minimum spanning tree obtained by picking iid $\Uniform([0,1])$ edge weights. Let $\mathrm{Tree}(G)$ refer to either $\UST(G)$ or $\MST(G)$. Removing an edge $e$ from $\Tree(G)$ gives a forest with two components, and hence an element of $P_2(G)$. If we repeat this algorithm, only selecting those edges that provide an $\epsilon$-balanced connected $2$-partition of $G$, then we obtain a distribution over $P^{\epsilon}_2(G)$. %
For both $\UST$ and $\MST$, this distribution over partitions has some favorable properties, such as its concentration on partitions with small edge-cuts, which have made it appealing as a tool for sampling districting plans \cite[Section 3.1.1]{deford2019redistricting}. \footnote{A fact partially reflected in \cite[Corollary 2] {kenyon2000asymptotic}; the connection being that the probability that a $\UST$-partition $(A,B) \in P_2(G)$ is chosen is proportional to $T(A) T(B) cut(A,B)$, where $T(\_)$ counts the number of spanning trees, and based on this one can rearrange the asymptotics in Kenyon's paper to deduce that among the rectilinear partitions those with smaller perimeter are asymptotically preferred. However, rectilinear partitions are a set of extremely small measure in the $\UST$-partition distribution, and so this explanation for concentration of $\UST$-Partition on smaller fundamental cut-sets is only partial.} %
We will call this distribution $\UST$-partition or $\MST$-partition, or $\Tree$-partition if we refer to either.%

We construct a sequence of graphs from the $36 \times 36$ grid graph $G$ by by triangulating a set of its faces in the following way. Fix some $w \in [0,3]$, and suppose that the nodes of the grid graph are labelled by $(i,j)$, $i,j \in [0,35]$. For each $(x,y) \in V(G)$, if $12 \leq y \leq 20$ and $0 \leq x \leq 6w$ or $34 - 6w \leq x \leq 34$, add an edge $((x,y), (x + 1, y +1))$ if $x$ is even, and an edge $((x,y), (x + 1, y - 1))$ if $x$ is odd. We call the resulting graph $G_w$, and $w$ is referred to as the \emph{width}. We think of partitions of $G_w$ as modelling the same state, but with different choices regarding adjacency between geographic units. By changing the width, we can see a change in the shape of a typical $\Tree$-partition. The results are displayed in \Cref{fig:treegate}, where we use a number in $[0,3]$ to quantify the width of each half of the gate; the entire graph has dimensions $36 \times 36$, and the length of each half-gate is $width \times 6$. %
The effect is very clearly that closing the gap in the squeezes the boundary of the partitions in between the two halves of the gate. %

In \Cref{fig:treegatevotes} we show that $\MST$-partitions and $\UST$-partitions produce different partisan outlier measurements, despite the similarity between the description of the algorithms. %
In particular, if the underlying graph is is the one with width $2$, and if the distribution chosen as a baseline for outlier analysis is $\MST$-partition, then most of the time the seat share is $1$, and many random samples from $UST$-partition would be considered extreme outliers. We discuss this further in \Cref{section:EOH}.

We recall from \Cref{section:CongressionalMotivation} that in the analysis of a districting plan, one starts with a geographic entity, a U.S. state, that is broken into small geographic units. The choice of how the state is broken into small units determines an adjacency graph, and we modelled partitions of this state as connected partitions of that adjacency graph. If we intend to draw conclusions about the political geography of the underlying geographic entity, then one might hope that the impact of this modelling step is relatively tame. \Cref{fig:treegatevotes} shows that this is not the case for the two examples above, potentially muddying the description of what outlier methods measure.

In political redistricting, there is occasionally reason to refine the basic geographic units, for example to balance population. A similar situation where changes to graph topology can occur is where choices have to be made about which units to connect across bodies of water, such as in \cite{calderamathematics}, or when deciding between using rook or queen adjacency in the construction of the state dual graph.\footnote{That is, between declaring two units adjacent if they have a common edge, or if they have any common point.} There are also a variety of resolutions on which maps can be viewed, from counties to tracts to census blocks. All of this could potentially impact outlier analysis in a way similar to the Frankengraph example and \Cref{fig:treegate}. Thus, someone sampling partitions of state dual graph to investigate the political geography of redistricting should keep in mind that they are making a potentially significant choice at the level of the choice of model graph. 
To comport with best practices in statistics \cite{gelman2017beyond}, these decisions should be made as transparently and impartially as possible.

In many U.S. states, precincts, which are the atomic geographic units where electoral data is reported, are drawn at the discretion of the local municipalities. The examples in this section show that the people who choose the smaller geographic units potentially have a lot of control over the results reported by outlier methods. If outlier methods become standard, it would open the possibility of \emph{metamandering} through the deliberate manipulation of these boundaries. No comprehensive analysis has been done to understand the impacts of these decisions.

\begingroup
\setlength{\tabcolsep}{.05pt}
\renewcommand{\arraystretch}{.5}
\begin{figure}[t]
    \centering
    \begin{tabular}{cc|cc}
       \includegraphics[scale = .1, angle=90]{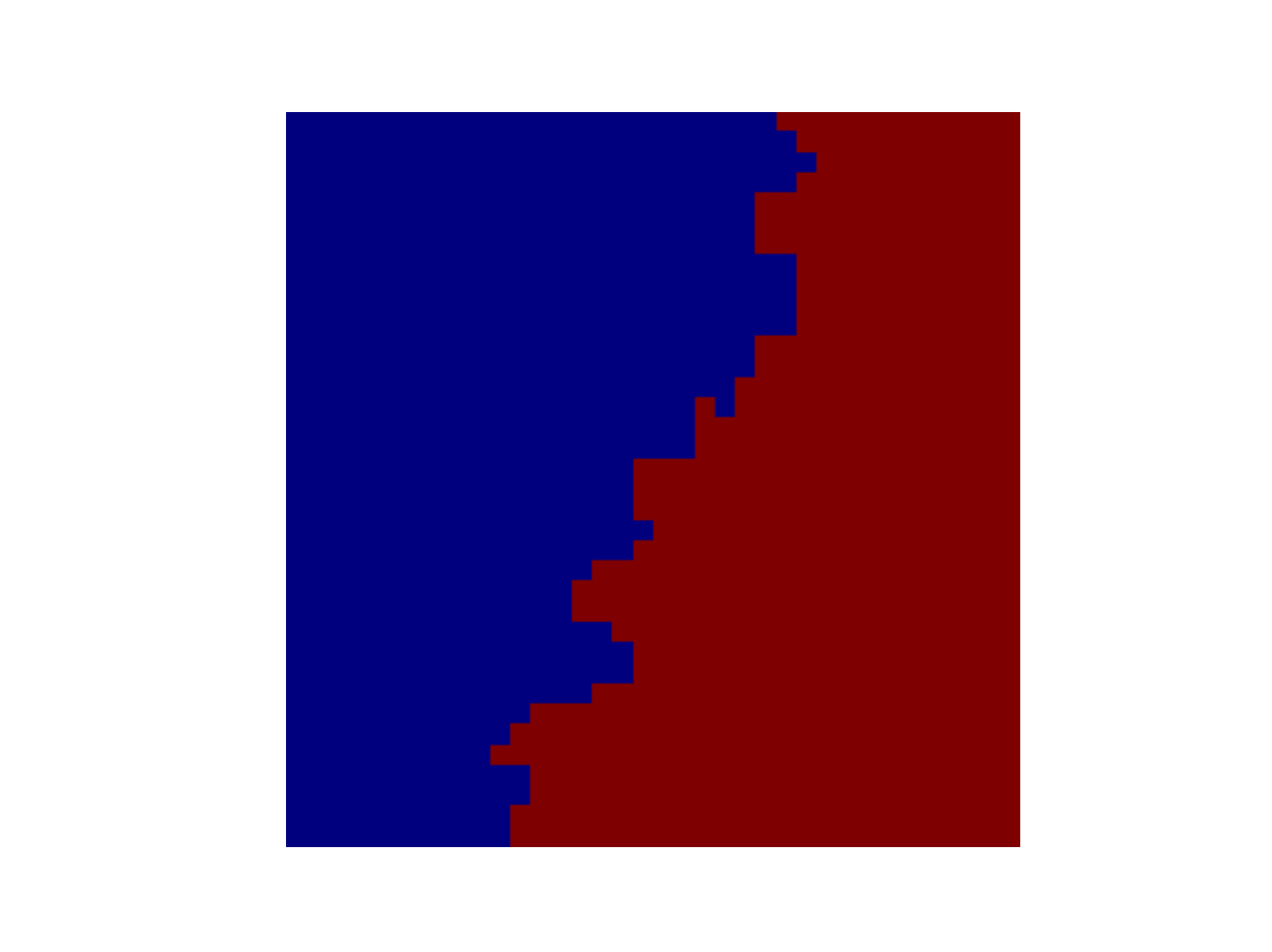} &  
       \includegraphics[scale = .2]{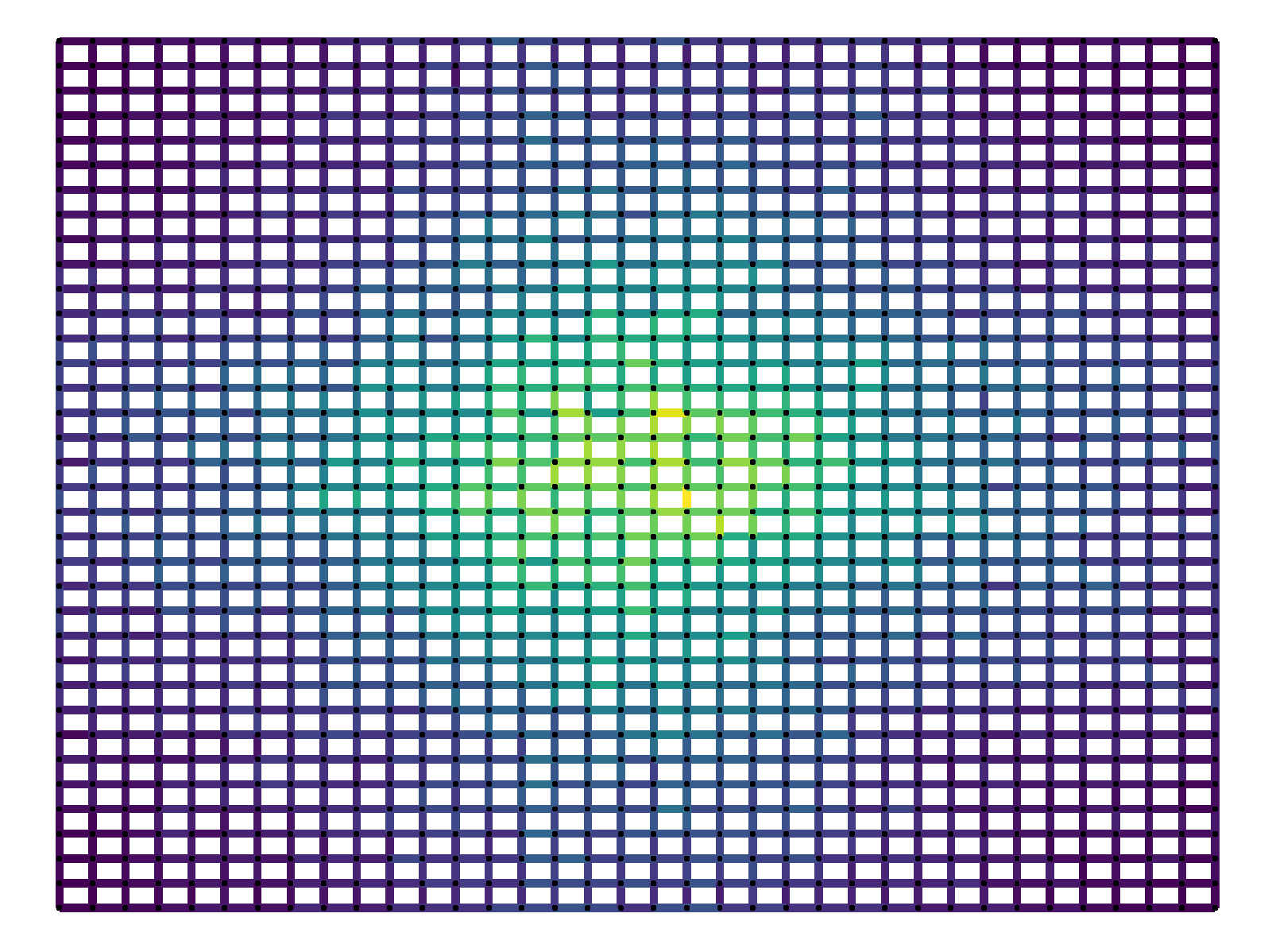}
       & \includegraphics[scale = .1, angle=90]{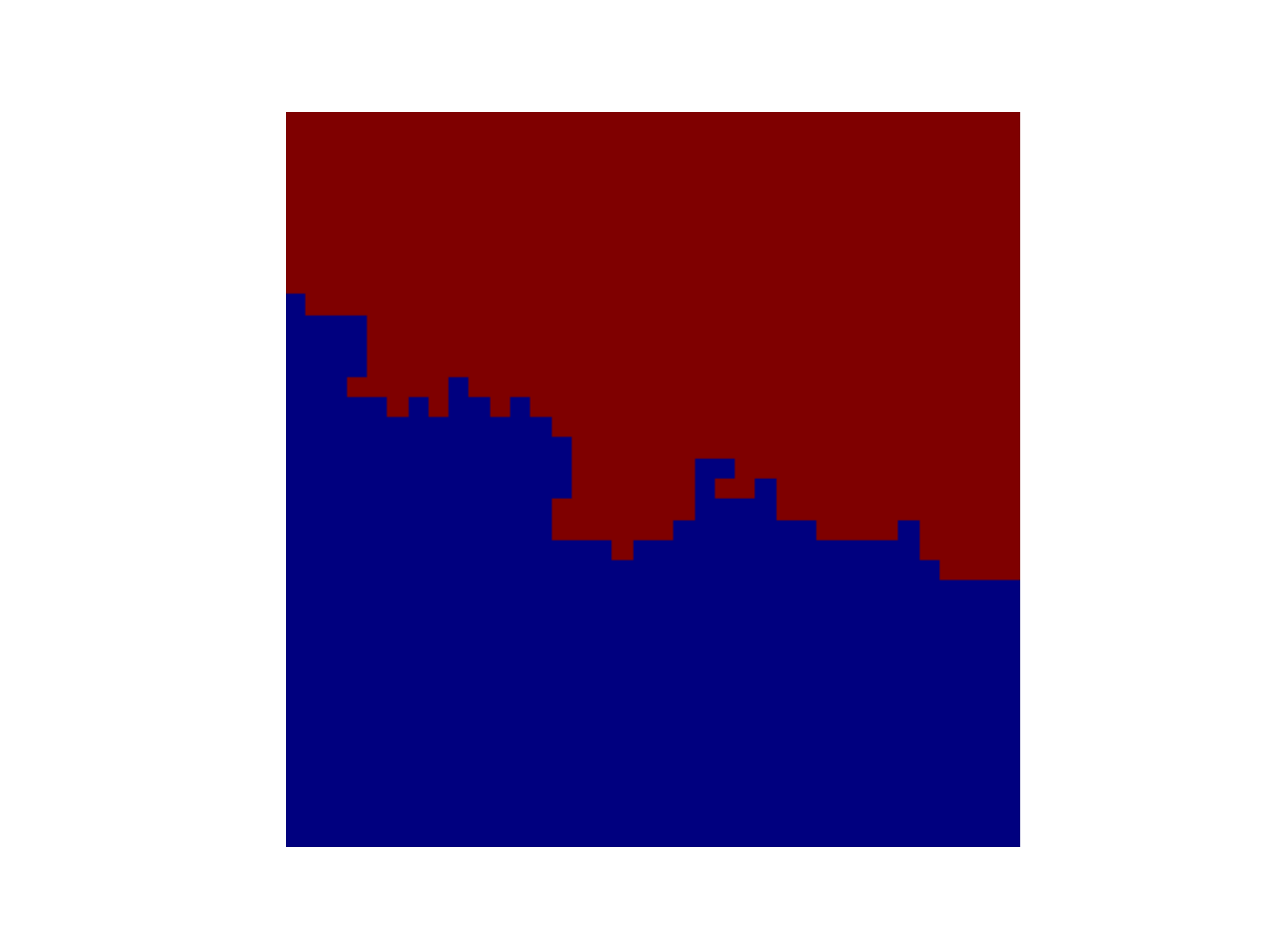} &
       \includegraphics[scale = .2]{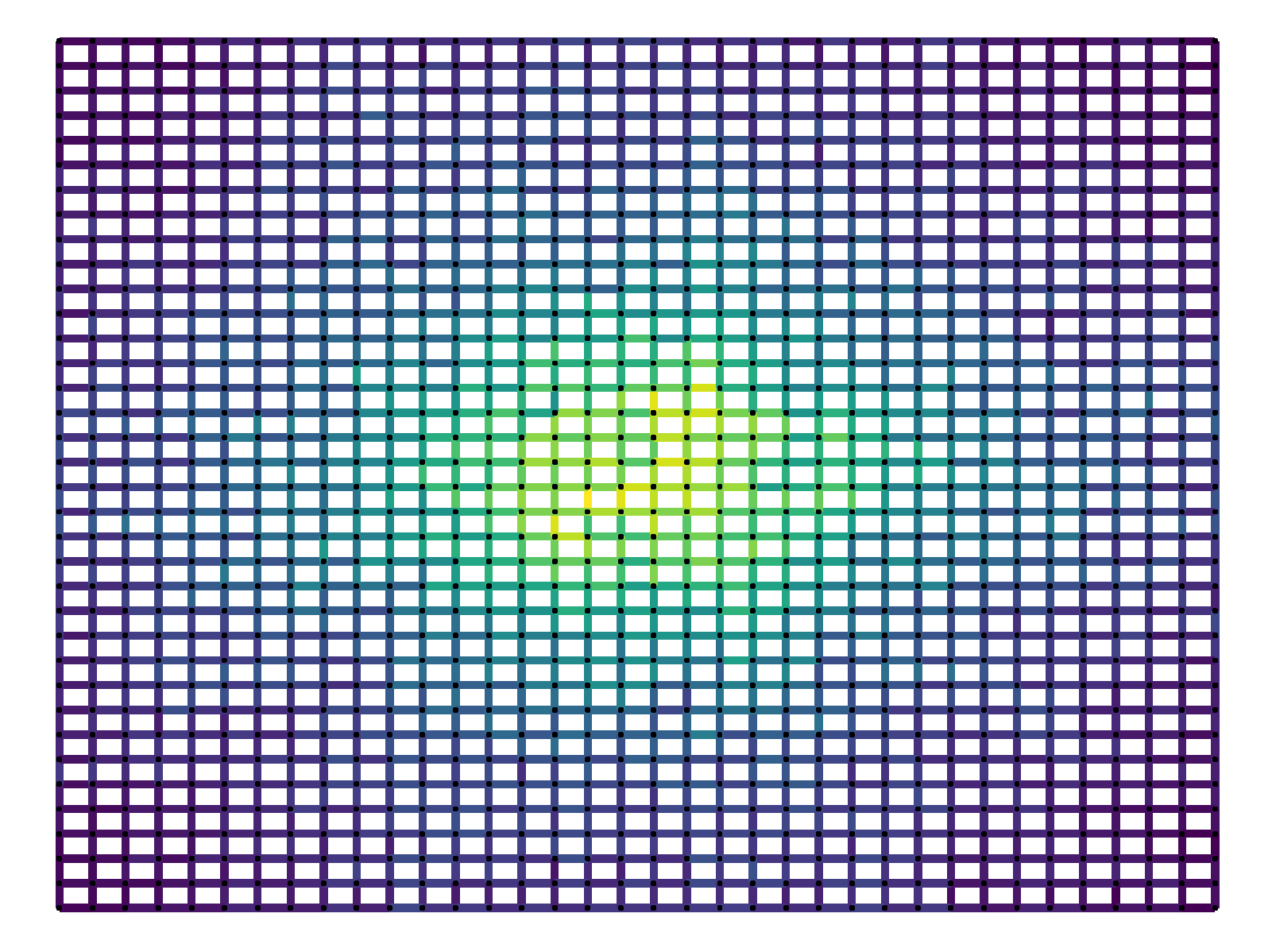}
       \\
       \includegraphics[scale = .1, angle=90]{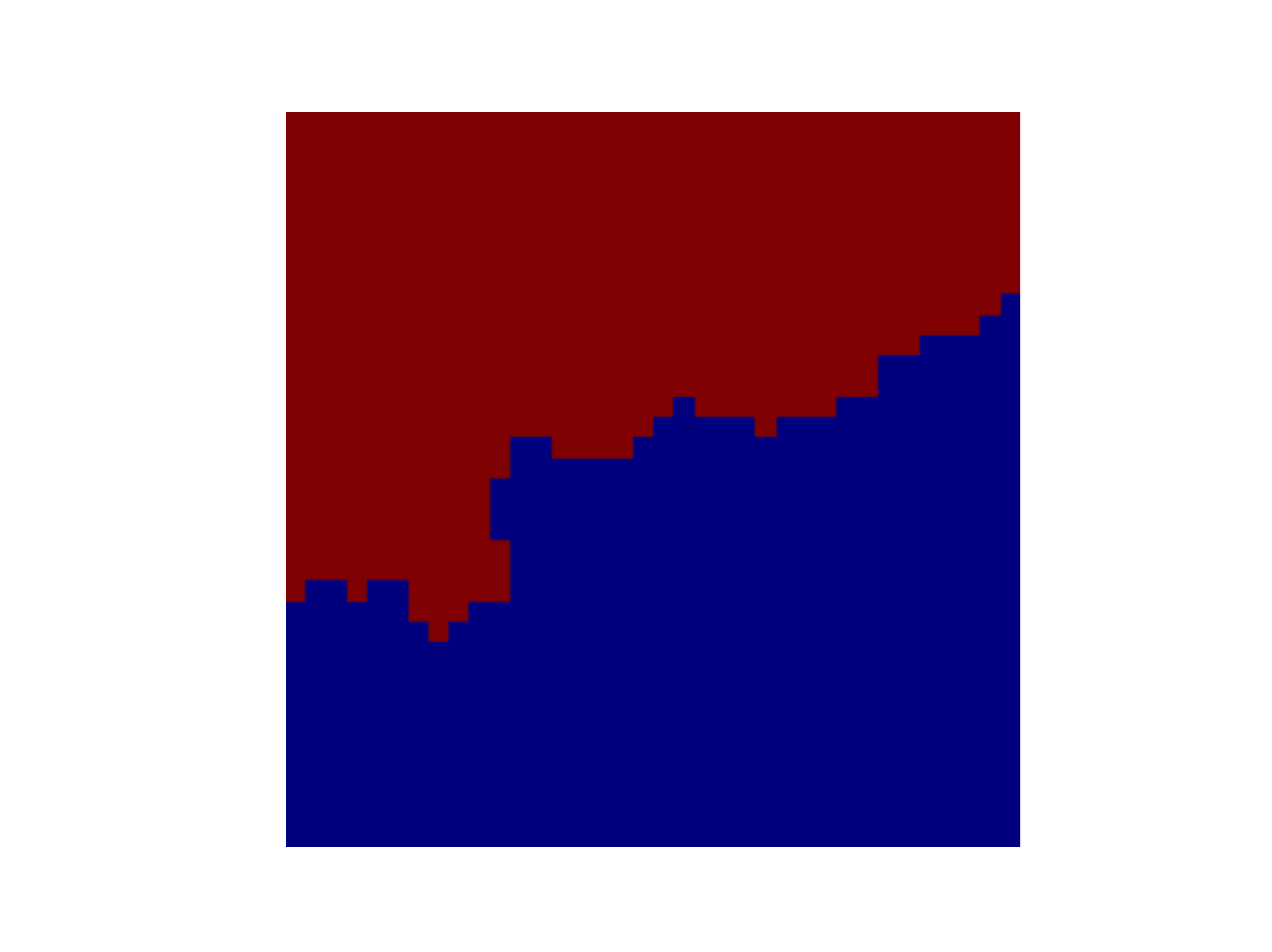} &
       \includegraphics[scale = .2]{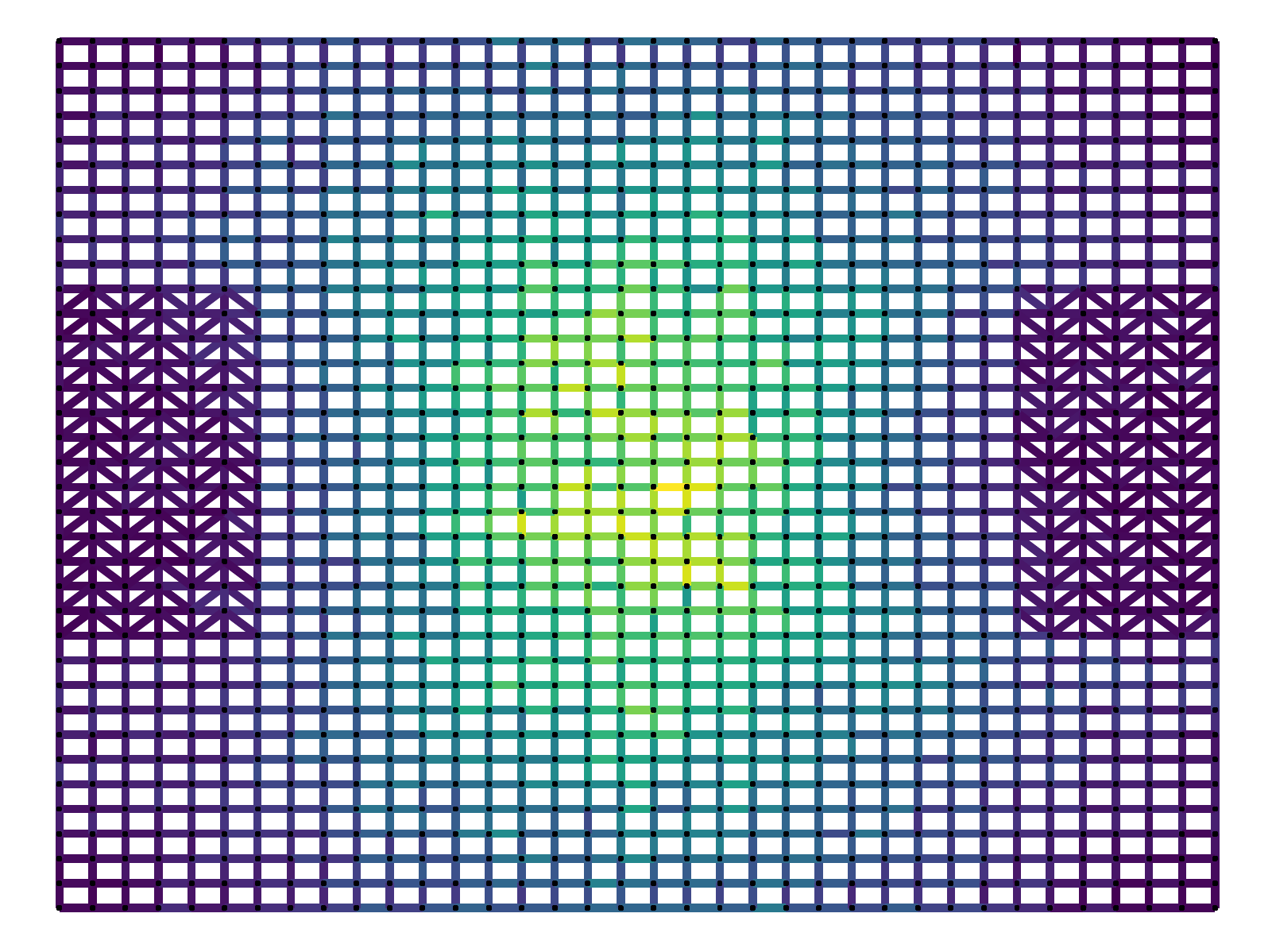}  & \includegraphics[scale = .1, angle=90]{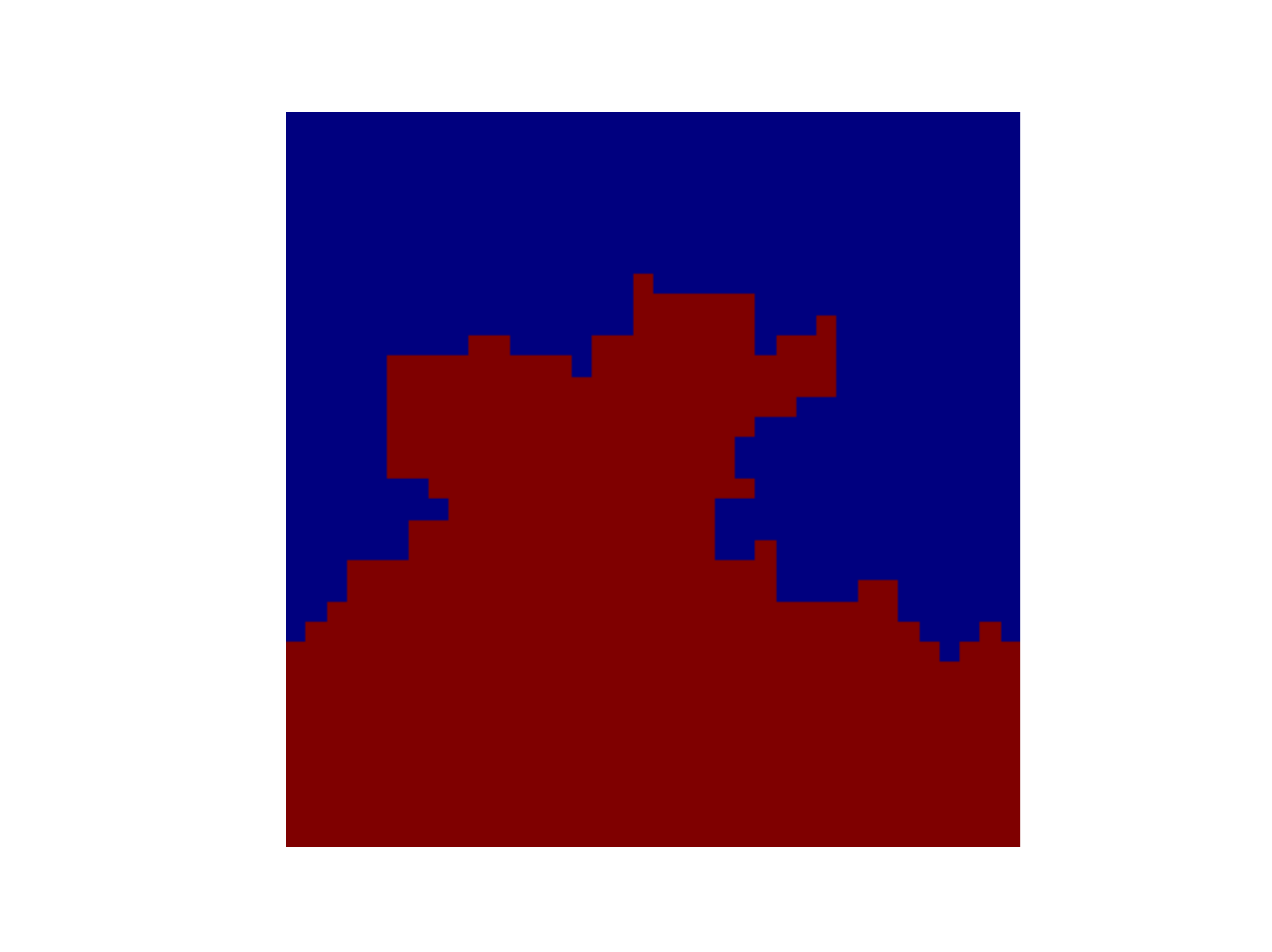}  &
       \includegraphics[scale = .2]{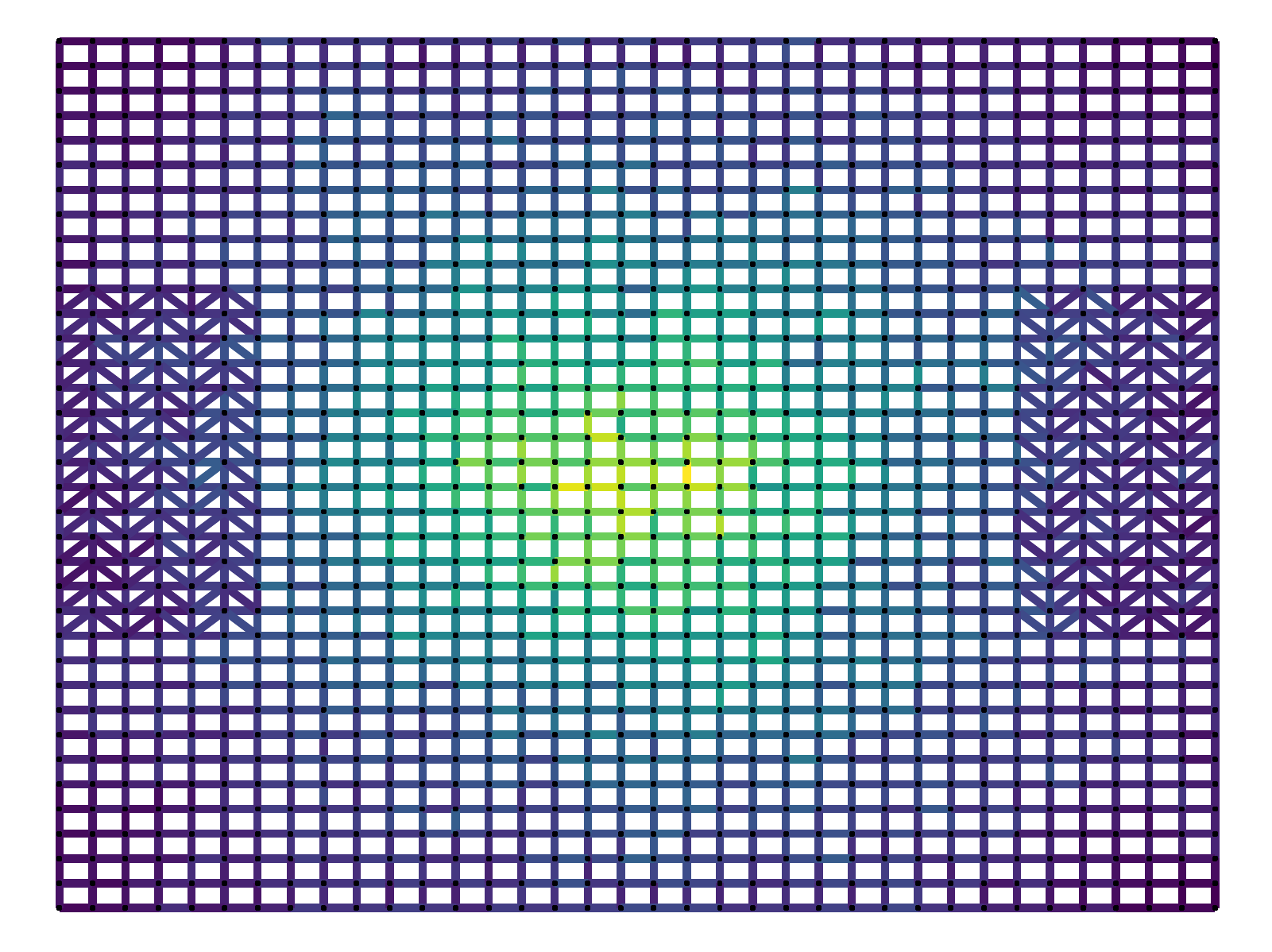} \\
       \includegraphics[scale = .1, angle=90]{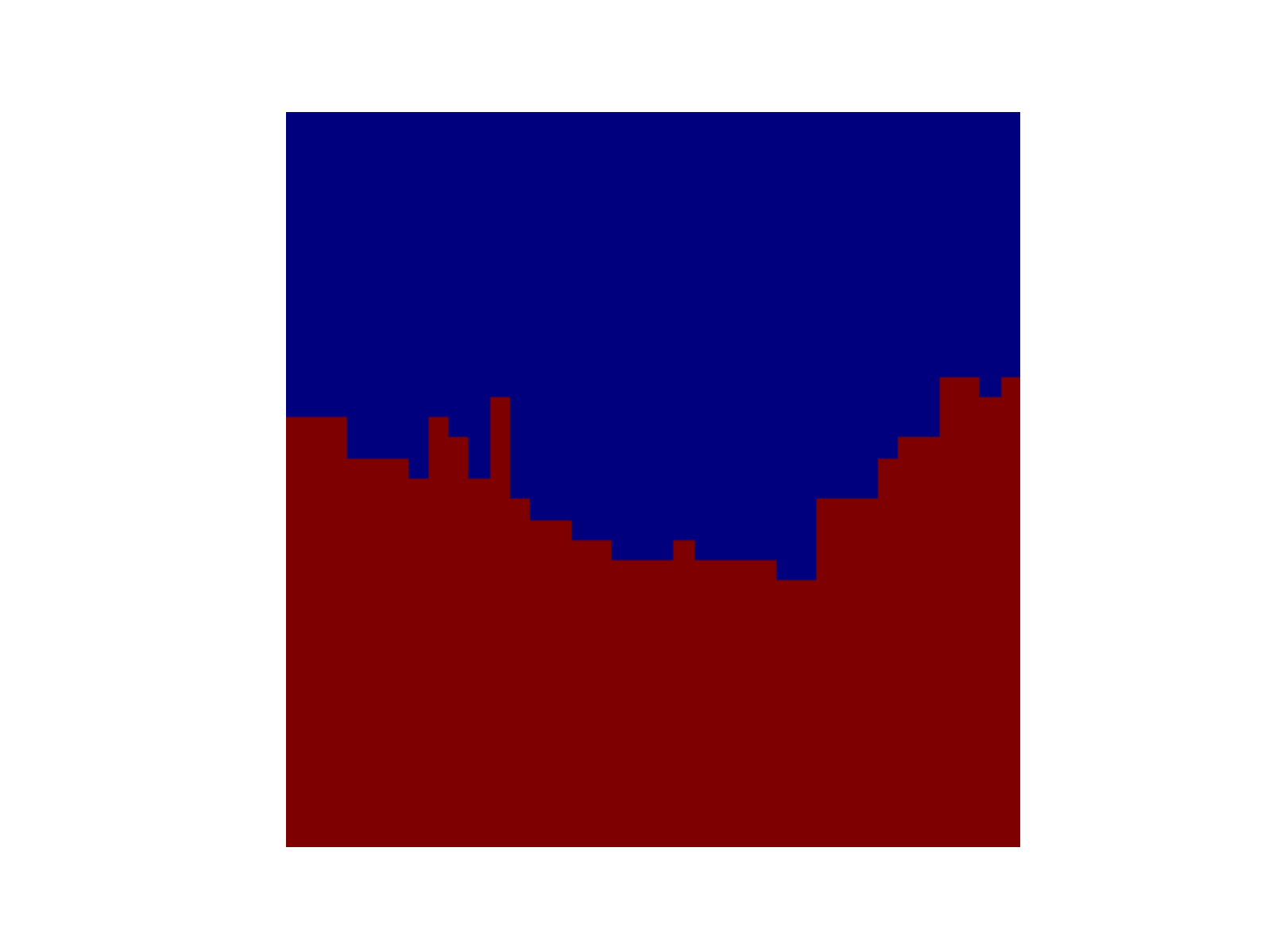} &
       \includegraphics[scale = .2]{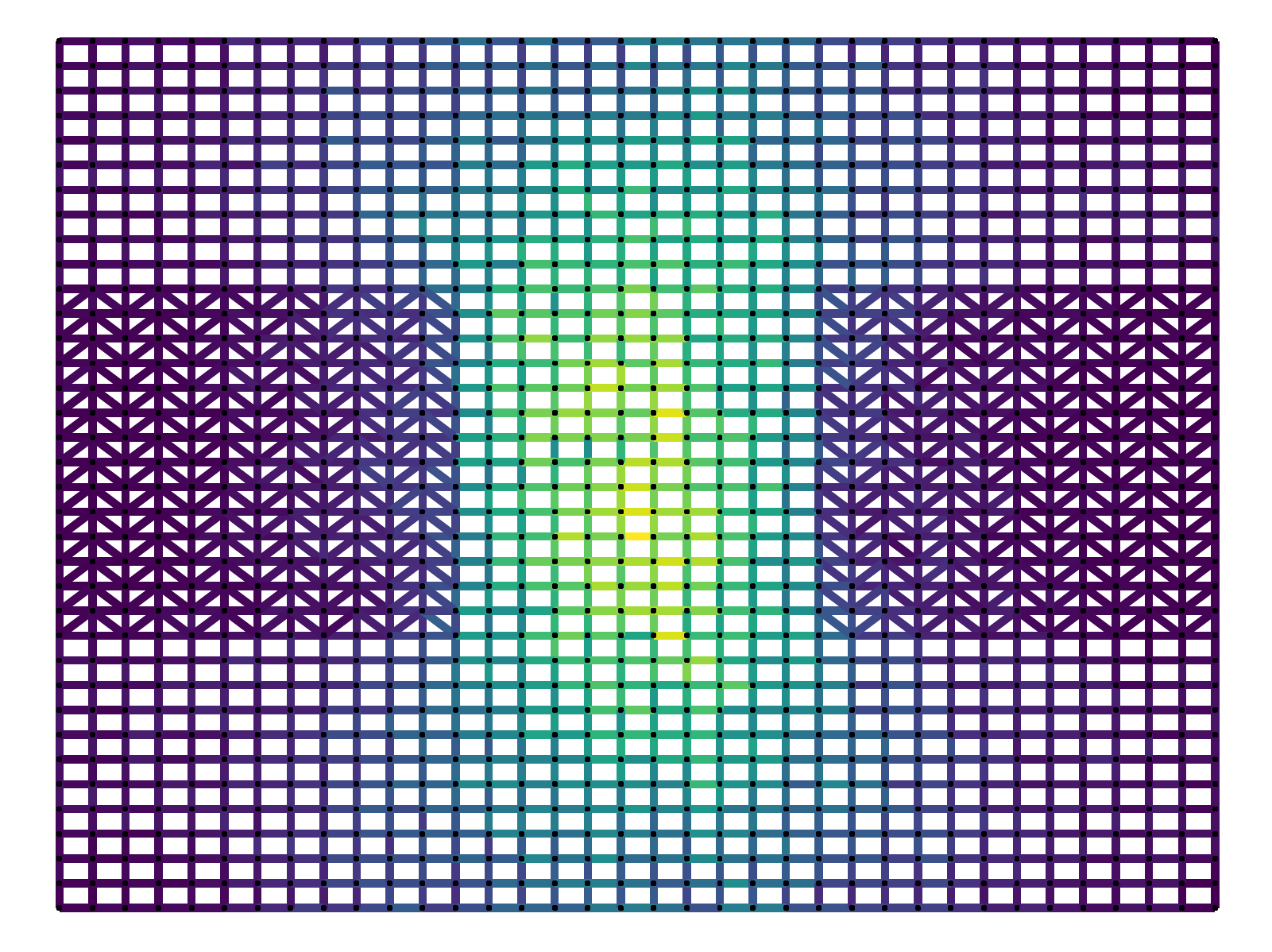}  & \includegraphics[scale = .1, angle=90]{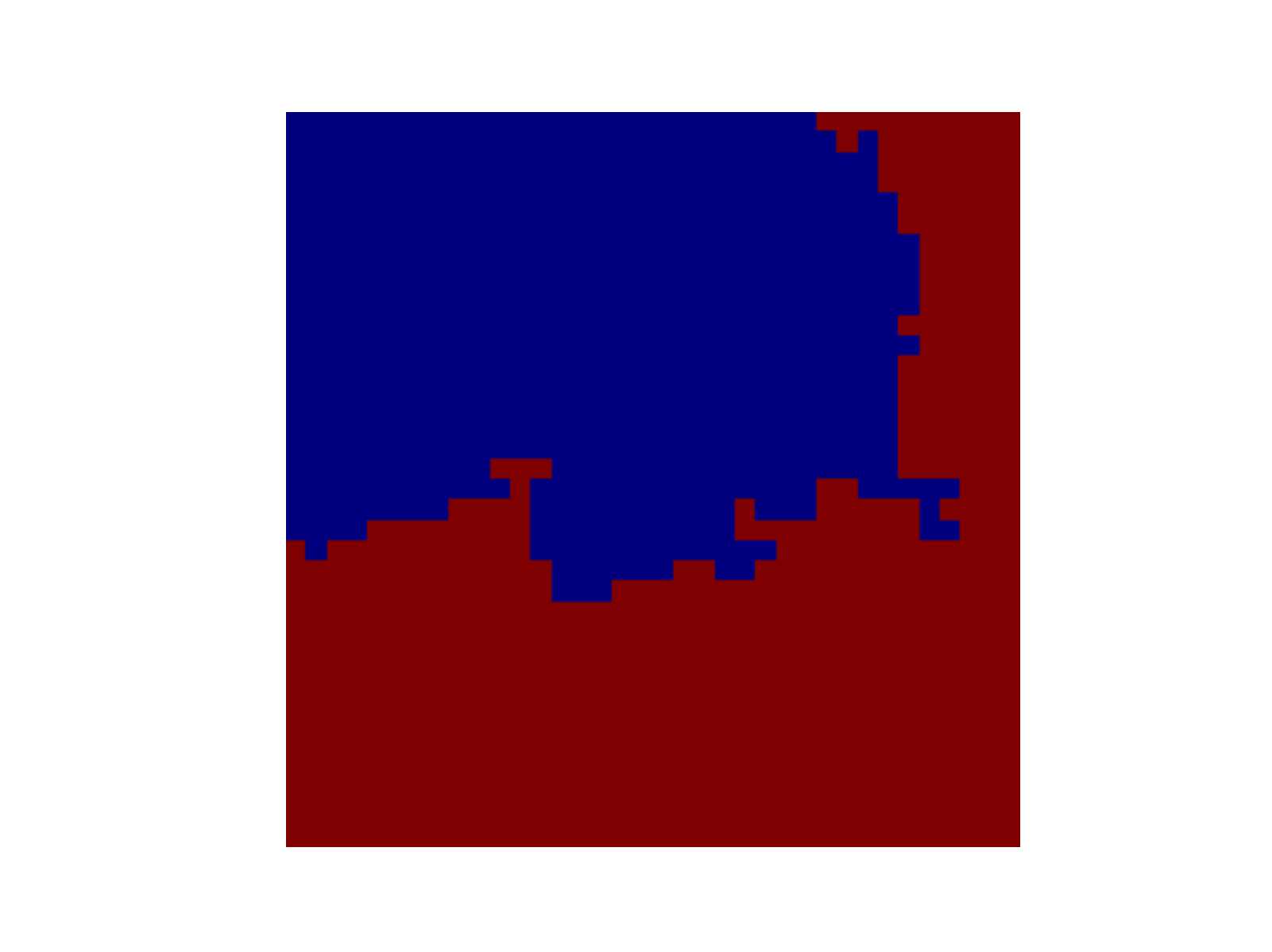}  &
       \includegraphics[scale = .2]{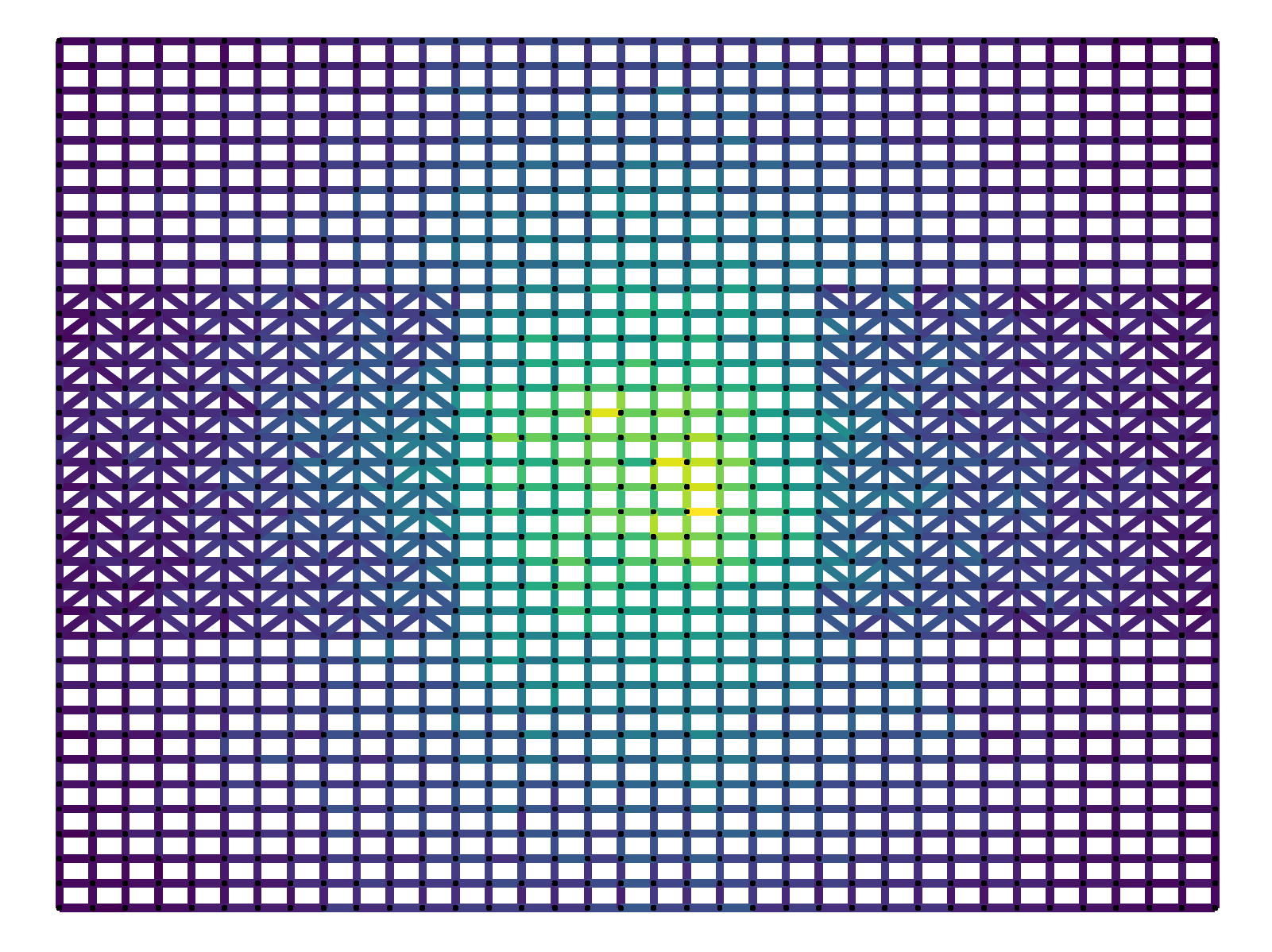} \\
       \includegraphics[scale = .1, angle=90]{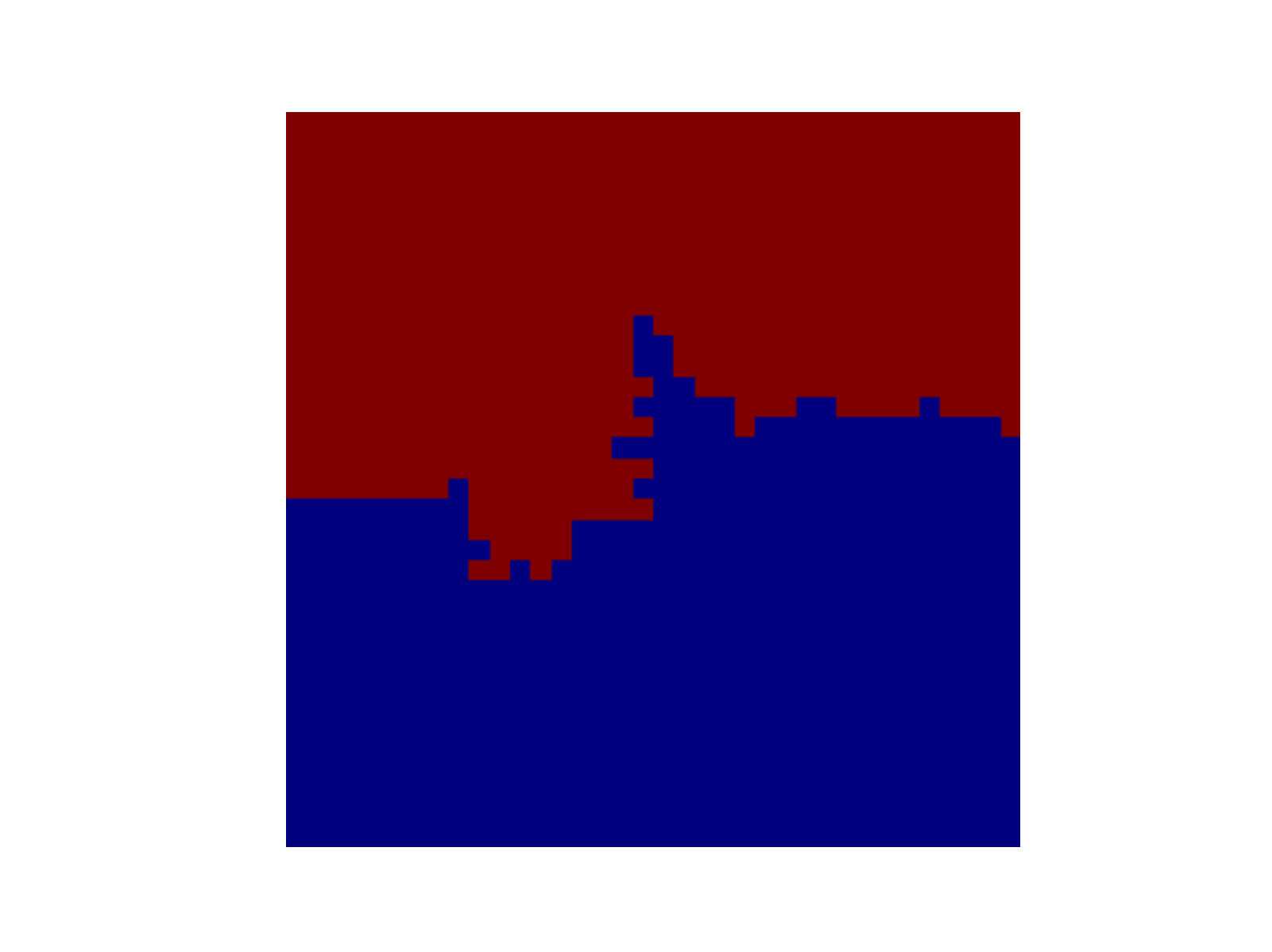} &
       \includegraphics[scale = .2]{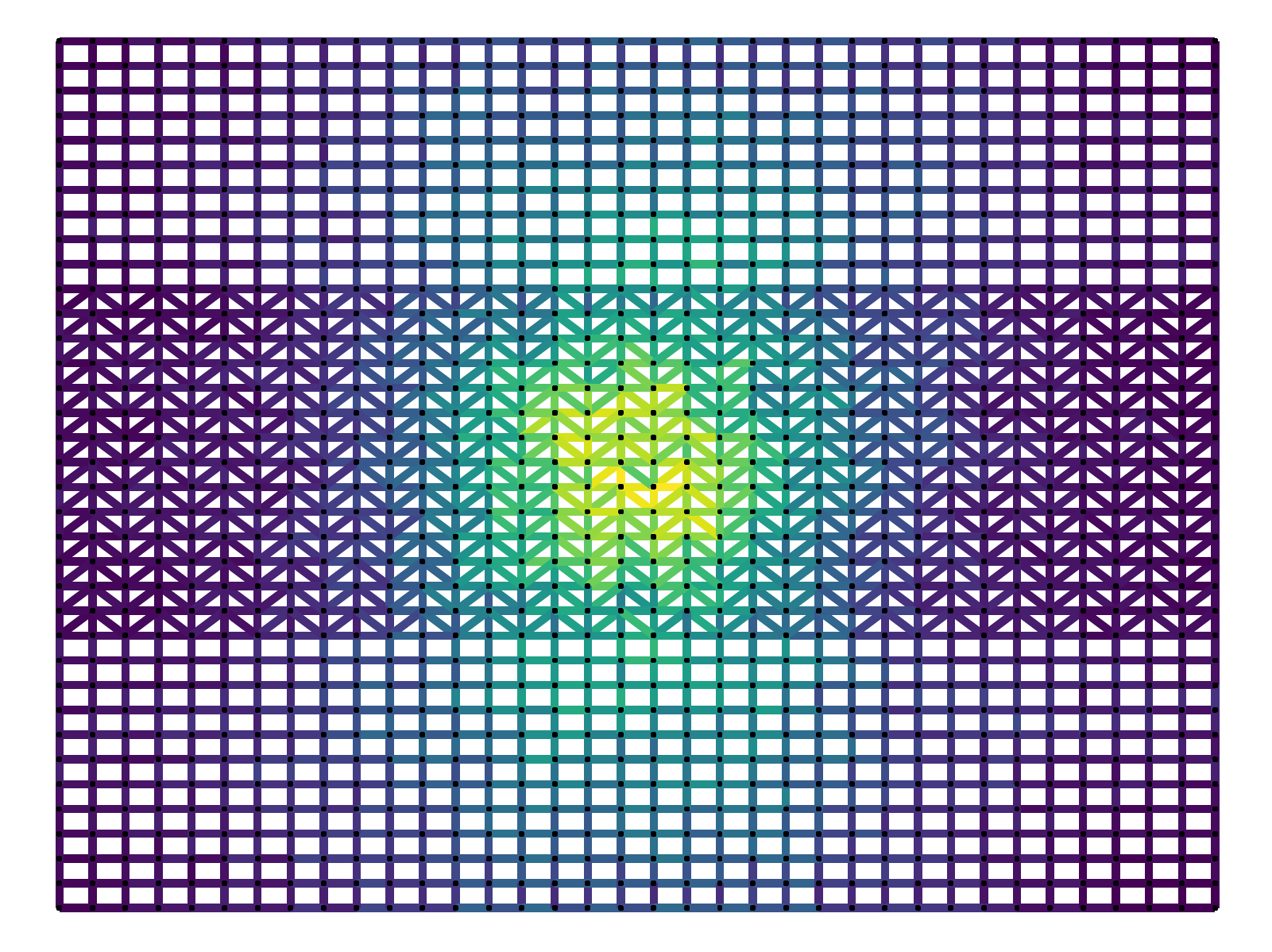}  & \includegraphics[scale = .1, angle=90]{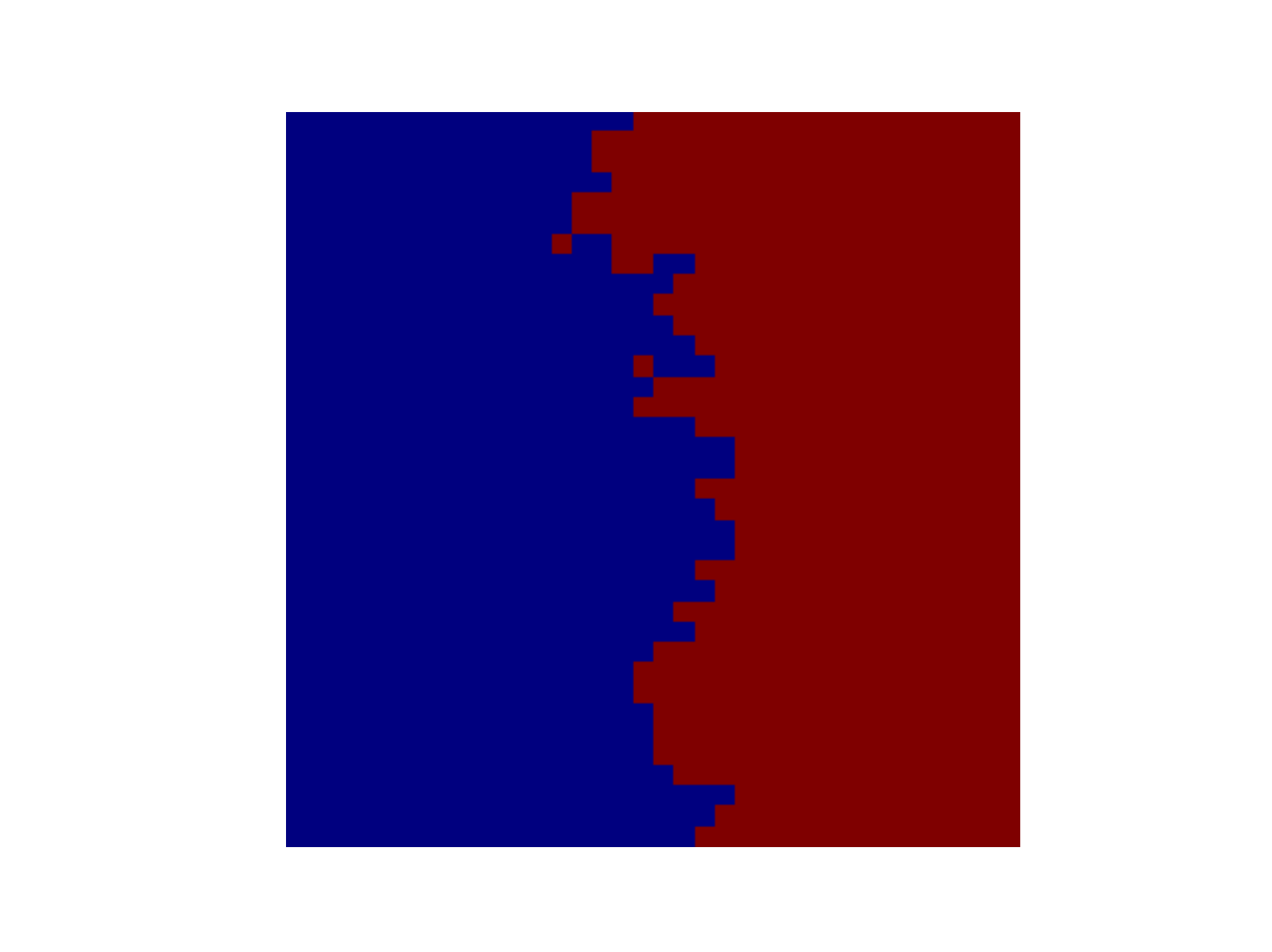}  &
       \includegraphics[scale = .2]{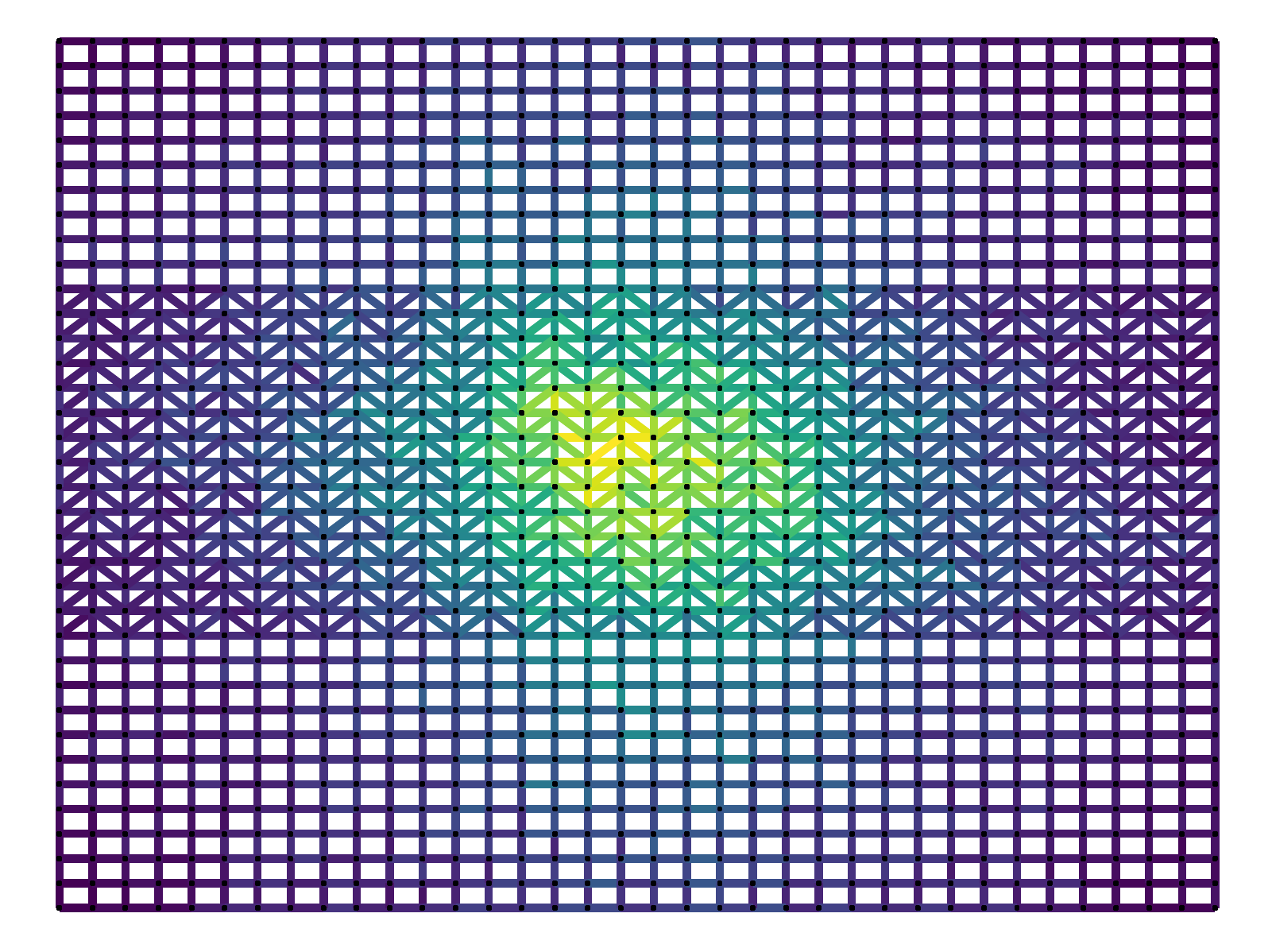} \\

    \end{tabular}
    
    \caption{An edge is yellower if it is more frequently a cut-edge of sampled partition. The left hand diagrams were made using $\MST$-partition, and the right hand side with $\UST$-partition. Both use $1000$ sampled plans, and plans are conditioned on having at most a $5 \%$ size deviation between the two blocks. %
    Some sample partitions are displayed next to the cut-edge picture. The UST-partitions tend to have slightly longer boundaries on average. %
    }%

    \label{fig:treegate}
\end{figure}
\endgroup

\begin{figure}
    \centering
    \begin{tabular}{c|c|c}
         Width & MST & UST \\
         0 & 1.346, 1.339 & 1.348, 1.273\\
         1 & 1.045, 1.479 & 1.220, 1.314\\
         2 & 1.003, 1.689 & 1.191, 1.425\\
         3 &  1.072, 1.559 & 1.239, 1.379\\
    \end{tabular}
    
    \caption{1000 samples, on a $36 \times 36$ node graph. Each half gate is $6 \times \text{width}$ vertices long, so width $3$ represents the gate cutting the graph entirely in half. The numbers reported are obtained in the following way: some nodes are set to be $1$ and others to $0$. Each district majority votes to decide the ``party'' of the representative in that region, which is a number in $\{0,1\}$. By summing the party across both districts, each plan is assigned a total party value in $\{0,1,2\}$. We have reported the mean total party value across $1000$ trials, according to two different voting population distributions. The left hand numbers reported are based on a distribution of voters where the left 60 $\%$ of the nodes in the square are party $1$, and the right hand numbers are based on a voter distribution where the bottom $60 \%$ are party $1$. Since the voting data forces the total party value into $\{1,2\}$, these are all Bernoulli random variables, and therefore the entire distribution can be read from the mean that we reported.
    }
    \label{fig:treegatevotes}
\end{figure}

%% file: Sections/5PositiveResults/ShortVersion.tex
\section{Positive results}\label{Section:PositiveResults}

In this section, we provide several results regarding the tractability of sampling connected $k$-partitions and simple cycles. Most of these results will follow from the tractability of counting connected $k$-partitions on various families of graphs. Unlike many $p$-relations, connected 2-partitions and simple cycles do not appear to be encodeable in a self-reducible way (in the sense used in 
\cite{JVV,Khullerselfreducible}), meaning that the equivalence between counting and sampling requires variations on the ideas explained in \cite{JVV}.\footnote{Using techniques similar to \cite{Khullerselfreducible} we can prove that at least one reasonable encoding of simple cycles is not self-reducible, unless $\Poly = \NP$. See \Cref{Section:Notselfreducible}.}
However, we can use the chain of bigons construction to evaluate the marginal probabilities that self-reducibility would normally reduce to a counting problem. We can also directly modify the counting algorithms we find to compute the marginals. First, we recall how to use certain marginal probabilities to sample from a probability distribution over subsets of a given set.

\subsection{Sampling from counting}

The following algorithm is a standard part of the equivalence between counting and sampling, and is usually stated in the context of self-reducible structures. In \Cref{appendix:inductivesampling} we provide a proof of correctness.

\begin{defn}[Marginal probabilities of distributions over subsets]\label{defn:thosemarginals}
Let $p$ be a probability distribution on $2^{[n]}$, the set of subsets of $[n] = \{1,2,3,\ldots, n\}$. Let $S$ be a random variable distributed according to $p$. For a set $J \subseteq [i-1]$, define $p( i | J) = \mathbb{P} (  i \in S |  S \cap [i - 1] = J)$; that is, the probability that $S$ contains $i$, conditioned on containing $J$ and being disjoint from $[i - 1] \setminus J$. (As a convention, take $[0] = \emptyset$.)
\end{defn}

The use of the previous definition is in the following algorithm for sampling from a probability distribution over subsets of a set, given access to the marginals of \cref{defn:thosemarginals}.

\begin{algorithm}[H]
\caption{\texttt{InductiveSampling}}\label{alg:inductivesampling}
\textbf{Input:} A probability distribution $p$ on $2^{[n]}$ described via an oracle $O$ that can compute $p( i | J)$ for any $i \in [n]$ and any $J \subseteq [n]$.\\
\textbf{Output:} A random element of $2^{[n]}$ distributed according to $p$. 
\begin{algorithmic}[1]
\STATE{Set $J_0 = \emptyset$}
\FOR{ $i \in [n]$ } \STATE{
Use $O$ to calculate $p(i | J_{i-1})$} \STATE{
With probability $p ( i |J_{i-1})$ , set $J_i = J_{i-1} \cup \{i\}$. Else, set $J_i = J_{i-1}$.
}   
 \ENDFOR
\STATE{Return $J_n$.}

\end{algorithmic}

\end{algorithm}

The correctness of this algorithm (\Cref{appendix:inductivesampling}) proves the following:

\begin{thm}\label{thm:MarginalToSamples}
Let $\mathscr{C}$ be a language encoding graphs (resp. node-weighted graphs), and suppose that there is a polynomial time Turing machine $M$ (resp. $M_B$), that on input $G \in \mathscr{C}$, $J, J' \subseteq E(G)$, computes the number of simple cycles containing $J$ and disjoint from $J'$ (resp. the number of balanced connected $2$-partitions whose cut set contains $J$ and is disjoint from $J'$). Then there is a polynomial time probabilistic Turing machine that uniformly samples from $SC(G)$  (resp. uniformly samples from $P_2^0(G)$) for $G \in \mathscr{C}$.
\end{thm}

\subsection{Remarks on algorithmic meta-theorems}

The next sections will be concerned with computing the marginals that are necessary for \Cref{thm:MarginalToSamples}. This will be done by showing that we can solve certain counting problems. The tractability of the counting problems on graphs of bounded treewidth will follow from extensions of Courcelle's theorem, such as those in \cite{arnborg1991easy}. In particular, the cut edges of a connected $k$-partition can be expressed in $\MSO_2$ (see \Cref{section:MSO2connectedkpartition}), and similarly the cut edges of a balanced connected $k$-partition can be expressed in $\EMS$, as defined in \cite{arnborg1991easy}.%
\footnote{For background on second order logic, the reader is referred to \cite[Chapter 7]{ebbinghaus2013mathematical}; for background on these meta-theorems and $\MSO_2$, the reader is referred to \cite{arnborg1991easy}. A brief summary of the meta-theorem that we use is given in \Cref{section:briefMSOsummary}.}
The constants in these meta-theorems are too large to be practically useful, as the automata on which they are based grows in size like a tower of exponentials in the size of the formula, and there are no general tricks to avoid this \cite{frick2004complexity}
. Therefore, although we appeal to these meta-theorems to conclude complexity theory statements, we emphasize practical approaches to solving these counting and sampling problems on series-parallel graphs, which give some directions for practical implementations on wider classes of graphs. For example, forthcoming work \cite{ignasi} extends the ideas applied to series-parallel graphs to arrive at a reasonably implementable algorithm for counting and sampling simple cycles fixed-parameter tractably in the treewidth.

\subsection{Simple cycles}

\begin{defn}[Marginal graph]
Given a graph $G = (V,E)$ and $J, J' \subseteq E$, let $G_{J,J'}(d)$ denote the graph where the edges in $J'$ are deleted, and the edges in $J$ are replaced by a chain of $d$ bigons. 
\end{defn}

The next lemma shows that for sufficiently large $d$, the number of simple cycles in $G$ containing $J$ and disjoint from $J'$ can be inferred from $|SC ( G_{J,J'}(d))|$, by using division with remainder and the same exponential growth rate comparisons that drove the intractability results:

\begin{defn}
If $\mathscr{C}$ is some language encoding graphs, and $k : \mathscr{C} \to \mathbb{N}$ some function, then we call $k$ a parametrized language of graphs.
\end{defn}

\begin{lem}\label{lem:Simplecycleremaindercount}
Let $k : \mathscr{C} \to \mathbb{N}$ be a parametrized language of graphs. Suppose that there is a polynomial $p$ and a computable function $f$ and a Turing machine $M$ that can calculate $| SC( G_{J, J'}(d) )| $ in time $f ( k(G_{J, J'}(d))) p(|G|,d)$ for all $G \in \mathscr{C}$ and $d \geq 1$ and for any $J, J' \subseteq E(G)$. Then, there is a polynomial $q$ and a TM which calculates $b_{J,J'} := | \{  T \in SC(G) : J \subseteq T, J' \cap T = \emptyset  \} |$ in time $O ( f( k ( G_{J, J'}( 36n^4) )) q( |G|))$ %
for all $G \in \mathscr{C}$ and $J, J \subseteq E(G)$.
\end{lem}
\begin{proof}
If $|J| = 0$, then $|SC( G_{J,J'}(d))| = b_{J,J'}$. We assume that $|J| \geq 1$, and let $n = |G|$. Observe, in the manner of \Cref{prop:SimpleCycleUniformHard}, that $$| SC( G_{J, J'}(d) )| = \sum_{k = 0}^{|J|} 2^{dk} | \{ X \in SC(G) : X \cap J' = \emptyset, |X \cap J| = k \} | + d|J.|$$ We define $a_d = |SC(G_{J, J'}(d))|$ and the remainder term $$R_d = \sum_{k = 0}^{|J| - 1} 2^{dk} | \{ X \in SC(G) : X \cap J' = \emptyset, |X \cap J| = k \} | + d|J|.$$ $R_d$ is bounded above by $2^{d (|J| - 1) } 2^{n^2} + d |J| \leq d n^2 2^{d (|J| - 1) + n^2}$. 
By \Cref{lem:polylarge}, for $d = 4 \ceil{ (n^2 + \log(n^2) + 1)^2}$, $2^{d |J| }  > d n^2 2^{d (|J| - 1) + n^2}$, and hence for such $d$, we have $a = 2^{d |J|} b_{J,J'} + R_d,$ where $2^{d |J| } > R_d$.
Since each term in that expression is an integer, $R_d$ is the remainder of dividing $a$ by $2^{d |J|}$, and $b_{J,J'}$ is the quotient. This division with remainder can be performed in $O( ( \log ( | SC(G_{J,J'}(d))|) \log (2^{d|J|}) ))$ time \cite[Theorem 3.3]{shoup2009computational}, which is polynomial in $(|G|, d)$. Since $d = 4 \ceil{ (n^2 + \log(n^2) + 1)^2} \leq 36 n^4$, the result follows. %
\end{proof}

We briefly recall a definition of treewidth:

\begin{defn}[$k$-trees, partial $k$-trees and treewidth]
A $k$-tree is any graph that can be recursively constructed in the following manner. We start with a tree $T_0$ that is a $k$ clique. Then, we obtain $T_n$ from $T_{n-1}$ by picking any $k$-clique $Q$ of $T_{n-1}$ %
adding a new vertex $v$ and connecting $v$ to each node of $Q$. A \emph{partial} $k$-tree is any subgraph of a $k$-tree. The treewidth of a graph $G$ is the smallest $k$ such that $G$ is a partial $k$-tree.
\end{defn}

The operation $G \to G_{J,J'}(d)$ preserves the class of series-parallel graphs, and in addition, does not increase the treewidth of any graph with treewidth $\geq 2$. %
There are polynomial time dynamic programming algorithms for counting the number of simple cycles of a series-parallel graph,\footnote{See \Cref{appendix:SCDP}} and it follows from Courcelle's theorem that counting the number of simple cycles is fixed-parameter tractable in the treewidth.\footnote{ This follows from Proposition 5.11 and Theorem 6.56 in \cite{courcelle2012graph}. We are grateful to \href{http://fc.isima.fr/~kante/}{Mamadou Moustapha Kant\'e} for \href{https://cstheory.stackexchange.com/a/43865/44995}{pointing this out on Stack Exchange} \cite{mamadou}. There is an explicit $\MSO_2$ formula for simple cycles at the same link.}
Thus, from \Cref{lem:Simplecycleremaindercount} and \Cref{thm:MarginalToSamples} we have the following:

\begin{thm}\label{thm:SCtreewidthsampling}
The problem of uniformly sampling simple cycles is $\FPT$ in the treewidth.
\end{thm}

Since the treewidth of a plane dual changes by at most one \cite{lapoire1996treewidth,bouchitte2001treewidth}, we obtain a similar result for sampling connected $2$-partitions. However, in the next section we will show how to apply Courcelle's theorem directly to show that sampling connected $k$-partitions is $\FPT$ in the treewidth.

The treewidth of a typical state dual graph used for redistricting (\Cref{section:CongressionalMotivation}) is on the order of $40$ to $60$; although this shows that treewidth is not a useful parameter for state dual graphs, this does not rule out the possibility that other parameters can make sampling from $P_2(G)$ tractable when $G$ is a state dual graph. We discuss this further in \Cref{section:openquestions}. 

With a little more work, one can show that many other distributions over simple cycles besides the uniform distribution can be efficiently sampled from, at least on graphs with treewidth $\leq 2$. We now introduce a definition to describe these distributions.

\begin{defn}[Edge weight probability]\label{defn:edgeweightprobability}
Let $G$ be a graph, and $ c : E(G) \to \mathbb{Q}_{\geq 0}$ some weight function. Let $N_c$ denote the measure on simple cycles that gives weight $N_c (C) = \prod_{\substack{e \in C}} c(e)$ to each simple cycle $C$, and let $\nu_c$ denote the probability distribution obtained by normalizing $N_c$.
\end{defn}

\begin{thm}\label{thm:mucsampling}
Sampling from $\nu_c$ is polynomial-time solvable on the class of graphs of treewidth $\leq 2$.
\end{thm}
\begin{proof}
See \Cref{lem:simplecycleweightedmarginals}, which shows how to directly compute the required marginal probabilities for sampling from $\nu_c$ via \Cref{alg:inductivesampling}, without using \Cref{lem:Simplecycleremaindercount}. 
\end{proof}

The space of distributions on $SC(G)$ is \emph{far} larger than the distributions described by $\nu_c$. We mention several other tractable distributions on $SC(G)$ and $P_2(G)$ in \Cref{section:otherdistributions}.

\subsection{Connected $k$-partitions}

First, we will briefly review $\MSO_2$ and the counting meta-theorem in \Cref{section:briefMSOsummary}. Next, in \Cref{section:MSO2connectedkpartition}, we will describe an $\MSO_2$ formula that defines connected $k$-partitions. Finally, we will tie these together to prove the following: %

\begin{theorem}\label{thm:uniformlysamplingkparts}
Uniformly sampling from $P_k(G)$ is $\FPT$ in the treewidth.
\end{theorem}

\subsubsection{Counting solutions to $\MSO_2$ formulas}\label{section:reviewcourcelles}\label{section:briefMSOsummary}

We mostly follow \cite{arnborg1991easy}. We consider a relational vocabulary $\mathcal{R} = (V,E,J, J', \inc)$, where $V,E,J,J'$ are unary relations, and $\inc$ is a binary relation. Additionally, we consider a set of formulas $\Gamma = \{\forall x V(x) \vee E(x), \forall x V(x) \leftrightarrow \lnot E(x), \forall x \forall y \inc(x,y) \to V(x) \wedge E(y), \forall x J(x) \vee J'(x) \to E(x)\}$. $\Mod(\Gamma)$ denotes the set of models of $\Gamma$. Given a model of $\Gamma$ with universe $A$, $A$ is partitioned by the two sets defined by $V$ and $E$, which we refer to as $V$ and $E$, by abuse of notation. We interpret $V$ as the set of vertices, $E$ as the set of edges, $J$ and $J'$ as two collections of edges, and $\inc$ as the incidence relation. That is, $\inc(v,e)$ is interpreted as meaning that vertex $v$ is incident to edge $e$. Thus, a model for $\Gamma$ is a graph along with two sets of edges. %

We denote by $\MSO_2$ the second order logic with signature $\mathcal{R}$ that allows only unary relational variables. %
Given a formula $\Phi(X)$ in $\MSO_2$ with a free variable $X$, the enumeration problem for $\Phi$ is that of computing $| \{ X : G \models \Phi(X) \}|$ for any given $G \in \Mod(\Gamma)$.

Then, we have:

\begin{theorem}{\cite[{Theorem 5.7}]{arnborg1991easy}}\label{thm:arnborgs}
For each $\MSO_2$ formula $\Phi(X)$, and for each class $K$ of graphs with universally bounded treewidth, the enumeration problem for $\Phi$ can be solved in $O( |G| \log ( |G|))$ time if $G$ is given with a tree-decomposition.
\end{theorem}

For the purposes of obtaining an $\MSO_2$ formula, it is convenient to represent a connected $k$-partition by the \emph{complement} of the cut-set, similarly to \Cref{appendix:duality}. 

\begin{defn}[Connected partitions as edge sets]
Given an (unordered) $k$-partition $P$, define $F(P) = \cut(P)^c$. Let $\Flats_k(G)$ be the set $\{F(P) : P \in P_k(G) \}$. %
For any $J', J \subseteq E$, define $F_{J, J'}(G) \subseteq \Flats_k(G)$  as those $Q \in \Flats_k(G)$ with $J \subseteq Q$ and $Q \cap J' = \emptyset$.

\end{defn}

Since a connected $k$-partition is determined by its cut set (\Cref{prop:matroidduality}), it follows that $F: \mathscr{P}_k(G) \to \Flats_k(G)$ is a bijection. We are describing (unordered) connected partitions as flats in the graphic matroid, hence the notation ``F'' and ``$\Flats_k$.''

\subsubsection{$\MSO_2$ formula for edge sets in $\Flats_k(G)$}\label{section:MSO2connectedkpartition}

Let $G = (V,E)$ be a graph. For $X \subseteq E$ we will build up to an $\MSO_2$ formula that checks if $X \in \Flats_k(G)$. Our building blocks are inspired by the examples in \cite[Chapter 7]{cygan2015parameterized}. First, we define a formula that checks if a set of nodes, $Y$, is contained in $G[X]:$

$$\In(Y,X) = \forall_{v \in Y} \exists_{e \in X} \inc(v,e).$$

Given two sets of vertices, $U$ and $W$, we define a formula that checks if there is an edge in $X$ connecting a node in $U$ to a node in $W$:

$$\Bridge(U,W,X) = \exists_{u \in U, w \in W, e \in E} \inc(u,e) \wedge \inc(w,e).$$

Next we define a formula that checks if $G[X]$ is connected, by checking whether there are any non-trivial $2$-partitions of $V(G[X])$ with no edges in $X$ between the different blocks.

$$\connE(X) := \forall_{Y \subseteq V} [\In(Y,X) \wedge Y \not = \emptyset \wedge[\exists_{U \subseteq V} \In(U,X) \wedge U \cap Y = \emptyset \wedge U \not = \emptyset \wedge U \cup Y = V]] \to \Bridge(Y,U,X).$$

We next define a formula that checks if an edge has both endpoints in the nodes of a subgraph induced by a set of edges $Y$:

$$\ep(e,Y) = \forall_{v \in V} \inc(e,v) \to (\exists_{e' \in Y} \inc(e',v))$$

Finally, for each $k \geq 1$, we define a formula that takes a collection of edges, $X$, and checks whether it is in $\Flats_k(G)$. This is accomplished by checking that every node of $G$ is incident to some edge in $X$ and that $X$ is a union of $k$ sets of edges, each of which induces a connected subgraph and so that any edge with both endpoints in one of those connected subgraphs is in $X$.

$$F'_k(X) = \In(V,X) \wedge ( \exists_{X_1, \ldots, X_k \subseteq E} (X = \bigcup X_i \wedge \bigwedge_i connE(X_i) \wedge (\forall_{e \in E} \bigwedge_i (\ep(e,X_i) \to e \in X_i ))) $$

Recall that we considered $J$ and $J'$ to be part of the relational structure $\mathcal{R}$, so we can define the $\MSO_2$ formula whose solution sets are the members of $F_{J,J'}$: 

\begin{equation}\label{eqn:MSOformula}
    F_k(X) = F'_k(X) \wedge (J \subseteq X) \wedge (J' \cap X = \emptyset)
\end{equation}

\begin{lemma}\label{lem:kpartMSOformula}
Let $A$ be a model of $\Gamma$, i.e. a graph $G = (V,E)$ with vertex-edge incidence matrix given by $\inc$ and two distinguished subsets of edges, $J$ and $J'$. $F_k(X)$ is true in $A$  if and only if $X = F(P)$ for some connected $k$-partition $P$ of $G$ and $X \cap J' = \emptyset$ and $J \subseteq X$.
\end{lemma}

We now prove \Cref{thm:uniformlysamplingkparts}.

\begin{proof}

Let $K$ be a class of graphs with universally bounded treewidth. For any $G \in K$ and $J,J' \subseteq E(G)$, by \Cref{thm:arnborgs} and \Cref{lem:kpartMSOformula}, we can count $|F_{J,J'}(G)|$ in time $O(|G|\log(|G|)$ with constant dependent only on the bound on the treewidth and the formula $F_k$. The conditions for running \Cref{alg:inductivesampling} are satisfied.
\end{proof}

\begin{remark}
It is easy to add a relational formula (see \cite{arnborg1991easy}) to \Cref{eqn:MSOformula} that restricts our count to only balanced connected $k$-partitions. In particular, the balanced connected $k$-partition problem is in \emph{extended} monadic second order logic (EMS). From this it should follow that so the counting and sampling problems are $\XP$ in the treewidth. However, as noted at \href{https://cstheory.stackexchange.com/q/44338/44995}{this Stack Exchange question}\cite{msomissing}, the corresponding meta-theorem appears to be missing from the literature. %
\end{remark}

\subsection{Balanced 2-partitions}

We mentioned in \Cref{section:balancedhard} that \cite{dyer_complexity_1985} proved that determining if a graph has a balanced connected $2$-partition is $\NP$-hard. That paper also describes a dynamic programming algorithm that determines if a series-parallel graph has a balanced connected 2-partition. This dynamic programming algorithm can be modified to produce an algorithm for \emph{counting} the number of balanced connected 2-partitions of a given node-weighted series-parallel graph $G$ in time polynomial in $G$ and pseudopolynomial in the weights. We present the details of this algorithm in \Cref{appendix:balancedDP}. To turn such a counting algorithm into an algorithm for calculating the marginals necessary for \Cref{thm:MarginalToSamples}, we proceed along similar lines as in the simple cycle case. %

\begin{defn}[$W^{J, J'}(G,w)(d)$]\label{defn:WJJ'}
Let $(G,w)$ be a node weighted graph, and let $J, J' \subseteq E(G)$.
Define $G^{J,J'}$ by replacing edges in $J$ with the doubled $d$-star gadgets from \Cref{defn:doubledstar} and contracting the edges in $J'$, deleting any self loops that arise in this way.
Assign the ``new nodes'' of $D_d(e)$ weight $0$ for each $e \in J$, and the old nodes the same weight as they had in $G$. The resulting node-weighted graph is denoted $W^{J,J'}(G,w)(d)$. 
\end{defn}

We now show that the marginals necessary for \Cref{alg:inductivesampling} can be computed from $|P_2^0(W^{J,J'}(G,w)(d))|$ (with notation as in \Cref{section:balancedhard}) using division with remainder and the exponential growth rate calculations that drove the intractability result in \Cref{section:balancedhard}:

\begin{prop}\label{lem:balancedpartitionremaindercount}
Let $(G,w)$ be a weighted graph. Then: %

\begin{equation}\label{eqn:balanceddecomp}
    |P_2^0(W^{J, J'}( G,w)(d))| = 2^{d |J|} | \{ X \in P_2^0(G,w) : J \subseteq \cut(X), \cut(X) \cap J' = \emptyset \}| + R_d,
\end{equation}

where $R_d$ is a non-negative integer with:  
\begin{equation}\label{eqn:remainderupperbound}
    R_d \leq 2^{n^2} 2^{ d ( |J| - 1) }. 
\end{equation}

\end{prop}
\begin{proof}
Let $G / J'$ denote the quotient graph obtained by identifying $u,v \in V(G)$ if $\{u,v\} \in J'$. First, we decompose 
\begin{equation}\label{eqn:decomposition1} 
P_2 (G/ J') = \bigcup_{k = 0}^{|J|} \{ (A,B) \in P_2(G/J') : | \cut(A,B) \cap J | = k \}. \end{equation}

We define $R_J : P_2( G^{J, J'}(d)) \to P_2( G / J')$ as $R_d$ is in \Cref{defn:Rdmap} by forgetting the assignment of new nodes, We pull back \Cref{eqn:decomposition1} along $R_J$ to obtain: $$P_2( G^{J, J'}(d)) = \bigcup_{k = 0}^{|J|} R_J^{-1} ( \{ (A,B) \in P_2(G/J') : | \cut(A,B) \cap J | = k \} ).$$

The map $\phi^* : P_2(G/ J') \to P_2(G)$ defined by $\phi^*( (A,B) ) = ( \phi^{-1}(A), \phi^{-1}(B))$ is an injection, and the image is $\{ (A,B) \in P_2(G) : \cut(A,B) \cap J' = \emptyset \}$.

Hence we have a \emph{partition} of $P_2(G^{J,J'}(d))$, $$P_2( G^{J, J'}(d)) = \bigcup_{k = 0}^{|J|} (\phi_{J'}^* \circ R_J) ^{-1} ( \{ (A,B) \in P_2(G) : \cut(A,B) \cap J' = \emptyset, | \cut(A,B) \cap J | = k \} ).$$

So far we have decomposed the set of partitions of $G^{J,J'}(d)$. Next, we  compute the $0$-balanced partitions in each block of that decomposition. %
The elements of $(\phi_{J'}^* \circ R_J) ^{-1} ( \{ (A,B) \in P_2(G) : \cut(A,B) \cap J' = \emptyset, | \cut(A,B) \cap J | = k \})$ are obtained by extending a partition in $\{ (A,B) \in P_2(G) : \cut(A,B) \cap J' = \emptyset, | \cut(A,B) \cap J | = k \}$ onto the new nodes. Since each new node has weight $0$, it is impossible to assign new nodes in such a way as to make unbalanced partitions of $G$ balanced. %

The balanced partitions that have $J$ contained in the cut have exactly $2^{d|J|}$ balanced extensions each. This proves \Cref{eqn:balanceddecomp}. %
We are left to show the upper bound of \Cref{eqn:remainderupperbound} for the reamining partitions, %
namely :$$\text{Rem}_d =  (\bigcup_{k = 0}^{|J| - 1} (\phi_{J'}^* \circ R_J) ^{-1} ( \{ (A,B) \in P_2(G) : \cut(A,B) \cap J' = \emptyset, | \cut(A,B) \cap J | = k \} )) \cap P_2^0( W^{J,J'}(G)(d))$$

We have that $R_d = |\text{Rem}_d|$. Suppose that $X$ is some balanced partition of $G$, with $|\cut(X) \cap J| \leq |J| - 1$. The number of ways to extend $X$ to the new nodes and get a balanced partition is at most $2^{ d ( |J| - 1) } $. Since $|P_2(G)| \leq 2^{n^2}$, this provides the upper bound on the remainder term.
\end{proof}

\begin{prop}\label{prop:balancedcountingmachineimpliessampler}
Let $\mathscr{C}$ be some class of graphs that is closed under the operation $G \to G^{J,J'}(d)$ of \Cref{defn:WJJ'}, for all $d \geq 1$. Let $p$ be a polynomial. %
Suppose that $M$ is a Turing machine which can compute $|P_2^0(G)|$ on all weighted graphs $(G,w)$ where $G \in \mathscr{C}$ and $w : V(G) \to \{0,1,\ldots,\}$, in time bounded by $p(|G|, w(G))$. Then there is a polynomial time probabilistic Turing machine that uniformly samples from $P_2^0(G,w)$ in time polynomial in $(|G|, w(G))$ for all $G \in \mathscr{C}$.
\end{prop}
\begin{proof}

Due to \Cref{alg:inductivesampling}, to sample in polynomial time it suffices to be able to compute $a_{J, J'} := | \{ X \in P_2^0(G,w) : | \cut(X) \cap J | = |J|, \cut(X) \cap J' = \emptyset \}|$ in polynomial time for any given $J, J' \subseteq E(G)$. We will do this by from computing $|P_2^0(W^{J, J'}( G,w)(d))|$ at a value of $d$ which is polynomially large in $|G|$. 

If $d = n^2 + 1$, then $2^{ d |J|} > 2^{n^2} 2^{ d (|J|- 1)}$.
Now, given $N_d =  |P_2^0(W^{J, J'}( G,w)(d))|$, from \Cref{lem:balancedpartitionremaindercount} we know that we can write $N_d = a_{J,J'} 2^{d|J|} + R_d$. Since we can efficiently compute $2^{d|J|}$ and $N_d$ in time polynomial in $(|G|, w(G))$, since we fixed $d = n^2 + 1$, by division with remainder we can compute $a_{J,J'}$ in time polynomial in $(|G|, w(G))$. Thus, we have calculated the marginal that we need for sampling.
\end{proof}

\begin{thm}\label{thm:balSCsampling}
There is an algorithm for uniformly sampling from the balanced partitions of a node weighted \emph{series-parallel} graph $(G,w)$, which runs in time polynomial in $(|G|, w(G))$.
\end{thm}
\begin{proof}
This follows from \Cref{prop:balancedcountingmachineimpliessampler} and the dynamic program for counting balanced partitions on series-parallel graphs presented in \Cref{appendix:balancedDP}, since the class of node weighted series-parallel graphs is closed under the operation $G \to G^{J,J'}(d)$ for all $d \geq 1$; this is because series-parallel graphs are closed under replacing edges by doubled $d$-trees, and under edge contractions (provided we eliminate self loops).
\end{proof}

\begin{remark}
It may be possible to extend this to an $XP$ in treewidth algorithm for sampling balanced $k$-partitions, using similar ideas as well as those mentioned in the conclusion of \cite{ito2006partitioning}.
\end{remark}

\subsection{Other families of distributions over partitions}\label{section:otherdistributions}

We conclude this section by pointing out that \emph{many} distributions on $P_k(G)$ and $P_k^0(G)$ are tractable to sample. 
A general strategy for building $k$-partitions of $G$ is to contract $G$ in some random way onto a simpler graph, $G'$, and then pull back $k$-partitions from the simpler graph. The following lemma shows that one can pull back connected $k$-partitions along quotient maps obtained by contracting connected partitions:

\begin{lem}
Let $G$ be a graph, and let $\phi : G \to G / R$ be a graph quotient map, where $R$ is an equivalence relation on the nodes such that the equivalence classes of $R$ induce connected subgraphs. Then for any $(A_1, \ldots, A_k) \in P_k(G/R)$, $( \phi^{-1}(A_1), \ldots, \phi^{-1}(A_k)) \in P_k(G)$. Moreover, if $w$ is a node weight function on $G$, then if we assign each equivalence class of $G/R$ the total weight of all its elements, $\phi^{-1}$ preserves the weight of blocks, and thus also pulls back balanced partitions.
\end{lem}

This lemma can be used to give a recipe for chaining together random partitions into of $G$ with many blocks into an algorithm for obtaining random partitions into $k$ blocks. For example, at each stage one can take $R$ to be an equivalence relation induced by random sets of edges, such as the edges of a random matching, or the  monochromatic edges in a sample from distribution over colorings, or a random forest, as we did in \Cref{section:thechoiceofmodel}. One could also imagine finding random quotients onto graphs of smaller treewidth, and then using the sampling algorithms from the previous sections. %
Additionally, it is known \cite[Theorem 11]{demaine2004bidimensional, demaine2008bidimensionality} %
that plane graphs can be \emph{contracted} onto partially triangulated grid graphs with similar treewidth, which suggests that understanding the connected $k$-partition sampling problem for partially triangulated grid graphs is an open problem with important implications with sampling connected partitions of state dual graphs. %

Another means of producing connected $2$-partitions is via min cuts, since min cuts are always connected. There are polynomial time algorithms for uniformly sampling min $s,t$-cuts genus $g$ graphs given in \cite{chambers2014countin}. %
On general graphs, one can also sample min-cuts in a way that is fixed parameter tractable in the size of the min-cut \cite{berge2019fixed}, but the running time of this algorithm is practical only for very small min-cut sizes.

We emphasize that even though these distributions can be efficiently sampled from, it is not clear how to characterize their properties in terms of interpretable features of districting plans.%
As we discussed in \Cref{section:thechoiceofmodel}, the properties of these distributions as distributions over partitions of the underlying geography may vary significantly with the discretization, and in a given application one needs to decide if this is acceptable.

%% file: Sections/6ConclusionsAcknowledgements/1Conclusions.tex
\section{Conclusions}

\pagebreak
\subsection{Broad overview of paper}

We motivated this paper 
by discussing attempts at characterizing outlier redistricting plans through ensemble methods (\Cref{section:CongressionalMotivation}). In practice this is applied by picking some statistic of interest and comparing the distribution of this statistic over sample maps to those of a proposed plan. %
The theoretical and experimental results in this paper indicate that additional considerations are needed before we can put full trust into statistical outlier analysis %
of gerrymandering. %
In particular: %

\begin{itemize}
    \item The flip walk proposal distribution used in practice is likely not rapidly mixing (\cref{Section:FlipChain} and \cref{subsection:gridgraphempirical}), even on classes of graphs where sampling is tractable (\cref{Section:PositiveResults}).
    \item The complexity results (\cref{thm:generalintractability3CCP}, \cref{prop:balancedhard} and \cref{thm:lambdadualkintractable}) show that for many classes of graphs and distributions, there are likely to be no efficient replacements for the flip walk when it comes to sampling from certain prescribed distributions.
    \item Even if it were possible to sample from an explicitly designed distribution, that distribution may undergo phase changes in its qualitative behavior if the description is slightly modified (\cref{section:phasetransitions}). Along similar lines, a fixed algorithm can produce dramatically different distributions of partitions over the underlying geography if different choices are made regarding the choice of model graph (\Cref{fig:Frankengraph} and \Cref{section:thechoiceofmodel}).  %
\end{itemize}

Altogether, these observations imply that inferential conclusions \emph{may} not be robust to small changes in the set up. %

\subsection{Future directions}

We now describe some directions that the computational redistricting community can go in to address these challenges.

\subsubsection{Other applications to redistricting}

There are local measurements of gerrymandering that do not require guarantees about sampling. One example \cite{chikina_assessing_2017, pegden1} is supported by a rigorous theory for a particular meaning of gerrymandering regarding ``carefully crafted'' plans. It remains an important question to investigate the extent to which the decisions we highlighted in \Cref{section:phasetransitions} and  \Cref{section:thechoiceofmodel} affect the interpretation of these tests. %

Additionally, it is possible to produce a large ensemble of partitions using the flip walk \cite{mattingly}, the means described in \Cref{section:otherdistributions} and other methods such as \cite{cho_sampling_2018, chen_unintentional_2013, recomb}. Despite the difficulties inherent in characterizing the statistical behavior of these ensembles, %
a large collection of plans may be amenable for other species of analysis. For example, questions about the existence of plans meeting certain criteria can sometimes be answered by identifying demonstration plans, rough characteristics of trade-offs between criteria can perhaps be calculated, and certain implications of proposed legislation can be evaluated \cite{deford2019redistricting}. %

\subsubsection{The extreme outlier hypothesis}\label{section:EOH}

Currently a consensus is developing that the notion of an \emph{extreme} partisan outlier is robust between different sampling methods. For example, handful of recent court decisions reference favorably the outcome of such outlier analysis \cite{courtopinion1, courtopinion2, courtopinion3}. %
The consensus asserts that different ensemble methods are measuring a consistent and interpretable feature of the political data because those methods used in practice seem to detect the same extreme partisan outliers. 

The hypothesis that there is a consistent and robust notion of an extreme outlier, which we will call the ``extreme outlier hypothesis'' (EOH), is likely to be a rich source of challenges and questions about ensemble based redistricting. This hypothesis was partially explored in \cite{deford2019redistricting}, where it was found that certain changes in the sampling algorithm had little effect on the tails of some chosen statistics, but that other changes caused certain tails to become exaggerated. However, the changes that exaggerated the tails resulted from interpreting different legal constraints for permissible maps; for example, section 3.4 of \cite{deford2019redistricting} shows some examples of how adherence to the voting rights act can have dramatic impacts on the distribution of partisan scores. 

It is easy to fabricate distributions where any specific plan appears to be an outlier, but such fabricated distributions may not reflect principles important in real world redistricting. To be useful, EOH must incorporate redistricting principles and not just mathematical abstractions. 
For example, in \Cref{section:phasetransitions} we obtained two distributions described by similar parameters, which nevertheless find different extreme outliers. However, this does not falsify the EOH, because supercritical partitions are not representative of legal maps. Therefore, one could reasonably object to the baseline provided by the supercritical distribution, while accepting the usefulness of the baseline provided by the subcritical distribution. On the other hand, as we saw in \Cref{section:thechoiceofmodel}, certain changes to the underlying discretization can shift its partisan properties without changing a typical district's geometry, a more subtle change that may or may not have bearing on the EOH in practice. Of our experiments, the most challenging for the EOH is the significant difference between the UST-partition and MST-partition in the third row of \cref{fig:treegatevotes}.%

A reasonable formulation of EOH hypothesizes that %
it is implausibly difficult to fabricate distributions over districting plans that are defensible as a baseline for redistricting, even under scrutiny by adversarial experts, but which report different \emph{extreme} outliers.
To falsify EOH, it would \emph{suffice} to find reasonable operationalizations of the same set of legal requirements into different sampling algorithms, which nonetheless report different extreme outliers. The problem of establishing precise guidelines for what constitutes a representative distribution over districting plans is understudied and critical to understanding the EOH.%

\subsubsection{Pragmatism and distributional design}\label{section:distributionaldesign}

A pragmatic resolution to the hard questions raised by the EOH may be to pick a handful of sampling algorithms that will be consistently recognized a baseline. 
Some states already incorporate restrictions on the use of partisan data in the drawing of districting plans, and ensemble based methods could be one additional way for those states to operationalize these intentions. However, there is likely not a single collection of distributions that will suit every geography and political culture. Beyond purely mathematical analysis, one route to finding suitable distributions is through empirical analysis and successive refinements under real world conditions. %

Although the procedure for determining which distributions to set as baselines is politically fraught, understanding the sampling algorithms to promote informed decisions is a hard scientific problem. 
In addition to characterizing the distributions generated by these algorithms, an analysis of algorithmic reliability, as in \Cref{Section:FlipChain} and \Cref{section:phasetransitions}, is important for developing standards that constrain data dredging and post-hoc analysis. Such reliability considerations are also important for reproducibility. %
Similarly, we should understand the dependence of partition sampling algorithms on the discretization, to constrain what we called \emph{metamandering} in \Cref{section:thechoiceofmodel}.

Using sampling algorithms to detect partisan intent and using sampling algorithms to constrain possibilities are two separate tasks. A clear danger in both approaches is that distributions could be chosen in a way that creates unnecessary constraints or undesirable bias.
Indeed, \emph{if} the EOH fails, %
then giving someone the power to choose a baseline distribution creates the opportunity for subtler manipulation of voting outcomes. %
Along with exploring the EOH, developing empirically motivated partition sampling algorithms and understanding their trade-offs is a key direction for future research in this area. %

\subsection{Open questions}\label{section:openquestions}

We summarize a handful of remaining questions about sampling connected partitions:%

\begin{itemize}
    \item What sort of ground truth models could be useful for assessing the accuracy of outlier methods? Of the districting plans historically or presently used, what proportion of them are flagged as outliers by suggested methodology? Does it correlate with other evidence for gerrymandering? Is the flagging consistent between methods? %
    \item Can the intractability results be unified and strengthened? It seems unlikely that \Cref{thm:facebounded3CCP} is optimal.
    \item Can a general and practically useful sufficient for the existence of a bottleneck in the flip walk be extracted from the examples in \Cref{Section:FlipChain}?
    \item Besides treewidth, are there other graph parameters that make uniformly sampling from $P_2(G)$ tractable?
    \item All of our intractability results relied on reductions from Hamiltonian cycle, by proving that any algorithm sampling from certain distributions can be modified to put large mass on the longer simple cycles of a graph. However, the partitions that are of interest to redistricting tend to have relatively short boundaries on the order of $\Theta(\sqrt{|V|})$, rather than $\Theta(|V|)$. As mentioned in \Cref{section:otherdistributions}, it is possible to sample from the set of min-cuts, FPT in the size of the min-cut. However, this is a different regime from sampling from the cuts of size $\Theta(\sqrt{|V|})$ for the graphs that arise as state dual graphs (\Cref{section:CongressionalMotivation}). Are there approaches to proving tractability or intractability of sampling such medium length cycles?
    \item Is it possible to uniformly sample from $P_2(L_n)$, where $L_n$ is the $n \times n$ grid graph from \Cref{Section:GridGraph}? What about if we consider the class of graphs obtained from $L_n$ by adding some diagonal edges as in \cref{fig:treegate}, or partially triangulated grid graphs \cite{demaine2008bidimensionality}.
    \item Are there families of graph with unbounded treewidth where $\nu_{\lambda}$ sampling $P_2(G)$ is tractable?
    \item Statistical evidence, included repeating the tests in \cite{kennedy2002monte} as well as the flip pictures in \Cref{subsection:gridgraphempirical}, suggests that the $\nu_{\lambda}$ Metropolis-Hastings weighted flip walk Markov chain on $P_2(L_n)$ mixes rapidly only at the critical value $\lambda = 1/\mu$. Is this true? What is the dependence on the population balance restriction?
    \item Which distributions over $P_k(G)$ can we efficiently sample from? Which of these distributions is robust to changes in the discretization? %
    \item Recalling the motivation (\cref{section:CongressionalMotivation}) and problems raised by discretization (\cref{section:thechoiceofmodel}), one may be inspired to abandon the discrete model and directly sample from partitions of the underlying geography. There are many sampling algorithms one can investigate here, such as some derived from Schramm-Loewner evolution, random lines or polynomial curves. What favorable or unfavorable properties do these sampling algorithms have? %
    \item Although there are many plans, many of them are similar in shape. This may lead one to guess that there is a small collection of plans that are near to every other plan; i.e. that there is an epsilon net in the space of reasonable plans. For low dimensional shapes, it is reasonable to find an epsilon net, but for shapes of dimension $d$, the number of points needed to form an epsilon net grows roughly like $(1/\epsilon)^d$. %
    It would be interesting to determine whether or not the space of reasonable plans was high or low dimensional from this point of view. The authors conjecture that this space will behave as if extremely high dimensional, but if there are ways to constrain it to be low dimensional, then the potential existence of a computable $\epsilon$-net opens another way to discuss typicality while being distribution agnostic, for example through the analysis of Pareto fronts between measurements.
\end{itemize}

%% file: Sections/6ConclusionsAcknowledgements/2Acknowledgements.tex
\section*{Acknowledgements}

We want to thank the following people for their patience, enthusiasm, eagerness to share knowledge, insightful questions and helpful discussions: 
Hugo Akitaya,
Eric Bach,
Assaf Bar-Natan,
Jin-Yi Cai,
Sarah Cannon,
Ed Chien,
Sebastian Claici,
Moon Duchin, 
Charlie Frogner,
Jordan Ellenberg,
Heng Guo,
David Hayden,
P\'{a}lv\"{o}lgyi D\"{o}m\"{o}t\"{o}r
 Honlapja,
Mamadou Moustapha Kant\'e,
Fredrik Berg Kjolstad,
Tianyu Liu, 
Aleksander M\c{a}dry,
Elchanan Mossel,
Marshall Mueller,
David Palmer,
Wes Pegden,
Sebastien Roch, 
Mikhail Rudoy,
Zach Schutzmann,
Allan Sly, and
Nike Sun. 

We want to thank  Jin-Yi Cai and Tianyu Liu for several in depth discussions that helped to guide this investigation, for catching some mistakes in an earlier version of \cref{section:balancedhard} and for drawing our attention to relevant literature. We want to thank Sebastien Roch for several in depth discussions about bottlenecks in the flip walk, and asking useful questions which helped to guide the direction of investigation. %

We also want to thank the creators, moderators and contributors to the Stack Exchange network; this project benefited greatly from the expertise of the individuals with whom that webpage connected us to: Heng Guo, P\'{a}lv\"{o}lgyi D\"{o}m\"{o}t\"{o}r
 Honlapja, Mamadou Moustapha Kant\'e, Mikhail Rudoy, ``Gamow,'' and ``Kostya\_I.'' %
 
We also wish to thank the students from VRDI 2018. We  especially thank the students in the graph partitioning group from week two, who suffered with us trying to find an algorithm that would solve the connected $2$-partition uniform sampling problem: Austin Eide, Victor Eduardo Chavez-Heredia, Patrick Girardet, Amara Jaeger, Ben Klingensmith, Bryce McLaughlin, Heather Newman, Sloan Nietert, Anna Schall, Lily Wang.
We also want to thank the team that developed Gerrychain at VRDI 2018, which we used extensively in \Cref{Section:Empirical}: Mary Barker, Daryl DeFord, Robert Dougherty-Bliss, Max Hully, Anthony Pizzimenti, Preston Ward.
 
\subsection*{Funding}

The first author was partially supported by the NSF RTG award DMS-1502553 and by U.S. National Science Foundation grants DMS-1107452,
DMS-1107263, DMS-1107367.
The authors acknowledge the generous support of NSF grant IIS-1838071 %
and the Prof. Amar G. Bose Research Grant.
This work was partially completed at the Voting Rights Data Institute in
the summer of 2018.

%% file: Sections/2Complexity/Appendix/TediousLemmaVerificationPicture.tex
\newpage

\section{Appendix for complexity results}
\subsection{Verifying \Cref{lem:tedious}}
\begin{figure}[h]
    \centering
    \hspace*{-6cm}  
    \scalebox{.93}{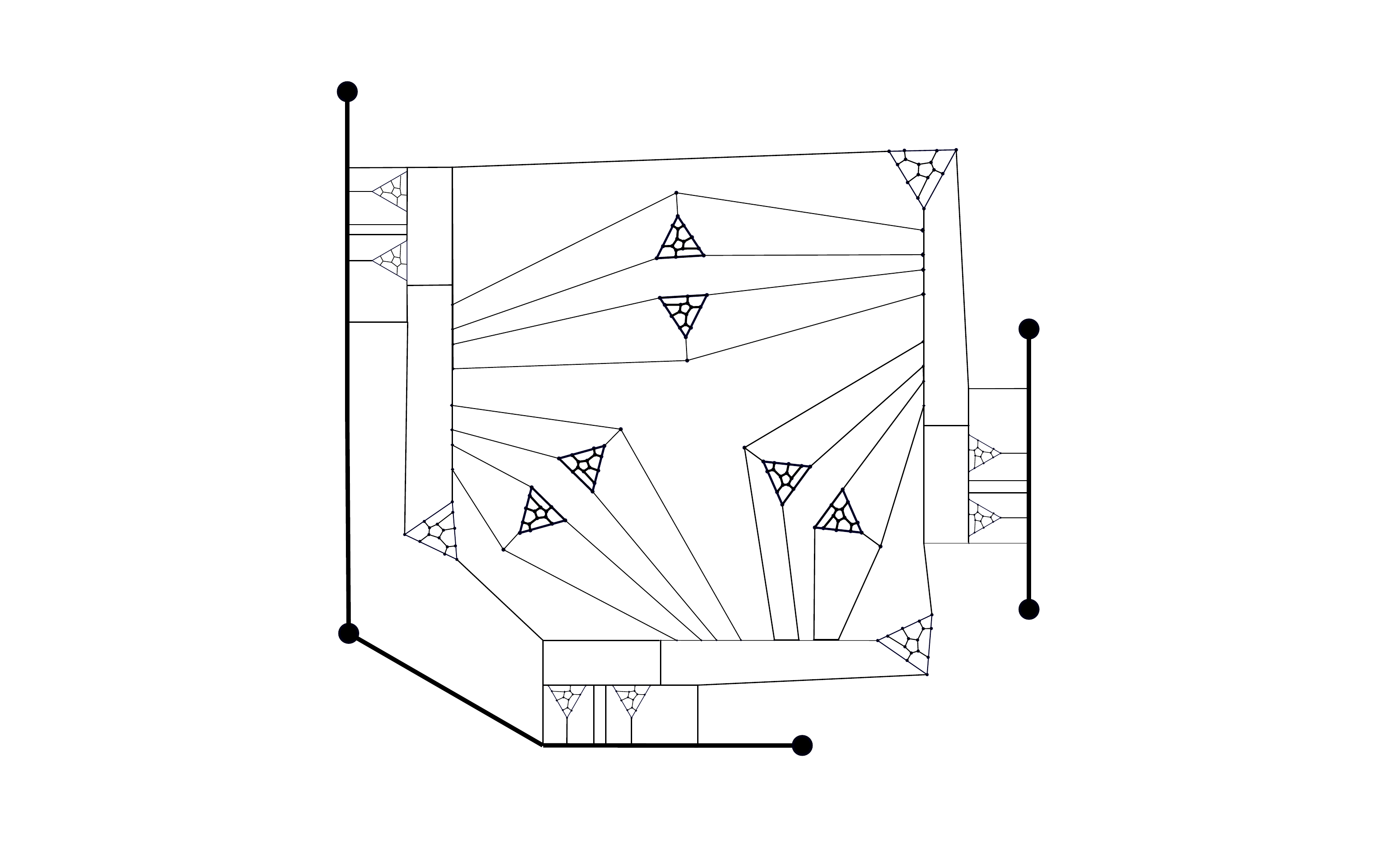}
    \caption{The $3OR$ subdivision as it appears in \Cref{lem:tedious}.}
    \label{fig:3ORDetailedForTedious}
\end{figure}

%% file: Images/3ORDetailedForTedious.pdf_tex
\begingroup%
  \makeatletter%
  \providecommand\color[2][]{%
    \errmessage{(Inkscape) Color is used for the text in Inkscape, but the package 'color.sty' is not loaded}%
    \renewcommand\color[2][]{}%
  }%
  \providecommand\transparent[1]{%
    \errmessage{(Inkscape) Transparency is used (non-zero) for the text in Inkscape, but the package 'transparent.sty' is not loaded}%
    \renewcommand\transparent[1]{}%
  }%
  \providecommand\rotatebox[2]{#2}%
  \newcommand*\fsize{\dimexpr\f@size pt\relax}%
  \newcommand*\lineheight[1]{\fontsize{\fsize}{#1\fsize}\selectfont}%
  \ifx\svgwidth\undefined%
    \setlength{\unitlength}{899.01227397bp}%
    \ifx\svgscale\undefined%
      \relax%
    \else%
      \setlength{\unitlength}{\unitlength * \real{\svgscale}}%
    \fi%
  \else%
    \setlength{\unitlength}{\svgwidth}%
  \fi%
  \global\let\svgwidth\undefined%
  \global\let\svgscale\undefined%
  \makeatother%
  \begin{picture}(1,0.62875809)%
    \lineheight{1}%
    \setlength\tabcolsep{0pt}%
    \put(0,0){\includegraphics[width=\unitlength,page=1]{3ORDetailedForTedious.pdf}}%
    \put(0.24916668,0.34996749){\color[rgb]{0,0,0}\makebox(0,0)[rt]{\lineheight{1.25}\smash{\begin{tabular}[t]{r}{\Huge $e_1$}\end{tabular}}}}%
    \put(0.41214587,0.05417326){\color[rgb]{0,0,0}\makebox(0,0)[t]{\lineheight{1.25}\smash{\begin{tabular}[t]{c}{\Huge $e_2$}\end{tabular}}}}%
    \put(0,0){\includegraphics[width=\unitlength,page=2]{3ORDetailedForTedious.pdf}}%
    \put(0.25765546,0.18085321){\color[rgb]{0,0,0}\makebox(0,0)[lt]{\lineheight{1.25}\smash{\begin{tabular}[t]{l}Pocket\end{tabular}}}}%
    \put(0.36491602,0.54673096){\color[rgb]{0,0,0}\makebox(0,0)[lt]{\lineheight{1.25}\smash{\begin{tabular}[t]{l}LargeFace\end{tabular}}}}%
    \put(0.60208107,0.11530506){\color[rgb]{0,0,0}\makebox(0,0)[lt]{\lineheight{1.25}\smash{\begin{tabular}[t]{l}LargeFace\end{tabular}}}}%
    \put(0.30064873,0.01312888){\color[rgb]{0,0,0}\makebox(0,0)[lt]{\lineheight{1.25}\smash{\begin{tabular}[t]{l}AdjacentFace\end{tabular}}}}%
    \put(0.19103098,0.30649206){\color[rgb]{0,0,0}\makebox(0,0)[lt]{\lineheight{1.25}\smash{\begin{tabular}[t]{l}Adjacent\end{tabular}}}}%
    \put(0.20679084,0.28550312){\color[rgb]{0,0,0}\makebox(0,0)[lt]{\lineheight{1.25}\smash{\begin{tabular}[t]{l}Face\end{tabular}}}}%
    \put(0.75335119,0.31371608){\color[rgb]{0,0,0}\makebox(0,0)[lt]{\lineheight{1.25}\smash{\begin{tabular}[t]{l}Adjacent\end{tabular}}}}%
    \put(0.77149461,0.29391891){\color[rgb]{0,0,0}\makebox(0,0)[lt]{\lineheight{1.25}\smash{\begin{tabular}[t]{l}Face\end{tabular}}}}%
    \put(0.75084961,0.24957747){\color[rgb]{0,0,0}\makebox(0,0)[lt]{\lineheight{1.25}\smash{\begin{tabular}[t]{l}{\Huge $e_3$}\end{tabular}}}}%
  \end{picture}%
\endgroup%

%% file: Sections/2Complexity/Appendix/6bProvingThatRdPreserves3CCP.tex
\subsection{Proving that $R_d$ preserves $3CCP$ graphs}\label{section:preserves3CCP}

By construction, $R_d(G)$ (\cref{defn:GtoRdGconstruction}) remains cubic, and $R_d(G)$ is planar if $G$ is planar. The next few lemmas show that if $G$ is $3$-connected, so is $R_d(G)$. We will let $\tilde{R}_d$ be the graph obtained from $R_d$ by adding $3$ leaf edges to each of $\{a_0, b_0, c_0\}$. The following lemma will show us that we can replace cubic vertices of $G$ with copies of $R_d$ and preserve $3$-connectedness:

\begin{figure}
    \centering
    \begin{tabular}{cc}
    \def\svgscale{.6}{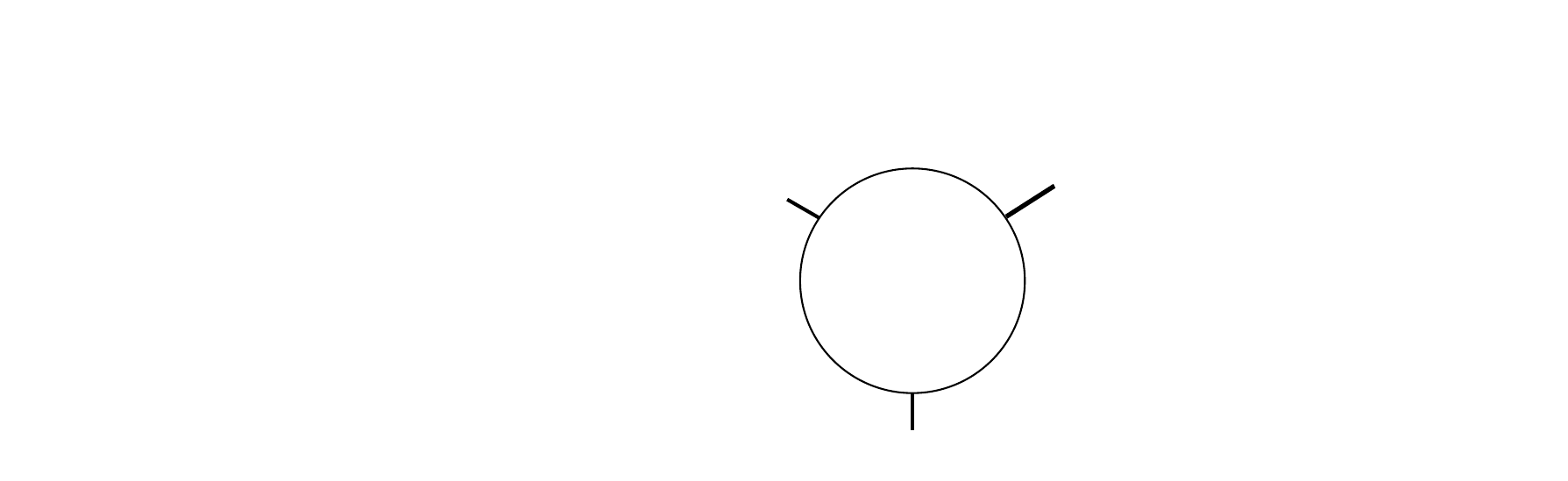}
    \end{tabular}
    \caption{Constructions described in \Cref{lem:3_connected_lemma}    }\label{fig:3ConnectednessConstruction}
\end{figure}

\begin{lem}\label{lem:3_connected_lemma}
Suppose that $A$ is a graph with $3$ leaf nodes, denoted by $L = \{l_1, l_2, l_3\}$.  Let $A'$ be the graph obtained by identifying the 3 leaf nodes of $A$. Let $G$ be some graph with a cubic vertex $v \in V(G)$. Let $A(G)$ be a graph obtained from $G$ by replacing $v$ by $A$: that is, by deleting $v$ from $G$ and choosing some identification between the $3$ leaf nodes of $A$ and the $3$ neighbors of $v$.  %
Then, if $G$ is 3 connected and $A'$ is 3 connected, $A(G)$ is 3 connected.
\end{lem}
\begin{proof}
Suppose that $a, b, x, y \in V(A(G))$. We will show that there is a path in $A(G) \setminus \{a,b\}$ between $x$ and $y$. Let $B = V(A) \setminus L$, and let $C : V(A(G)) \to V(G)$ be the map that contracts $B$ back to $v$: for $s \not \in B$, $C(s) = s$, and for $s \in B$, $C(s) = v$. There are two cases to consider:

\begin{enumerate}
    \item If $C(x) \not = v$ and $C(y) \not = v$, then there is a path between $C(x)$ and $C(y)$ in $G \setminus \{ C(a), C(b) \}$, since $G$ is 3-connected. This can be lifted to a path between $x$ and $y$ in $A(G) \setminus \{a, b\}$.
    \item If $C(x) = v$, it is always possible to find a path in $A(G) \setminus \{a,b\}$ from $x$ to some $x'$ with $C(x') \not = v$. There are three cases:

\begin{enumerate}
\item If $| L \cap \{a,b\}| = 0$: As $A'$ is $3$-connected, there is a path in $A' \setminus (B \cap \{a,b\})$ from $x$ to $w$. This gives a path in $A(G)$ from $x$ %
to a node %
$L$. 

\item If $| L \cap \{a,b\}| = 1$: At most one node of $\{a,b\}$ can be contained in $B = A \setminus L$. Thus, $A' \setminus (B \cap \{a,b\})$ is $2$-connected, so there are two paths in $A' \setminus (B \cap \{a,b\})$  from $x$ to $w$. These give paths in $A(G)$, which only intersect $\{a,b\}$ at $L$, and of these paths connects to a node in $L \setminus \{a,b\}$.

\item If $| L \cap \{a,b\}| = 2$: Since $A'$ is $3$ connected, there are three node disjoint paths from $x$ to $w$. In $A(G)$, this corresponds to a path to each of the three leaf nodes, one of which is not contained in $\{a,b\}$.

\end{enumerate}
    
Likewise, if $C(y) = v$, then we can connect $y$ to some $y'$, with $C(y') \not = v$. Once we have connected $x$ to $x'$ and $y$ to $y'$ outside of $A(v)$, we are back in Case 1.

\end{enumerate}

\end{proof}

The following lemma is well known; it is one of the Barnette-Grunbaum (BG) operations, introduced in \cite[Proof of Theorem 2]{barnette1969steinitz}. See also \cite{schmidt2011structure}. %

\begin{lem}[BG-operation] \label{lem:addedge3connected}
Let $G$ be a $3$-connected graph. Let $e_1$ and $e_2 \in E(G)$. Suppose that $G'$ is the graph obtained from $G$ by subdividing each $e_i$ by introducing a vertex $x_i$, and then adding an edge from $x_1$ to $x_2$. Then $G'$ is $3$-connected.
\end{lem}

\begin{lem}
If $G$ is 3-connected, then so is $R_d(G)$.
\end{lem}
\begin{proof}

If we show that $(\tilde{R}_d)'$ (in the notation of \Cref{lem:3_connected_lemma} and the paragraph preceding it) is $3$ connected, then the claim follows from \Cref{lem:3_connected_lemma} by considering $R_d(G)$ as obtained by replacing each node of $G$ by an $\tilde{R}_d$ one at a time in the sense given by \Cref{lem:3_connected_lemma}. To prove that $(\tilde{R}_d)'$ is 3-connected we argue by induction. In the base case, $(\tilde{R}_0)'$ is a $K_4$ graph, so it is $3$-connected. Let $Q_d$ be obtained from $(\tilde{R}_d)'$ by adding a single node $c$ in the center connected with three edges subdividing edges of the inner circle of $R_d$. If $(\tilde{R}_d)'$ is $3$-connected, then it follows from applying the BG-operation of \Cref{lem:addedge3connected} twice that $Q_d$ is also $3$-connected. From this it follows that $(\tilde{R}_{d+1})'$ $3$-connected, because $(\tilde{R}_{d+1})'$ is obtained by replacing $c$ with an $\tilde{R}_0$, so it is also $3$-connected by \Cref{lem:3_connected_lemma}.

\end{proof}

\begin{remark}
\Cref{lem:3_connected_lemma} and \Cref{lem:addedge3connected} make it relatively straightforward to check that inserting the $3OR$ gadget preserves $3$-connectedness, which is stated in \cite{garey_plane_1976} without proof.
\end{remark}

\begin{lem}
Suppose that $H$ is a cubic planar graph, with face degree bounded by $d$. Then $R_d(H)$ has face degree bounded by $3d$.
\end{lem}
\begin{proof}
For each face, each vertex along that face supplies two additional edges when we replace vertices with copies of $R_d$. Thus, the face degree multiplies by $3$. The claims follows. %
\end{proof}

Altogether, we have shown that for all $d$, the construction $G \to R_d(G)$ sends $\mathscr{C}_{m}$ into $\mathscr{C}_{3m}$, where $\mathscr{C}_m$ is as in \Cref{defn:Cm3CCP}.

%% file: Images/3ConnectednessConstruction.pdf_tex
\begingroup%
  \makeatletter%
  \providecommand\color[2][]{%
    \errmessage{(Inkscape) Color is used for the text in Inkscape, but the package 'color.sty' is not loaded}%
    \renewcommand\color[2][]{}%
  }%
  \providecommand\transparent[1]{%
    \errmessage{(Inkscape) Transparency is used (non-zero) for the text in Inkscape, but the package 'transparent.sty' is not loaded}%
    \renewcommand\transparent[1]{}%
  }%
  \providecommand\rotatebox[2]{#2}%
  \newcommand*\fsize{\dimexpr\f@size pt\relax}%
  \newcommand*\lineheight[1]{\fontsize{\fsize}{#1\fsize}\selectfont}%
  \ifx\svgwidth\undefined%
    \setlength{\unitlength}{515.0271463bp}%
    \ifx\svgscale\undefined%
      \relax%
    \else%
      \setlength{\unitlength}{\unitlength * \real{\svgscale}}%
    \fi%
  \else%
    \setlength{\unitlength}{\svgwidth}%
  \fi%
  \global\let\svgwidth\undefined%
  \global\let\svgscale\undefined%
  \makeatother%
  \begin{picture}(1,0.30655887)%
    \lineheight{1}%
    \setlength\tabcolsep{0pt}%
    \put(0,0){\includegraphics[width=\unitlength,page=1]{3ConnectednessConstruction.pdf}}%
    \put(0.49428326,0.195094){\color[rgb]{0,0,0}\makebox(0,0)[t]{\lineheight{1.25}\smash{\begin{tabular}[t]{c}$l_1$\end{tabular}}}}%
    \put(0.67488268,0.19113532){\color[rgb]{0,0,0}\makebox(0,0)[lt]{\lineheight{1.25}\smash{\begin{tabular}[t]{l}$l_2$\end{tabular}}}}%
    \put(0.58329953,0.00415528){\color[rgb]{0,0,0}\makebox(0,0)[t]{\lineheight{1.25}\smash{\begin{tabular}[t]{c}$l_3$\end{tabular}}}}%
    \put(0.53408256,0.10172349){\color[rgb]{0,0,0}\makebox(0,0)[lt]{\lineheight{1.25}\smash{\begin{tabular}[t]{l}$A$\end{tabular}}}}%
    \put(0,0){\includegraphics[width=\unitlength,page=2]{3ConnectednessConstruction.pdf}}%
    \put(0.08298426,0.12064147){\color[rgb]{0,0,0}\makebox(0,0)[lt]{\lineheight{1.25}\smash{\begin{tabular}[t]{l}$v$\end{tabular}}}}%
    \put(0.07226956,0.03973975){\color[rgb]{0,0,0}\makebox(0,0)[t]{\lineheight{1.25}\smash{\begin{tabular}[t]{c}$G$\end{tabular}}}}%
    \put(0,0){\includegraphics[width=\unitlength,page=3]{3ConnectednessConstruction.pdf}}%
    \put(0.31004702,0.0456981){\color[rgb]{0,0,0}\makebox(0,0)[t]{\lineheight{1.25}\smash{\begin{tabular}[t]{c}$A(G)$\end{tabular}}}}%
    \put(0,0){\includegraphics[width=\unitlength,page=4]{3ConnectednessConstruction.pdf}}%
    \put(0.31130789,0.12514159){\color[rgb]{0,0,0}\makebox(0,0)[t]{\lineheight{1.25}\smash{\begin{tabular}[t]{c}$A$\end{tabular}}}}%
    \put(0,0){\includegraphics[width=\unitlength,page=5]{3ConnectednessConstruction.pdf}}%
    \put(0.96213789,0.171484){\color[rgb]{0,0,0}\makebox(0,0)[lt]{\lineheight{1.25}\smash{\begin{tabular}[t]{l}$A'$\end{tabular}}}}%
    \put(0,0){\includegraphics[width=\unitlength,page=6]{3ConnectednessConstruction.pdf}}%
    \put(0.82763734,0.28186836){\color[rgb]{0,0,0}\makebox(0,0)[t]{\lineheight{1.25}\smash{\begin{tabular}[t]{c}$w$\end{tabular}}}}%
    \put(0,0){\includegraphics[width=\unitlength,page=7]{3ConnectednessConstruction.pdf}}%
  \end{picture}%
\endgroup%

%% file: Sections/2Complexity/Appendix/7bProofOfDualityTheorem.tex
\subsection{Duality for connected $k$-partitions}\label{appendix:duality}

In this section, we prove \Cref{thm:kpartitionduality:ref}.
There are three main steps in the proof, which answer three questions: 
\begin{enumerate}
    \item Can you recover a connected partition from its edge boundary? 
    \item What does the edge boundary of a partition of a plane graph look like in the dual graph? 
    \item In what way is the number of blocks of a connected partition reflected in its representation in the dual graph?
\end{enumerate}
We state and prove the theorems that answer these questions in the next three subsections, and the end result is \Cref{thm:kpartitionduality:ref}. The reader may note that 1) is answered by matroid duality between the graphic and cographic matroids (specifically between flats, which correspond to connected partitions, and unions of circuits)\footnote{We wish to acknowledge a helpful \href{https://mathoverflow.net/q/316132/41873}{MathOverflow discussion} that drew our attention to this connection to matroids \cite{matroidmathoverflow}.}, that the answer to 2) follows quickly from the usual bond-cycle duality, and that 3) is a discrete version of Alexander duality.

\subsubsection{Connected partitions and edge cuts}

We first recall the bijection between connected partitions and edge-cuts:

\begin{defn}[Unordered connected partitions]
Let $\mathscr{P}(G)$ denote the set of \emph{unordered} partitions of $V(G)$. Let $\mathscr{P}^c(G) \subseteq \mathscr{P}(G)$ denote the set of partitions such that each block induces a connected subgraph.
That is, $\mathscr{P}^c(G) = \bigcup_{k = 1}^{|V(G)|} \mathscr{P}_k(G)$.

\end{defn}

\begin{defn}[Edge cut]
Let $P$ be a partition of $V(G)$.  If $P = \{A_1, \ldots, A_k\}$, then we refer to the $A_i$ as the \emph{blocks} of $P$. Let $\cut(P)$ denote the set of edges of $G$ with endpoints in different blocks of $P$. %
\end{defn}

\begin{defn}[Cut sets]
Let $\Cuts(G)$ be the set of the cuts of partitions of $V(G)$. That is, $\Cuts(G) = \{ \cut( P ) : P \in \mathscr{P}(G) \}$. The elements of $\Cuts(G)$ are called cut sets.
\end{defn}

\begin{defn}[Component map]
Given $J \in \Cuts(G)$, define a partition $\comp(J) \in \mathscr{P}^c(G)$ as the connected components of $G \setminus J$. This defines a function $\comp : \Cuts(G) \to \mathscr{P}^c(G)$.
\end{defn}

\begin{prop}\label{prop:matroidduality}
The functions $\comp$ and $\cut$ induce a bijection between $\Cuts(G)$ and $\mathscr{P}^c(G)$.
\end{prop}
\begin{proof}

To show that $\cut$ is surjective, we observe that if $P$ is any partition, we can define a connected partition $P'$, whose blocks are the connected components of the blocks of $P$, and $\cut{P} = \cut{P'}$.
We will conclude by showing that $\comp \circ \cut = id$. 
First, observe that $\comp \circ \cut$ does not merge any blocks, since every path in $G$ between two blocks has to cross a cut edge. %
Second, observe that $\comp \circ \cut$ does not split any blocks, since two points in any block \emph{of a connected partition} are always connected by a path that does not use any cut edges. %

\end{proof}

So far we have established that a connected partition is determined by the boundaries between its blocks. Next, we work towards characterizing the shapes that can arise as such boundaries, by treating them as subgraphs of the planar dual. 

\subsubsection{Dual connected partitions and connected partitions}

The following straightforward lemma is useful for proving the duality theorem:

\begin{lemma}\label{lem:equivbridgeunion}
Let $G$ be a graph, and $J \subseteq E(G)$. Then each connected component of $G[J]$ is two edge-connected if and only if each connected component of $G[J]$ has no bridge edges if and only if $G[J]$ is a union of not-necessarily disjoint simple cycles.
\end{lemma}

\begin{defn}[Dual connected partitions] Let $\mathscr{E}_2(G)$ denote the set of subsets of edges of $G$ that are unions of not-necessarily disjoint simple cycles. We will call these the \emph{dual connected partitions}. %
\end{defn}

The purpose of the next few propositions is to show that dual connected partitions are plane duals of the cuts of connected partitions. First we recall the bijection between the edges of a plane graph and the edges of its dual:

\begin{defn}[Dual edges]
Let $G$ be a plane graph. For an edge $e  \in E(G)$, let $e^*$ denote the edge in $G^*$ with the property that the two endpoints of $e^*$ are the shores of $e$, i.e., the two faces that are separated by $e$. For a set $J \subseteq E(G)$ , denote by $J^*$ the corresponding set of edges in $G^*$. We define a function $D(J) = J^*$, which is a bijection $2^{E(G)} \to 2^{E(G^*)}$
\end{defn}

We aim to prove a plane duality between $\Cuts(G)$ and $\mathscr{E}_2(G^*)$. In particular, we want show that $D$ induces a bijection between $\Cuts(G)$ and $\mathscr{E}_2(G^*)$. Towards that, we will recall the plane duality between even subgraphs and the edge boundaries, which will be useful for controlling the topology of $D(J)$ for $J \in \Cuts(G)$.

\begin{defn}[Edge boundary] Let $G = (V,E)$ be a graph, and $A \subseteq V$. Denote by $\cut(A) = \cut( \{A, A^c\} ) =  \partial_E(G)$, the edge boundary of $A$.
\end{defn}

\begin{defn}[Even Subgraphs]\label{defn:evensubgraphs}
Let $G = (V,E)$ be a graph. A subset $J \subseteq E$ defines an even subgraph $G[J]$ if the degree of each node of $G[J]$ is even. Let $Even(G) = \{ J \subseteq 2^{E(G)} : G[J] \text{ is an even subgraph} \}$.
\end{defn}

\begin{prop}[Proposition 2.1 in \cite{Erickson}\footnote{Beware that his terminology is different from ours; specifically, he refers to what we call an edge boundary as an edge cut.}] \label{prop:evensubgraphs}
Let $G$ be a connected plane graph, and let $H \subseteq E(G)$. Then, $H$ is an even subgraph if and only if $H^*$ is an edge boundary. Moreover, $H$ is a simple cycle if and only if $H^*$ is the cut of a connected $2$-partition.
\end{prop}

The following well-known lemma will be useful for relating $Even(G)$ to $\mathscr{E}_2(G)$:
\begin{lem}[Euler]\label{prop:Euler}
Let $G$ be a graph. Then $J \subseteq E(G)$ is an even subgraph if and only if $J$ is a union of pairwise disjoint simple cycles.
\end{lem}

The previous theorem characterized edge boundaries, which are the cut sets of not-necessarily connected 2-partitions, using the planar dual. %
The next proposition will characterize the cut sets of connected $k$-partitions using the planar dual.

\begin{prop}\label{prop:DBijection}
Let $G$ be a connected plane graph. Let $H \subset E(G)$. Then $H \in \mathscr{E}_2(G)$ if and only if $H^*$ is a cut set. In particular, $D$ gives a bijection between $\mathscr{E}_2(G)$ and $\Cuts(G^*)$.
\end{prop}
\begin{proof}

\begin{itemize}
    \item [$\Rightarrow$]
    Suppose that $H \in \mathscr{E}_2(G)$, and let $P \in \mathscr{P}^c(G^*)$ be the connected partition defined by the connected components of $G^* \setminus H^*$. We show that $\cut(P) = H^*$. Let $e \in H$, and let $C \subseteq H$ be a cycle containing it. Then the shores of $e$ are necessarily in different components of $G^* \setminus C^*$ by \Cref{prop:evensubgraphs}, and thus in different components of $G^* \setminus H^*$. Thus, $e^* \in \cut(P)$, so $H^* \subseteq \cut(P)$. On the other hand, if $e^* \in \cut(P)$, then the shores of $e$ are in different components of $G^* \setminus H^*$, so every path in $G^*$ between the two shores of $e$ must pass through $H^*$. In particular, the path $e^*$ must pass through $H^*$, which means that $e^* \in H^*$. Thus, $\cut(P) \subseteq H^*$.

    \item [$\Leftarrow $]

    Now suppose that $H^* = \cut(P)$ for some $P \in \mathscr{P}^c(G)$. Take any $e \in H$. We want to show that $e$ is not a bridge edge. Suppose that $A$ and $B$ are the blocks of $P$ containing the faces in the two shores of $e$. Now, create a not necessarily connected $2$-partition $P'$ by reassigning all of the blocks of $P$ that are not $A$ or $B$ to be part of block $A$. Now, $\cut(P') \subseteq \cut(P)$ %
    and $e \in \cut(P')$%
    . Since $\cut(P')^*$ is a union of edge disjoint cycles (\Cref{prop:Euler}), there is a simple cycle $C$ in $\cut(P')^*$ that contains $e$. In particular, $\cut(P)^* = H \supseteq C$, so $e$ could not have been a bridge edge of $H$.
\end{itemize}

\end{proof}

We summarize the previous two results in a single duality statement:
\begin{prop}[Duality between connected partitions and dual connected partitions]\label{prop:geometricpartofduality}The functions $\comp \circ D^{-1}$ and $D \circ \cut$ are mutual inverses, inducing a bijection between $\mathscr{P}^c(G)$ and $\mathscr{E}_2(G^*)$.
\end{prop}
\begin{proof}
Since $D$ induces a bijection between $\mathscr{E}_2(G^*)$ and $\Cuts(G)$ (\Cref{prop:DBijection}) and $\comp$ and $\cut$ give a bijection between $Cuts(G)$ and $\mathscr{P}^c(G)$ (\Cref{prop:matroidduality}), the claim follows.
\end{proof}
\subsubsection{The number of blocks and the circuit rank}
Now we will review some facts that relate the number of blocks in a connected partition of a plane graph to the circuit rank of the corresponding dual connected partition.
\begin{defn}[$h_1$ and $h_0$] Let $G$ be a graph. Then $h_1(G)$ denotes the circuit rank of $G$, $h_0(G)$ denotes the number of connected components of $G$. 
\end{defn}
\begin{prop}\label{prop:Eulersformula} Let $G = (V,E)$ be a plane graph. %
Then $1 + h_0(G) = |V(G)| - |E(G)| + |F(G)|$. 
\end{prop}
\begin{proof} One can add $h_0 - 1$ edges connecting the components, without changing the number of faces. If the new graph has $E'$ edges, then $E' = E + h_0 -1$, and we have $V - E' + F = 2$, from which the formula follows.\end{proof}
\begin{prop}\label{prop:ranknullity}
If $G$ is a graph, then $h_1(G) - h_0(G) = |E(G)| - |V(G)|$. %
\end{prop}
\begin{proof}Since $h_1$ is the cycle rank, which is the dimension of the kernel of the boundary map $\partial: \mathbb{F}_2^E \to \mathbb{F}_2^V$, and $h_0$ is the rank of the cokernel of $\partial$, this is just a statement of the rank-nullity theorem. %
\end{proof}
\begin{prop}\label{prop:topologicalpartofduality}
Let $G$ be a connected plane graph. For $P \in \mathscr{P}^c(G)$ and $J \in \mathscr{E}_2(G)$, $h_1( G^*[\cut( P)^*]) = |P| - 1$ and $| \comp(J^*)| - 1 = h_1(G[J])$. Here $|P|$ counts the number of blocks of $P$.
\end{prop}
\begin{proof}
Let $P \in \mathscr{P}^c(G)$. Let $J = \cut(P)^*$. %
\Cref{prop:Eulersformula} and \Cref{prop:ranknullity} together yield that %
$h_1( G^*[J] ) = h_0 - V + E =|F(G^*[J])| - 1$. Since the number of faces of $G^*[J]$ is the number of components of $G \setminus J$, and since $\comp(J) = \comp \circ \cut (P) = P$, we obtain $h_1( G^*[\cut( P)^*]) = |P| - 1$. %
Now consider $J \in \mathscr{E}_2(G^*)$, and let $P = \comp(J^*)$. Since $J = \cut(P)^*$, it follows from $h_1( G^*[\cut( P)^*]) = |P| - 1$ that $h_1( G^*[J]) = |P| - 1$, so the claim follows from $(G^*)^* = G$. %
\end{proof}
Finally, we %
present the duality theorem:
\begin{defn}[Dual $k$-partition]
We define $P_k^*(G) = \{ J \in \mathscr{E}_2(G) : h_1( G[J]) = k - 1 \}$. We call the elements of this set dual $k$-partitions.
\end{defn}
\begin{thm}[Duality between $P_k(G)$ and $P_k^*(G^*)$]\label{thm:kpartitionduality}
The map $D \circ \cut : P_k(G) \to P_k^*(G^*)$ is a bijection, with $\comp \circ D^{-1} : P_k^*(G^*) \to P_k(G)$ as its inverse. Both are computable in polynomial time.
\end{thm}
\begin{proof}
The bijection follows from \Cref{prop:topologicalpartofduality} and \Cref{prop:geometricpartofduality}. It is well known that $D$ and $\comp$ and $\cut$ can be computed in polynomial time.
\end{proof}

%% file: Sections/5PositiveResults/Appendix/SamplingUsingMarginals.tex
\section{Positive results}

\subsection{Using marginal counts}\label{appendix:inductivesampling}

In this section, we prove the correctness of \Cref{alg:inductivesampling}.%

\begin{prop}\label{prop:InductiveSampling}
The output of \Cref{alg:inductivesampling} is a random variable $J$ valued in $2^{ [n]}$ and drawn with distribution $p$. Moreover, if a call to the oracle $O$ takes time $O(T(n))$, then the total runtime of \Cref{alg:inductivesampling} is $O(nT(n))$.
\end{prop}
\begin{proof}

Let $S$ be a random variable distributed according to $p$. Let $\mathbb{P}$ be the probability measure underlying the process of the algorithm and the random variable $S$. Let $J_k$ be the random set $J$ on the $k$th step of \Cref{alg:inductivesampling}. Using induction, we will show that, for all $m \in [n]$,
\begin{equation}\label{eqn:inductivehypothesisforsampling}
\mathbb{P} ( J_m = W) = \mathbb{P} ( S \cap [m] = W).
\end{equation}
The desired conclusion is the case $m = n$. In the base case, when $k = 1$,  \Cref{eqn:inductivehypothesisforsampling} holds because $\mathbb{P} ( J_1 = \{1\} ) = p(1 | \emptyset) =  \mathbb{P} ( 1 \in S ) = \mathbb{P} ( S \cap \{1\} = \{1 \})$. Now, suppose that for some $m \geq 1$ with $m < n$, it holds that $\mathbb{P}( J_m = W) = \mathbb{P} ( S \cap [m] = W)$ for all $W \subset [m]$. Recall that from the definition of \Cref{alg:inductivesampling} we have that $\mathbb{P}(J_{m+1} =  W \cup \{m+1\} | J_m = W) = \mathbb{P}(  m + 1 \in S | S \cap [m] = W)$.

The inductive step now follows by a computation:
\begin{align*}
    \mathbb{P}(J_{m+1} = W \cup \{m+1\}) &= \mathbb{P}(J_{m+1} =  W \cup \{m+1\} | J_m = W)) \mathbb{P}(J_m = W)\\ &= \mathbb{P}(  m + 1 \in S | S \cap [m] = W) \mathbb{P}( S \cap [m] = W) \\ &= \mathbb{P}( S \cap  [m+1] = W \cup \{m + 1\}).
\end{align*}
Likewise $\mathbb{P}( J_{m+1}  =W ) = \mathbb{P}( S \cap [m+1] = W)$. %
\end{proof}

%% file: Sections/5PositiveResults/Appendix/SeriesParallelGraphsBackground.tex
\subsection{Series-parallel graphs}

We recall the definition of a series-parallel graph, a class of graphs well suited to dynamic programming algorithms.

\begin{defn}[Two-terminal graphs]
A two-terminal graph $G$ is a graph with two distinguished nodes: a source, $\sigma(G)$, and a sink $\tau(G)$. A pair of two-terminal graphs, $G$ and $H$, are said to be isomorphic, $G \cong H$, if there is an isomorphism of the underlying graphs which maps the source to the source and the sink to the sink.
\end{defn}

\begin{example}
The complete graph on $\{0,1\}$ is naturally a two-terminal graph, where we set $\sigma(K_2) = 0$ and $\tau(K_2) = 1$. We denote it by $K_2$
\end{example}

\begin{defn}[Series Composition]\label{defn:series}
Let $G_1$ and $G_2$ be two-terminal graphs. We define $G_1 \circ G_2$ as the graph obtained from the disjoint union of $G_1$ and $G_2$ by identifying $\tau_1$ and $\sigma_2$, and we make it into a two-terminal graph by setting $\sigma ( G_1 \circ G_2) = \sigma(G_1)$ and $\tau(G_1 \circ G_2) = \tau(G_2)$.
\end{defn}

\begin{defn}[Parallel Composition]
In the notation of \Cref{defn:series}, we define $G_1 \parallelsum G_2$ as the graph obtained from the disjoint union of $G_1$ and $G_2$ by making the identifications $\sigma_1 \sim \sigma_2$ and $\tau_1 \sim \tau_2$. We make $G_1 \parallelsum G_2$ into a two-terminal graph by defining $s( G_1 \parallelsum G_2) = [\sigma(G_1)] = [\sigma(G_2)]$ and $t( G_1 \parallelsum G_2) = [\tau(G_1)] = [\tau(G_2)]$.
\end{defn}

\begin{defn}[Series-Parallel graphs]
We define the class of series-parallel graphs as the smallest class of two-terminal graphs that is closed under Parallel Composition and Series Composition, and which contains the two-terminal graph $K_2$.
\end{defn}

The feature which makes series-parallel graphs convenient for dynamic programming is that we can record the series and parallel composition operations into a tree, called the $SP$-tree, around which we can organize dynamic programs:

\begin{defn}[$SP$-tree]
An $SP$ tree is a rooted binary tree, where the children of any internal node have an ordering, and where each internal node of the tree is labelled $P$ or $S$. We assign to each leaf a copy of the two-terminal graph $K_2$. Then, to each internal node we assign the graph obtained by applying either $S$eries Composition or $P$arallel composition to its children, depending on the label of that internal node; for the series composition (labelled $S$), the order of the composition is in the order on the children prescribed by the tree, and for parallel composition (labelled $P$) the order does not matter. If $T$ is an $SP$-tree, define $G(T)$ as the two-terminal graph assigned to the root of $T$. If $X = (G, \sigma, \tau)$ is a series-parallel graph, we say that an $SP$-tree $T$ is an $SP$-tree for $X$ if $X \cong G(T)$.
\end{defn}

\begin{lem}[Theorem 4.1 of \cite{bodlaender1996parallel} ] 
Given a graph $G$, determining if $G$ is series-parallel and if so building an $SP$-tree for it can be done in linear time.
\end{lem}

%% file: Sections/5PositiveResults/Appendix/SimpleCycleDP.tex
\subsection{Computing marginal probabilities on graphs of treewidth $2$}\label{appendix:SCDP}

In this section we give a polynomial time algorithm for counting simple cycles on graphs of treewidth $2$. We also prove as a byproduct of the method that it is possible to efficiently sample from a much broader family of distributions than just uniform, namely those defined by \Cref{defn:edgeweightprobability}. For these computations, it will be convenient to extend the concept of a network to allow edge weights in other rings. These results are extended in \cite{ignasi} to graphs of bounded treewidth.

\begin{defn}[$R$-network]
Let $G$ be a graph and $R$ a ring, and let $ w :E (G) \to R$ be a function. Then we will call $(G,w)$ an $R$-network.
\end{defn}

Let $\mathbb{Q}$ be the rationals. Let $(G,w)$ be a $\mathbb{Q}$-network with non-negative weights, and let $J, J' \subseteq E(G)$. For sampling with \Cref{alg:inductivesampling} we would like to be able to compute the total $\nu_w$ (\cref{defn:edgeweightprobability}) mass of the simple cycles containing $J$ and disjoint from $J'$. The approach here will be to encode that mass as the evaluation of a generating function (\Cref{defn:scgf}), which we can evaluate efficiently on series-parallel graphs by dynamic programming.

\begin{defn}[Simple cycle generating function]\label{defn:scgf}
Let $(G,w)$ be an $R$-network. Let $f_{SC}(G,w) := \sum_{ C \in SC(G) } \prod_{e \in C} w(e)$ denote the generating function of the simple cycles of $G$ evaluated at the weights $w$.%
\end{defn}

Let $(G,c)$ be a $\mathbb{Q}$-network. To sample from $\nu_c$ of $G$, the marginals we need in the course of \Cref{prop:InductiveSampling} are easily computed from $N_c ( \{ C \in SC(G) : C \cap J' = \emptyset, C \supseteq J \} ) = \sum_{ \substack{C \in SC(G) , C \cap J' = \emptyset, C \supseteq J}} \prod_{e \in C} c(e)$. To obtain these measurements for any given $J, J'$, let $x$ be some formal variable, and set $w(e) = x c(e)$ for $e \in J$, $w(e') = 0$ for $e' \in J'$ and $w(f) = c(e)$ otherwise. Then the coefficient of the $x^{|J|}$ term of $f_{SC}(G,w)$ is $\sum_{ C \in SC(G) : J \subseteq C, J' \cap C = \emptyset} \prod_{e \in C} c(e)$, which is the $N_c$ mass of all the simple cycles that are disjoint from $J'$ and that contain $J$. We next show that we can compute $f_{SC}(G,w)$ \emph{if} $G$ is a series-parallel graph, via a dynamic programming algorithm which runs in time polynomial in $|G|$ and $|w|$. We also need to keep track of the corresponding generating function for the simple paths, which we define next.

\begin{defn}[Simple path generating function]
Let $(G,w)$ be a series-parallel $R$-network, with source $\sigma$ and sink $\tau$. Then we define $f_{SP}(G) = \sum_{ \gamma \in SP_{\sigma, \tau} (G) } \prod_{e \in \gamma} w(e)$, where $SP_{\sigma, \tau}(G)$ is the set of simple paths from $\sigma$ to $\tau$ in $G$, where a path is a \emph{sequence} of \emph{edges}.
\end{defn}

\begin{lem}\label{lem:combinationrules}

Let $(G_1,w_1)$ and $(G_2,w_2)$ be series-parallel $R$-networks. Let $w : E(G_1) \sqcup E(G_2) \to R$ be the unique weight function that restricts to $w_1$ and $w_2$. Then, let $w$ make both $G_1 \circ G_2$ and $G_1 \parallelsum G_2$ into $R$-networks, using that both have edge set $E(G_1) \cup E(G_2)$. Then:

\begin{equation}\label{eqn:SCunderComp}
    f_{SC}(G_1 \circ G_2, w) = f_{SC}(G_1, w_1) + f_{SC}(G_2, w_2)
\end{equation}
\begin{equation}\label{eqn:SPunderComp}
    f_{SP} (G_1 \circ G_2 , w) = f_{SP}(G_1, w_1) f_{SP}(G_2, w_2)
\end{equation}
\begin{equation}\label{eqn:SCunderparallel}
    f_{SC} ( G_1 \parallelsum G_2 ,w) = f_{SC}(G_1, w_1) + f_{SC}(G_2, w_2) + f_{SP}(G_1, w_1) f_{SP}(G_2,w_2)
\end{equation}
\begin{equation}\label{SPunderparalell}
    f_{SP}(G_1 \parallelsum G_2, w) = f_{SP}(G_1, w_1) + f_{SP}(G_2, w_2)
\end{equation}
\end{lem}
\begin{proof}
In each case, the equality on generating functions will follow from a bijection of sets.

\begin{proof}[Proof of \Cref{eqn:SCunderComp}]
$SC(G_1 \circ G_2) = SC(G_1) \cup SC(G_2)$.

\end{proof}
\begin{proof}[Proof of \Cref{eqn:SPunderComp}]
$SP(G_1 \circ G_2) = \{ \gamma_1 \circ \gamma_2 : \gamma_i \in SP(G_i) \}$, where $\circ$ between paths denotes concatenation.

\end{proof}
\begin{proof}[Proof of \Cref{eqn:SCunderparallel}]

$SC(G_1 \parallelsum G_2) = SC(G_1) \cup SC(G_2) \cup \{ Set(\gamma_1 \circ Reverse(\gamma_2)) : \gamma_i \in SP(G_i) \}$, where the $Reverse$ of a path is that path run backwards, and $Set(\_)$ takes the sequence of edges and turns it into a subset of $E(G_1 \parallelsum G_2)$.%

\end{proof}
\begin{proof}[Proof of \Cref{eqn:SCunderparallel}]
$SP(G_1 \parallelsum G_2) = SP(G_1) \cup SP(G_2)$

\end{proof}

\end{proof}

Now we recall a key lemma that will allow us to reduce the treewidth $2$ case to the series-parallel case, and prove the theorem about efficiently evaluating $f_{SC}$:

\begin{lem}[\cite{wald1983steiner} Theorem 4.1]\label{lem:embedinSPgraph}
If $H$ is a graph of treewidth $\leq 2$, then there is a series-parallel graph $G$ and an embedding $i : H \to G$. Both $H$ and $i$ are computable from $H$ in linear time.
\end{lem}

\begin{lem}\label{lem:simplecycleweightedmarginals}
Suppose $(H,w)$ is a $\mathbb{Q}[x]$-network with treewidth $\leq 2$. Then $f_{SC}(H,w)$ can be computed in time polynomial in $|(H,w)|$. In particular, if $(H,c)$ is a $\mathbb{Q}$-network, then $N_c ( \{ C \in SC(H) : C \cap J' = \emptyset, C \supseteq J \} )$ can be computed in time polynomial in $|(H,c)|$.
\end{lem}
\begin{proof}
By \Cref{lem:embedinSPgraph}, there is a series-parallel graph $G$, which contains $H$ as a subgraph, and moreover $G$ and $i : H \to G$ can be constructed in linear time. Extend $w$ to all the edges of $G$ by defining $w(e) = 0$ for $e \in E(G) \setminus E(H)$. %
Then we have that $f_{SC}(G,w) = f_{SC}(H,w)$ since all the cycles of $G$ that are not cycles of $H$ have weight zero.
We let $T_G$ be a binary $SP$-tree of $G$, which has $O(|G|)$ nodes. Using \cref{lem:combinationrules}, we compute $f_{SC}$ and $f_{SP}$ at each node. The cost of the calculation at each node is bounded by the cost of multiplying and adding the corresponding generating functions for the nodes children. Each generating function has degree at most $O ( |G| ( \max_{e \in E} \deg ( w(e)))$, and has coefficients that have binary encoding whose length is polynomial in $|(G,w)|$. Thus, the first claim follows. The second claim follows because if we set $w(e) = c(e)x1_{x \in J} + c(e)1_{x \not \in J' \cup J}$, then the coefficient of the $x^{|J|}$ term of $f_{SC}(G,w)$ is $N_c ( \{ C \in SC(H) : C \cap J' = \emptyset, C \supseteq J \} )$. %
\end{proof}

%% file: Sections/5PositiveResults/Appendix/DynamicProgramForBalancedPartitions.tex
\subsection{Dynamic program for counting balanced partitions on series-parallel graphs}\label{appendix:balancedDP}

In this section we show how to set up a dynamic program that will count the number of balanced connected $2$-partitions of a series-parallel graph. We will be interested in the case where nodes have weights valued in $\mathbb{N} = \{0,1,2,\ldots \}$, but it will be convenient extend these weights to take values in a larger monoid, $N$, which will also keep track of when a set of nodes is non-empty.

\begin{defn}[The monoid $N$]
Let $E$ be the commutative monoid given by $E = \{ \{\emptyset, \lnot \emptyset \}, \cup, \emptyset \}$ where $\emptyset$ is the identity element, and  $\lnot \emptyset \cup \lnot \emptyset = \lnot \emptyset$. %
Let $\mathbb{N}$ be the monoid of natural numbers with addition. Let $N = \mathbb{N} \times E$ and by abuse of notation we let $0 \in N$ denote the additive identity, and $+$ the binary operation in $N$. Let $n : N \to \mathbb{N}$ and $e : N \to E$ be the natural projections. 
\end{defn}

\begin{defn}[(Admissible) node-weighted graphs]
$(G,w)$ will denote a graph $G = (V,E)$ along with a function $w: V(G) \to N$. If $e(w(v)) = \lnot \emptyset$ for all $v \in V$, then we call $(G,w)$ admissible. For any $A \subseteq V(G)$, define $w(A) = \sum_{a \in A}w(a)$. 
\end{defn}

\begin{defn}[The DP-table $X(G,w)$] For a series-parallel graph $G$ with source $\sigma$ and sink $\tau$, and weight function $w : V(G) \to N$, we imitate \cite{dyer_complexity_1985} and define:
\begin{align*}
X(G, w) = \{ (a_1, a_2, a_3, m) \in N^3 \times \mathbb{N} &| \text{there are exactly $m$ partitions of $V(G)$ into blocks $V_1, V_2, V_3$ such that:} \\
& \text{(i) $w(V_i) = a_i$ and $G[V_i]$ is connected for $i = 1,2,3$}\\
& \text{(ii) $\sigma \in V_1$} \\
& \text{(iii) $a_3 = 0$ implies that $\tau(G) \in V_1$ and $a_3 \not = 0$ implies that $\tau(G) \in V_3$. } \}
\end{align*}
\end{defn}
That is, $X(G,w)$ is a function $N^3 \to \mathbb{N}$ that counts the number of partitions of $G$ into connected $3$-partitions which blocks of given weights. If we know $X(G,w)$, we will be able to calculate $|P_2^0(G,w)|$ (\cref{thm:calculatingbalancedpartitions}). %
\begin{defn}[Series and parallel composition for weighted SP graphs]\label{defn:SPweightscomposition}
Let $(G_1, w_1)$ and $(G_2,w_2)$ be series-parallel graphs with $N$ valued weights on the nodes. Define weights $w$ on the nodes of $G = G_1 \circ G_2$ and $G = G_1 \parallelsum G_2$ by first thinking of $w_1$ and $w_2$ as functions on $V(G)$, through extending them to the other nodes by assigning them the value $(0,  \lnot \emptyset)$, and then setting $w = w_1 + w_2$.
\end{defn}
\begin{defn}[Naming conventions for series-parallel compositions]
In the case of $G_1 \circ G_2$, we let $\epsilon$ denote the node $\tau_1 = \sigma_2$, $\sigma = \sigma_1$ and $\tau = \tau_2$. In the case of $G_1 \parallelsum G_2$, $\sigma = \sigma_1 = \sigma_2$ and $\tau = \tau_1 = \tau_2$.
\end{defn}
Let $(G,w)$ be a node-weighted series-parallel graph. The next several propositions will show that we can compute $X( G,w)$ by a dynamic programming algorithm on a binary $SP$-tree of $G$. Specifically, we will show how to compute $X( ( G_1, w_1) \circ (G_2, w_2) )$ and $ X (( G_1, w_1) \parallelsum (G_2, w_2) ) $ from $X( G_1, w_1)$ and $X(G_2, w_2)$, using a algorithms that are polynomial time in $G$ and pseudopolynomial time in the total weights. To prove these algorithms to be correct, we will compare the dynamic calculation of $X(G,w)$ with a dynamic enumeration of the partitions in question, using the following definition:
\begin{defn}[The enumeration version of $X(G,w)$]\begin{align*}
\widetilde{X}(G) = \{ (V_1, V_2, V_3) \in \mathcal{P}(V(G))^3 &: (V_1, V_2, V_3) \text{ form a partition of } V(G) \text{ such that }\\
& \text{(i) $G[V_i]$ is connected for $i = 1,2,3$}\\
& \text{(ii) $\sigma \in V_1$} \\
& \text{(iii) $V_3 = \emptyset$ implies that $\tau(G) \in V_1$ and $V_3 \not = \emptyset$ implies that $\tau(G) \in V_3$. } \}
\end{align*}
\end{defn}

In the proof of correctness for the dynamic program for evaluating $X( (G, w))$, we will explain how to compute $\tilde{X}(G_1 \circ G_2)$ from $\tilde{X}(G_1)$ and $\tilde{X}(G_2)$ by merging the blocks sharing the node $\epsilon$, and accepting the output when it results in an element of $\tilde{X}(G_1 \circ G_2)$. %

\begin{algorithm}[H]
\caption{\texttt{SeriesPartitions}}\label{alg:seriespartitions}
\textbf{Input:} $X( (G_1, w_1) )$ and $X( (G_2, w_2) )$ for $(G_1, w_1)$ and $(G_2, w_2)$.\\
\textbf{Output:} $X( (G_1,w_1) \circ (G_2, w_2) )$
\begin{algorithmic}[1]
\STATE{Set $f$ as constantly zero on $\{ (a_1, a_2, a_3) \in \mathbb{N}^3 : 0 \leq a_i , \sum a_i = w( (G_1, w_1) \circ (G_2, w_2))\} $.} \FOR{$(a_1, a_2, a_3, m_1) \in X( G_1, w_1)$ \textrm{ and }$(b_1, b_2, b_3, m_2) \in X(G_2, w_2)$ } 
\IF{($a_3 = 0$ and $b_3 = 0$) and ($a_2 = 0$ or $b_2 = 0$)} 
 \STATE{ $f(a_1 + b_1, a_2 + b_2, 0) \pluseq m_1 m_2$ }\ENDIF

\IF{($a_3 = 0$ and $b_3 \not = 0$) and ($a_2 = 0$ or $b_2 = 0$)} \STATE{
$f(a_1 + b_1, a_2 + b_2, b_3) \pluseq m_1 m_2$ }
\ENDIF 

\IF{($a_3 \not = 0$ and $b_3 = 0$) and ($a_2 = 0$ or $b_2 = 0$) } \STATE{ $f(a_1, a_2 + b_2, a_3 + b_1) \pluseq m_1 m_2$ 
} \ENDIF 
\IF{$a_3 \not = 0$ and $b_3 \not = 0$ and ($a_2 = 0$ and $b_2 = 0$)} \STATE{ $f(a_1, a_3 + b_1,b_3) \pluseq m_1 m_2$ }
\ENDIF   \ENDFOR
\STATE{Return $f$}
\end{algorithmic}
\end{algorithm}
\begin{figure}[H]
    \centering
    
    \begin{tabular}{cc}
    \def\svgscale{.28}{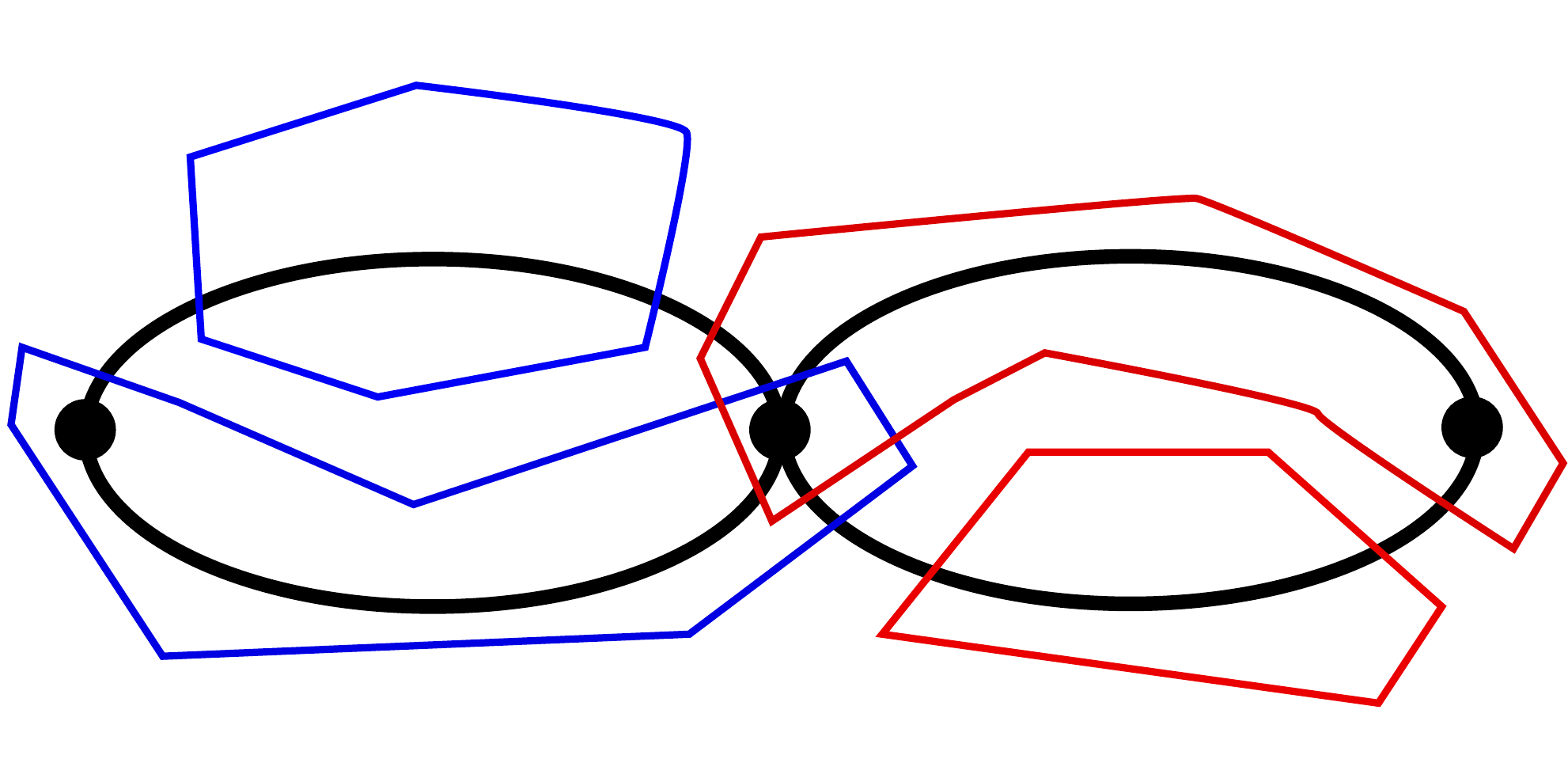} &
        \def\svgscale{.28}{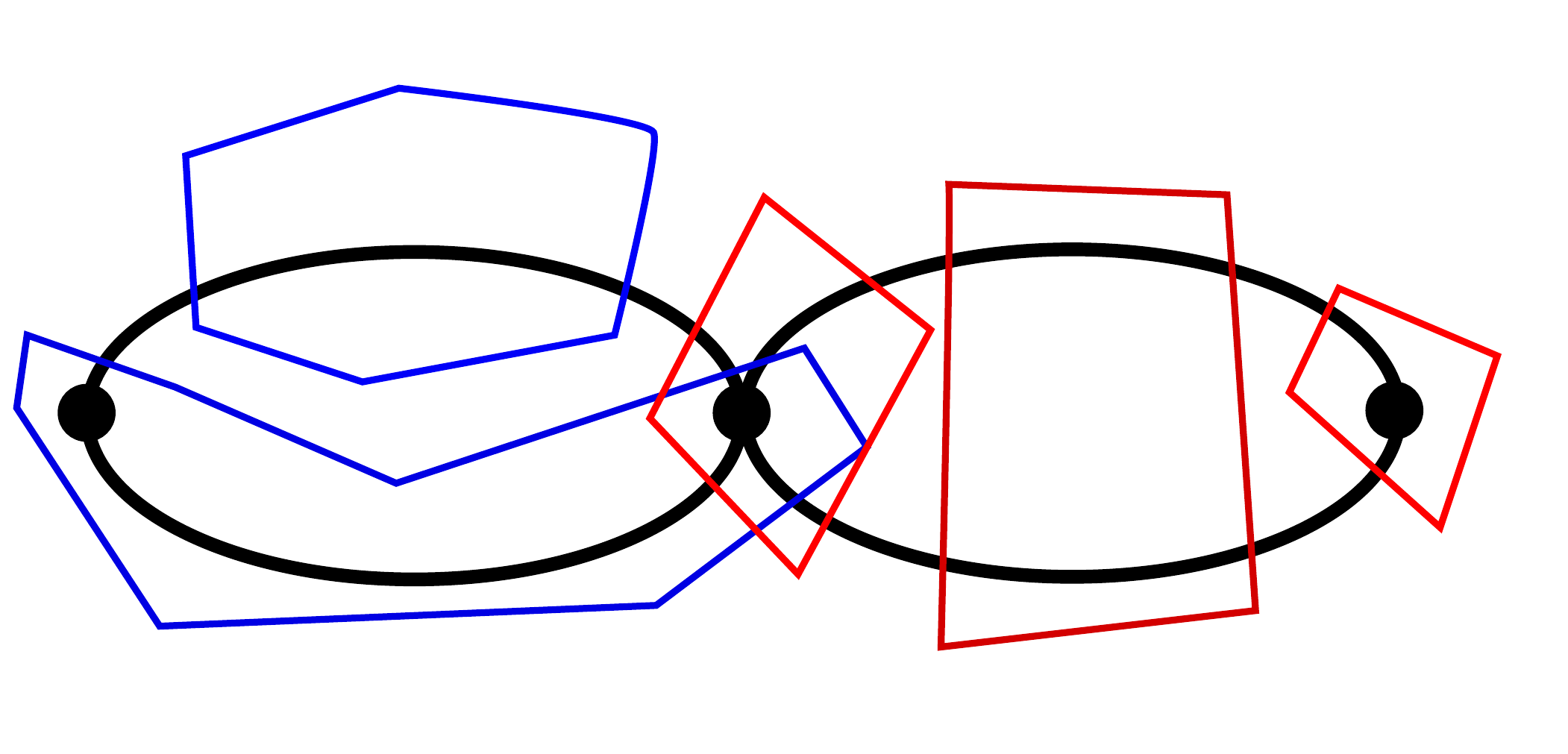} \\     \def\svgscale{.28}{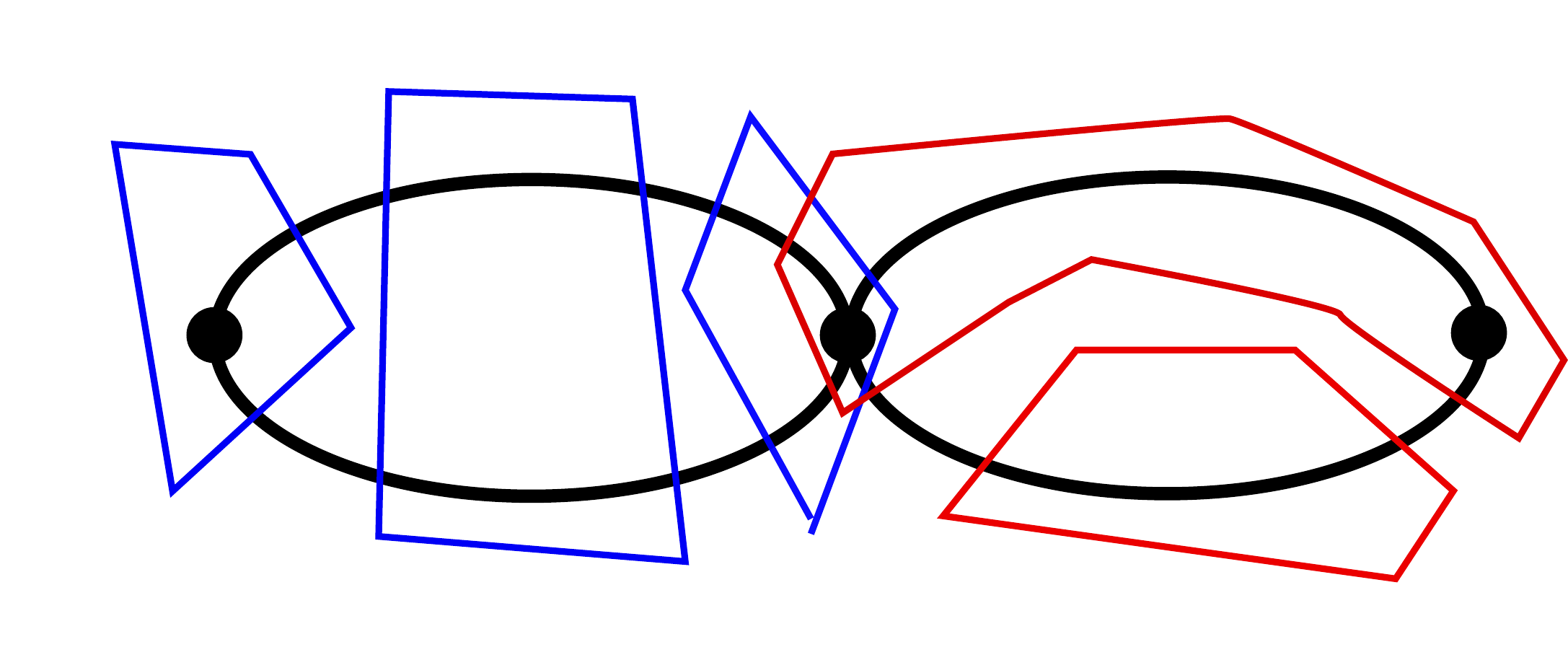} &     \def\svgscale{.28}{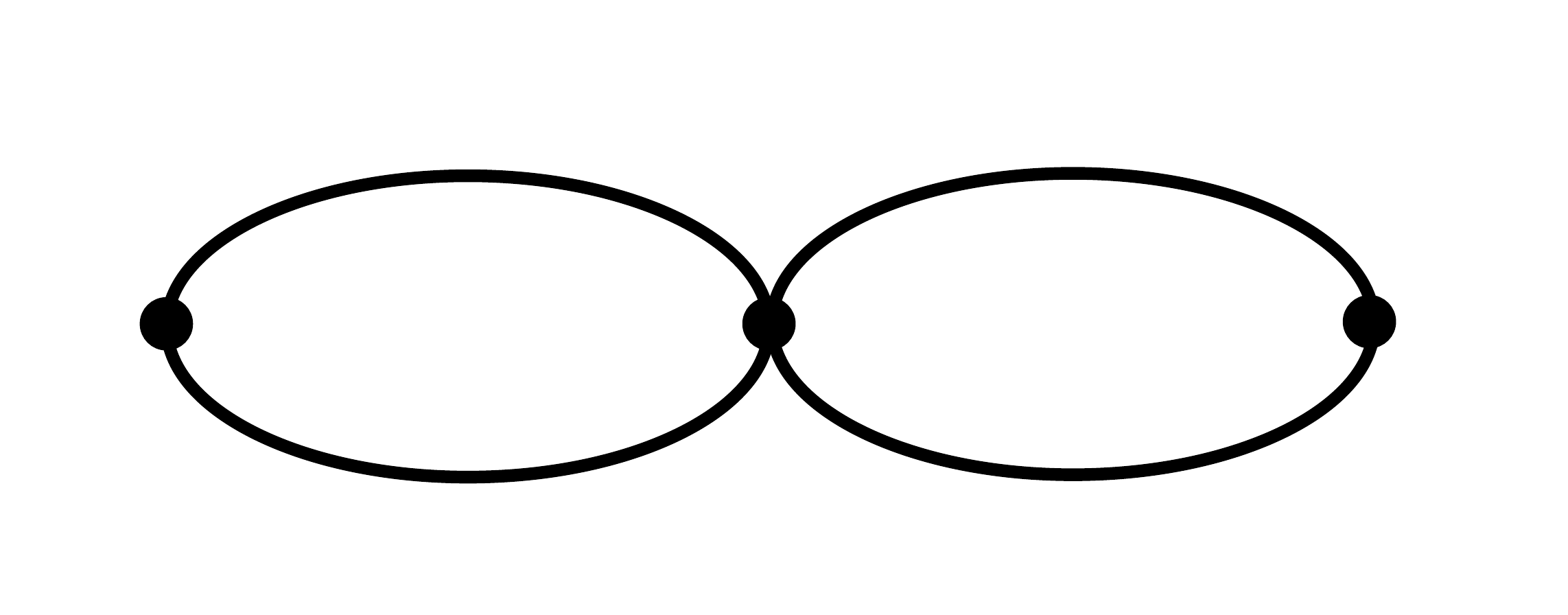} 
        \\
    (a) Case 1 & (b) Case 2 \\
    (c) Case 3 & (d) Case 4
    \end{tabular}
    
    \caption{The four cases in \Cref{alg:seriespartitions}}\label{fig:seriespartition}

\end{figure}

\begin{prop}\label{prop:verifyseries}
If $(G,w)$ is admissible, then \Cref{alg:seriespartitions} runs correctly and in polynomial time in $(|G|, n(w(G)) )$.\footnote{This is pseudopolynomial in the total weight.} %
\end{prop}
\begin{proof}

We need to check that each element of $P_2(G)$ is counted exactly once. To verify this, we explain another algorithm, \Cref{alg:setlevelseriespartitions}, that computes $\widetilde{X}(G_1 \circ G_2)$ from $\widetilde{X}(G_1)$ and $\widetilde{X}(G_2)$, but in exponential time. We will verify that in the course of this algorithm each element of $\widetilde{X}(G_1 \circ G_2)$ is computed exactly once. Finally, we will explain that the correctness of \Cref{alg:seriespartitions} can be seen by coupling it with a accelerated version of \Cref{alg:setlevelseriespartitions}, and we compute the time it takes for \Cref{alg:seriespartitions} to run.

\begin{algorithm}[H]
\caption{\texttt{SetLevelSeriesPartitions}}\label{alg:setlevelseriespartitions}
\textbf{Input:} series-parallel graphs $G_1$ and $G_2$ along with sets $\widetilde{X}( G_1 )$ and $\widetilde{X}( G_2 )$ \\
\textbf{Output:} $\widetilde{X}( G_1 \circ G_2 )$ 
\begin{algorithmic}[1]
\STATE{Initialize $F$ as the zero function on $\mathcal{P}(V(G))^3$, where $\mathcal{P}$ denotes the powerset.}
\FOR{$(X_1, X_2, X_3) \in \widetilde{X}( G_1)$ \textrm{ and }$(Y_1, Y_2, Y_3) \in \widetilde{X}(G_2)$ } 

\IF{$X_3 = \emptyset $ AND $Y_3 = \emptyset $ AND ($X_2 = \emptyset$ OR $Y_2 = \emptyset$)} \STATE{
   $F(X_1 \cup Y_1, X_2 \cup Y_2, \emptyset) \pluseq 1$} \ENDIF
   
\IF{$X_3 = \emptyset $ AND $Y_3 \not = \emptyset $ AND ($X_2 = \emptyset$ OR $Y_2 = \emptyset$)}
   \STATE{ $F(X_1 \cup Y_1, X_2 \cup Y_2, Y_3) \pluseq 1$}
   \ENDIF
   
\IF{$X_3 \not = \emptyset $ AND $Y_3 = \emptyset $ AND ($X_2 = \emptyset$ OR $Y_2 = \emptyset$)}
   \STATE{ $F(X_1, X_2 \cup Y_2, X_3 \cup Y_1) \pluseq 1$}
   \ENDIF

\IF{$X_3 \not = \emptyset $ AND $Y_3 \not = \emptyset $ AND $X_2 = \emptyset$ AND $Y_2 = \emptyset$} \STATE{ $F(X_1 , X_3 \cup Y_1, Y_3) \pluseq 1$} \ENDIF
\ENDFOR
\STATE{Return $F$}
\end{algorithmic}
   
\end{algorithm}
Our first goal is to verify that the function $F$ returned by \Cref{alg:setlevelseriespartitions} is the indicator of $\tilde{X}(G)$ in $\mathcal{P}(V(G))^3$.
First, it is straightforward to check that $\supp(F) \subseteq \tilde{X}(G)$, by examining each case. Let $(Z_1, Z_2, Z_3) \in \tilde{X}(G)$. %
There are four cases based on how $Z$ separates $(\sigma, \epsilon, \tau)$. The cases are:

\begin{enumerate}
    \item There is a block containing $\sigma$, $\epsilon$ and $\tau$.
    \item There is a block containing $\sigma$ and $\epsilon$, and a distinct block containing $\tau$.
    \item There is a block containing $\sigma$, and a distinct block containing $\epsilon$ and $\tau$.
    \item $\sigma, \epsilon$ and $\tau$ are each in a different block.
\end{enumerate}

This accounts for $4$ of the $5$ equivalence relations on $\{\sigma, \epsilon, \tau \}$. The missing equivalence relation on $\{\sigma, \epsilon, \tau\}$ would put $\sigma$ and $\tau$ in a block, and $\epsilon$ in a different block. This case cannot occur due to the requirement that the blocks are connected. One can now go through the four cases, and observe that for each $Z = (Z_1,Z_2,Z_3)$,   $Z$ can be produced in exactly one of the four cases in the algorithm, because each step is distinguished by how the partitions they produce separate $\{\sigma, \epsilon, \tau\}$. Moreover, because within each case the $X_i$ and $Y_i$ can be recovered by taking appropriate intersections of the $Z_i$ with $G_1$ and $G_2$, there is a unique pair of partitions of $G_1$ and $G_2$ that produce each $(Z_1, Z_2, Z_3)$. Thus, $F( Z_1, Z_2, Z_3) = 1$. We now go through the cases in more detail:

\begin{itemize}
    \item The case $X_3 = \emptyset$ and $Y_3 = \emptyset$, equivalently, $\sigma, \tau \in Z_1$. See \Cref{fig:seriespartition}(a).

    Recovery: $X_1 = Z_1 \cap G_1$, $X_2 = Z_2 \cap G_1$, $X_3 = \emptyset$, and $Y_1 = Z_1 \cap G_2$, $Y_2 = Z_2 \cap G_2$, $Y_3 = \emptyset$.

    \item The case $X_3 = \emptyset$, $Y_3 \not = \emptyset$, equivalently, $\sigma, \epsilon \in Z_1$, $\tau \in Z_3$. See \Cref{fig:seriespartition}(b).
    
    Recovery:  $X_1 = Z_1 \cap G_1$, $X_2 = Z_2 \cap G_1$, $X_3 = \emptyset$, and $Y_1 = Z_1 \cap G_2$, $Y_2 = Z_2 \cap G_2$ $Y_3 = Z_3$.

    \item The case $X_3 \not = \emptyset$, $Y_3 = \emptyset$, equivalently, $\sigma \in Z_1$, $\epsilon, \tau \in Z_3$. See \Cref{fig:seriespartition}(c).
    
    Recovery:  $X_1 = Z_1 \cap G_1$, $X_2 = Z_2 \cap G_1$, $X_3 = Z_3 \cap G_1$, and $Y_1 = Z_3 \cap G_2$, $Y_2 = Z_2 \cap G_2$, $Y_3 = \emptyset$.
    
    \item The case $X_3 \not = \emptyset$, $Y_3 \not = \emptyset$, equivalently, $\sigma, \epsilon, \tau$ are each in a different block. See \Cref{fig:seriespartition}(d).
    
    Recovery: $X_1 = Z_1 \cap G_1$, $X_2 = \emptyset$, $X_3 = Z_2 \cap G_1$, and $Y_1 = Z_2 \cap G_2$, $Y_2 = \emptyset$, $Y_3 = Z_3 \cap G_2$.

\end{itemize}

Finally, we examine the relationship between \Cref{alg:setlevelseriespartitions} and \Cref{alg:seriespartitions}. In particular, we will couple them together by sorting the elements of $\tilde{X}((G_1,w_1)$ so that those with the same sequence of weights appear together, and the same for $\tilde{X} ( ( G_2, w_2))$. The weights of the two $3$-partitions we are merging uniquely determines which of the four cases \Cref{alg:setlevelseriespartitions} is in since, by construction of the $E$ coordinate of every node weight, a set is empty iff its weight is zero. The weights of the the two merging partitions also determine the weights of the resulting $3$-partition, say $(c_1, c_2, c_3)$. Thus, while the set level algorithm \Cref{alg:setlevelseriespartitions} putters through $\Theta(m_1m_2)$ set-operations and $m_1m_2$ updates to a function, the algorithm \Cref{alg:seriespartitions} computes $m_1m_2$ and adds that that to the value of $f$ over some efficiently computable tuple $(c_1, c_2,c_3)$. %

Finally, after having verified that \Cref{alg:seriespartitions} has the correct output, we remark that it consists of a single loop over the product of two sets, which has size bounded by $O( w(G_1)^2 w(G_2)^2)$, since the number of elements in $X(G)$ is in general bounded by $w(G)^2$. Within each loop, each $m_i \leq 2^{3|V(G)|}$, so the cost of multiplications and additions are polynomial in $|G|$. This concludes the proof.
\end{proof}

We next explain how to compute $X( (G_1, w_1) \parallelsum (G_2, w_2))$ from $X( (G_1,w_1))$ and $X((G_2, w_2))$. %

\begin{algorithm}[H]
\caption{\texttt{ParallelPartitions}}\label{alg:parallelpartition}

\textbf{Input:} $X( (G_1, w_1) )$ and $X( (G_2, w_2) )$ for $(G_1, w_1)$ and $(G_2, w_2)$\\
\textbf{Output:} $X( (G_1,w_1) \parallelsum (G_2, w_2) )$.
\begin{algorithmic}
\STATE{Set $f$ to be constantly zero on $\{ (a_1, a_2, a_3) : 0 \leq a_i , \sum a_i = w( (G_1,w_1) \parallelsum (G_2, w_2) ) \}$}

\FOR{$(a_1, a_2, a_3, m_1) \in X( G_1, w_1)$ \textrm{ and }$(b_1, b_2, b_3, m_2) \in X(G_2, w_2)$ }

\IF{$a_3 = 0$ AND $b_3 = 0$ AND ($a_2 = 0$ OR $b_2 = 0$ )}
   \STATE{ $f(a_1 + b_1 , a_2 + b_2, 0) \pluseq m_1 m_2$ }
   \ENDIF
   
\IF{$a_3 = 0$ AND $b_3 \not = 0$ AND ($a_2 = 0$ OR $b_2 = 0$)} \STATE{ $f(a_1 + b_1 + b_3, a_2 + b_2, 0) \pluseq m_1 m_2$ } \ENDIF

\IF{$a_3 \not = 0$ AND $b_3 = 0$ AND ($a_2 = 0$ OR $b_2 = 0$)}
   \STATE{ $f(a_1 + b_1 + a_3, a_2 + b_2, 0) \pluseq m_1 m_2$ } \ENDIF

\IF{$a_3 \not = 0$ AND $b_3 \not = 0$ AND ($a_2 = 0$ OR $b_2 = 0$)}
   \STATE{ $f(a_1 + b_1, a_2 + b_2, a_3 + b_3) \pluseq m_1 m_2$ 
} \ENDIF

\ENDFOR

\STATE{Return $f$}
\end{algorithmic}
\end{algorithm}

\begin{figure}[H] \centering \begin{tabular}{cc}

\def\svgscale{.4}{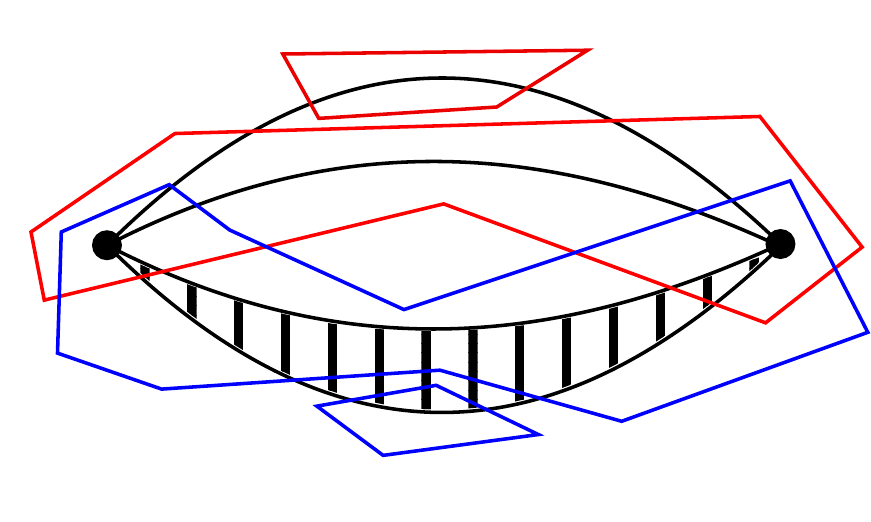} &\def\svgscale{.4}{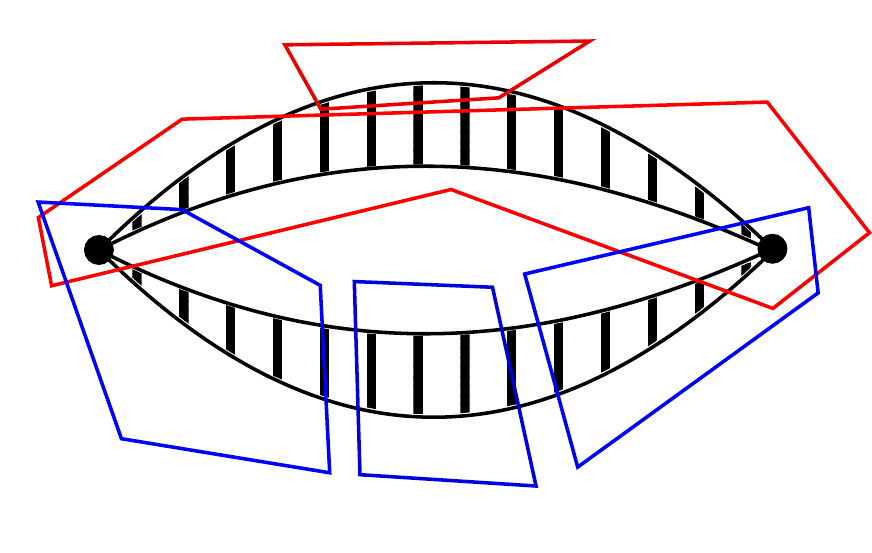} \\\def\svgscale{.4}{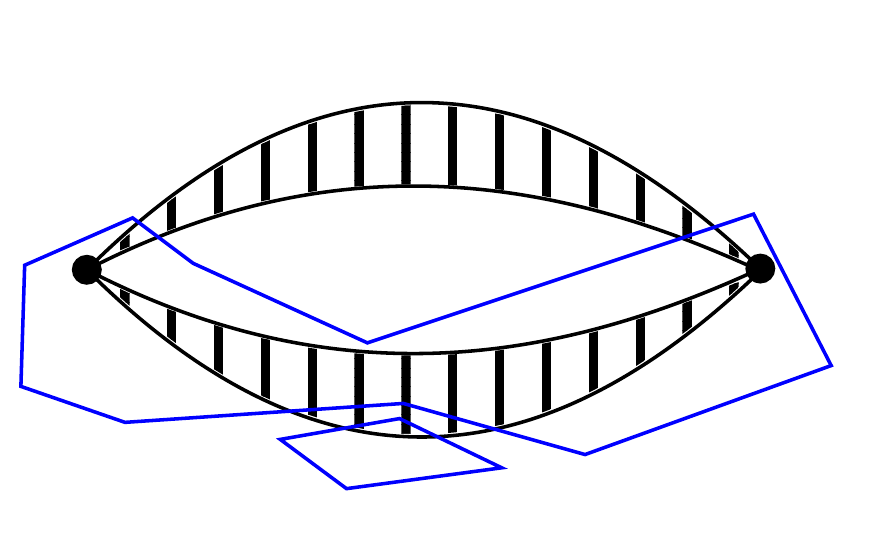} &    \def\svgscale{.4}{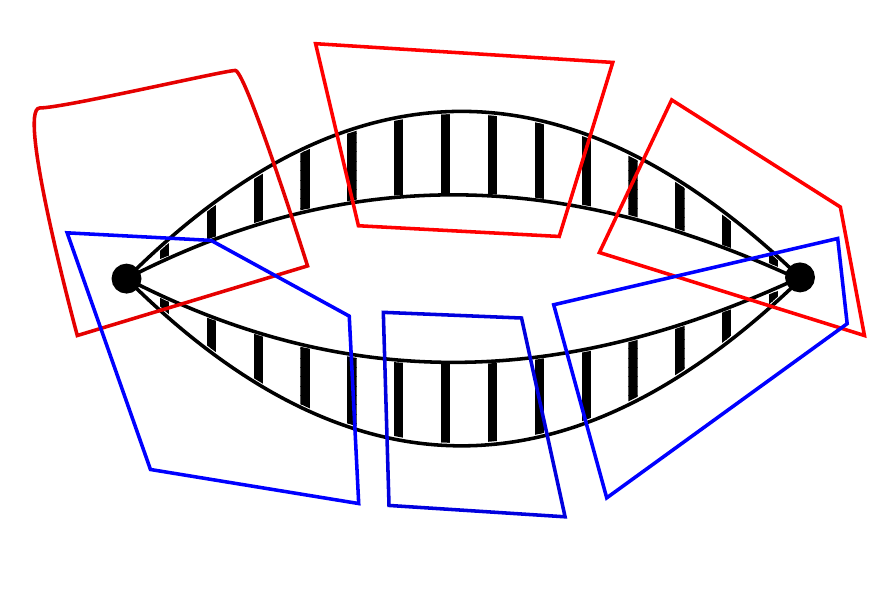} \\(a) Case 1 & (b) Case 2 \\(c) Case 3 & (d) Case 4\end{tabular}\caption{The four cases in \Cref{alg:parallelpartition}}\label{fig:parallelpartition}\end{figure}

\begin{prop}\label{prop:verifyparallel}
If $(G,w)$ is admissible, then \Cref{alg:parallelpartition} runs correctly and in time polynomial in $(|G|, n(w(G)))$.
\end{prop}
\begin{proof}
This proof follows the same structure as \Cref{prop:verifyseries}.
\begin{algorithm}
\caption{\texttt{SetLevelParallelPartitions}}\label{alg:setlevelparallelpartitions}
\textbf{Input:} series-parallel graphs $G_1$ and $G_2$ along with $\widetilde{X}( G_1 )$ and $\widetilde{X}( G_2 )$ \\
\textbf{Output:} $\widetilde{X}( G_1 \parallelsum G_2 )$. (Convention to align with pictures $G_1$ will be the top SP-graph, whose partitions are denoted with $X_i$.)

\begin{algorithmic}[1]
\STATE{Initialize $F$ as the zero function on $\mathcal{P}(V(G))^3$}

\FOR{$(X_1, X_2, X_3) \in \widetilde{X}( G_1)$ \textrm{ and }$(Y_1, Y_2, Y_3) \in \widetilde{X}(G_2)$ }

\IF{$X_3 = \emptyset $ and $Y_3 = \emptyset $ AND ($X_2 = \emptyset$ or $Y_2 = \emptyset$) }
   \STATE{ $F(X_1 \cup Y_1, X_2 \cup Y_2, \emptyset) \pluseq 1$}
   \ENDIF
   
\IF{$X_3 = \emptyset $ and $Y_3 \not = \emptyset $ AND ($X_2 = \emptyset$ or $Y_2 = \emptyset$)}
   \STATE{ $F(X_1 \cup X_3 \cup Y_1, X_2 \cup Y_2, \emptyset) \pluseq 1$} \ENDIF
   
\IF{$X_3 \not = \emptyset $ and $Y_3 = \emptyset $ AND ($X_2 = \emptyset$ or $Y_2 = \emptyset$)}
    \STATE{ $F(X_1 \cup Y_1 \cup Y_3, X_2 \cup Y_2, \emptyset) \pluseq 1$} \ENDIF
   
\IF{$X_3 \not = \emptyset $ and $Y_3 \not = \emptyset $ AND ($X_2 = \emptyset$ or $Y_2 = \emptyset$) }
   \STATE{ $F(X_1 \cup Y_1 , X_2 \cup Y_2, X_3 \cup Y_3) \pluseq 1$}
   \ENDIF

\ENDFOR
\STATE{Return $F$}
\end{algorithmic}
\end{algorithm}
We verify that the function $F$ returned by \Cref{alg:setlevelparallelpartitions} is the indicator of $\tilde{X}(G)$ in $\mathcal{P}(V(G))^3$.
First, it is straightforward to check that $\supp(F) \subseteq \tilde{X}(G)$. Next, let $(Z_1, Z_2, Z_3) \in \tilde{X}(G)$. There are four cases based on how $Z$ connects $\sigma$ and $\tau$. %

\begin{enumerate}
    \item $Z_1$ has a path through $G_1$ and $G_2$ from $\sigma$ to $\tau$.
    \item $Z_1$ has a path through only $G_1$ from $\sigma$ to $\tau$.
    \item $Z_1$ has a path through only $G_2$ from $\sigma$ to $\tau$.
    \item No paths, that is: $\sigma \in Z_1$ and $\tau \in Z_3$.
\end{enumerate}
One can now observe that for each $Z = (Z_1,Z_2,Z_3)$,   $Z$ can be produced in exactly one of the four cases in \Cref{alg:setlevelparallelpartitions}, because each case is distinguished what kind of paths in $Z_1$ there are from $\sigma$ to $\tau$. Moreover, because the $X_i$ and $Y_i$ can be recovered by taking appropriate intersections of the $Z_i$ with $G_1$ and $G_2$, there is a unique pair of partitions of $G_1$ and $G_2$ that produce each $(Z_1, Z_2, Z_3)$. Thus, $F( Z_1, Z_2, Z_3) = 1$. We now list the cases of the algorithm in more detail:
\begin{itemize}
    \item The case $X_3 = \emptyset $ and $Y_3 = \emptyset $, equivalently, $Z_1$ has a path through $G_1$ and $G_2$ from $\sigma$ to $\tau$. See \Cref{fig:parallelpartition}(a).
    Recovery: $X_1 = Z_1 \cap G_1$, $X_2 = Z_2 \cap G_1$, $X_3 = \emptyset$, and $Y_1 = Z_1 \cap G_2$, $Y_2 = Z_2 \cap G_2$, $Y_3 = \emptyset$.
    \item The case $X_3 = \emptyset $ and $Y_3 \not = \emptyset $, equivalently, $Z_1$ has a path through only $G_1$ from $\sigma$ to $\tau$. See \Cref{fig:parallelpartition}(b).
    Recovery:  $X_1 = Z_1 \cap G_1$, $X_2 = Z_2 \cap G_1$, $X_3 = \emptyset$, and $Y_1$ the component of $Z_1 \cap G_2$ containing $\sigma$, $Y_2 = Z_2 \cap G_2$ and $Y_3$ is the component of $Z_1 \cap G_2$ containing $\tau$.
    \item The case $X_3 \not = \emptyset $ and $Y_3 = \emptyset $, equivalently, $Z_1$ has a path through only $G_2$ from $\sigma$ to $\tau$. See \Cref{fig:parallelpartition}(c).
    Recovery:  $X_1$ is the component of $Z_1 \cap G_1$ containing $\sigma$, $X_2 = Z_2 \cap G_1$, $X_3$ is the component of $Z_3 \cap G_1$ containing $\tau$, and $Y_1 = Z_1 \cap G_2$, $Y_2 = Z_2 \cap G_2$, $Y_3 = \emptyset$.
    \item The case $X_3 \not = \emptyset $ and $Y_3 \not = \emptyset $, equivalently, $\sigma \in Z_1$ and $\tau \in Z_3$. See \Cref{fig:parallelpartition}(d).
    Recovery: $X_1 = Z_1 \cap G_1$, $X_2 = Z_2 \cap G_1$, $X_3 = Z_3 \cap G_1$, and $Y_1 = Z_1 \cap G_2$, $Y_2 = Z_2 \cap G_2$ and $Y_3 = Z_3 \cap G_2$.
\end{itemize}
The relationship between  \Cref{alg:setlevelparallelpartitions} and \Cref{alg:parallelpartition} is the same as in \Cref{prop:verifyseries}, proving that \Cref{alg:parallelpartition} has the correct output. Moreover, it consists of a single loop over the product of two sets, which has size bounded by $O( w(G_1)^2 w(G_2)^2)$, since the number of elements in $X(G)$ is in general bounded by $w(G)^2$. Within each loop, each $m_i \leq 2^{3|V(G)|}$, so the cost of multiplications and additions are polynomial in $|G|$.
\end{proof}

\begin{thm}\label{thm:calculatingbalancedpartitions}
Let $(G,w) $ be a node $\mathbb{N}$-weighted series-parallel graph. Then $|P_2^0(G,w)|$ can be calculated in time polynomial in $(|G|, w(G))$.\footnote{The input to the polynomial is the size of $w(G)$, not the binary encoding of $w(G)$.}

\end{thm}
\begin{proof}

We extend $w$ to weights valued in $N =  E \times \mathbb{N}$, by setting $w'(a) = (\lnot \emptyset, w(a))$ for all $a \in V(G)$. Thus,$(G,w')$ is admissible. %
We let $T$ be a binary $SP$-tree for $G$. This is a binary tree with $|E(G)|$ leaves, so it has $O(|E(G)|)$ nodes. %
Each node of $T$ is associated with a subgraph of $G$, and we make them into node-weighted series-parallel graphs by setting the $E$ component to be $\lnot \emptyset$ on all nodes, and by setting the $\mathbb{N}$ component of the weight function in any way that adds up correctly using \Cref{defn:SPweightscomposition}; for example, if $H$ is a node in $T$, with left child $L$ and right child $R$, we can assign the $\mathbb{N}$ part of the weight on $L$ to be the restriction of $w$ to $L$, and on $R$ to be the restriction of $w$ to $R \setminus L$, and zero elsewhere. The resulting node-weighted graphs are all admissible by construction.

Moreover, the graph at each $P$ node is the node-weighted parallel composition of the graphs at each child node, and the graph at each $S$ node is the node-weighted series-composition of the graphs at each child node.
Computing $X( (H,w'))$ at each of the leaves takes time $O(N(w(G)))$. Computing the value of $X( (H',w'))$ for each $P$ or $S$ node, given the values at the children, takes time $O(p(|G|, w(G)))$ for some polynomial fixed $p$ (given by \Cref{prop:verifyseries} and \Cref{prop:verifyparallel}). Thus, the total time to compute $X( (H',w'))$ at each node of the tree by memoization is $O( |E(G)| p( |G|, w(G)))$.

From $X( (G,w'))$ we can calculate $|P_2^0(G,w)|$ as $$|P_2^0(G,w)| = X(G,w')( (\frac{w(V)}{2}, \lnot \emptyset) , (\frac{w(V)}{2}, \lnot \emptyset), 0) +  X(G,w')((\frac{w(V)}{2}, \lnot \emptyset), 0, (\frac{w(V)}{2}, \lnot \emptyset)) $$ %

\end{proof}

%% file: Images/SeriesCase1.pdf_tex
\begingroup%
  \makeatletter%
  \providecommand\color[2][]{%
    \errmessage{(Inkscape) Color is used for the text in Inkscape, but the package 'color.sty' is not loaded}%
    \renewcommand\color[2][]{}%
  }%
  \providecommand\transparent[1]{%
    \errmessage{(Inkscape) Transparency is used (non-zero) for the text in Inkscape, but the package 'transparent.sty' is not loaded}%
    \renewcommand\transparent[1]{}%
  }%
  \providecommand\rotatebox[2]{#2}%
  \newcommand*\fsize{\dimexpr\f@size pt\relax}%
  \newcommand*\lineheight[1]{\fontsize{\fsize}{#1\fsize}\selectfont}%
  \ifx\svgwidth\undefined%
    \setlength{\unitlength}{570.61906499bp}%
    \ifx\svgscale\undefined%
      \relax%
    \else%
      \setlength{\unitlength}{\unitlength * \real{\svgscale}}%
    \fi%
  \else%
    \setlength{\unitlength}{\svgwidth}%
  \fi%
  \global\let\svgwidth\undefined%
  \global\let\svgscale\undefined%
  \makeatother%
  \begin{picture}(1,0.4957712)%
    \lineheight{1}%
    \setlength\tabcolsep{0pt}%
    \put(0,0){\includegraphics[width=\unitlength,page=1]{SeriesCase1.pdf}}%
    \put(-0.00408789,0.01513201){\color[rgb]{0,0,0}\makebox(0,0)[lt]{\lineheight{1.25}\smash{\begin{tabular}[t]{l}$X_1$\end{tabular}}}}%
    \put(0.09088847,0.45835502){\color[rgb]{0,0,0}\makebox(0,0)[lt]{\lineheight{1.25}\smash{\begin{tabular}[t]{l}$X_2$\end{tabular}}}}%
    \put(0.58870363,0.43021386){\color[rgb]{0,0,0}\makebox(0,0)[lt]{\lineheight{1.25}\smash{\begin{tabular}[t]{l}$Y_1$\end{tabular}}}}%
    \put(0.62387999,0.0116144){\color[rgb]{0,0,0}\makebox(0,0)[lt]{\lineheight{1.25}\smash{\begin{tabular}[t]{l}$Y_2$\end{tabular}}}}%
  \end{picture}%
\endgroup%

%% file: Images/SeriesCase2.pdf_tex
\begingroup%
  \makeatletter%
  \providecommand\color[2][]{%
    \errmessage{(Inkscape) Color is used for the text in Inkscape, but the package 'color.sty' is not loaded}%
    \renewcommand\color[2][]{}%
  }%
  \providecommand\transparent[1]{%
    \errmessage{(Inkscape) Transparency is used (non-zero) for the text in Inkscape, but the package 'transparent.sty' is not loaded}%
    \renewcommand\transparent[1]{}%
  }%
  \providecommand\rotatebox[2]{#2}%
  \newcommand*\fsize{\dimexpr\f@size pt\relax}%
  \newcommand*\lineheight[1]{\fontsize{\fsize}{#1\fsize}\selectfont}%
  \ifx\svgwidth\undefined%
    \setlength{\unitlength}{605.15948673bp}%
    \ifx\svgscale\undefined%
      \relax%
    \else%
      \setlength{\unitlength}{\unitlength * \real{\svgscale}}%
    \fi%
  \else%
    \setlength{\unitlength}{\svgwidth}%
  \fi%
  \global\let\svgwidth\undefined%
  \global\let\svgscale\undefined%
  \makeatother%
  \begin{picture}(1,0.48074178)%
    \lineheight{1}%
    \setlength\tabcolsep{0pt}%
    \put(0,0){\includegraphics[width=\unitlength,page=1]{SeriesCase2.pdf}}%
    \put(-0.00385456,0.02090205){\color[rgb]{0,0,0}\makebox(0,0)[lt]{\lineheight{1.25}\smash{\begin{tabular}[t]{l}$X_1$\end{tabular}}}}%
    \put(0.0923346,0.44546119){\color[rgb]{0,0,0}\makebox(0,0)[lt]{\lineheight{1.25}\smash{\begin{tabular}[t]{l}$X_2$\end{tabular}}}}%
    \put(0.86357114,0.0806057){\color[rgb]{0,0,0}\makebox(0,0)[lt]{\lineheight{1.25}\smash{\begin{tabular}[t]{l}$Y_3$\end{tabular}}}}%
    \put(0.42906144,0.38575754){\color[rgb]{0,0,0}\makebox(0,0)[lt]{\lineheight{1.25}\smash{\begin{tabular}[t]{l}$Y_1$\end{tabular}}}}%
    \put(0.59158798,0.01095149){\color[rgb]{0,0,0}\makebox(0,0)[lt]{\lineheight{1.25}\smash{\begin{tabular}[t]{l}$Y_2$\end{tabular}}}}%
  \end{picture}%
\endgroup%

%% file: Images/SeriesCase3.pdf_tex
\begingroup%
  \makeatletter%
  \providecommand\color[2][]{%
    \errmessage{(Inkscape) Color is used for the text in Inkscape, but the package 'color.sty' is not loaded}%
    \renewcommand\color[2][]{}%
  }%
  \providecommand\transparent[1]{%
    \errmessage{(Inkscape) Transparency is used (non-zero) for the text in Inkscape, but the package 'transparent.sty' is not loaded}%
    \renewcommand\transparent[1]{}%
  }%
  \providecommand\rotatebox[2]{#2}%
  \newcommand*\fsize{\dimexpr\f@size pt\relax}%
  \newcommand*\lineheight[1]{\fontsize{\fsize}{#1\fsize}\selectfont}%
  \ifx\svgwidth\undefined%
    \setlength{\unitlength}{626.09477361bp}%
    \ifx\svgscale\undefined%
      \relax%
    \else%
      \setlength{\unitlength}{\unitlength * \real{\svgscale}}%
    \fi%
  \else%
    \setlength{\unitlength}{\svgwidth}%
  \fi%
  \global\let\svgwidth\undefined%
  \global\let\svgscale\undefined%
  \makeatother%
  \begin{picture}(1,0.412405)%
    \lineheight{1}%
    \setlength\tabcolsep{0pt}%
    \put(0,0){\includegraphics[width=\unitlength,page=1]{SeriesCase3.pdf}}%
    \put(0.22069149,0.37830412){\color[rgb]{0,0,0}\makebox(0,0)[lt]{\lineheight{1.25}\smash{\begin{tabular}[t]{l}$X_2$\end{tabular}}}}%
    \put(-0.00372567,0.03815169){\color[rgb]{0,0,0}\makebox(0,0)[lt]{\lineheight{1.25}\smash{\begin{tabular}[t]{l}$X_1$\end{tabular}}}}%
    \put(0.43264209,0.01091089){\color[rgb]{0,0,0}\makebox(0,0)[lt]{\lineheight{1.25}\smash{\begin{tabular}[t]{l}$X_3$\end{tabular}}}}%
    \put(0.67617881,0.35490561){\color[rgb]{0,0,0}\makebox(0,0)[lt]{\lineheight{1.25}\smash{\begin{tabular}[t]{l}$Y_1$\end{tabular}}}}%
    \put(0.68451471,0.0105853){\color[rgb]{0,0,0}\makebox(0,0)[lt]{\lineheight{1.25}\smash{\begin{tabular}[t]{l}$Y_2$\end{tabular}}}}%
  \end{picture}%
\endgroup%

%% file: Images/SeriesCase4.pdf_tex
\begingroup%
  \makeatletter%
  \providecommand\color[2][]{%
    \errmessage{(Inkscape) Color is used for the text in Inkscape, but the package 'color.sty' is not loaded}%
    \renewcommand\color[2][]{}%
  }%
  \providecommand\transparent[1]{%
    \errmessage{(Inkscape) Transparency is used (non-zero) for the text in Inkscape, but the package 'transparent.sty' is not loaded}%
    \renewcommand\transparent[1]{}%
  }%
  \providecommand\rotatebox[2]{#2}%
  \newcommand*\fsize{\dimexpr\f@size pt\relax}%
  \newcommand*\lineheight[1]{\fontsize{\fsize}{#1\fsize}\selectfont}%
  \ifx\svgwidth\undefined%
    \setlength{\unitlength}{657.89735916bp}%
    \ifx\svgscale\undefined%
      \relax%
    \else%
      \setlength{\unitlength}{\unitlength * \real{\svgscale}}%
    \fi%
  \else%
    \setlength{\unitlength}{\svgwidth}%
  \fi%
  \global\let\svgwidth\undefined%
  \global\let\svgscale\undefined%
  \makeatother%
  \begin{picture}(1,0.39215966)%
    \lineheight{1}%
    \setlength\tabcolsep{0pt}%
    \put(0,0){\includegraphics[width=\unitlength,page=1]{SeriesCase4.pdf}}%
    \put(0.21002325,0.35970721){\color[rgb]{0,0,0}\makebox(0,0)[lt]{\lineheight{1.25}\smash{\begin{tabular}[t]{l}$X_2$\end{tabular}}}}%
    \put(-0.00354558,0.03599779){\color[rgb]{0,0,0}\makebox(0,0)[lt]{\lineheight{1.25}\smash{\begin{tabular}[t]{l}$X_1$\end{tabular}}}}%
    \put(0.48156085,0.33499431){\color[rgb]{0,0,0}\makebox(0,0)[lt]{\lineheight{1.25}\smash{\begin{tabular}[t]{l}$Y_1$\end{tabular}}}}%
    \put(0.64326298,0.01158988){\color[rgb]{0,0,0}\makebox(0,0)[lt]{\lineheight{1.25}\smash{\begin{tabular}[t]{l}$Y_2$\end{tabular}}}}%
    \put(0,0){\includegraphics[width=\unitlength,page=2]{SeriesCase4.pdf}}%
    \put(0.87450745,0.07261884){\color[rgb]{0,0,0}\makebox(0,0)[lt]{\lineheight{1.25}\smash{\begin{tabular}[t]{l}$Y_3$\end{tabular}}}}%
    \put(0,0){\includegraphics[width=\unitlength,page=3]{SeriesCase4.pdf}}%
    \put(0.41172831,0.01007361){\color[rgb]{0,0,0}\makebox(0,0)[lt]{\lineheight{1.25}\smash{\begin{tabular}[t]{l}$X_3$\end{tabular}}}}%
  \end{picture}%
\endgroup%

%% file: Images/ParallelCase1.pdf_tex
\begingroup%
  \makeatletter%
  \providecommand\color[2][]{%
    \errmessage{(Inkscape) Color is used for the text in Inkscape, but the package 'color.sty' is not loaded}%
    \renewcommand\color[2][]{}%
  }%
  \providecommand\transparent[1]{%
    \errmessage{(Inkscape) Transparency is used (non-zero) for the text in Inkscape, but the package 'transparent.sty' is not loaded}%
    \renewcommand\transparent[1]{}%
  }%
  \providecommand\rotatebox[2]{#2}%
  \newcommand*\fsize{\dimexpr\f@size pt\relax}%
  \newcommand*\lineheight[1]{\fontsize{\fsize}{#1\fsize}\selectfont}%
  \ifx\svgwidth\undefined%
    \setlength{\unitlength}{250.64613159bp}%
    \ifx\svgscale\undefined%
      \relax%
    \else%
      \setlength{\unitlength}{\unitlength * \real{\svgscale}}%
    \fi%
  \else%
    \setlength{\unitlength}{\svgwidth}%
  \fi%
  \global\let\svgwidth\undefined%
  \global\let\svgscale\undefined%
  \makeatother%
  \begin{picture}(1,0.59954218)%
    \lineheight{1}%
    \setlength\tabcolsep{0pt}%
    \put(0,0){\includegraphics[width=\unitlength,page=1]{ParallelCase1.pdf}}%
    \put(0.06309103,0.44774392){\color[rgb]{0,0,0}\makebox(0,0)[t]{\lineheight{1.25}\smash{\begin{tabular}[t]{c}$X_1$\end{tabular}}}}%
    \put(0.47719077,0.56681897){\color[rgb]{0,0,0}\makebox(0,0)[t]{\lineheight{1.25}\smash{\begin{tabular}[t]{c}$X_2$\end{tabular}}}}%
    \put(0.10012126,0.10245065){\color[rgb]{0,0,0}\makebox(0,0)[t]{\lineheight{1.25}\smash{\begin{tabular}[t]{c}$Y_1$\end{tabular}}}}%
    \put(0.47863252,0.01015766){\color[rgb]{0,0,0}\makebox(0,0)[t]{\lineheight{1.25}\smash{\begin{tabular}[t]{c}$Y_2$\end{tabular}}}}%
    \put(0,0){\includegraphics[width=\unitlength,page=2]{ParallelCase1.pdf}}%
  \end{picture}%
\endgroup%

%% file: Images/ParallelCase2.pdf_tex
\begingroup%
  \makeatletter%
  \providecommand\color[2][]{%
    \errmessage{(Inkscape) Color is used for the text in Inkscape, but the package 'color.sty' is not loaded}%
    \renewcommand\color[2][]{}%
  }%
  \providecommand\transparent[1]{%
    \errmessage{(Inkscape) Transparency is used (non-zero) for the text in Inkscape, but the package 'transparent.sty' is not loaded}%
    \renewcommand\transparent[1]{}%
  }%
  \providecommand\rotatebox[2]{#2}%
  \newcommand*\fsize{\dimexpr\f@size pt\relax}%
  \newcommand*\lineheight[1]{\fontsize{\fsize}{#1\fsize}\selectfont}%
  \ifx\svgwidth\undefined%
    \setlength{\unitlength}{251.12524115bp}%
    \ifx\svgscale\undefined%
      \relax%
    \else%
      \setlength{\unitlength}{\unitlength * \real{\svgscale}}%
    \fi%
  \else%
    \setlength{\unitlength}{\svgwidth}%
  \fi%
  \global\let\svgwidth\undefined%
  \global\let\svgscale\undefined%
  \makeatother%
  \begin{picture}(1,0.63383783)%
    \lineheight{1}%
    \setlength\tabcolsep{0pt}%
    \put(0,0){\includegraphics[width=\unitlength,page=1]{ParallelCase2.pdf}}%
    \put(0.10316837,0.48232934){\color[rgb]{0,0,0}\makebox(0,0)[t]{\lineheight{1.25}\smash{\begin{tabular}[t]{c}$X_1$\end{tabular}}}}%
    \put(0.51647811,0.60117704){\color[rgb]{0,0,0}\makebox(0,0)[t]{\lineheight{1.25}\smash{\begin{tabular}[t]{c}$X_2$\end{tabular}}}}%
    \put(0.06137541,0.09798851){\color[rgb]{0,0,0}\makebox(0,0)[t]{\lineheight{1.25}\smash{\begin{tabular}[t]{c}$Y_1$\end{tabular}}}}%
    \put(0.49462914,0.01013828){\color[rgb]{0,0,0}\makebox(0,0)[t]{\lineheight{1.25}\smash{\begin{tabular}[t]{c}$Y_2$\end{tabular}}}}%
    \put(0.85538364,0.13548623){\color[rgb]{0,0,0}\makebox(0,0)[t]{\lineheight{1.25}\smash{\begin{tabular}[t]{c}$Y_3$\end{tabular}}}}%
  \end{picture}%
\endgroup%

%% file: Images/ParallelCase3.pdf_tex
\begingroup%
  \makeatletter%
  \providecommand\color[2][]{%
    \errmessage{(Inkscape) Color is used for the text in Inkscape, but the package 'color.sty' is not loaded}%
    \renewcommand\color[2][]{}%
  }%
  \providecommand\transparent[1]{%
    \errmessage{(Inkscape) Transparency is used (non-zero) for the text in Inkscape, but the package 'transparent.sty' is not loaded}%
    \renewcommand\transparent[1]{}%
  }%
  \providecommand\rotatebox[2]{#2}%
  \newcommand*\fsize{\dimexpr\f@size pt\relax}%
  \newcommand*\lineheight[1]{\fontsize{\fsize}{#1\fsize}\selectfont}%
  \ifx\svgwidth\undefined%
    \setlength{\unitlength}{252.88430412bp}%
    \ifx\svgscale\undefined%
      \relax%
    \else%
      \setlength{\unitlength}{\unitlength * \real{\svgscale}}%
    \fi%
  \else%
    \setlength{\unitlength}{\svgwidth}%
  \fi%
  \global\let\svgwidth\undefined%
  \global\let\svgscale\undefined%
  \makeatother%
  \begin{picture}(1,0.62820626)%
    \lineheight{1}%
    \setlength\tabcolsep{0pt}%
    \put(0,0){\includegraphics[width=\unitlength,page=1]{ParallelCase3.pdf}}%
    \put(0.09065268,0.10154373){\color[rgb]{0,0,0}\makebox(0,0)[t]{\lineheight{1.25}\smash{\begin{tabular}[t]{c}$Y_1$\end{tabular}}}}%
    \put(0.46581565,0.01006775){\color[rgb]{0,0,0}\makebox(0,0)[t]{\lineheight{1.25}\smash{\begin{tabular}[t]{c}$Y_2$\end{tabular}}}}%
    \put(0,0){\includegraphics[width=\unitlength,page=2]{ParallelCase3.pdf}}%
    \put(0.06253264,0.54977769){\color[rgb]{0,0,0}\makebox(0,0)[t]{\lineheight{1.25}\smash{\begin{tabular}[t]{c}$X_1$\end{tabular}}}}%
    \put(0.51957412,0.59577266){\color[rgb]{0,0,0}\makebox(0,0)[t]{\lineheight{1.25}\smash{\begin{tabular}[t]{c}$X_2$\end{tabular}}}}%
    \put(0.9374882,0.46154711){\color[rgb]{0,0,0}\makebox(0,0)[t]{\lineheight{1.25}\smash{\begin{tabular}[t]{c}$X_3$\end{tabular}}}}%
  \end{picture}%
\endgroup%

%% file: Images/ParallelCase4.pdf_tex
\begingroup%
  \makeatletter%
  \providecommand\color[2][]{%
    \errmessage{(Inkscape) Color is used for the text in Inkscape, but the package 'color.sty' is not loaded}%
    \renewcommand\color[2][]{}%
  }%
  \providecommand\transparent[1]{%
    \errmessage{(Inkscape) Transparency is used (non-zero) for the text in Inkscape, but the package 'transparent.sty' is not loaded}%
    \renewcommand\transparent[1]{}%
  }%
  \providecommand\rotatebox[2]{#2}%
  \newcommand*\fsize{\dimexpr\f@size pt\relax}%
  \newcommand*\lineheight[1]{\fontsize{\fsize}{#1\fsize}\selectfont}%
  \ifx\svgwidth\undefined%
    \setlength{\unitlength}{252.88428249bp}%
    \ifx\svgscale\undefined%
      \relax%
    \else%
      \setlength{\unitlength}{\unitlength * \real{\svgscale}}%
    \fi%
  \else%
    \setlength{\unitlength}{\svgwidth}%
  \fi%
  \global\let\svgwidth\undefined%
  \global\let\svgscale\undefined%
  \makeatother%
  \begin{picture}(1,0.69237016)%
    \lineheight{1}%
    \setlength\tabcolsep{0pt}%
    \put(0,0){\includegraphics[width=\unitlength,page=1]{ParallelCase4.pdf}}%
    \put(0.09835275,0.10154374){\color[rgb]{0,0,0}\makebox(0,0)[t]{\lineheight{1.25}\smash{\begin{tabular}[t]{c}$Y_1$\end{tabular}}}}%
    \put(0.4735153,0.01006776){\color[rgb]{0,0,0}\makebox(0,0)[t]{\lineheight{1.25}\smash{\begin{tabular}[t]{c}$Y_2$\end{tabular}}}}%
    \put(0.86989184,0.14089913){\color[rgb]{0,0,0}\makebox(0,0)[t]{\lineheight{1.25}\smash{\begin{tabular}[t]{c}$Y_3$\end{tabular}}}}%
    \put(0.06253264,0.62029682){\color[rgb]{0,0,0}\makebox(0,0)[t]{\lineheight{1.25}\smash{\begin{tabular}[t]{c}$X_1$\end{tabular}}}}%
    \put(0.58100813,0.65993656){\color[rgb]{0,0,0}\makebox(0,0)[t]{\lineheight{1.25}\smash{\begin{tabular}[t]{c}$X_2$\end{tabular}}}}%
    \put(0.9374882,0.53206614){\color[rgb]{0,0,0}\makebox(0,0)[t]{\lineheight{1.25}\smash{\begin{tabular}[t]{c}$X_3$\end{tabular}}}}%
  \end{picture}%
\endgroup%

%% file: Sections/5PositiveResults/Appendix/9NonSelfReducible.tex
\subsection{The natural $p$-relation for simple cycles and un-ordered $2$-partitions is not self-reducible}\label{Section:Notselfreducible}
Here we prove that a natural encoding of the simple cycles of a graph is not self-reducible, unless $\Poly = \NP$. %
We take our inspiration from \cite{Khullerselfreducible}, where it is proven that a particular $p$-relation encoding $4$ coloring a planar graph is not self-reducible. We use the same definition for self-reducibility as in \cite{Khullerselfreducible}.

We encode a graph $G = (V,E)$ in such a way that the edges are ordered. A solution $X \in SC(G)$ is described by an element of $2^{|E|}$, a binary sequence of length $|E|$, where the $k$th term is $1$ if and only if the $k$th edge of $G$ is in $X$. We call this $p$-relation $R_{SC}$.

Let $\Sigma$ be some alphabet, where $\Sigma$ has some ordering, and we extend that ordering to $\Sigma^n$ for all $n$ using the lexicographic order. We let $R \subseteq \Sigma^* \times \Sigma^*$ be any $p$-relation with $|y| = p(|x|)$ for all $(x,y) \in R$ for some polynomial $p(n)$. Now we define the following problem (see also \cite[Definition 3.2]{grosse2001relating}):

\begin{computationalproblem}{LF-$R$ (Lexicographically First)}
Input: $x \in \Sigma^*$

Output: The lexicographically first element of $R(x)$, provided $R(x) \not = \emptyset$.

\end{computationalproblem}

For example, LF-$R_{SC}$ is the problem of finding the lexicographically first simple cycle under the particular encoding of $R_{SC}$. The following proposition is well known:

\begin{prop}[\cite{Khullerselfreducible}, Prefix-Search]\label{prop:prefix}
Suppose that $R$ is a self-reducible $p$-relation, and suppose that there is a polynomial time Turing machine which, given $x$, answers if $R(x) \not = \emptyset$. Then there is a polynomial time algorithm for LF-$R$.
\end{prop}

The idea in \cite{Khullerselfreducible} is to show that a certain $p$-relation $R$ is not self-reducible, as long as $\Poly \not = \NP$, by showing that checking $R(X) \not = \emptyset$ is in $P$, but that LF-$R$ is $NP$-hard. We will follow the same approach by reducing to LF-$R_{SC}$ from the following problem, which we will shortly show is $NP$-complete.

\begin{computationalproblem}{ExtendingToSimpleCycle}\label{problem:extending}

Input: An undirected graph $G$ and a set of edges $J \subseteq E(G)$.

Output: YES if $J$ can be extended to a simple cycle of $G$. NO, otherwise.

\end{computationalproblem}

\begin{figure}
    \centering
    \begin{tabular}{cc}
          \includegraphics[scale = .6]{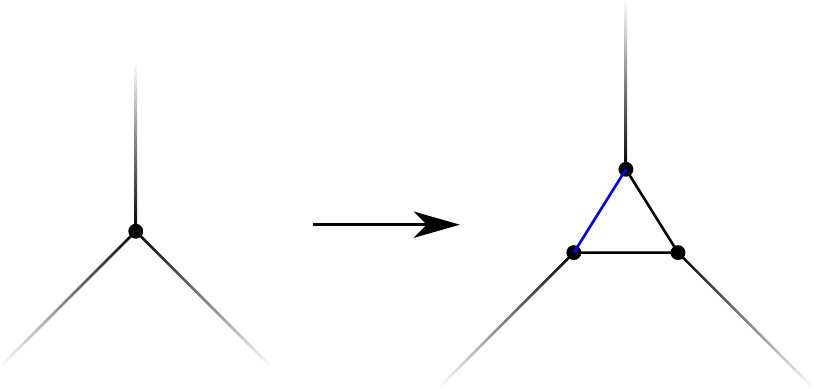}   &  
          \includegraphics[scale = .6]{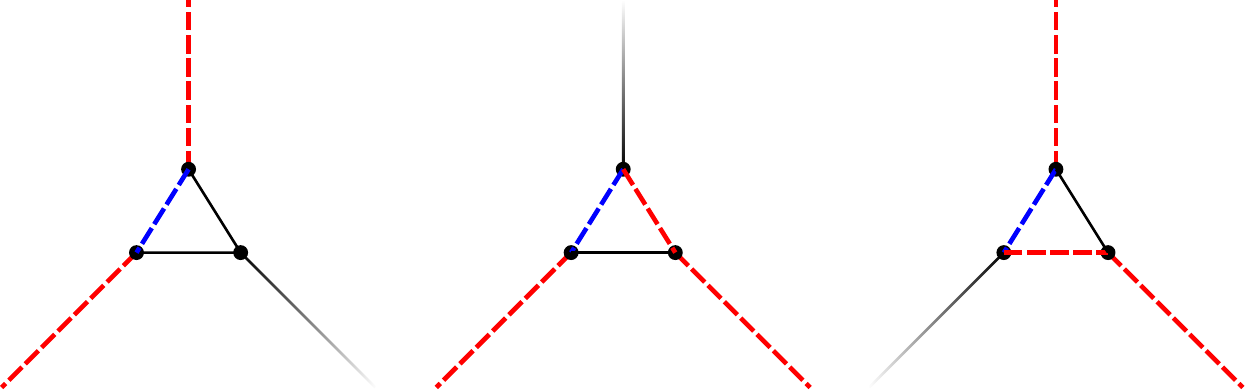}
    \end{tabular}
    \caption{a)Replacing each vertex of $G$ with a triangle. Here the blue edge is the one added to $J$. b) The routing rules illustrating that any Hamiltonian cycle in $G$ gives an extension of $J$ to a simple cycle in $G'$, and any extension of $J$ in $G'$ gives a Hamiltonian cycle of $G$.}
    \label{fig:trianglereplacementextension}
\end{figure}

\begin{prop}
\texttt{ExtendingToSimpleCycle} is $NP$-complete on the class of $3CCP$ graphs with face degree bounded by $531$.
\end{prop}\label{prop:extendingSChard}
\begin{proof}
Let $G \in \mathscr{C}_{177}$. Construct $G'$ by replacing each vertex of $G$ with a triangle as in \Cref{fig:trianglereplacementextension}a); the result remains $3CCP$ by \Cref{lem:3_connected_lemma}, and has face degree bounded by $531$. Build a set $J$ by taking one edge from each of those triangles. By examining the local routings in figure \Cref{fig:trianglereplacementextension}b), we can see that $J$ has an extension to a simple cycle of $G'$ iff $G$ has a Hamiltonian cycle. Since the latter problem is $NP$-complete by \Cref{thm:facebounded3CCP}, the proposition follows.

\end{proof}

\begin{thm}\label{thm:notselfreducible}
Fix any class of graphs $K$ which contains $\mathscr{C}_{531}$. Then the relation $R_{SC}$ is not self-reducible on $K$ assuming $\Poly \not = \NP$.
\end{thm}
\begin{proof}
Let $G, J$ be some instance of \texttt{ExtendingToSimpleCycle}. We order the edges of $G$ so that $J$ are the first edges. If there is an extension of $J$ to a simple cycle, then one of the simple cycles extending $J$ is the lexicographically first simple cycle among all simple cycles of $G$. %
If we assume that $R_{SC}$ is self-reducible, then \Cref{prop:prefix} guarantees that we can determine the lexicographically first simple cycle in polynomial time, which puts the extension problem into $\Poly$.
This contradicts $\Poly \not = \NP$ because we have proven that \texttt{ExtendingToSimpleCycle} is $\NP$-hard.
\end{proof}

We remark that, since we have proven \Cref{thm:notselfreducible} in the context of plane graphs, we obtain results about the non-self reducibility of encodings of connected $2$-partitions. One encoding that immediately reduces to the theorem just proven is to encode a connected $2$-partition as the set of cut edges. %